%% file: mythesis.tex
\documentclass[12pt,a4paper]{Thesis} 

\graphicspath{%
	{./Pictures/}%
	{./Figures/}%
}
\DeclareMathOperator{\Tr}{Tr}

\let\degree\relax

\usepackage{adjustbox}
\usepackage{xcolor}
\usepackage{amsmath}
\usepackage{dsfont}
\usepackage[amssymb]{SIunits}
\usepackage{amssymb}
\usepackage{multirow}
\usepackage{wrapfig}
\usepackage{subcaption}
\usepackage{mathtools}
\usepackage{lipsum}
\usepackage{tikz}
\newcommand{\id}{\mathds{1}}

\usepackage{glossaries}
\usepackage{amsmath}
\usepackage[ruled,vlined,noresetcount]{algorithm2e}
\usepackage[square, numbers, comma, sort&compress]{natbib} 

\title{\ttitle} 

\usepackage{float}
\usepackage{color,soul}               
\usepackage{enumerate}
\usepackage{Styles/mydefs}
\usepackage{tabularx}     
\usepackage{rotating} 
\usepackage{physics}

\newcommand{\A}{\mathcal{A}}
\newcommand{\X}{\mathcal{X}}
\newcommand{\W}{\mathcal{W}}

\newcommand{\K}{\mathcal{K}}
\newcommand{\F}{\mathcal{F}}
\DeclareMathOperator{\regret}{Rgrt}
\newcommand{\argmax}{\operatorname*{argmax}}
\DeclareMathOperator*{\argmin}{\arg\!\min}

\usepackage{upgreek}
\usepackage{amsmath,bm}
\def\gbm#1{{\let\pi\uppi \let\phi\upphi \let\lambda\uplambda \let\mu\upmu \let\rho\uprho \let\sigma\upsigma \let\tau\uptau \let\theta\uptheta \let\eta\upeta \bm{#1}}}
\newcommand{\D}{\text{D}}
\DeclareMathOperator{\polylog}{polylog}

\newcommand{\dissipation}{W_{\text{diss}}}
\newcommand{\Ex}{\mathbb{E}}
\newcommand{\hquad}{\hspace{0.5em}} 
\begin{document}

\frontmatter 

\setstretch{1.3} 

\pagestyle{fancy} 
\fancyhead{}   
\fancyhead[LO]{\sl{\leftmark}}
\fancyhead[RE]{\sl{\rightmark}}
\fancyhead[LE,RO]{\thepage}



\maketitle
\maketitleforreview


\thesisdeclareOriginality{I hereby certify that the work embodied in this thesis is the result of original research done by me except where otherwise stated in this thesis. The thesis work has not been submitted for a degree or professional qualification to any other university or institution. I declare that this thesis is written by myself and is free of plagiarism and of sufficient grammatical clarity to be examined. I confirm that the investigations were conducted in accord with the ethics policies and integrity standards of Nanyang Technological University and that the research data are presented honestly and without prejudice.}{06 Aug. 2025}{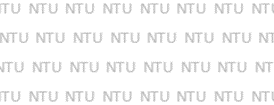}

\thesisdeclareSupervisor{I have reviewed the content and presentation style of this thesis and declare it of sufficient grammatical clarity to be examined. To the best of my knowledge, the thesis is free of plagiarism and the research and writing are those of the candidate’s except as acknowledged in the Author Attribution Statement. I confirm that the investigations were conducted in accord with the ethics policies and integrity standards of Nanyang Technological University and that the research data are presented honestly and without prejudice.}{06 Aug. 2025}{Styles/Signature_NTU.png}

\thesisdeclareAuthorship{
This thesis contains material from 2 papers published in the following peer-reviewed journals / from papers accepted at conferences in which I am listed as an author.}
{

	Chapter 3, 4 and Appendix A are published as {\color{blue}Huang, R. C., Riechers, P. M., Gu, M., Narasimhachar, V. (2023). Engines for predictive work extraction from memoryful quantum stochastic processes. Quantum, 7, 1203. DOI: 10.22331/q-2023-12-11-1203}

	The contributions of the co-authors are as follows:
	\begin{enumerate}
		\item A/Prof Mile provided the initial project direction and edited the manuscript drafts.
		\item I prepared the manuscript drafts.  The manuscript was revised by Dr Narasimhachar and Dr. Riechers.
		\item I co-designed the protocol together with Dr. Paul Riechers and Dr. Varun Narasimhachar.
        \item I coded the computer simulation for the working of the engine.
		
		\item Dr. Paul Riechers provided analytical proof for upper bounds on extraction as well as the analytical functions used.
		\item Dr. Varun Narasimhachar assisted in understanding and provided meaningful insights throughout the project.
	\end{enumerate}
    
    Chapter 3, 6 and Appendix C are published as {\color{blue}Quantum state-agnostic work extraction (almost) without dissipation. arXiv preprint arXiv:2505.09456.}

	The contributions of the co-authors are as follows:
	\begin{enumerate}
		\item A/Prof Mile and Josep provided the initial project direction.
		\item I prepared the manuscript drafts, along with Josep and Yanglin. The manuscript was revised by Assoc Prof Gu Mile and Prof Marco Tomamichel.
        \item I carried out detailed error and rate of convergence analysis together with Yanglin. 
        \item Assoc Prof Gu Mile and Prof Marco Tomamichel provided direction for tightening of bounds.
		
	\end{enumerate}

}{06 Aug. 2025}{Styles/Signature_NTU.png}


\addtotoc{Abstract} 

\abstract{\addtocontents{toc}{\vspace{0.8em}}}


This thesis investigates the fundamental question of whether temporal correlations in quantum systems can be harnessed to extract thermodynamic work. Motivated by the conceptual legacy of Maxwell’s demon, we develop an agential framework in which a classical agent—without access to quantum memory—interacts with temporally correlated quantum states to perform adaptive work extraction. Unlike traditional resource-theoretic approaches, which assume complete knowledge of the system's quantum state, the agential approach centers on belief formation, inference, and decision-making under uncertainty.

We introduce a family of \texorpdfstring{$\rho^*$}{}-ideal protocols and demonstrate that an adaptive agent can outperform non-adaptive strategies by exploiting memory effects in quantum processes. Through dynamic programming, we define the Time-Ordered Free Energy (TOFE), a new quantity that upper bounds the work extractable under causal, adaptive operations. This reveals a fundamental thermodynamic gap between idealized and memory-constrained extraction strategies, one which could potentially be quantified by a novel measure of discord, \emph{adaptive ordered discord}.

In the second part of the thesis, we explore the challenge of learning and extracting work from unknown quantum sources. Drawing inspiration from reinforcement learning, we extend multi-armed bandit algorithms to quantum thermodynamics. We show that an agent can simultaneously identify an unknown i.i.d. quantum state and extract work, with cumulative dissipation scaling only polylogarithmically in the number of observations—significantly improving over traditional tomography-based approaches.

Overall, this work lays the foundation for a predictive and decision-theoretic approach to quantum thermodynamics, opening new directions in the study of adaptive agents, temporal correlations, and learning-based work extraction in quantum systems.

\clearpage


\setstretch{1.3} 

\acknowledgements{\addtocontents{toc}{\vspace{0.8em}} 
I would like to express my deepest gratitude to my advisor, Associate Professor Gu Mile, for his unwavering support and guidance throughout the course of my doctoral research. His patience, openness to discussion, and insightful suggestions consistently challenged me to approach problems from fresh and often unexpected perspectives. His mentorship has been instrumental to the development of my work.

I am also immensely grateful to two former research fellows in our group, Dr. Paul Riechers and Dr. Varun Narasimhachar. Dr. Riechers generously shared his deep insights into quantum thermodynamics and computational mechanics, which offered me valuable new ways of thinking about fundamental problems. Dr. Narasimhachar contributed a rich understanding from the perspective of quantum resource theories; his ideas, as well as his philosophical approach to both quantum mechanics and life, have left a lasting impact on me. I am truly thankful for the time, patience, and thoughtful guidance both have extended during various stages of this journey.

I would also like to thank my collaborator, Dr. Josep Lumbreras, whose intellectual contributions and initiative in shaping our manuscripts made the collaborative process both productive and enjoyable.

Last but not least, I am grateful to my fiancée--to--be, Ms. Sim Shi Yuan, for her unwavering support and love. Her constant presence and encouragement have sustained me through the inevitable moments of solitude and doubt that accompany the pursuit of scientific research.

}


\pagestyle{empty} 
\emph{``Prediction is very difficult, especially about the future."}

\begin{flushright}
---Niels Bohr
\end{flushright}
\null\vfill 

\begin{center}
\large{To my dear family and friends}
\end{center}
\vfill\vfill\null 
\cleardoublepage 

\pagestyle{fancy} 
\fancyhead{}   
\fancyhead[LO]{\sl{\leftmark}}
\fancyhead[RE]{\sl{\rightmark}}
\fancyhead[LE,RO]{\thepage}


\pagestyle{fancy} 

\tableofcontents 

\listoffigures 

\listoftables 


%
%


%
%



\listofnomenclature{ll} 
{
\multicolumn{2}{l}{\LARGE{\textbf{Symbols}}}\\[0.618cm]
$\mathbb{Z}$                       &Integer\\
$\mathbb{R}$                       &Real Numbers\\
$\mathcal{H}$                      &Hilbert Space\\
$\norm{\cdot}$                     &the 2-norm of a vector or matrix \\
$|\cdot|$                       &absolute value\\

$\otimes$                          &tensor product\\
$\ket{\psi}$                       &ket notation of column vector\\
$\bra{\psi}$                       &bra notation of row vector\\
$\braket{\psi}{\phi}$              &inner product between bra and ket\\
$a^*$                        &complex conjugate of $a$\\
$a^T$                         &transpose of $a$\\
$a^\dagger$                    &Hermitian conjugate of $a$\\
$\id$                          &identity matrix\\
$\tr$                           & trace of a matrix\\
$O(f(x))$             &asymptotic upper bound scaling with $f(x)$\\
$\Omega(f(x))$       &asymptotic lower bound scaling with $f(x)$\\
$\Theta(f(x))$       &asymptotic growth scaling exactly as $f(x)$\\
$\Ex(W)$              &expectation value of random variable $W$ with respect to some distribution\\ 
$\rho=\sum p_i\ket{\lambda_i}\!\bra{\lambda_i}$  &Spectral decomposition of $\rho$ with eigenvalue $p_i$ and eigenvector $\ket{\lambda_i}$\\
$\ket{0},\ket{1}$ &Computational basis states\\
$\beta=(k_BT)^{-1}$ &inverse temperature, $k_B$ is the Boltzmann constant\\
$\gamma/\gamma_\beta$ &Thermal/Gibbs state at a inverse temperature $\beta$\\
$Z=\tr(e^{-\beta H})$  &Partition function of a given Hamiltonian $H$\\
$\F_{\text{eq}}$ &Equilibrium free energy, defined as $-\beta^{-1} \ln Z$\\
$\D(\rho\|\sigma)$ &Quantum relative entropy between $\rho$ and $\sigma$\\
$\Delta_{\epsilon}(\rho)=(1-\epsilon)\rho+\epsilon \mathbf{1}/2$ &Depolarization channel with strength $\epsilon$\\
$\gbm\eta$ &Agent's belief state (probability distribution over latent states $\mathcal{S}$)\\
$\xi_t$ &expected quantum state at time $t$ based on the current belief $\eta_{t-1}$\\
$\mathcal{C}_t$ &Confidence region constructed based on $t$ least square estimators\\
[0.618cm]

\multicolumn{2}{l}{\LARGE{\textbf{Acronyms}}}\\[0.618cm]
i.i.d                      & independent and identically distributed\\
$a.s.$           & almost sure convergence of a random sequence\\
HMM &Hidden Markov Model, defined by the 4-tuple $\{\mathcal{S},\{x\},\{T^{(x)}\},\gbm\pi\}$\\
CPTP &Completely Positive Trace-Preserving maps\\
CPTNI &Completely Positive Trace Non-Increasing maps\\
POVM &Positive Operator-Valued Measure\\
PVM &Projection-Valued Measure\\
RL / QMAB &Reinforcement Learning and Quantum Multi-Armed Bandit\\
}


\mainmatter       
\pagenumbering{arabic}
\setstretch{1.3}  

\fancyhead{}   
\fancyhead[LO]{\sl{\leftmark}}
\fancyhead[RE]{\sl{\rightmark}}
\fancyhead[LE,RO]{\thepage}


\input{./Chapters/Chapter1}

\input{./Chapters/Chapter2}

\input{./Chapters/Chapter3}
\input{./Chapters/Chapter4}
\input{./Chapters/Chapter5}

\input{./Chapters/Chapter6}
\input{./Chapters/Chapter7}


\addtocontents{toc}{\vspace{0.8em}} 
\label{Bibliography}
\setstretch{1}
\bibliographystyle{unsrtnat}
\bibliography{References/References}

\appendix 


\input{./Appendices/Appendix3}

\input{./Appendices/Appendix4}
\input{./Appendices/Appendix5}

\clearpage











\backmatter


\end{document}

%% file: Chapters/Chapter1.tex

\chapter{Introduction} 
\label{ch:introduction}
\vspace{-1em} 
\noindent\rule{\textwidth}{0.4pt}
\vspace{1.5em} 
\setlength{\parindent}{4ex}

\noindent \textit{This introductory chapter outlines the motivation, objectives, and structure of the thesis. The work is inspired by Maxwell’s demon—one of the first agent-based constructs in thermodynamics—and its modern reinterpretation, which underscores the thermodynamic relevance of information. We explore the role of predictive agents that interact with systems that exhibit temporal correlations and provide a concise overview of the interplay between information and thermodynamics. The central aim of this thesis is to examine the potential for extracting work from temporally correlated quantum systems through an agent-based framework.}

\newpage
An \emph{agent} is a system capable of performing feedback operations based on information acquired from its \emph{environment}, with the aim of achieving specific goals. The environment may range from stochastic noise and thermal reservoirs to quantum systems, whether static or dynamically evolving. In this thesis, we investigate whether an agent can extract work from an environment that evolves according to some underlying, possibly hidden, structure. We focus on agents equipped only with classical memory, explicitly excluding access to quantum memory. This assumption is motivated by both practical and theoretical considerations. Experimentally, preserving quantum coherence over long durations and implementing joint operations on multiple quantum systems remain significant challenges~\cite{lvovsky2009optical, heshami2016quantum}. Theoretically, quantum memory could store quantum states across time steps, mapping temporal correlations onto spatial degrees of freedom. This would enable local operations to leverage future information to manipulate past states, which would be a violation of temporal causality, though it is permissible in spatial systems, where bidirectional interactions are natural.

In this context, Maxwell’s demon can be regarded as the earliest agent-like concept in thermodynamics: it performs feedback based on information on the relative velocities of particles in a box, enabling the creation of a temperature gradient across its two halves.

\section{Maxwell's Demon}
\label{sec:demon}
The extraction of useful work from available resources has been a central pursuit of thermodynamics since the steam engine era. The advent of statistical mechanics deepened our understanding of how microscopic interactions among particles give rise to macroscopic thermodynamic phenomena such as temperature and entropy. Over time, the notion of a ``resource" has expanded—from classical fuels such as coal and oil to renewable sources such as solar energy. Of particular relevance to this thesis is the recognition, emerging from thought experiments like Maxwell’s demon, that information itself can act as a thermodynamic resource, enabling the extraction of work.

To illustrate this idea, consider a box of volume $V$, divided into two equal compartments. In equilibrium, the velocities of the particles inside are distributed according to the Maxwell-Boltzmann distribution~\cite{callen1980thermodynamics,blundell2010concepts,sethna2021statistical}.  Imagine an agent—the \emph{demon}—who possesses perfect knowledge of each particle’s position and velocity. The demon operates a massless, frictionless trapdoor in the partition and opens it selectively, allowing fast-moving particles into the left compartment and slow-moving ones into the right. Over time, this creates a temperature gradient, with high temperatures on one side and low temperatures on the other. This gradient could, in principle, be harnessed to perform work. But where does this work originate? No heat is added to the system, and no external mechanical force is applied. The emergence of usable energy seemingly violates the first law of thermodynamics.

The resolution lies in the role of information. The demon’s knowledge of the particles’ microstates constitutes a form of resource. When processed and acted upon, this information enables the extraction of free energy. This foundational idea—that information has thermodynamic value—is central to this thesis~\cite{maxwell2012theory}.

To formalize this resolution and ensure consistency with thermodynamic laws, we consider the Szilard engine~\cite{szilard1929entropieverminderung}, a simplified model that captures the essence of Maxwell’s demon. In this setup, a single particle is confined within a box. Upon observing the particle’s position, the demon inserts a frictionless partition and attaches a weight to the side containing the particle. The particle then undergoes an isothermal expansion, performing mechanical work equal to $k_BT\ln2$, where $T$ is the temperature of the environment. This work is stored as gravitational potential energy in the lifted weight. Once again, it appears as though energy is extracted without input.

The paradox is resolved through a key insight, articulated by Landauer and expanded by Bennett and Zurek~\cite{landauer1961irreversibility, bennett1982thermodynamics, zurek1989thermodynamic}. At the end of the Szilard engine’s cycle, the demon no longer retains information about the particle’s position—information that was essential for extracting work. To complete a thermodynamic cycle and return the demon to its original state, its memory must be reset. Resetting one bit of information incurs a thermodynamic cost of at least $k_BT\ln2$, thereby exactly offsetting the work extracted. This preserves the first and second laws of thermodynamics. An illustration is shown in Fig.~\ref{fig:maxwell}.
\begin{figure}
    \centering
    \includegraphics[width=0.8\linewidth]{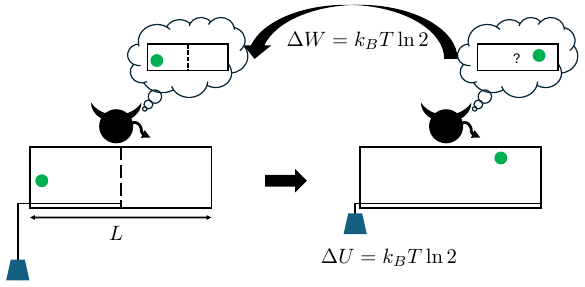}
    \caption{Illustration of Szilard’s Engine. The box has an initial volume $V=\alpha L$. After isothermal expansion of the single particle within the box, the attached weight gains energy $\Delta U=k_BT\ln2$. However, in order for the demon to acquire knowledge of the particle’s position again, a minimum energy cost of $\Delta W=k_BT\ln2$ must be expended. The net energy change is at most zero, thereby preserving the second law of thermodynamics.}
    \label{fig:maxwell}
\end{figure}
This thought experiment underscores the profound and fundamental link between information and thermodynamics. It reveals that information is not merely an abstract or mathematical construct, but a physical entity with real thermodynamic consequences.  In particular, it demonstrates that information can be harnessed to extract work from systems out of equilibrium. In this way, information itself functions as a thermodynamic resource, one that can drive physical processes when properly utilized.
Moreover, if the agent possesses incorrect information about the particle’s position, the result is not merely a failure to extract work: the system actually performs negative work. In this case, the weight attached to the partition is dropped downward during the expansion, leading to a loss of potential energy. This emphasizes that reliable information is not just beneficial but essential for thermodynamic advantage.

\section{Time-varying resources}
As the philosopher Heraclitus once said, \emph{``change is the only constant"}. The ``resources" discussed in Sec.~\ref{sec:demon} are not necessarily stationary; they may vary over time. For instance, the angle of sunlight changes, and quantum states may evolve due to decoherence or external interactions. More concretely, suppose the state of a system (e.g., the particle's location in a box) at time $t$ differs from that at time $t+1$. In such cases, real-time feedback and monitoring are typically required to adapt the work extraction protocol dynamically.
\begin{figure}[ht]
    \centering
    \includegraphics[width=0.9\linewidth]{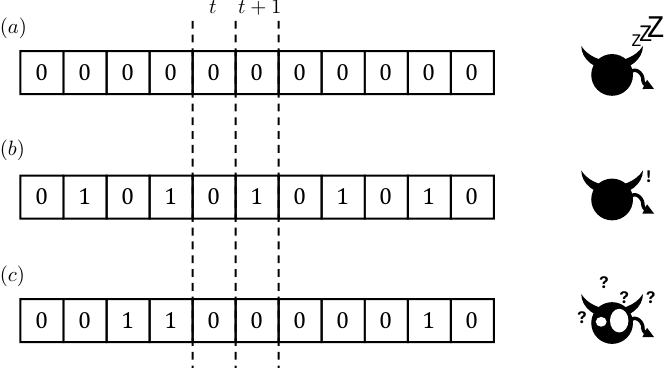}
    \caption{Depiction of temporally correlated systems. Panel (a) shows a sequence of boxes that remains invariant over time; no feedback control is required to extract work from such a sequence. Panel (b) illustrates a sequence with an alternating pattern, where the agent must retain information about the preceding box to extract work effectively. Panel (c) depicts a system with more complex temporal correlations, requiring a larger memory or inference mechanism for optimal work extraction.}
    \label{fig:agent}
\end{figure}
Consider two simple examples shown in Fig.~\ref{fig:agent}. Suppose that an agent is not interacting with a single system, but a sequence of systems (or ``boxes"), each labeled ``0" or ``1" depending on whether the particle is on the left or right. If all of the boxes are identical (e.g., each contains a particle confined to the left half as shown in Fig.~\ref{fig:agent}(a)), the same work extraction protocol can be repeated at every time step, harvesting all available free energy. In fact, an agent is not even necessary since no feedback control is needed. On the other hand,  suppose the particle positions alternate—left in the first box, right in the second, and so on as shown in Fig.~\ref{fig:agent}(b). Applying the same protocol will yield exactly 0 work extractions on average. An agent, in this case, must be present to keep track of past outcomes and tailor the protocol for the next state's configuration accordingly. 

These examples presume that the sequence of states is temporally correlated—a common feature in physical systems due to memory effects, dynamical constraints, or structured environments~\cite{breuer2016colloquium, megier2021memory, wisniewski2024memory}. Such non-i.i.d. (independent and identically distributed) behavior can be modeled using stochastic processes like Hidden Markov Models (HMMs). In these models, latent states represent hidden degrees of freedom influencing observable events—akin to environmental structure or memory. Knowing the stochastic structure allows an agent to infer future distributions and adjust its actions accordingly. 

\section{Agential approach}
This agent-based perspective contrasts sharply with conventional frameworks in quantum thermodynamics, such as resource theories or ergotropy-based methods. Traditional approaches assume complete knowledge of the quantum state and seek to maximize work extraction under specified constraints.

In contrast, the agential approach focuses on the agent’s belief or estimate of the system's state. This shift is motivated by physical realism: perfect state knowledge is rare due to preparation errors, decoherence, and measurement noise. This uncertainty becomes even more pronounced when the quantum states evolve stochastically over time.

The agential framework emphasizes protocol design based on \emph{estimated} states, necessitating an analysis of performance under uncertainty. In work extraction, this leads to deviations in both average and distribution of extractable work compared to fully known-state scenarios. These distinctions will be explored further in \cref{ch:work_extraction}.

\section{Learning and Update of Belief}
How does an agent form and update its beliefs about its environment? 
From a statistical standpoint, an agent's belief can be viewed as a prior distribution, which is iteratively updated based on the evidence gathered through interactions with its environment. In the context of work extraction from temporally correlated quantum systems, belief plays a critical role in two distinct but complementary ways.

\subsection{Belief about the future distribution}
The first concerns the agent’s belief about future outcomes. Suppose the agent has partial knowledge of the underlying generative process—for instance, the structure of a hidden Markov model (HMM). In that case, how should it update its internal memory to best predict the next quantum state? To address this, we appeal to computational mechanics, a framework designed to track the evolution of an agent’s belief conditioned on past observations. This formalism has been instrumental in formulating studies of work extraction from classical sequences in terms of information ratchets~\cite{boyd2016identifying,boyd2017leveraging,he2022information}, a topic we will review in \cref{ch:3}.

Within this framework, an agent interacts with a sequence of classical symbols generated by a stochastic process. It aims to extract work via a policy governing its decision-making, leveraging
past outcomes to influence future actions. The agent’s performance hinges on its ability to predict future states of the system, typically modeled using the $\epsilon$-machine representation of the HMM. However, these results do not translate straightforwardly to the quantum regime. In general, unlike classical symbols, quantum states cannot be measured without disturbance (in general): any observation induces irreversible state collapse, thereby altering the very statistics that guide prediction. Moreover, non-orthogonal quantum states cannot be perfectly distinguished, unless the agent has access to an infinite number of identical copies~\cite{hiai1991proper,audenaert2007discriminating}. These features violate core assumptions in computational mechanics, making prediction in non-i.i.d. quantum processes significantly more challenging and fundamentally limiting the agent's predictive capacity.

This leads us to our first central question: 
\begin{itemize}
    \item \emph{To what extent can an agent harness temporal correlations in quantum systems for sequential work extraction?}
\end{itemize}

\subsection{Belief about process structure}
The second role of belief pertains to the structure of the underlying process itself. If the agent lacks prior information about the generative dynamics, can it learn this structure over time? To approach this problem, we turn to reinforcement learning (RL), where an agent interacts with an environment, collects feedback through rewards, and iteratively learns a policy that maximizes long-term returns. This framework is particularly important in scenarios where current actions influence not only immediate outcomes but also future rewards.
At the core of reinforcement learning lies the exploration–exploitation trade-off: the agent must balance trying new actions to gather information (exploration) with leveraging existing knowledge to optimize performance (exploitation). The standard models here are the Markov Decision Process (MDP) and its generalization, the Partially Observable Markov Decision Process (POMDP), which account for fully and partially observable environments, respectively. However, both frameworks are computationally demanding due to the discretizations of state and action space. Instead, we turn our attention to the Multi-armed bandit framework used in RL to exemplify the exploration–exploitation tradeoff.

In its classical formulation, this framework describes an agent that interacts sequentially with a set of stochastic reward sources—collectively referred to as the environment—by choosing from a set of possible actions, aiming to minimize a \emph{regret} function. This function measures the cumulative loss resulting from not consistently choosing the optimal action. Only recently has the framework been extended to the quantum domain, particularly in quantum metrology~\cite{lumbreras2022multi,lumbreras24pure}, where it has been applied to derive new bounds on quantum fidelity between the estimated and true quantum states over finite time horizons.

This raises our second key question: 
\begin{itemize}
    \item \emph{Can such learning algorithms be adapted to sequential work extraction, enabling the agent to learn the unknown state or process while simultaneously extracting work?}
\end{itemize}

This thesis aims to address the overarching question of whether temporal correlations in quantum systems can be exploited for work extraction. Furthermore, we wish to investigate if an agent can utilize the framework of multi-armed bandits to learn the identity of unknown quantum states or processes while extracting work at the same time. To do so, we break it down into the following subproblems.
\begin{enumerate}
\item \textbf{Adaptivity without Quantum Memory:} Given a sequence of temporally correlated quantum states, can an adaptive agent operating without access to quantum memory extract more work than a nonadaptive strategy? Under what conditions does adaptivity provide a thermodynamic advantage?

\item \textbf{Thermodynamic Limits of Adaptive Strategies:} What is the maximum amount of work that such an adaptive agent can extract from a temporally correlated quantum sequence? Can this strategy approach or saturate the upper bounds imposed by the second law of thermodynamics?

\item \textbf{Learning from Unknown Quantum Sequences:} In the absence of any prior knowledge about the sequence of quantum states, is it possible for the agent to extract work and learn such a process simultaneously? If so, what is the operational limit?
\end{enumerate}
While work extraction is historically rooted in classical thermodynamics, it remains vital in the quantum regime as an operational metric for the thermodynamic value of information. Current hybrid quantum technologies rely on classical algorithms to measure and control quantum states. By quantifying the thermodynamic penalties incurred when relying strictly on classical memory, this research establishes a precise benchmark for evaluating future fully-quantum memories.

\section{Major Contributions}\label{sec:contribution}
To each of the questions posed, we provide affirmative and constructive responses. Our key contributions are as follows: 
\begin{enumerate}
\item We present a systematic framework for constructing autonomous, adaptive agents capable of extracting work from temporally correlated quantum states. These agents operate without requiring quantum memory or coherent control across multiple time steps. In certain regimes, they outperform agents that lack predictive capabilities, and in all other cases, they perform at least as well.

\item We establish a fundamental upper bound on the amount of work that an adaptive agent with classical memory can extract sequentially from temporally correlated quantum states. This bound, which we term the \emph{Time-Ordered Free Energy} (TOFE), upper-limits the performance of agents described in Contribution 1. We demonstrate that TOFE is generally less than or equal to the true non-equilibrium free energy, revealing a fundamental performance gap rooted in classical information processing constraints and causal ordering.

\item We show that an adaptive agent can simultaneously learn and extract work from independent and identically distributed (i.i.d.) copies of an unknown quantum state. Remarkably, the agent’s cumulative dissipation scales only polylogarithmically with the number of copies, which represents an exponential improvement over traditional approaches such as full quantum state tomography.

\end{enumerate}

\section{Outline of the Thesis}
The remainder of the thesis is structured as follows:

\cref{ch:terms_notations} introduces foundational concepts and definitions essential for this work, covering quantum information theory, quantum thermodynamics, and stochastic processes.

\cref{ch:work_extraction} defines and motivates the class of $\rho^*$-ideal work extraction protocols, which form a core component of the thesis. We provide formal definitions, discuss performance guarantees, and illustrate implementation through a concrete example.

\cref{ch:3} investigates work extraction from temporally correlated quantum states. We construct a formal framework in which an adaptive agent follows a structured algorithm to exploit temporal correlations. The agent’s performance is compared to that of both non-adaptive agents and agents restricted to thermal operations.

\cref{ch:4} expands on Chapter 4 by incorporating dynamic programming to optimize the average extracted work. This leads to the definition of the Time-Ordered Free Energy (TOFE) and its comparison to conventional free energy. We identify a fundamental gap between these quantities, hypothesized to be quantified by a new measure of quantum correlation: adaptive multipartite thermal discord. This gap reflects limits imposed by causality and classical communication.

\cref{ch:learning} proposes a novel framework that unifies work extraction and online learning of an unknown quantum state. The approach draws from quantum multi-armed bandits in reinforcement learning, enabling the agent to manage the exploration–exploitation trade-off. This framework generalizes work extraction theory to scenarios with incomplete information and can be extended to other quantum resources.

\cref{ch_conclusions} summarizes the main findings of the thesis and discusses promising directions for future research.

%% file: Chapters/Chapter2.tex

\chapter{Theoretical background} 
\chaptermark{Terms and notations}  
\label{ch:terms_notations} 
\noindent\rule{\textwidth}{0.4pt}
\vspace{1.5em} 
\setlength{\parindent}{4ex}

\noindent \textit{In this chapter, we present an overview of the foundational topics relevant to this thesis. We begin with the basics of quantum information theory and explore how certain information-theoretic quantities are intimately connected to concepts in quantum thermodynamics. We then turn to stochastic processes, focusing on how computational mechanics provides a principled framework for constructing minimal predictive models—known as $\epsilon$-machines—that capture the underlying structure of temporal correlations.}
\newpage

\section{Quantum Information}
In this section, we provide a concise overview of foundational concepts in quantum information theory, focusing on the mathematical characterization of quantum states~\cite{nielsen2010quantum,wilde2013quantum}.
\begin{definition}[Quantum State]
    A quantum state $\rho$ is a positive semi-definite operator with unit trace acting on a $d$-dimensional Hilbert space $\mathcal{H}$. It admits a spectral decomposition of the form
\begin{equation}
\rho = \sum_{i=1}^d p_i \ket{\lambda_i}\bra{\lambda_i},
\end{equation}
where $\{\ket{\lambda_i}\}_{i=1}^d$ form an orthonormal eigenbasis of $\mathcal{H}$, satisfying $\braket{\lambda_i}{\lambda_j} = \delta_{ij}$. The eigenvalues $\{p_i\}_{i=1}^d$ are non-negative and sum to one, i.e.,
\begin{equation}
p_i \geq 0,\quad \sum_{i=1}^d p_i = 1,
\end{equation}
in accordance with the positivity and unit trace constraints.
\end{definition}

A state $\rho$ is called pure if and only if it is rank-one; that is, if exactly one of the eigenvalues is equal to one and the rest are zero. In this case, $\rho$ can be written as $\rho = \ket{\psi}\!\bra{\psi}$. Otherwise, $\rho$ is referred to as a mixed state, which can be interpreted as a classical probabilistic mixture of pure states with weights ${p_i}$. A general finite $d$-dimensional \emph{pure} state can be parametrized by $2(d-1)$ real parameters, $(\theta_1\cdots\theta_{d-1},\phi_1\cdots\phi_{d-1})$ in the form of
\begin{equation}
    \ket{\varphi}= \sum_{i=0}^{d-1}\alpha_i\ket{i}~,
\end{equation}
where $\alpha_i\in\mathbb{C}$ are given by:
\begin{equation}
    \begin{split}
    \alpha_0 &= \cos\frac{\theta_1}{2}\\
    \alpha_1 &=  e^{i\phi_1}\sin\frac{\theta_1}{2}\cos\frac{\theta_2}{2}\\
    &\vdots\\
    \alpha_{d-1}&=e^{i\phi_{d-1}}\sin\frac{\theta_1}{2}\sin\frac{\theta_2}{2}\cdots\sin\frac{\theta_{d-1}}{2}~.
    \end{split}
\end{equation}
For a 2-dimensional quantum state or a ``qubit", its parametrization is given by 
\begin{equation}
    \ket{\psi} = \cos\frac{\theta}{2}\ket{0} + e^{i\phi}\sin\frac{\theta}{2}\ket{1}~.
\end{equation}
where $\theta\in[0,\pi]$ and $\phi\in[0,2\pi]$ are the angles measured from $Z$ and $X$-axis, respectively. This representation allows a qubit to be visualized as a point in a Bloch sphere shown in Fig.~\ref{fig:bloch}. 
\begin{figure}
    \centering
    \includegraphics[width=0.4\linewidth]{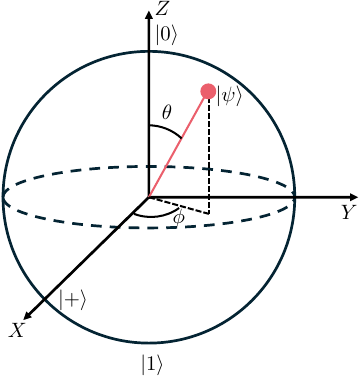}
    \caption{Diagrammatic representation of qubits on a Bloch sphere. $\theta$ is the angle from the $Z$-axis and $\phi$ is the angle measured from the $X$-axis. Pure states reside on the surface while mixed states occupy the interior of the sphere. }
    \label{fig:bloch}
\end{figure}
The basis along the $Z$-axis, i.e., $\ket{0},\ket{1}$ is commonly known as computational basis whereas the basis along $X$-axis are the $\ket{+}$ and $\ket{-}$, with a relation
\begin{equation}
    \ket{\pm} = \frac{1}{\sqrt{2}}(\ket{0}\pm\ket{1})~.
\end{equation}
The set of all pure states will occupy the surface of the Bloch sphere, whereas the set of all mixed states fills the interior of the sphere. The maximally mixed state, $\id/2$, corresponds to the center of the sphere. While this geometric representation is specific to qubits, generalizations to higher-dimensional systems, \emph{qudits}, are possible, though the visualization becomes increasingly abstract.

A commonly used set of operators in quantum information theory is the set of Pauli matrices, defined as:
\begin{equation}
    \sigma_X=\begin{pmatrix}
        0&1\\1&0
    \end{pmatrix},\quad \sigma_Y=\begin{pmatrix}
        0&-i\\
        i&0
    \end{pmatrix},\quad \sigma_Z=\begin{pmatrix}
        1&0\\0&1
    \end{pmatrix}~.
\end{equation}
These matrices play central roles in quantum state tomography, quantum error correction, quantum cryptography, and many other applications. The computational basis states, $\ket{0},\ket{1}$ are eigenvectors of $\sigma_Z$ while $\ket{+},\ket{-}$ are the eigenvectors of $\sigma_X$. Although the eigenstates of $\sigma_Y$ are less commonly used in practice, they are similarly well-defined and relevant in specific contexts.

When describing the evolution or manipulation of quantum states, it is essential to restrict attention to operations that are completely positive and trace-preserving (CPTP). The action of such operations, $\mathcal{E}$, on a quantum state $\rho$ can be expressed as 
\begin{equation}
    \mathcal{E}: \rho \mapsto \mathcal{E}(\rho)~.
\end{equation}
The condition of complete positivity requires that a quantum operation, $\mathcal{E}$, maps valid quantum states to valid quantum states even when acting as part of a larger system. Formally, $\mathcal{E} \otimes \mathcal{I}_n$ must be a positive map for all $n$, where $\mathcal{I}_n$ denotes the identity map on an auxiliary $n$-dimensional Hilbert space. This requirement is crucial for the physical consistency of quantum operations in composite systems. The condition of trace-preserving requires that $\tr[\mathcal{E}(\rho)]=1$ for all normalized input density operators $\rho$, preserving total probability.  More generally, we also consider completely positive and trace non-increasing (CPTNI) maps, which satisfy $\tr[\mathcal{E}(\rho)]\leq 1$. These maps are particularly relevant for modeling quantum processes that involve post-selection, such as quantum measurements or probabilistic quantum gates.  In this context, the trace of the output state $\tr[\mathcal{E}(\rho)]$ corresponds to the probability that a particular event (e.g., a measurement outcome) occurs.

While we are on the topic of measurement, it is important to introduce one of the most general and widely used classes of quantum measurements: the positive operator-valued measure (POVM). A POVM is specified by a set of positive semi-definite operators $\{M_i\}_i$ acting on a Hilbert space $\mathcal{H}$, where each operator $M_i$ corresponds to a possible measurement outcome labeled by $i$. These operators satisfy the completeness relation
\begin{equation}
    \sum_i M_i^\dagger M_i = \id~.
\end{equation}
POVM elements need not be orthogonal projectors, nor are they required to be of rank one. However, if each POVM element, $M_i$, is idempotent ($M_i^2=M_i$), Hermitian, and mutually orthogonal (i.e., $M_iM_j = \delta_{ij}$), then the measurement is referred to as a projective measurement, or more formally, a projection-valued measure (PVM).
From a physical perspective, any POVM can be realized as a projective measurement on a higher-dimensional Hilbert space. This insight is formalized by the Stinespring dilation theorem, which states that any POVM on a system can be implemented by coupling the system to an ancillary system, applying a unitary transformation, and then performing a projective measurement on the joint system. While this dilation viewpoint is conceptually fundamental, we do not explore it further here (see Ref.~\cite{stinespring1955positive} for further details).

The outcome probabilities for quantum measurements are computed using Born's rule. For a general POVM, $\{M_i\}_i$, the probability of obtaining outcome $i$ when measuring a quantum state $\rho$ is given by
\begin{equation}
    \Pr(O = i | \rho,\{M_j\}_j) =\tr(M_i\rho M_i^\dagger) 
\end{equation}
In the special case of PVMs, this simplifies to
\begin{equation}
    \Pr(O = i | \rho,\{M_j\}_j) =\tr(M_i\rho )~.
\end{equation}

\subsection{Multi-partite states}
\label{sec:multi-partite}
This thesis focuses on temporal correlations, which necessitate a clear specification of the subsystems between which such correlations are defined. In quantum information theory, when two well-defined quantum states reside in distinct Hilbert spaces or are ``spatially separated", their joint system residing in a composite Hilbert space can be represented by the tensor product of the local Hilbert spaces.
For instance, let $\rho_A$ be a quantum state in Hilbert space $\mathcal{H}_A$, and $\sigma_B$ be a quantum state in Hilbert space $\mathcal{H}_B$. The joint state of the composite Hilbert space is then given by
\begin{equation}
\label{eq:bipartite}
\rho_{AB} = \rho_A \otimes \sigma_B,
\end{equation}
where the subscripts denote the respective subsystems. States that can be written in the form of Eq.~\ref{eq:bipartite} are referred to as bipartite states. Whereas a classical mixture of such states in the form of 
\begin{equation}
\label{eq:separable}
\rho_{AB} = \sum_ip_i\rho^{(i)}_A \otimes \sigma^{(i)}_B,
\end{equation}
are referred to as separable states. Such states do not contain quantum entanglement between subsystems $A$ and $B$, although they may still exhibit classical correlations and other forms of non-classical forms of correlation, such as quantum discord. 

The tensor product formalism similarly extends to quantum operations. A joint or coherent operation on both subsystems can be described by a map $\mathcal{E}_{AB}$ acting on the combined Hilbert space $\mathcal{H}_A \otimes \mathcal{H}_B$. If the operations on subsystems $A$ and $B$ are independent, it can be expressed as
\begin{equation}
\mathcal{E}_{AB} = \mathcal{N}_A \otimes \mathcal{M}_B,
\end{equation}
where $\mathcal{N}_A$ and $\mathcal{M}_B$ act exclusively on $\mathcal{H}_A$ and $\mathcal{H}_B$, respectively. The subscripts here indicate the domains of action for each map.

Conversely,  consider the reverse situation in which a joint quantum state $\rho_{AB}$ is defined over two spaces, and we are interested in describing the state of just one subsystem, say $A$. To obtain this, we apply the operation known as the partial trace over the other subsystem. The reduced state of subsystem $A$ is defined as
\begin{equation}
\label{eq:partial_trace}
    \rho_A = \tr_B\rho_{AB}=\sum_i (\id_A\otimes\bra{i}_B )\rho_{AB}(\id_A\otimes\ket{i}_B)~,
\end{equation}
where $\{\ket{i}_B\}_i$ are the orthonormal basis for subsystem $B$.
Operationally, this corresponds to performing a projective measurement on subsystem $B$ in the $\{\ket{i}_B\}_i$ basis and disregarding the measurement outcomes. The result is the marginal state of the subsystem $A$.

\subsection{Information-theoretic quantities}
Next, we review several fundamental quantities commonly used in quantum information theory. 

One of the most central concepts in both classical and quantum information theory is entropy, which quantifies uncertainty or information content. In classical information theory, the uncertainty associated with a discrete random variable $X\in \{x_i\}_i^n$ with probability distribution $\Pr(X=x_i)=p_i$ is captured by Shannon entropy:
\begin{equation}
H(X) = -\sum_i p_i \log p_i.
\end{equation}
This quantity measures the average amount of information required to describe the outcome of $X$. For a pair of possibly correlated random variables $X$ and $Y$, we can also define the conditional entropy as
\begin{equation}
H(X|Y) = H(X,Y) - H(Y),
\end{equation}
where $H(X,Y)$ is the joint entropy of the pair $(X, Y)$, and $H(Y)$ is the marginal entropy of $Y$. The conditional entropy $H(X|Y)$ quantifies the residual uncertainty about $X$ given knowledge of $Y$. In the context of temporal processes, if $X$ represents the future and $Y$ the past, then $H(X|Y)$ characterizes how uncertain an agent remains about the future when the past is known.

Another key quantity is mutual information, which measures the amount of information shared between two random variables. It is defined as
\begin{equation}
\label{eq:mutual_info}
\begin{split}
I(X;Y) &= H(X) + H(Y) - H(X,Y) \\
&= H(X) - H(X|Y) = H(Y) - H(Y|X).
\end{split}
\end{equation}
Mutual information is always non-negative and equals zero if and only if $X$ and $Y$ are statistically independent. It captures the reduction in uncertainty of one variable due to the knowledge of the other. In the context of temporal correlations—where $X$ represents the future and $Y$ the past—mutual information $I(X;Y)$ quantifies the degree to which the past constrains or informs the future. Mutual information captures the minimum amount of uncertainty about the future given the past. Operationally, it captures the minimal amount of memory an agent must retain about the past in order to predict future statistics. This interpretation is especially relevant in thermodynamics, where information erasure expends energy. A helpful visualization of these relationships is shown in Fig.~\ref{fig:venn_diagram}.
\begin{figure}
    \centering
    \includegraphics[width=0.6\linewidth]{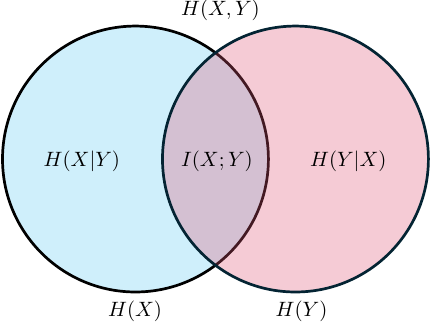}
    \caption{Venn diagram of entropic quantities. The blue and red circles represent the entropies of random variables $X$ and $Y$, respectively. Their union corresponds to the joint entropy of $H(X,Y)$, and the intersection represents the mutual information $I(X;Y)$.}
    \label{fig:venn_diagram}
\end{figure}
It is worth emphasizing that Eq.~\eqref{eq:mutual_info} admits several equivalent classical formulations using marginal and joint entropies. However, in the quantum case, these expressions are no longer equivalent, a key distinction that we explore later.

In quantum information theory, the analogue of Shannon entropy is the von Neumann entropy, defined for a quantum state $\rho$ as
\begin{equation}
    S(\rho) = -\tr(\rho\log\rho) = -\sum_i\lambda_i\log\lambda_i
\end{equation}
where $\{\lambda_i\}_i$ are the eigenvalues of $\rho$. Operationally, this corresponds to the Shannon entropy of the probability distribution obtained by performing a projective measurement in the eigenbasis of $\rho$. The classical mutual information in Eq.~\eqref{eq:mutual_info} can be extended to the quantum setting by substituting von Neumann entropies in place of Shannon entropies. Given a bipartite quantum state $\rho_{AB}$, the quantum mutual information is defined as
\begin{equation}
    I(A;B) = S(A)_{\rho}+S(B)_{\rho}-S(A,B)_{\rho}~,
\end{equation}
where $\rho_A = \tr_B\rho_{AB}$ and $\rho_B = \tr_A\rho_{AB}$ are the reduced density matrices for subsystems $A$ and $B$, respectively. This quantity captures the total correlations between $A$ and $B$. Alternatively, mutual information can also be expressed using the conditional entropy, $S(A|B)_\rho = S(\rho_{AB}) - S(\rho_B)$, leading to another seemingly equivalent formulation:
\begin{equation}
\label{eq:another_mutual}
    \mathcal{J}(A;B) = S(A)_\rho - S(A|B)_\rho~.
\end{equation}
Physically, this means the difference between the uncertainty of $A$ and the uncertainty of $A$ conditioned on knowledge of $B$, which is intuitive. However, unlike in the classical setting where the two formulations are equivalent, accessing information about $B$ in the quantum case requires performing a measurement. Given the quantum system in $B$, the outcome of the measurement may not perfectly elucidate the identity of $B$. Specifically, suppose we perform a PVM  $\{\Pi_i\}_i$ on subsystem $B$.  Upon observing outcome $i$, the post-measurement state takes the form
\begin{equation}
    \rho_{A|\Pi_i^{B}} \coloneqq \frac{\Pi_i^B\rho_{AB}\Pi_i^B}{\tr(\Pi_i^B\rho_{AB}\Pi_i^B)}~.
\end{equation}
The measurement-based conditional entropy term is
\begin{equation}
    S(A|\{\Pi_i^B\}_i) = \sum_i p_i S(\rho_{A|\Pi_i^{B}})~,
\end{equation}
where $p_i=\tr[(\id_A\otimes\Pi_i^{B})\rho_{AB}]$ is the probability of outcome $i$. Note that $S(A|\{\Pi_i^B\}_i)\geq S(A|B)_\rho$ with equality holding only when subsystem $B$ is diagonalized in the basis aligning with the measurement $\{\Pi_i\}_i$. Motivated by this distinction, one can define a measurement-based mutual information as
\begin{equation}
    \mathcal{J}(A;B)_{\{\Pi^B_i\}} \coloneqq S(\rho_A)+S(\rho_B)-[H(B)^\Pi_{\rho}+S(A|\{\Pi_i^B\}_i)]~,
\end{equation}
where $H(B)^\Pi_{\rho}$ is the Shannon entropy associated with the outcome distribution when subjecting $\rho$ to measurements $\{\Pi_i\}_i$ according to Born's rule. 
The discrepancy between this measurement-based mutual information and the standard quantum mutual information gives rise to the concept of quantum discord~\cite{ollivier2001quantum}:
\begin{equation}
\begin{split}
    \delta_{\{\Pi^B_i\}} &= I(A;B)_{\rho} - \mathcal{J}(A;B)_{\{\Pi^B_i\}}\\
    &=H(B)^\Pi_{\rho}+S(A|\{\Pi_i^B\}_i) -S(A,B)_\rho ~.
\end{split}    
\end{equation}
Accordingly, one can remove the measurement dependency by finding the optimal set $\{\Pi_i\}_i$ that minimizes the first quantity, i.e.
\begin{equation}
\label{eq:D2}
    \delta \coloneqq \min_{\{\Pi_i\}_i}[H(B)^\Pi_{\rho}+S(A|\{\Pi_i^B\}_i)] -S(A,B)_\rho ~.
\end{equation}
It is important to note here that the definition of discord we used in Eq.~\eqref{eq:D2} is not unique. Multiple definitions of quantum discord exist in the literature~\cite{henderson2001classical,brodutch2010quantum,dakic2010necessary,modi2012classical}. The formulation in Eq.~\eqref{eq:D2} corresponds to the $D_2$-variant discussed in~\cite{brodutch2010quantum}.  This particular form is chosen because of its thermodynamic interpretation:  it quantifies the work deficit when classical communication is restricted from subsystem $A$ to $B$, assuming both parties have full knowledge of the joint state $\rho_{AB}$~\cite{zurek2003quantum}.

\section{Quantum Thermodynamics}
Instead of expressing temperature directly as $T$, we typically use the inverse temperature, defined as $\beta = (k_BT)^{-1}$, where $k_B$ is the Boltzmann constant. This quantity conveniently appears in many thermodynamic expressions and simplifies certain calculations. Crucially, it also helps distinguish between time and temperature, two quantities that play central roles throughout this thesis. For a quantum system with Hamiltonian $H=\sum_i E_i\ket{E_i}\!\bra{E_i}$, the thermal equilibrium state at inverse temperature $\beta$ is given by the thermal (or Gibbs) state:
\begin{equation}
    \gamma_\beta = \sum_i p_i\ket{E_i}\!\bra{E_i}, \quad p_i = \frac{e^{-\beta E_i}}{Z}~,
\end{equation}
where $Z = \Tr(e^{-\beta H})$ is the partition function of the system. This state represents a classical statistical mixture of energy eigenstates weighted according to the Boltzmann distribution. Note that for the rest of the thesis, $\gamma$ and $\gamma_\beta$ will be used interchangeably; this is to prevent cluttering of notation that may be used to indicate the specific subsystem that the thermal state is in, and its dependence on inverse temperature is always assumed implicitly. 

As in classical thermodynamics, when the system is in equilibrium, one can define the equilibrium free energy:
\begin{equation}
\label{eq:free_energy}
    \mathcal{F}_{\text{eq}} =U-TS =  -\beta^{-1}\ln Z~.
\end{equation}
$U$ here refers to the average internal energy of the system, whereas $S$ here refers to the thermodynamic entropy of the system. $T$ refers to the ambient temperature that the system is in.
The equilibrium free energy serves as a reference point for all other quantum states governed by the same Hamiltonian and inverse temperature. Thermal states are completely passive, meaning that no work can be extracted from it, even when there are multiple copies of the state available~\cite{pusz1978passive,allahverdyan2004maximal}. 
In contrast, if the quantum system is not in equilibrium, i.e., the state $\rho\neq\gamma_\beta$, where $\gamma_\beta$ is the thermal state at inverse temperature $\beta$, then one can define the non-equilibrium free energy.

This quantifies the potential for work extraction or other thermodynamic tasks. This quantity is closely related to a key concept in quantum information theory: the quantum relative entropy. Specifically, the non-equilibrium free energy is given by
\begin{equation}
\begin{split}
    \D(\rho||\gamma_\beta) &= \tr(\rho\log\rho) -\tr(\rho\log\gamma_\beta)\\
    & = -S(\rho) - \tr(\rho \beta(\mathcal{F}_{\text{eq}}\id-H))\\
    & = \beta\tr(\rho H)- S(\rho) - \beta\mathcal{F}_{\text{eq}}~.
\end{split}
\end{equation}
In the second equality, we used the identity $\log\gamma_\beta=\beta(\mathcal{F}_{\text{eq}}\id-H)$, note that $\mathcal{F}_{\text{eq}}$ is a scalar and thus multiplies the identity operator.
This decomposition provides a physical interpretation of the relative entropy. The first term, $\tr(\rho H)$ represents the average energy of $\rho$, corresponding to the quantity $U$ in Eq.~\eqref{eq:free_energy}. The second term $S(\rho)$ represents the entropy of $\rho$. Recalling the expression in Eq.~\eqref{eq:free_energy}, this matches the definition of free energy of state $\rho$. Thus, the relative entropy quantifies the excess free energy in $\rho$ relative to the thermal state $\gamma_\beta$, it captures how far $\rho$ is from equilibrium and, correspondingly, how much thermodynamic advantage it may offer for performing useful work.

To this end, we first establish an operational definition of extracted work. Considering the state $\rho$ as a source of free energy, an agent applies an operation $\mathcal{W}$ to the state, utilizing a battery and any necessary ancilla. The agent's objective is to transform $\rho$ into the thermal state $\gamma$ while increasing the energy of the battery by an amount equal to the non-equilibrium free energy of $\rho$. We define this net increase in the battery's energy as the \emph{work} extracted. A generic schematic of this process is shown in Fig.~\ref{fig:work_circuit}. Imposing additional constraints on the allowed operations would further reduce the upper bound on the amount of work an agent can extract.

\begin{figure}[htbp]
    \centering
    \includegraphics[width=0.4\linewidth]{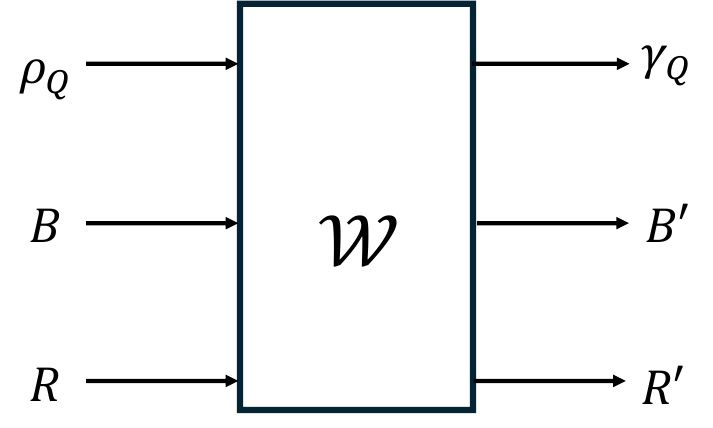}
    \caption{A circuit diagram representation of the work extraction protocol. The system $Q$ represents the system where free energy is drawn, $B$ is a battery, and $R$ represents a thermal reservoir as an ancillary system. The protocol aims to transform $\rho_Q$ to a thermal state $\gamma_Q$ with the help of the thermal states from the reservoir; the free energy lost in system $Q$ will be balanced by the increase in energy of the battery $B$.}
    \label{fig:work_circuit}
\end{figure}

Under this operational definition, the non-equilibrium free energy, quantified by the relative entropy, serves as an upper bound on the \emph{average work} extractable from the state:
\begin{equation}
    \langle W
    \rangle \leq \beta^{-1}\D(\rho\|\gamma_\beta)~.
\end{equation}
Alternatively, in the single-shot regime, work extraction is characterized by the maximum guaranteed work. This is quantified by the min-relative entropy $\D_{\min}(\rho\|\gamma)$ and its smoothed version, $\D^\epsilon _{\min}(\rho\|\gamma)$, if a failure probability $\epsilon$ is tolerated. 
While a detailed discussion of single-shot thermodynamics is beyond the scope of this thesis (interested readers may refer to Refs.~\cite {horodecki2013fundamental,aaberg2013truly,brandao2015second}), it is instructive to note the mathematical connection between these regimes. The min-relative entropy belongs to a family of divergences known as the $\alpha$-Rényi divergence, $\D_\alpha$, which is a generalization of the standard quantum relative entropy. The min-relative entropy corresponds to $\alpha=0$, whereas the standard relative entropy corresponds to $\alpha=1$. A property of $\alpha$-Rényi divergence is that $\D_\alpha\leq\D_\beta$ if $\alpha\leq\beta$. It naturally follows that the guaranteed single-shot work extracted will always be less than or equal to the average extractable work.

\subsection{Thermal operation}
\label{sec:thermal_ops}
One common approach to study quantum thermodynamics is through the resource theory framework~\cite{janzing2000thermodynamic,horodecki2013fundamental,brandao2013resource,brandao2015second}. In any resource theory, we define a set of free states—states that contain no resource—and a set of free operations—transformations that can be implemented without any cost. A core principle of resource theories is that free operations cannot generate resourceful states from free states. That is, the resource cannot be created ``for free." A canonical example is the resource theory of entanglement, where entanglement itself is the resource. In this setting, local operations and classical communication (LOCC) are considered free. It is well established that LOCC alone cannot generate entanglement from separable states. 
In quantum thermodynamics, the free states are thermal (Gibbs) states, and the class of free operations is called thermal operations. These operations take the general form:
\begin{equation}
    \mathcal{E}(\rho_S) = \tr_E U_{SE}(\rho_S\otimes\gamma_E)U_{SE}^\dagger~,
\end{equation}
where $\rho_S$ is the initial non-equilibrium state and $\gamma_E$ is a thermal state in the environment.  The operation is implemented via a global unitary $U_{SE}$ acting jointly on the system and environment, and is subject to strict energy conservation, meaning: $[U_{SE},\mathcal{H}_S+\mathcal{H}_E]=0$. Note that the notation of $\mathcal{H}_S+\mathcal{H}_E$ is not a literal sum but rather 
\begin{equation}
\label{eq:no_int_H}
    \mathcal{H}_A+\mathcal{H}_B = \mathcal{H}_A\otimes\id_B + \id_A\otimes \mathcal{H}_B~.
\end{equation}
This is a commonly used notation for non-interacting Hamiltonians. For brevity, we will use the shorthand $\mathcal{H}_S+\mathcal{H}_E$ throughout the remainder of the thesis, with the understanding that it refers to the full tensor-product form given in Eq.~\ref{eq:no_int_H}. 
The commutation condition ensures energy conservation, implying that no external work is required to implement the operation. However, thermal operations have a significant limitation: they cannot extract work from the quantum coherence present in the energy eigenbasis of the system. This restriction arises because the energy conservation constraint prohibits $U_{SE}$ from generating or exploiting coherence between different energy levels. As a result, the upper bound on the average amount of work extractable from a single copy of $\rho_S$ under thermal operations is given not by its non-equilibrium free energy, $\D(\rho\|\gamma)$, but by the relative entropy between the dephased version of $\rho$ and the corresponding thermal state: 
\begin{equation}
    \langle W\rangle \leq \beta^{-1}\D(\Delta(\rho)\|\gamma)~,
\end{equation}
where $\Delta(\rho)=\sum_i\bra{E_i}\rho\ket{E_i}\ket{E_i}\!\bra{E_i}$ is the dephased version of $\rho$ in the energy eigenbasis of the system. This phenomenon is known as ``work-locking"~\cite{lostaglio2015quantum,korzekwa2016extraction}.

For this thesis,  we turn our attention to a different class of protocols known as the $\rho^*$-ideal work extraction protocol. These protocols are subject to less stringent constraints compared to the strict energy conservation required by thermal operations. A detailed discussion of this framework is presented in~\cref{ch:work_extraction}.

\subsection{Collective vs single-copy}
\label{sec:collective}
An important notion for this thesis is the difference between collective processing and single-copy (local) processing. Given $N$ copies of a quantum state $\rho^{\otimes N}$, distributed across subsystems, $A_1,\cdots,A_N$, collective processing refers to applying a joint operation $\mathcal{E}_{Q_1,\cdots,Q_N}$ to the entire composite system. In contrast, local processing restricts the agent to operations acting on individual subsystems, of the form $\{\id_{A_1}\cdots\otimes\mathcal{E}_{A_i}\otimes \cdots\id_{A_N}\}_i$ where each $\mathcal{E}_{A_i}$ acts only on the individual copy $A_i$. A diagrammatic representation of the comparison is presented in Fig.~\ref{fig:collectvssingle}.
\begin{figure}[b]
    \centering
    \includegraphics[width=0.9\linewidth]{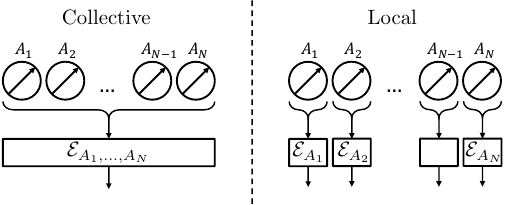}
    \caption{Distinction between collective processing vs single-copy(local) processing. Left panel: collective processing, where subsystems $A_1,\ldots,A_N$ are jointly acted upon by a global operation $\mathcal{E}_{A_1,\ldots,A_N}$. Right panel: local processing, where only individual subsystems are operated on separately, with operations restricted to one subsystem at a time.} 
    \label{fig:collectvssingle}
\end{figure}
This distinction plays a critical role in various domains of quantum information, including state discrimination and, more importantly for this thesis, quantum work extraction. Under collective operations—even restricted to thermal operations—it has been shown that in the asymptotic limit of using a large number of quantum states, $\rho^{\otimes N}$, an agent using a collective operation can extract an average amount of work equal to the non-equilibrium free energy of the state $\rho$ in the limit of $N\to\infty$~\cite{horodecki2013fundamental,brandao2015second,faist2015minimal,gour2015resource}. Formally,
\begin{equation}
   \lim_{N\to\infty} \frac{1}{N}\D(\Delta(\rho^{(N)})\|\gamma^{\otimes N}) = \D(\rho\|\gamma)~.
\end{equation}
Physically, this advantage arises because the collective Hamiltonian governing the tensor-product state $\rho^{\otimes N}$ exhibits a highly degenerate energy spectrum. By applying global coherent operations that commute with this collective Hamiltonian, an agent can couple degenerate energy levels that remain strictly inaccessible when operating on a single copy of $\rho$. This global access allows the agent to extract work from coherence that would otherwise be locked.

This asymptotic enhancement is further illuminated by considering single-shot work extraction with a permitted failure probability $\epsilon$. In the asymptotic regime, where global coherent operations are applied across infinitely many copies, the smoothed min-relative entropy $\D^\epsilon_{\min}(\rho\|\gamma)$ converges to the standard relative entropy $\D(\rho\|\gamma)$ when an agent is allowed to operate coherently across all time steps in the asymptotic regime. This foundational result is governed by the quantum asymptotic equipartition property (QAEP)~\cite{tomamichel2009fully}.

Remarkably, the performance gap between collective and single-copy processing persists even when the agent lacks full information about the quantum state. It has been shown that in the limit $N\to\infty$, collective strategies outperform single-copy ones—even when the identity of the quantum state from which work is extracted is unknown~\cite{watanabe2024black,watanabe2025universal}. A similar phenomenon is observed in the study of ergotropy~\cite{allahverdyan2004maximal,perarnau2015extractable}, where the distinction between passive and completely passive states reflects the same trend: collective processing generally yields greater work extraction than single-copy approaches~\cite{pusz1978passive,alicki2013entanglement}. This advantage largely stems from the ability to exploit quantum correlations, such as entanglement and discord, which are inaccessible to local operations. That said, results from quantum state discrimination suggest that adaptive, sequential strategies—in which operations are updated based on prior outcomes—can, in some cases, match the performance of collective measurements~\cite{martinez2021quantum}.  

The distinction between collective and local processing becomes especially significant when considering the sequential extraction of work from temporally correlated quantum states. Suppose the agent is given a joint state $\rho_{A_{1:T}}$, defined over a sequence of time steps $1,\cdots, T$. Due to physical constraints, each marginal state $\rho_{A_t}$ is only available for a limited duration, i.e., expired states from $t-1$ cannot be accessed at $t$ and for causality to be obeyed, agent at $t$ has no access to information about future states at $t+1$. Its decision must be made based only on past and present information. Under these constraints, the agent must act sequentially: at each time step $t$, it applies a local operation to the reduced state $\rho_{A_t} = \tr_{A_{1:T\backslash t}}(\rho_{A_{1:T}})$, based only on the history up to time $t$. Since the agent knows the process generating these correlated states, it can deduce the exact ensemble of states across all time steps. This allows the agent to implement global operations, which offer the theoretical advantage of harnessing global correlations, their practical implementation is technically demanding. Maintaining coherent quantum memory over extended timescales and performing coherent operations on many qubits are significant experimental challenges~\cite {lvovsky2009optical, heshami2016quantum}, making such strategies costly in real-world scenarios. For this reason, sequential protocols, which operate under realistic constraints, are the focus of this thesis—particularly with respect to their capacity to extract work from temporally structured quantum states. The sequential nature also opens up the possibility for the agent to adapt the future operations based on the work that is extracted in the past time step. To this end, the next section will aim to provide an overview of stochastic processes, in particular the hidden Markov Model (HMM), which will be used to generate and characterize these temporal correlations. This would also introduce the mathematical objects that a sequential agent requires to perform said adaptation. 

\section{Stochastic Processes}
\label{sec:stochastic_process}
As previously discussed, the central goal of this work is to investigate where an agent can extract work from temporal correlations in physical systems. To model such correlations rigorously, we employ the framework of \emph{stochastic processes}, with a particular focus on \emph{discrete-time} stochastic processes.
These are formally represented as sequences of random variables $\{X_t\}_{t\in T}$, where each variable is governed by a conditional probability distribution of the form $\Pr(X_t|X_{t-1}\cdots X_{-\infty})$.  In general, the statistical dependence of future variables on past observations can extend arbitrarily far into the history of the process. This leads to an important property of a stochastic process: its \emph{Markov order}. This quantity quantifies the length of past observation that influences future statistics~\cite{cover1999elements}. Specifically, a process is said to have Markov order $k$ if the conditional distribution of future events depends only on the most recent 
$k$ outcomes, rather than the entire past. Formally, this is expressed as:
\begin{equation}
    \Pr (X_t|X_{t-1}X_{t-2}\cdots X_{-\infty})=\Pr(X_t|X_{t-1}X_{t-2}\cdots X_{t-k})~.
\end{equation}
A process with Markov order $0$ is referred to as memoryless, meaning each random variable in the sequence is statistically independent of all others. Such processes are also known as \emph{independent and identically distributed} (i.i.d.). A process with Markov order 1 is called \emph{Markovian}, indicating that the future state depends solely on the immediate past state.
More generally, processes with Markov order greater than $1$ are said to exhibit \emph{non-Markovian behavior}, as their future statistics depend on longer histories rather than just the most recent outcome. This hierarchy of memory dependence plays a critical role in understanding the informational and thermodynamic structure of temporal processes. 

In this thesis, we focus specifically on a class of models known as \emph{Hidden Markov Model} (HMM). HMMs are particularly valuable because they offer a compact, state-based description of processes whose statistical dependencies may otherwise be difficult to capture explicitly~\cite{rabiner1989tutorial}. Unlike models that rely solely on observable sequences, HMMs introduce internal \emph{latent} states, enabling a more efficient representation of complex temporal dependencies. For instance, the well-known \emph{Even Process}~\cite{shalizi2001computational}, despite having infinite Markov order, can be succinctly represented using a finite-state HMM, as illustrated in the first row of Table~\ref{tab:table_msp}.

Formally, an HMM, or more specifically an edge-emitting HMM, is defined as the following~\cite{riechers2018spectral,jurgens2021shannon}.
\begin{definition}[Hidden Markov Model]
    An edge-emitting hidden Markov model (HMM) $\mathcal{M}$, is defined by a 4-tuple
    \begin{equation}
     \mathcal{M}=\{\mathcal{S},\{x\}_{x\in\X},\{T^{(x)}\}_{x\in\mathcal{X}},\gbm \pi\}
    \end{equation}
    where each component encodes a different aspect of the underlying stochastic process. The set $\mathcal{S}$ represents the finite set of latent (hidden) states, which are not directly observable. The set $\mathcal{X}$ denotes the observable alphabet, comprising the symbols that may be emitted during state transitions. For each symbol $x\in\mathcal{X}$, the matrix $T^{(x)}$ defines a labeled-sub-stochastic matrix, encoding the probabilities of transitioning between hidden states while emitting the symbol $x$. These matrices satisfy the condition 
\begin{equation}
\label{eq:stochastic_condition}
    T = \sum_{x\in\mathcal{X}}T^{(x)}~,
\end{equation} 
where $T$ is the overall state-to-state transition matrix, which must be row-stochastic, i.e., each row sums to 1. Finally, $\gbm\pi$ denotes the stationary distribution over the hidden states, representing the long-term probabilities of occupying each latent state. This distribution is uniquely determined by the transition matrix, $T$, satisfying $\gbm\pi T=\gbm \pi$. 
\end{definition}
Note that there is also another class of HMM known as state-emitting HMM, which assumes that the transition between latent state and emission is independent of each other. However, the edge-emitting HMM tend to be minimal in terms of memory requirement~\cite{hopcroft2001introduction,amato2010finite}.

However, the representation of a stochastic process via an HMM is generally \emph{not unique}. It is possible to construct alternative HMMs by introducing redundant latent states that do not alter the observable statistics of the output process. Such representations may obscure the underlying structure and inflate the model's complexity without providing additional predictive power. To address this, we introduce the $\epsilon$-machine representation—an HMM that is minimal, unifilar, and maximally predictive~\cite{crutchfield1989inferring,shalizi2001computational}.
To fully understand what this means, we first introduce the concept of causal state, which concerns the memory requirement of the hidden Markov Model
\begin{definition}[Causal State] 
Causal state is a class of past histories that are \emph{causally equivalent} to each other.
Let $\overleftarrow{x} = x_{-\infty},\cdots x_{t-2},x_{t-1}$ be a string of past symbols and $\overleftarrow{x}'$ be a different string. They are considered causally equivalent if they both lead to identical distributions of future observations. Mathematically, we say that the 2 past strings, $\overleftarrow{x}$ and $\overleftarrow{x}'$, are causally equivalent if and only if
    \begin{equation}
    \overleftarrow{x} \sim \overleftarrow{x}' \iff\Pr(\overrightarrow{X}|\overleftarrow{x}) = \Pr(\overrightarrow{X}|\overleftarrow{x}')~,
\end{equation}
where ``$\sim$" symbols represent causal equivalence. If strings $A$ and $B$ are causally equivalent to each other, they belong to the same causal state.
\end{definition}
Next, we need to introduce the concept of unifilarity, which concerns the transitions within the HMM.
\begin{definition}[Unifilarity]
    A process is said to be \emph{unifilar} if the next state is uniquely determined by the current state and observed symbol:
\begin{equation}
\label{eq:unifilar}
    H(S_{t+1}|X_{t+1},S_{t}) = 0~.
\end{equation}
A corollary of a process being minimal and unifilar is that it guarantees the observation-induced synchronization,
\begin{equation}
\label{eq:asymp_synch}
    \lim_{N\to\infty} H(S_N|X_{0:N}) = 0~,
\end{equation}
where we used the shorthand $X_{0:N}$ to represent the string $X_0,X_1\cdots,X_N$~\cite{crutchfield2009time,crutchfield2010synchronization}. This property is particularly important for prediction as it allows an agent to eventually synchronize its memory with the latent state in order to make meaningful predictions. 
\end{definition}

\begin{definition}[$\epsilon$-Machine]
    An $\epsilon$-machine is the \emph{minimal}, \emph{unifilar}, edge-emitting HMM that perfectly generates a given stationary stochastic process.
\end{definition}

The term \emph{minimal} here refers to an $\epsilon$-machine having the lowest statistical complexity $C_\mu$ among all unifilar representations of a process, while still preserving all statistically relevant structure necessary for optimal prediction.
\begin{equation}
    C_\mu \coloneqq -\sum_{s\in\mathcal{S}}\Pr(s)\log_2\Pr(s)~,
\end{equation}
where the distribution $\Pr(s)$ is taken over the stationary distribution $\gbm\pi$. Statistical complexity quantifies the amount of past information a system must store in order to reliably predict the future behavior of a stochastic process. Beyond its role in prediction, statistical complexity carries important thermodynamic implications: it sets a lower bound on the energetic cost required by any physical system that implements or simulates the process. This connection arises from the fundamental link between information storage and dissipation, as highlighted by Landauer's principle discussed in \cref{ch:introduction}. The $\epsilon$-machine achieves this precisely by keeping track of \emph{causal states} rather than keeping track of all past outputs. 
As a result of this, any unifilar minimal HMM will have the same structure as its $\epsilon$-machine, and its latent states correspond to the causal states.
The $\epsilon$-machine representation of an HMM can be obtained by minimizing the recurrent component of any HMM's mixed state presentation (MSP). We will discuss this in more detail in Section~\ref{sec:belief}~\cite{ellison2009prediction}. 

For this thesis, we investigate the quantum states generated by a semi-quantum stochastic process. These are processes where the underlying dynamics follow a classical HMM, but instead of emitting classical symbols, the model emits non-orthogonal quantum states, as illustrated in Fig.~\ref{fig:PC_diagram}. While such processes are unifilar from a modeling standpoint, measurements on the quantum outputs introduce measurement-induced non-unifilarity, since non-orthogonal quantum states cannot be perfectly distinguished~\cite{venegas2020measurement,venegas2023optimality}.  This undermines the synchronization property described in Eq.~\ref{eq:asymp_synch}, posing challenges for agents attempting to align their internal state estimates with the process. We address this challenge in Section~\ref{sec:synchronization}, where we present a Bayesian framework for updating the agent's belief state over time, enabling synchronization even in the presence of quantum uncertainty.

%% file: Chapters/Chapter3.tex

\chapter{Work extraction protocol} 
\chaptermark{Work extraction}  
\label{ch:work_extraction} 
\noindent\rule{\textwidth}{0.4pt}
\vspace{1.5em} 
\setlength{\parindent}{4ex}

\noindent \textit{In this chapter, we introduce and discuss a class of protocols, the $\rho^*$-\emph{ideal work extraction protocol}, which will be frequently referenced and employed throughout the thesis. We begin by outlining the motivation behind their development and examining their performance characteristics. Particular attention is given to the implications of relaxing strict energy conservation in this framework, and we discuss the consequences this may cause, as well as how this may potentially be resolved through careful energy accounting. Furthermore, we present a concrete realization of such a protocol, accompanied by a detailed analysis of its corresponding work distribution.}
\newpage

\section{Motivation}
\label{sec:motivation}
Work extraction with zero entropy production is of particular importance, as it represents the maximum amount of work that can be extracted from a quantum state. This serves as an ideal benchmark against which all other protocols can be evaluated. Achieving this limit requires that the thermodynamic process be precisely tailored to the input state~\cite{riechers2021impossibility}.

This tailoring is often implicitly assumed in many work extraction protocols, although it is rarely emphasized. This is primarily because much of the existing literature focuses on scenarios where the quantum state is already known, whether within the framework of thermal operations or in the study of ergotropy. While recent works have addressed state-agnostic work extraction---where the state is unknown---the methods primarily involve tomography or coarse-graining of the unknown state prior to extraction~\cite{vsafranek2023work,watanabe2024black,watanabe2025universal}. However, the thermodynamic penalties associated with these methods are seldom discussed.

In contrast, the agential approach highlights a critical distinction: the agent's belief about the state may differ from the actual state, due to incomplete knowledge or the stochastic nature of the underlying process. In such situations, it becomes essential to study protocols that bridge the gap between belief and reality, especially when extracting work under uncertainty.

Before delving into the formal definitions, it is important to introduce the concept of entropy production, which plays a key role in this context. Entropy production is defined as the difference between the heat supplied to the system and the actual change in the system's entropy. In other words, it is given by: 
\begin{equation}
    \langle\Sigma\rangle=\frac{\langle \tilde W\rangle-\Delta\mathcal{F}}{T}~.
\end{equation}
where $\tilde W$ represents the work exerted by the system, and $\Delta\F$ is the change in the system's free energy during the process.

It is important to note that entropy production is distinct from entropy itself. Entropy production measures the amount of entropy generated due to dissipation or irreversibility during a process. It serves as a measure of how irreversible the process is. For a reversible process, entropy production is zero, indicating no dissipation. For example, one can imagine a pendulum swinging back and forth in a vacuum, where its motion appears identical whether viewed forwards or backwards in time, as there is no dissipation. In contrast, for an irreversible process, where $\Sigma>0$, consider the story of Humpty Dumpty, who fell from the wall and could not be restored. The noise and heat generated during the collision with the ground represent dissipation.

If an agent wishes to design a work extraction protocol to extract \emph{all} non-equilibrium free energy from a quantum state without violating the laws of thermodynamics, the change in free energy must equal the work extracted (the negative of the work exerted)~\cite{parrondo2015thermodynamics}. This is only achievable when the entropy production is zero, which, in turn, implies that the operation across all systems involved must be unitary.

\section{General Set-up}
We consider a general setup for work extraction, where a system $Q$ supplies quantum states and acts as the source of free energy, a battery $B$ stores the extracted free energy, and a thermal reservoir $R$ is maintained at an inverse temperature $\beta$. We follow the same approach as in thermal operations, where the agent can acquire an unbounded number of thermal states from the reservoir. The agent performs the extraction by transforming the state in $Q$ into a thermal state with the assistance of any number of thermal states, $\gamma_\beta$, taken from the reservoir. During this process, the free energy lost from the state in $Q$ is transferred to the battery $B$. A circuit representation of this protocol is shown in Fig.~\ref{fig:work_circuit}.

In the following,we adopt several idealized assumptions. First, we treat the battery as an ideal system with a bi-infinite energy spectrum that spans all real values well above the ground state. In a more realistic scenario, however, a physical battery possesses a bounded Hamiltonian. When the battery's energy is low, corrections to the free energy must be considered~\cite{lipka2021second}, generally resulting in lower work extraction. Furthermore, a battery with finite capacity causes the extraction process to eventually terminate.

In the case of a battery with discrete energy levels, any energy not matching the specific gaps between levels cannot be extracted and is dissipated as heat, thereby reducing the total extractable work~\cite{aaberg2013truly,horodecki2013fundamental}. Such discrete systems may also necessitate a coherence investment~\cite{aaberg2014catalytic}.

Beyond the ideal battery, we also assume the thermal reservoir is infinitely large and memoryless. If the reservoir is finite or possesses memory, it may retain information from previous operations. This leads to work extraction rates below the Landauer bound, even in state-aware cases~\cite{reeb2014improved}. Additionally, while a finite bath allows for temporary violations of the Landauer bound during the extraction process, the bound remains valid for the cumulative work~\cite{pezzutto2016implications}. Such fluctuations disturb the transition rules between an agent’s belief states, as these rules depend on the actual work values rather than just the expectation value. To mitigate this, the agent’s update rules can be pre-programmed to account for these stochastic variations in work.

\section{\texorpdfstring{$\rho^*$}{}-work extraction protocol}
\label{sec:rho_ideal}
We now formally introduce the class of operations known as $\rho^*$-ideal work extraction protocols. Based on the analysis of entropy production in~\cite{riechers2021initial}, any non-unitary transformation designed to transform $\rho_0$ into $\rho_\tau$, where $\tau$ is the time interval over which the process occurs, when applied to an alternate state $\sigma_0$ will incur additional dissipation, quantified by:
\begin{equation}
\label{eq:entropy_production}
    \Sigma_{\rho_0} - \Sigma_{\sigma_0} = k_B \D(\rho_0\|\sigma_0)-k_B\D(\rho_\tau\|\sigma_\tau)~.
\end{equation}

Using this result, and the fact that zero entropy production is required for the agent to extract all non-equilibrium free energy, we can express the entropy production for a protocol tailored to $\rho_0$ as:
\begin{equation}
\label{eq:3.3}
    \Sigma_{\sigma_0} = k_B\D(\gamma_\beta\|\sigma_\tau)-k_B\D(\rho_0\|\sigma_0)~,
\end{equation}
where we enforce that $\Sigma_{\rho_0}=0$ since the protocol is tailored for $\rho_0$.
We now turn our attention to a class of protocols tailored to $\rho^*$,  which transforms it into a thermal state $\gamma$, extracting all its non-equilibrium free energy into a battery system. This will define the ``$\rho^*$-ideal work extraction protocol".
\begin{definition}[$\rho^*$-ideal work extraction protocol]
    A $\rho^*$-ideal work extraction protocol, $\mathcal{W}_{\rho^*}$ must satisfy the following two conditions:
    \begin{enumerate}
    \item When the initial state of the system is $\rho^*$, the protocol achieves zero entropy production, transferring on average all non-equilibrium addition to free energy $\beta^{-1} \text{D}( \rho^* \| \gamma)$ to a battery, $B$;
    \item It conserves energy globally when acting on any eigenstate of $\rho^*$.
\end{enumerate} 
\end{definition}
We now justify the reasoning behind these criteria. The first condition is straightforward, as discussed in Sec.~\ref{sec:motivation}. The second condition, requiring average energy conservation when acting on eigenstates of $\rho^*$, is a looser constraint compared to thermal operations. Recall from Sec.~\ref{sec:thermal_ops} that the restriction on allowable operations in thermal operations is strict energy conservation, i.e., the operations must commute with the combined Hamiltonian of the joint systems. This leads to an inability to extract work from coherence with respect to the energy eigenbasis. As a result, the agent cannot extract all available free energy. This constraint can be overlooked when one considers a system with a degenerate Hamiltonian, since the agent is then allowed to rotate the quantum state in any arbitrary way without violating energy conservation. It was shown in~\cite{skrzypczyk2014work} that the looser constraint imposed allows the agent to extract all non-equilibrium free energy. However, this comes with the trade-off that this condition can be violated if the input state is not diagonal in the eigenbasis of $\rho^*$. To mitigate this, we restrict our examination to systems $Q$ with a fully degenerate Hamiltonian, which ensures that this set of operations overlaps with thermal operations. The consequences of this condition, as well as cases involving non-degenerate Hamiltonians, will be discussed in Sec.~\ref{sec:consequence} and \ref{sec:non-degen}.

Using the definition of the $\rho^*$-ideal extraction protocol, and incorporating Eq.~\eqref{eq:3.3}, we can derive an upper bound on the average work extracted when such a protocol, tailored to $\rho^*$, is applied to a state $\sigma\neq \rho^*$:
\begin{equation}
\label{eq:work_for_rho_ideal}
    \langle W\rangle\leq \beta^{-1}\left[\D(\sigma\|\gamma)-\D(\sigma\|\rho^*)+\D(\gamma_\beta\|\sigma_\tau)\right]~.
\end{equation}
The first term, $\D(\sigma\|\gamma)$, represents the non-equilibrium free energy of the state $\sigma$. The second term, $\D(\sigma\|\rho^*)-\D(\gamma_\beta\|\sigma_\tau)$, arises from additional entropy production due to the agent incorrectly tailoring the protocol to $\rho^*$ instead of $\sigma$. This is referred to as dissipation due to \emph{misaligned expectations}.
Importantly, when $\rho^*=\sigma$, we recover the well-known result that the extractable work is upper-bounded by the non-equilibrium free energy: $\beta^{-1}\D(\sigma\|\gamma)$. This ideal scenario corresponds to the case where the agent has perfect knowledge of the input state.

\section{Explicit protocol}
\label{sec:actual_protocol}
We now discuss an explicit protocol which is instructive to how a $\rho^*$-ideal work extraction protocol can be carried out~\cite{skrzypczyk2014work}. For the protocol, we consider the following subsystems.
\begin{enumerate}
    \item System $Q$ where the temporally correlated quantum states occupy. Its Hamiltonian is $\mathcal{H}_Q = E_0\ket{E_0}\!\bra{E_0}+E_1\ket{E_1}\!\bra{E_1}$.
    \item Thermal Reservoir $R$ at specific inverse temperature $\beta$. It has a tunable Hamiltonian $\mathcal{H}_R(\nu)=e_0\ket{0}\!\bra{0}+(e_0+\nu)\ket{1}\!\bra{1}$, where $\nu$ represents the tunable energy gap between the 2 energy levels in order to reach the Maxwell-Boltzmann distribution needed for the protocol.
    \item Battery $B$ modeled after classical weight shifting up and down to gain or lose potential energy. It can be operated on by raising operation $\Gamma_y$ such that $\Gamma_y\ket{x}=\ket{x+y}$. Its Hamiltonian is defined as $\mathcal{H}_B\ket{x}\!\bra{x}_B=x\ket{x}\!\bra{x}$.
\end{enumerate}
Note that we will drop all time indices in this section since we are just dealing with extraction in a particular time step. 

Before the operation begins, the agent first chooses a state to tailor the protocol to, $\rho^*=\sum_i\lambda_i\ket{\lambda_i}\!\bra{\lambda_i}$, where $\lambda_i\geq\lambda_{i+1}$ are the ordered eigenvalues. The protocol then proceeds in 2 stages. Stage 1 involves the application of unitary rotation of the quantum state, and stage 2 involves a sequence of $M$ swap operations between the system and a specially tuned reservoir qubit. 
The first stage of the protocol proceeds by implementing a unitary rotation on the system qubit in an attempt to diagonalize it in the energy eigenbasis of the system Hamiltonian $\{\ket{E_i}\}_i$. The raising operator will balance the energy change by lifting/dropping the energy level of the battery. Mathematically, the operation is written as 
\begin{equation}
    U^{(\rho^*)}_{QB} = \sum_i\ket{E_i}\bra{\lambda_i}_Q\otimes\Gamma^{(B)}_{\epsilon_i}~,
\end{equation}
where we use the superscript in parenthesis $U^{(\rho^*)}$ to indicate that the operation is tailored for state $\rho^*$. Here, $\epsilon_i=\bra{\lambda_i}\mathcal{H}_Q\ket{\lambda_i}-E_i$ accounts for the difference in energy in the two states, and $\Gamma$ is an operator to raise the potential energy of the weight. It is important to point out that this is the stage in the whole protocol that violates the strict energy conservation imposed by thermal operations. More discussion on this operation can be found in the latter part of this chapter.
\begin{figure}
    \centering
    \includegraphics[width=1\linewidth]{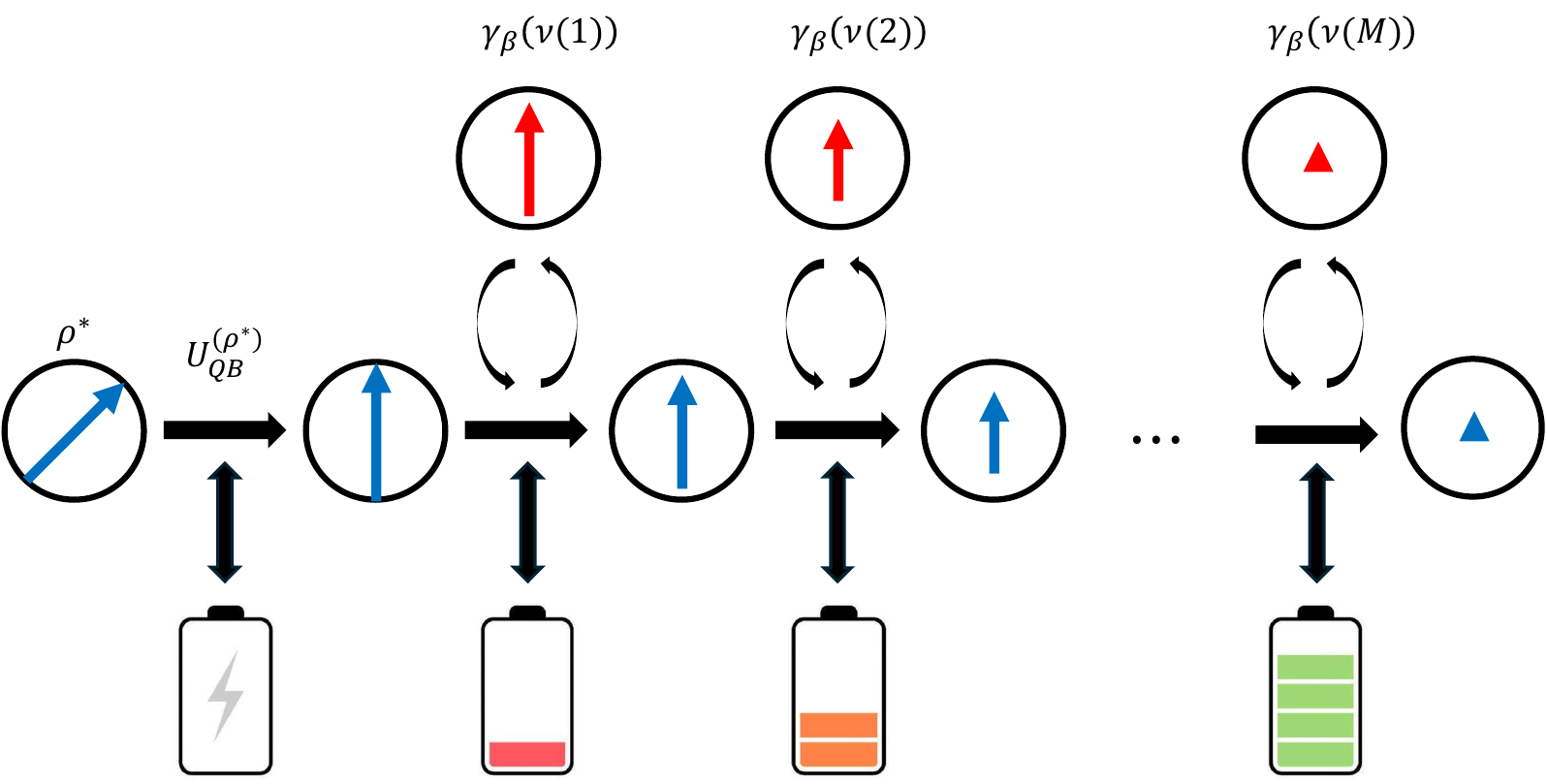}
    \caption{An illustration of how the protocol works. At stage 1, $U_{QB}^{(\rho^*)}$ is applied to the system-battery joint state; it effectively attempts to diagonalize $\rho^*$ in the energy eigenbasis. Stage 2 consists of $M$ swap operations with the tunable thermal reservoir; the thermal qubit in each step becomes increasingly mixed. All energy changes during the operations are stored in the battery, $B$.}
    \label{fig:protocol_gen}
\end{figure}
For the second stage of the protocol, the agent performs a thermal operation by repeatedly appending a reservoir qubit $R$, applying a strict energy-conserving unitary on the combined system $Q B R$, and discarding $R$ afterwards. Specifically, for each of $M$ repetitions (indexed by $\tau \in [M]$), the agent:
    \begin{enumerate}
        \item Sets the energy gap of the reservoir qubit to $\nu (\tau)$ (see Eq.~\eqref{eq:intro_gap_parametrization}) and gets a fresh thermal state $\gamma_\beta(\nu(\tau))$.
        \item Applies the following unitary 
        \begin{align}\label{eq:unitary_intro}
            V^{(\rho^* , \tau)} = \sum_{i,j} \ket{i}\!\bra{j}_Q \otimes \ket{j}\!\bra{i}_R \otimes \Gamma^{(B)}_{(i-j)\nu(\tau)} ,
        \end{align}
       satisfying $[\mathcal{H}_Q+\mathcal{H}_B+\mathcal{H}_R, V^{(\rho^* , \tau)}]=0$ to $QBR$.
        \item Discards the reservoir qubit.
    \end{enumerate}
After completing the $M$ iterations, the agent measures the energy of the battery in its energy eigenbasis and records the energy $W$. A diagrammatic representation of this protocol can be found in Fig.~\ref{fig:protocol_gen}

To approximate a quasi-static process, we repeat the thermal interaction $M$ times within each round. In the limit $M \to \infty$, the process becomes effectively reversible. At each repetition $\tau \in [M]$, we introduce a slight mismatch between the occupation of energy levels of system $A$ and the reservoir $R$ by varying the energy gap. We use the following parametrization of the gap
\begin{equation}
\label{eq:intro_gap_parametrization}
    \nu(\tau)=\beta^{-1}\log{\frac{\lambda_0-\tau\delta p}{\lambda_1+\tau\delta p}},
\end{equation}
where $\delta p = \frac{\lambda_0-e^{-\beta e_0}/Z}{M}$. 

The requirement for a quasi-static process --- that is, a process that is slow and nearly reversible. This arises from the hope to maximize the extractable non-equilibrium free energy from system $Q$. Under such conditions, the system qubit remains close to the thermal state of the reservoir's Hamiltonian, thereby suppressing heat flow during the interaction.  Conversely, if the process is carried out rapidly, i.e., the system state deviates from the reservoir's thermal state, the system will appear out of thermal equilibrium with the reservoir, resulting in heat exchange, which contributes to the entropy production during the protocol. The fundamental trade-off between efficiency and power in such protocols is well established: in general, the reduction in work output scales as $1/\tau_0$ where $\tau_0$ denotes the time taken for the extraction protocol~\cite{van2022finite}. This can also be characterized by a discrete number of interactions required to implement the transformation~\cite{taranto2023landauer}. It is also proven in the original literature that the protocol is able to extract the full non-equilibrium free energy up to $O(\delta p^2)$, which also implies a scaling of $O(\frac{1}{M^2})$.

\section{Analysis of work distribution}
In this section, we analyze the distribution of extracted work in the protocol outlined in Sec.~\ref{sec:actual_protocol}. The analysis involves examining entropy production and the scaling of error probabilities, particularly when the protocol is executed with a finite number of iterations $M$.

Notice that the protocol presented will transform \emph{any} quantum state to $\gamma_\beta$, simplifying Eq.~\eqref{eq:work_for_rho_ideal} expression:
\begin{equation}
    \langle W\rangle\leq \beta^{-1}\left[\D(\sigma\|\gamma)-\D(\sigma\|\rho^*)\right]~.
\end{equation}
The next theorem fully characterizes the set of extracted work values 
and their probabilities for any such protocol.
\begin{theorem}[The Work Distribution of $\rho^*$-Ideal Protocols]
\label{thm2}
Each $\rho^*$-ideal work-extraction protocol that thermalizes all $d$-dimensional quantum states and exhibits at most $d$ distinct extracted-work values. These extracted-work values can be expressed,
in terms of the ideal input's spectral decomposition $\rho^* = \sum_{n} \lambda_n \ket{\lambda_n} \bra{\lambda_n}$, as
\begin{align}
	w^{(n)} \coloneqq \bra{\lambda_n} \mathcal{H}_Q \ket{ \lambda_n}  + \beta^{-1} \ln \lambda_n - \mathcal{F}_{\text{eq}} 
	\label{eq:WorkExtractionValues}
\end{align}
where $\mathcal{F}_{\text{eq}}$ is the equilibrium free energy and $H$ is the Hamiltonian of the system. The associated probabilities are
\begin{equation}
    \Pr \bigl( W \! = w^{(n)}  | \sigma \bigr)
	= \sum_{m} \bra{\lambda_m }\sigma\ket{\lambda_m}  \delta_{ w^{(n)}, w^{(m)}  }
	\label{eq:WorkExtractionProbs}
\end{equation}
This set of work values is independent of the actual $d$-dimensional quantum state $\sigma$ input to the protocol, although the input state determines the probabilities of each outcome.
\end{theorem} 
\begin{proof}

In the limit of zero-entropy-production work extraction from $\rho^*$, the net unitary time evolution of the system--battery--baths joint system must take a special form.
In particular, 
The state of the battery will change deterministically when
the initial state of the system is an eigenstate $\ket{\lambda_n}$ of $\rho^* = \sum_n \lambda_n \ket{\lambda_n} \bra{\lambda_n}$,
almost-surely independent of the initial realization of the thermal reservoirs.  
This implies that the net unitary time evolution
will be of the form
\begin{equation}
\label{eq:UnitaryForm}
U = \sum_{\varepsilon, n, r}
\left( 
\ket{\varepsilon+ w^{(n)}}_B \otimes 
 \ket{f_\varepsilon(n, r)}_{Q,R} 
\Bigr)  \Bigl( 
\bra{\varepsilon}_B \otimes 
\bra{\lambda_n}_Q \otimes 
\bra{r}_R
\right)
\end{equation}
for some $w^{(n)} \in \mathbb{R}$.
Above, $\ket{\varepsilon}$ and $\ket{\varepsilon + w^{(n)} }$
are energy eigenstates of the battery, while $\ket{r}$ 
is an energy eigenstate of the thermal reservoir, while $\ket{f_\varepsilon(n, r)}$ is the effective joint state of the battery and bath after the evolution.
It will be useful to note that
$\braket{f_\varepsilon(n, r)}{f_\varepsilon(n', r')} = \delta_{n, n'} \delta_{r,r'}$
since unitary operations preserve orthogonality between states.

To determine $w^{(n)}$ via the initial-state dependence of entropy production,
we let $\langle\Sigma\rangle_{\rho^*}$ denote the expectation value for entropy production, given initial system-state $\rho^*$, under the fixed work-extraction protocol optimized for $\rho^*$. In our case, with a single heat bath at temperature $T$, the expected entropy production can be defined as usual as 
$\langle\Sigma\rangle_{\rho^*} = \bigl( \langle\tilde W\rangle_{\rho^*} - \Delta \mathcal{F}_t \bigr) / T$.
This is the entropy production for a fixed protocol operating on the initial state $\rho^*$,
where $\Delta \mathcal{F}$ is the change in 
nonequilibrium free energy over the course of the protocol,
and $\tilde{W}$ is the work \emph{exerted}, which is just the negative of the extractable work~\cite{parrondo2015thermodynamics}.

Since all initial states map to $\gamma_\beta$ by the end of the work-extraction protocol,
we know
from Ref.~\cite{riechers2021initial} that
\begin{align}
	\langle\Sigma\rangle_{\sigma} - 	\langle\Sigma\rangle_{\rho^*} = k_B \text{D}( \sigma \| \rho^*)  ~.
\end{align}	
In this case,
$\langle\Sigma\rangle_{\rho^*} = 0$
and 
$T \langle\Sigma\rangle_{\sigma} = \langle\tilde{W}\rangle_\sigma - \Delta \mathcal{F} = \langle\tilde{W}\rangle_\sigma + \beta^{-1} \text{D}(\sigma \| \gamma_\beta)$. 

Hence, 
\begin{align}
	\beta \langle\tilde{W}\rangle_\sigma 
	&=  \text{D}( \sigma \| \rho^*) -  \text{D}( \sigma \| \gamma )\\
	&= \tr(\sigma \ln \gamma) - \tr(\sigma \ln \rho^*) ~.
\end{align}	
In particular, 
let $\sigma = \ket{\lambda_n} \bra{\lambda_n}$,
and note that 
$\ln \gamma = \ln (e^{-\beta H} / Z) = \beta(\F_{\text{eq}}-H)$.
This yields 
$\langle\tilde{W}\rangle_{\ket{\lambda_n} \bra{\lambda_n}} = \mathcal{F}_{\text{eq}} - \bra{\lambda_n}H\ket{\lambda_n} - k_B T \ln \lambda_n$.
The deterministic work-extraction value, given initial pure state $\ket{\lambda_n}$,
must be the same as its expected value $w^{(n)}  = -\langle\tilde{W}\rangle_{\ket{\lambda_n} \bra{\lambda_n}} $, and is thus given by
\begin{align}
	w^{(n)} =  \bra{\lambda_n}H\ket{\lambda_n} + k_B T \ln \lambda_n - \mathcal{F}_{\text{eq}} ~.
\end{align}	

The probability of obtaining the work-extraction value $w$, given any input state 
$\sigma = \sum_{n, m} \ket{\lambda_n} \bra{\lambda_m} \bra{\lambda_n}\sigma\ket{\lambda_m}$, 
can be calculated as
\begin{align}
\Pr(W = w | \sigma)
&=
\tr \Bigl[ \bigl( \ket{\varepsilon_0 + w} \bra{\varepsilon_0 + w} \otimes I \bigr) U \bigl( \ket{\varepsilon_0} \bra{\varepsilon_0} \otimes \sigma \otimes \ket{r} \bra{r}  \bigr) U^\dagger \Bigr] \\
&= \sum_{n, m}
  \bra{\lambda_n}\sigma\ket{\lambda_m}
  \tr \Bigl[ \bigl( \ket{\varepsilon_0 + w} \bra{\varepsilon_0 + w} \otimes I \bigr)\\
  &\quad\quad\quad\quad\quad\quad\quad\quad\quad\quad U \bigl( \ket{\varepsilon_0} \bra{\varepsilon_0} \otimes \ket{\lambda_n} \bra{ \lambda_m } \otimes \ket{r} \bra{r}  \bigr) U^\dagger \Bigr] \\
&= \sum_{n, m}
\bra{\lambda_n}\sigma\ket{\lambda_m}
\tr \Bigl[ \bigl( \ket{\varepsilon_0 + w} \bra{\varepsilon_0 + w} \otimes I \bigr)   \\
&\quad\quad\quad\bigl( \ket{\varepsilon_0 + w^{(n)}}\otimes \ket{f_{\varepsilon_0}(n, r)} \bigr) \bigl( \bra{\varepsilon_0 + w^{(m)}} \otimes \bra{f_{\varepsilon_0}(m, r)}  \bigr)  \Bigr]   \\
&= \sum_{n, m}
\bra{\lambda_n}\sigma\ket{\lambda_m}
\delta_{w, w^{(n)}}
\delta_{w, w^{(m)}}
\braket{f_{\varepsilon_0}(m, r) }{ f_{\varepsilon_0}(n, r) }  \\
&= \sum_{n, m}
\bra{\lambda_n}\sigma\ket{\lambda_m}
\delta_{w, w^{(n)}}
\delta_{w, w^{(m)}}
\delta_{n, m} \\
&= \sum_{n}
\bra{\lambda_n}\sigma\ket{\lambda_n}
\delta_{w, w^{(n)}}
\end{align}
independent of the initial energy state of the battery $\ket{\varepsilon_0}$, and almost-surely independent of the initial realization $\ket{r}$ of the thermal reservoirs in the probability theoretic sense. 
We see that $\Pr(W = w | \sigma) = 0$
unless $w \in \{ w^{(n)} \}_n$.
The probability distribution over these
allowed work-extraction values are 
\begin{equation}
\Pr(W = w^{(n)} | \sigma) = \sum_{m}
\bra{\lambda_m}\sigma\ket{\lambda_m}
\delta_{w^{(n)}, w^{(m)}} ~.
\end{equation}

When there is some entropy production, the probability density of work extraction will have more diffuse peaks. However, for sufficiently low entropy production, the peaks will still be well-separated and so effectively discrete for the purpose of Bayesian updating. The above derivation is valid whether or not 
$\rho^*$ has degenerate eigenvalues.
Notably, the above sums are taken over the eigenstates and their associated eigenvalues, rather than summing over the eigenvalues directly.
\end{proof}

Let us build intuition for the specific forms of $w^{(m)}$. To do so, we consider an energy-degenerate Hamiltonian, which can be trivially extended to a non-degenerate counterpart. Consider $w^{(0)}$ as an example, this represents the extracted work given $\ket{\lambda_0}\!\bra{\lambda_0}$. The term $\bra{\lambda_0}H\ket{\lambda_0}-\mathcal{F}_{\text{eq}}$ reflects the non-equilibrium free energy of $\ket{\lambda_0}\!\bra{\lambda_0}$, which corresponds to the extracted work in the quasi-static limit. The logarithmic term arises because the process is not entirely quasi-static --- it is quasi-static for all but the first repetition. During the first repetition, the system qubit $\ket{\lambda_0}\!\bra{\lambda_0}$ is swapped with the reservoir qubit $\lambda_0\ket{0}\!\bra{0} + 
\lambda_1\ket{1}\!\bra{1}$. This swap generates entropy $\Delta S=-\lambda_0\ln \lambda_0-\lambda_1\ln\lambda_1$ and yields work extraction $\Delta U =-\beta^{-1}\lambda_1\nu(1) =-\beta^{-1}\lambda_1\ln\frac{\lambda_1}{\lambda_0}$. Combining both, the resulting change in free energy of the system and the battery is $\Delta U- \beta^{-1}\Delta S = \ln \lambda_0$, which precisely accounts for the logarithmic term. A similar argument holds for $w^{(1)}$.

Now, we will demonstrate the concentration bound on the work distribution as well as the error scaling of the protocol in Section~\ref{sec:actual_protocol}. We will use the Lagrange mean value theorem and the first mean value theorem for definite integrals~\cite{strang2019calculus}. We will present the main conclusions, while detailed derivations will be presented in Appendix~\ref{apd:work_distribution}.
For the protocol presented, it can be shown that the rate of convergence of the random variable $\Delta W$ towards its expectation value is
\begin{align}
    \Pr[|\Delta W - \Ex[\Delta W]|\geq \zeta ] \leq 2 e^{-\frac{\zeta^2}{\sum_{\ell=1}^{M-1} \left(\frac{2}{\lambda_{1}} \delta p \right)^2}} \leq 2 e^{-\frac{\lambda_{1}^2 \zeta^2 M}{(2\lambda_{0}-1)^2}}~,
\end{align}
where we used $\delta p =O(\frac{1}{M})$ in the second inequality. It is then easy to see that taking $M\to\infty$, we would obtain 
\begin{equation}
    \lim_{M\to\infty }\Pr[|\Delta W - \Ex[\Delta W]|\geq \zeta ]= 0~.
\end{equation}
This then also indicates that the error probability is upper bounded
\begin{align}
\label{eq:error_bound}
    \Pr(\text{error}) \leq 2e^{-\frac{\lambda_1^2\zeta^2 M}{(2\lambda_0-1)^2}} 
\end{align}
This result is especially important when we discuss the determination of the number of iterations $M$ required in~\cref{ch:learning}.

\section{Consequence of looser constraint}
\label{sec:consequence}
We now examine the consequences of relaxing the strict energy conservation constraint. Average energy conservation is only maintained when $\sigma$ is diagonal in the eigenbasis of $\rho^*$. This becomes evident when analyzing the expectation value of extracted work. Combining Eq.~ \eqref{eq:WorkExtractionValues} and \eqref{eq:WorkExtractionProbs} yields:
\begin{equation}
\begin{split}
\label{eq: average_energy}
     \Ex(W) &= \sum_i\bra{\lambda_i}\sigma\ket{\lambda_i}\left(\bra{\lambda_i}\mathcal{H}_Q\ket{\lambda_i}+\beta^{-1}\ln\lambda_i-\mathcal{F}_{\text{eq}}\right)\\
     &= \sum_i\bra{\lambda_i}\sigma\ket{\lambda_i}\left(\bra{\lambda_i}\mathcal{H}_Q-\mathcal{F}_{\text{eq}}\id\ket{\lambda_i}\right) +\beta^{-1}\tr(\sigma\ln\rho^*)\\
     &=-\beta^{-1}\sum_i\bra{\lambda_i}\sigma\ket{\lambda_i}\bra{\lambda_i}\ln\gamma\ket{\lambda_i}+\beta^{-1}\tr(\sigma\ln\rho^*)\\
     &= \beta^{-1}\left[\tr(\sigma\ln\rho^*)-\tr(\Delta(\sigma)\ln\gamma)\right]\\
     &= \beta^{-1}\left[\D(\Delta(\sigma)\|\gamma)-\D(\sigma\|\rho^*)+\D(\sigma\|\Delta(\sigma))\right]~,
\end{split}
\end{equation}
where we used the identity $\ln\gamma=\beta(\mathcal{F}_{\text{eq}}\id-H)$ in the third equality and $\D(\sigma\|\Delta(\sigma))=S(\Delta(\sigma))-S(\sigma)$ in the final line~\cite{baumgratz2014quantifying}. Here, $\Delta(\sigma)$ represents $\sigma$ dephased into the eigenbasis of $\rho^*$. When $\sigma$ is diagonal in this eigenbasis, $\Delta(\sigma)=\sigma$ and we recover $\Ex(W) = \beta^{-1}\D(\sigma\|\gamma)-\beta^{-1}\D(\sigma\|\rho^*)$. The general case reveals a discrepancy between these quantities, indicating that implementing the protocol requires energy exchange between the agent and quantum system - necessitating an external energy source. This issue disappears when $\mathcal{H}_Q$ is fully degenerate. In such cases, $\rho^*$-ideal protocols may qualify as thermal operations discussed in Section~\ref{sec:thermal_ops}, as the agent can freely choose the energy eigenbasis without being constrained by a unique system Hamiltonian.

For non-degenerate Hamiltonians, operations not commuting with the total Hamiltonian can generate coherence in the energy eigenbasis. This coherence breaks time-translation symmetry and typically violates global energy conservation~\cite{lostaglio2015quantum}. While knowledge of the true quantum state prevents coherence generation, applying the operation to an unknown input state inevitably creates energy-basis coherence, requiring external energy input. Consequently, this thesis primarily considers $\rho^*$-ideal protocols for degenerate system Hamiltonians. However, we show that analogous results can be obtained for non-degenerate Hamiltonians when proper accounting of external energy inputs is included.

\section{Non-degenerate Hamiltonian}
\label{sec:non-degen}
We now analyze how the protocol's behavior differs between degenerate and non-degenerate Hamiltonians. The second stage of work extraction (Section~\ref{sec:actual_protocol}) maintains strict energy conservation regardless of the Hamiltonian, so the key distinction arises in the first stage during implementation of the unitary operation $U^{(\rho^*)}_{QB}$.
For degenerate Hamiltonians, this unitary preserves the average energy of the system $Q$ and acts as an identity operation on the battery. However, for non-degenerate cases ($E_0 \neq E_1$), we must consider the energy difference between the initial and final states of the combined system. The initial energy is $\tr(\sigma \mathcal{H}_Q)$ (assuming the battery starts at zero energy for simplicity), while the post-rotation energy becomes $\sum_i \bra{\lambda_i}\sigma\ket{\lambda_i}\bra{\lambda_i}\mathcal{H}_Q\ket{\lambda_i}$. The energy change is given by:
\begin{equation}
\begin{split}
    \Delta U &= \sum_i \bra{\lambda_i}\sigma\ket{\lambda_i}\bra{\lambda_i}H\ket{\lambda_i} - \tr(\sigma \mathcal{H}_Q)\\
    & = \tr(\Delta(\sigma)\mathcal{H}_Q) - \tr(\sigma \mathcal{H}_Q)~,
\end{split}
\end{equation}
indicating that the protocol requires at least $\Delta U$ units of external energy input.

Combining this with the average energy extraction from $\rho^*$-ideal protocols (Eq.~\eqref{eq: average_energy}), we obtain the net work extraction:
\begin{equation}
\begin{split}
    \Ex(W_\text{net})&=\beta^{-1}\left[\tr(\sigma\ln\rho^*)-\tr(\Delta(\sigma)\ln\gamma)\right]-\Delta U \\
    &=\beta^{-1}\tr(\sigma\ln\rho^*) -\tr(\Delta(\sigma)(\mathcal{F}_{\text{eq}}-\mathcal{H}_Q))-\tr(\Delta(\sigma)\mathcal{H}_Q) + \tr(\sigma \mathcal{H}_Q)\\
    &=\beta^{-1}\tr(\sigma\ln\rho^*) - \mathcal{F}_{\text{eq}}+\tr(\sigma \mathcal{H}_Q)\\
    &= \beta^{-1}\left[\D(\sigma\|\gamma)-\D(\sigma\|\rho^*)\right]~,
\end{split}
\end{equation}
recovering Eq.~\eqref{eq:work_for_rho_ideal}. 

The original proposal for implementing this stage-1 unitary operation \cite{skrzypczyk2014work} involves a time-dependent interaction Hamiltonian, requiring precise temporal control and smooth Hamiltonian variations \cite{woods2023autonomous}. Subsequent work has established that perfect implementation demands unbounded coherence resources \cite{aaberg2014catalytic,korzekwa2016extraction}. While laser-based approximations exist, they require careful energy accounting of the laser input. Nevertheless, this protocol remains theoretically valuable as it establishes fundamental performance limits under ideal conditions, yielding non-trivial and physically significant results.

\section{Summary}
In this chapter, we discussed the class of work extraction known as the $\rho^*$-ideal work extraction protocol in detail. This class of protocol is operationally motivated for an agential approach of thermodynamics, the fact that an agent has to design a protocol for a potentially unknown quantum state. Its simple work distribution is also derived rigorously; its implication on breaking strict energy conservation as well as its consequences, are also discussed. We then discussed the most general strategy that an agent with classical memory can operate in this sequential operation setup.  With this out of the way, we are now ready to properly start studying how an agent, utilizing such protocol, is able to extract work from temporally correlated quantum states, following different possible policies.

%% file: Chapters/Chapter4.tex

\chapter{Extraction of work from temporally correlated quantum states} 
\label{ch:3}
\noindent\rule{\textwidth}{0.4pt}
\vspace{1.5em} 
\setlength{\parindent}{4ex}

\noindent \textit{In this chapter, we present the first of our main results: a systematic construction of how an autonomous agent can extract work from a \emph{finite-state source of quantum states}. We begin by providing a formal definition of finite-state sources of quantum states. We then detail the construction of the agent’s operational protocol within this framework. Central to our discussion is the interplay between the dynamics of the source process and the meta-dynamics governing the agent’s belief state updates. Finally, we demonstrate the advantages of such an agent—specifically, the enhanced capacity to extract work when equipped with classical memory and quantum processing capabilities, in contrast to agents lacking these features.}
\newpage

\section{Finite-state sources of quantum states}
We first discuss the most important thing for all work extraction, the ``source" of free energy, which in our case is the so-called ``finite-state sources of quantum states". 
Suppose we are given a classical HMM, $\mathcal{M}=\{\mathcal{S},\{T^{(x)}\}_{x\in\mathcal{X}},\gbm \pi\}$, as a generator of sequences consisting of $\{0,1\}$s, but instead of emitting classical symbols $X_0,X_1\cdots$, this particular generator is able to emit quantum states $\sigma_0,\sigma_1\cdots$. An example of this can be found in Fig.~\ref{fig:PC_diagram}, in particular, we will discuss the case of the ``Perturbed Coin" depicted in Fig.~\ref{fig:PC_diagram}(a). One can draw a parallel between these quantum states and the boxes in the Szilard engine thought experiment; as long as they are individually out of equilibrium, i.e. $\{\sigma_t\}_t\neq\gamma$, then it is possible for an agent to extract useful work from it. 
\begin{figure}[b]
    \centering
    \includegraphics[width=0.7\linewidth]{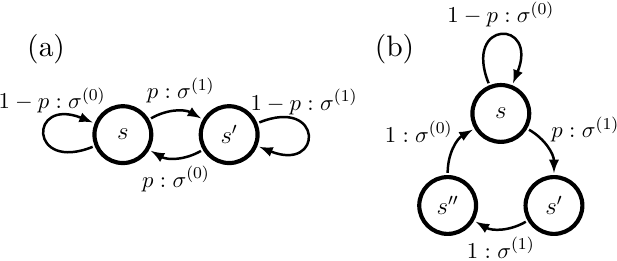}
    \caption{Latent-state sources of correlated quantum processes. 
Each arrow represents a transition between latent states; the label $p:\sigma^{(x)}$ indicates that the transition happens with probability $p$ and
produces a quantum state $\sigma^{(x)}$.
(a) Perturbed-coin process. (b) 2-1 golden-mean process.}
    \label{fig:PC_diagram}
\end{figure}
Mathematically, the quantum states generated by the HMM can be described formally by the density operator
\begin{equation}
    \label{eq: fuel}
\rho_{\overleftrightarrow Q}=\sum_{\overleftrightarrow x} \Pr \left(\overleftrightarrow x\right)\bigotimes_{t\in\mathbb{Z}}\sigma^{(x_t)}_{Q_t},
\end{equation}
where each time step $t$ is associated with a unique elementary physical system $Q_t$, and $\overleftrightarrow x = \dots x_{-1}x_0x_1\dots$ denotes a bi-infinite string over $\mathcal{X}$. This is what we define as \emph{finite-state sources of quantum state}, since they are quantum states that are generated by a finite-state HMM.
The joint quantum state is separable amongst the $Q_t$s, but 
can have non-classical correlations in the form of quantum discord~\cite{Modi10Unified}. These memoryful quantum sources generalize the kindred ``classically controlled qubit sources" of Ref.~\cite{venegas2020measurement}.

The time-indexed sequence of stochastically-generated quantum outputs $\sigma^{(x_t)}_{Q_t}$ of such a process is given to the agent one $Q_t$ at a time, in a temporal order where $Q_t$ precedes $Q_{t+1}$. 
Now, if the agent has access to an arbitrarily big quantum memory, quantum states from the past could be stored for a later time, then it would be straightforward to extract all non-equilibrium additions to free energy from such temporally correlated quantum systems, by acting coherently on the full sequence all at once~\cite{SSP14}, since the multipartite system could then be treated as a single quantum system. However, as we discussed in Sec.~\ref{sec:collective}, such an approach is costly and the decoherence time tends to be rather short, making it unfeasible. Instead, we consider a setting where each elementary physical system $Q_t$ has an immediate expiration date, and the agent is only equipped with a classical memory. This forces the agent to interact with $Q_t$ at the present time $t$, and then never again. This restriction is shared with the repeated-interaction~\cite{Stras17Quantum} frameworks and information-ratchet~\cite{mandal2012work, Boyd17PRE}, the latter being the \emph{classical analogue} of the problem we aim to study. Accordingly, we provide a brief overview of key results from the information ratchet literature.

\section{Information Ratchet}
The information ratchet framework provides a formal model for studying autonomous agents performing sequential work extraction from a classical input source, often referred to as a ``fuel tape"~\cite{mandal2012work,parrondo2015thermodynamics,boyd2016identifying}. It consists of a sequence of classical symbols $Y_{0:N}\in\mathcal{Y}^{\otimes N}$, which are fed one-by-one into an agent (the ratchet). The agent is equipped with an internal state $X_n\in \X$, and interacts with the symbol on the tape, depending on its internal state.

As a result of this interaction, each tape symbol $Y_n$ is transformed into an output symbol $Y_n'$. Depending on the Hamiltonian governing the agent and tape, this transformation can involve a change in energy, which is interpreted as work extracted from or injected into the system. A representation of such a ratchet can be found in Fig.~\ref{fig:ratchet}.
\begin{figure}
    \centering
    \includegraphics[width=0.9\linewidth]{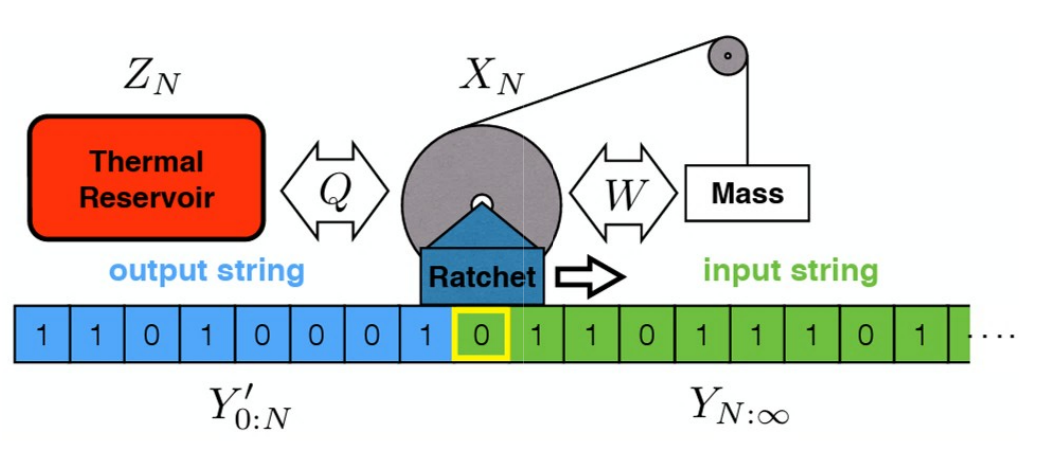}
    \caption{Basic form of a information ratchet considered in~\cite{boyd2016identifying}. The ratchet contains within itself some internal state $X\in\X$, it interacts with the input tape consisting of symbols $Y\in \mathcal{Y}$ according to some predetermined policy. The thermal reservoir provides the heat exchange necessary for work extraction.}
    \label{fig:ratchet}
\end{figure}
The combined evolution of the internal state and the tape forms a finite-state Markov process. The long-term (asymptotic) performance of the ratchet—particularly its capacity for work extraction—can be analyzed by studying the stationary distribution over the agent-tape joint states. This framework becomes especially useful when considering more complex settings, such as when the tape exhibits temporal correlations. Which is why it can be viewed as the classical analogue to the sequential work extraction we discuss in the quantum regime, where the fuel tape is replaced by sequences of correlated quantum states. 

It has been shown that the asymptotic rate of work extraction, $w$, by such a ratchet can be upper bounded by the difference in entropy rates between the input and output tapes:
\begin{equation}
    w\leq -\beta^{-1}\ln2 (h_\mu'-h_\mu)~,
\end{equation}
where $h_\mu$ is the entropy rate of the input tape,
\begin{equation}
\label{eq:ent_rate}
    h_\mu \coloneqq\lim_{n\to\infty}\frac{H(Y_{0:n})}{n}~,
\end{equation}
with $H(Y_{0:n})$ denoting the block entropy of a $n$-symbol segment~\cite{boyd2016identifying,boyd2017transient}. $h_\mu'$ 
is defined analogously for the output tape. Notably, for input processes that are unifilar, the entropy rate in  Eq.~\eqref{eq:ent_rate} admits a closed-form expression~\cite{jurgens2021shannon}. However, such analytic simplifications are generally unavailable in the quantum generalization of this framework. On top of that, temporally correlated quantum states tend to have non-classical correlations such as quantum discords, which could further reduce the extractable work.

\section{Set-up}
As discussed, we wish to study the case where the ``fuel tape" is made up of quantum states as in Eq.~\eqref{eq: fuel}. Drawing parallels from an information ratchet, 
the agent will be equipped with an internal classical memory $M$ and given access to a thermal reservoir $R$ at some fixed inverse temperature $\beta$. Its objective is to extract work by raising the internal energy of a battery $B$. To accomplish this, it can bring each system $Q_t$ closer to its thermal state, $\gamma = e^{- \beta H}/Z$. For simplicity of presentation, we assume that each subsystem $Q_t$ is subject to the same Hamiltonian $\mathcal{H}_Q$ for all $t$, but our results generalize in an obvious way
if we allow different Hamiltonians for each subsystem.
The agent requires only the description of the HMM, $\mathcal{M}$,
the Hamiltonian $\mathcal{H}$ and the quantum states that are emitted $\{\sigma^{(0)},\sigma^{(1)}\}$. 
Here we only impose the constraint that the stochastic generator has to be time-independent, i.e., the transition matrix, $\{T^{(x)}\}_{x\in\mathcal{X}}$, governing the dynamics of the HMM cannot change with time. The purity of the finite-dimensional emitted states, $\sigma^{(0)}$ and $\sigma^{(1)}$, is not restricted, but for simplicity of analysis, we will consider the case where both are pure and given by $\sigma^{(0)}=\ket{0}\bra{0}$, $\sigma^{(1)}=\ket{\psi}\bra{\psi}$ where $\ket{\psi}=\sqrt{r}\ket{0}+\sqrt{1-r}\ket{1}$. These are chosen as we can conveniently parametrize their fidelity $F(\sigma^{(0)},\sigma^{(1)})= \left(\Tr(\sqrt{\sqrt{\sigma^{(0)}}\sigma^{(1)}\sqrt{\sigma^{(0)}}})\right)^2=r$

Given this setup, the central question we aim to address is whether an agent with only classical memory can extract work that is otherwise locked within temporal correlations. If so, how does the performance of such an agent compare to that of a non-adaptive agent—i.e., one with no memory at all? To answer this, we first define the most general thing that an agent equipped with classical memory can do sequentially. 

\section{Classical-Causal Strategy}
\label{sec:classical_causal}
For an agent that only has classical memory and has to carry out operations sequentially, we can define a generic strategy characterized by a 4-tuple $\mathbf{S}=(\mathcal{K}, \A,\tau,\Lambda)$. Where
\begin{enumerate}
    \item $\mathcal{K}$ - set of memory states that is kept within the agent's classical memory.
    \item $\A$ -  set of allowed actions that the agent can take, which is a subset of the set of thermal operations on $Q,R,B$ when $\mathcal{H}_Q$ is energy degenerate.
    \item $\Lambda$ - policy, where $\Lambda: \mathcal{K}\rightarrow \mathcal{A}$ is a map from the agent's memory to the set of allowed actions. 
    \item $\tau$ - transition rule, this specifies how a memory state transitions to another after the agent obtains outcome $O_t$ when taking action $A_t$, i.e., $K_t=\tau(K_{t-1},A_t,O_t)$.
\end{enumerate}  
At each time step $t$, the agent utilizes its current memory state $K_{t-1}$ to select an action $A_t=\Lambda(K_{t-1})$, aiming to extract free energy from $\sigma_{Q_t}$ and deposit it into the battery $B$. Measuring $B$ in the energy eigenbasis then quantifies the extracted work, $W_t$.

The most na\"ive strategy involves the agent recording its entire history of actions and outcomes. Here, the memory state $\mathcal{K} \ni K_t=(a_0, a_1, \dots, a_t, o_t)$ is initialized as $K_0=()$. Under this approach, the transition rule $\tau$ is a simple concatenation, $K_t=(K_{t-1}, a_t, o_t)$, where the outcome $o_t$ is simply the extracted work $W_t$. This causal strategy with classical memory is depicted schematically in Fig.~\ref{fig:Bayes_Net}.
\begin{figure}[htbp]
    \centering
    \includegraphics[width=0.6\linewidth]{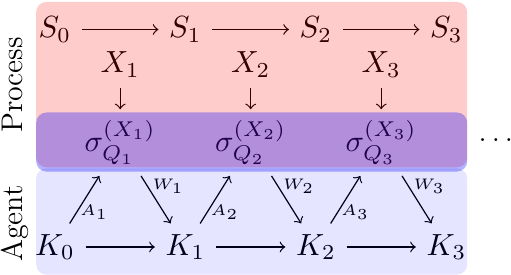}
    \caption{Schematic diagram illustrating the evolution of belief states over time. At each time step, the agent accesses its internal belief state $K_{t-1}$, which determines its action $A_{t}$. The agent then interacts with the quantum state $\sigma^{(X_t)}_{Q_t}$ and receives a corresponding reward $W_t$.}
    \label{fig:Bayes_Net}
\end{figure}

A fundamental drawback of this history-dependent approach is that the required memory capacity scales linearly with time, rendering it physically inefficient. To circumvent this, Sec.~\ref{sec:belief} demonstrates that it is sufficient for the agent to track its internal \emph{belief} state. This conceptual shift ensures that the necessary memory footprint remains constant, independent of the number of operational steps, $T$. The action space is defined by the set of $\rho^*$-work extraction protocols optimized for specific quantum states. Throughout this thesis, we evaluate two distinct decision-making policies: a locally optimizing policy ($\Lambda_\text{LO}$) and a globally optimal policy ($\Lambda^*$). Within this framework, the state transition $\tau$ operates via Bayesian updating, a mechanism explored in detail alongside belief states in Sec.~\ref{sec:belief}.

To rigorously assess whether classical memory provides an operational advantage over a memoryless agent, we must establish a comparative metric. We introduce the cumulative work extracted over a horizon of $T$ steps as 
\begin{equation}
    \mathbf{W}(T,\Lambda)\coloneqq\mathbb{E}_{\Lambda}\sum_{t=1}^T \mathbb{E}(W_t|A_t=\Lambda(K_{t-1}))~.
\end{equation}
However, finite-time metrics are often obscured by transient boundary effects, such as the agent's initially biased belief or the lag in synchronizing with the underlying HMM~\cite{boyd2017transient}. To isolate the fundamental performance, our comparative analysis focuses on the asymptotic rate of work extraction:
\begin{equation}
\label{eq:work_rate}
    w(\Lambda) = \lim_{T\to\infty}\frac{1}{T} \mathbf{W}(T,\Lambda)~.
\end{equation}

\section{Belief}
\label{sec:belief}
In the agential approach to many problems, an agent's belief about an underlying process or distribution governs their actions. This concept is central to both computational mechanics and Bayesian inference. In the Bayesian framework, these beliefs are formalized as prior distributions—representing the agent's expectations about certain events before new data is observed.
For example, suppose a certain disease is known to affect  $10\%$ of the population. An agent with no additional information might adopt an uninformative prior, assigning a $10\%$. However, if the agent experiences symptoms or receives partial test results, their prior would be adjusted accordingly, resulting in an informative prior that reflects this new evidence. When an agent claims to act based on a prior distribution $\gbm\eta_0$,  it means they assign probabilities to events in accordance with $\gbm\eta_0$. In the disease example, an uninformative prior would correspond to $\gbm\eta_0=(0.9,0.1)$, where $0.1$ is the prior probability of being infected. In contrast, an agent who has observed concerning symptoms might adopt a prior like $\gbm\eta_0 = (0.3,0.7)$, assigning higher likelihood to infection. Naturally, the decisions made by these two agents are likely to differ, given their distinct beliefs. In computational mechanics, we study how these beliefs evolve as actions are taken and outcomes are observed, governed by Bayesian update rules so as to remain consistent with the observed statistics.

Throughout the remainder of the thesis, beliefs are denoted by $\gbm \eta$, which can be interpreted by a vector with individual elements representing the prior distribution of a certain event $i$ occurring with probability $\{p_i\}_i$. Each of these elements must be positive, and they sum to unity. In other words
\begin{equation}
    \gbm \eta \coloneqq \{p_i\}_i,\quad p_i\geq0\quad\forall \quad i  , \quad\sum_ip_i=1~.
\end{equation}

To update one's belief, we employ Bayes’ theorem, which formalizes how prior beliefs are revised in light of new data. In its basic form:
\begin{equation}
    \Pr(H|E) = \frac{\Pr(E,H)}{\Pr(E)}=\Pr(E|H)\frac{\Pr(H)}{\Pr(E)}~.
\end{equation}
Here, $H$ denotes a \emph{hypothesis} that governs the distribution of observed data or evidence $E$.
\begin{enumerate}
    \item  $\Pr(H)$ is the prior: the agent’s belief about the hypothesis before observing any data, for instance, the probability that the agent has a disease.
    \item $\Pr(E|H)$ is the likelihood: the probability of observing data $E$ given that $H$ is true, for example, the probability of experiencing a symptom given that the agent has the disease.
    \item $\Pr(E)$ is the marginal likelihood: the total probability of observing the data under all possible hypotheses, such as the overall probability of experiencing the symptom whether or not the agent has the disease.
    \item Finally, $\Pr(H|E)$ is the posterior: the updated belief about the hypothesis after observing data, e.g., the probability that the agent has the disease given that they are experiencing symptoms.
\end{enumerate}
 
In the context of stochastic processes modeled by HMMs, an agent's \emph{belief at time $t$}, $\gbm\eta_t$ is the conditional distribution over the latent state, given all past observations. Mathematically, $\gbm\eta_t$ is a row vector with the number of elements equal to the number of latent states within the HMM. The $i$-th element is given by
\begin{equation}
    \gbm\eta_{t,i} := \Pr(S_t=S_i | O_1 \dots O_t = o_1 \dots o_t, S_0 \sim \gbm\pi )~.
\end{equation}
Here, the belief is conditioned on the initial latent states following the stationary distribution $\gbm\pi$ and the sequence of observed symbols. Upon observing symbol $x_t$, the belief updates via:
\begin{equation}
\label{eq:belief_update}
    \gbm\eta_{t+1}= \frac{\gbm\eta_t T^{(x_{t})}}{\gbm\eta_t T^{(x_t)}\mathbf{1}}
\end{equation}
where $\mathbf{1}$ is a column vector of $1$s and $T^{(x)}$ is the labeled-sub-stochastic matrix found in Eq.~\eqref{eq:stochastic_condition}. The new row vector $\gbm \eta_{t+1}$ is the posterior but also serves as the updated prior for the next time step. 

A central application of belief in HMMs is the construction of unifilar representations from non-unifilar models. This transformation is known as the mixed-state presentation (MSP), introduced by Crutchfield and collaborators~\cite{crutchfield2009time}. To avoid confusion with mixed states in quantum mechanics, we refer to these as belief states. Regardless of whether the original HMM is unifilar, one can always define a new unifilar model over belief states.  Initially, the agent's belief is the stationary distribution $\gbm\pi$. After observing each symbol, the belief updates via Eq.~\eqref{eq:belief_update}. As the agent continues observing, new belief states are generated, but after some time, the system enters a recurrent regime: the same belief states begin to reappear.
The recurrent portion of the belief-state space forms a unifilar representation of the original process. While this representation is always possible in principle, it may contain infinitely many belief states, especially if the original model is non-unifilar.

Two illustrative examples are presented in Table~\ref{tab:table_msp}.
\begin{sidewaystable}[htbp]
    \caption{Table of how MSP can be transformed into an $\epsilon$-machine. The first row shows a process known as ``Even process", this is a unifilar process. The second row shows a process that is non-unifilar. The first column shows the HMM of the process, the second shows how MSP is generated after each symbol is observed, and the last column shows the recurrent belief states which correspond to the $\epsilon$-machine.}
    \label{tab:table_msp}
\includegraphics[width=1\linewidth]{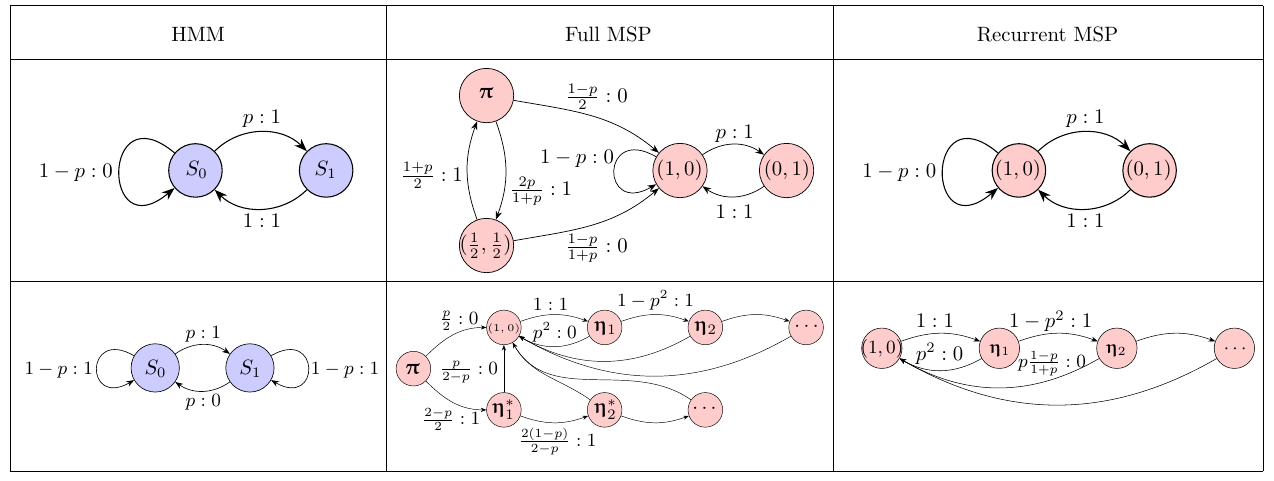}
\end{sidewaystable}
\begin{itemize}
    \item \textbf{First row} The Even Process, a unifilar HMM with stationary distribution $\gbm\pi = \left(\frac{1}{1+p}, \frac{p}{1+p} \right)$. The MSP is built by recursively applying Eq.~\eqref{eq:belief_update} for different $x\in\X$ and is shown in the middle panel. Its recurrent portion (right panel) is structurally identical to the original HMM, reflecting the fact that minimal unifilar HMMs are faithful to the underlying process\\
    \item \textbf{Second row}: A non-unifilar HMM. As shown, its MSP (middle panel) grows indefinitely due to uncertainty in the symbol-to-state transition. Even when only the recurrent portion is kept (right panel), it still contains an infinite number of belief states. Notice that observing a 0 perfectly collapses the belief to state $S_0$, but observing a 1 merely nudges the belief toward $(0,1)$ without fully resolving the latent state, an effect which is intuitive when we look at the structure of the non-unifilar model.
    
\end{itemize}
This transformation effectively shifts the dynamics from the latent states of the original HMM to a meta-dynamics over belief states—one that an agent can meaningfully track and update over time. This framework becomes especially crucial in the quantum setting, where direct observations are typically avoided to prevent the collapse of the quantum state.

\section{Synchronization to the process}
\label{sec:synchronization}
From Sec.~\ref{sec:rho_ideal}, we know that in order to extract all non-equilibrium free energy from any quantum state, one has to tailor the protocol for that specific quantum state; any deviation will then result in dissipation as shown in Eq.~\eqref{eq:work_for_rho_ideal}. In our case, we cannot directly measure the quantum states, nor do we know what the underlying latent is at any time $t$. In order to best tailor the protocol to extract work from $\sigma^{(x_t)}_{Q_t}$, the agent has to predict what $\sigma^{(x_t)}_{Q_t}$ will be, which in turn requires the agent's belief to be synchronized to the underlying HMM in order to make the best prediction.
It might be immediately clear to some readers that the challenge with such an HMM is precisely that of synchronization. Recall that we mentioned in Sec.~\ref{sec:stochastic_process}, in order to guarantee synchronization, the process has to be unifilar. The Perturbed Coin process in its classical formulation, with $\sigma^{(0)}$ being orthogonal to $\sigma^{(1)}$, is indeed unifilar when measured in their eigenbasis. However, in the case where both quantum states are non-orthogonal, any measurement on the state will introduce non-unifilarity~\cite{venegas2020measurement}. To illustrate this, consider a simple case where $\sigma^{(0)}=\ket{0}\!\bra{0}$ and $\sigma^{(1)}=\ket{+}\!\bra{+}$, suppose we choose to measure in the basis of $\{\ket{0},\ket{1}\}$. The resultant dynamic of the HMM will appear to be what is depicted in Fig.~\ref{fig:non_unifilar_coin}.
\begin{figure}
    \centering
    \includegraphics[width=0.6\linewidth]{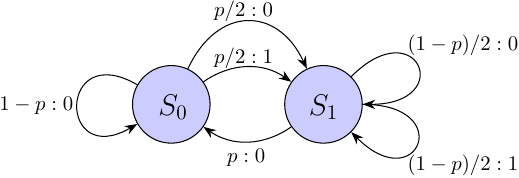}
    \caption{The effective dynamic of Perturbed Coin in Fig.~\ref{fig:PC_diagram}(a), if the quantum outputs are given by $\sigma^{(0)}=\ket{0}\!\bra{0}$ and $\sigma^{(1)}=\ket{+}\bra{+}$ and the agent chooses to measure along the $\ket{0},\ket{1}$ basis, which yields outcome of $0$ and $1$ respectively. }
    \label{fig:non_unifilar_coin}
\end{figure}
For any latent state, $S_0$ or $S_1$, if an agent observes a $0$ as the measurement outcome, the transitioned state cannot be uniquely determined, hence violating Eq.~\eqref{eq:unifilar}.

Furthermore, our objective is not to perform state estimation, but rather to extract work from the quantum states. According to the data-processing inequality, applying any quantum channel—including measurements—cannot increase the distinguishability of quantum states. The belief update method presented in Eq.~\eqref{eq:belief_update} is no longer adequate in this context. To address this, we generalize the update rule to incorporate the posterior probability distribution over the identity of the quantum state, conditioned on the measurement outcomes. 
\begin{theorem}[Optimal Recursive Belief Update]
\label{thm1} 
For any POVM which yields outcome $O_t$ on the quantum state of the system at time $t$, the optimal belief state --- about the latent state  of the quantum source --- updates iteratively according to
\begin{equation}
    \gbm\eta_{t} = z_t^{-1} \! \sum_{x \in \mathcal{X} } \!
\Pr( O_t = o_t | X_t = x, K_{t-1} = \gbm\eta_{t-1} ) \, \gbm\eta_{t-1} T^{(x)}
\label{eq:generalise_update}
\end{equation}
where $K_t$ is the random variable for the belief state,
and 
\begin{equation}
    z_t = \sum_{x' \in \mathcal{X} } \Pr( O_t = o_t | X_t = x', K_{t-1} = \gbm\eta_{t-1} ) \, \gbm\eta_{t-1} T^{(x')} \mathbf{1}
\end{equation}
is a normalizing factor.
\end{theorem}
The derivation of Theorem~\ref{thm1} will be shown in Appendix~\ref{AppendixA} for interested readers. The expression in Eq.~\eqref{eq:generalise_update} is very much a re-formulation of the familiar Bayes rule that we introduced earlier. Unlike the original formulation of the belief update in Eq.~\eqref{eq:belief_update}, where the observations are restricted to the same domain of $\mathcal{X}$. This generalized belief update allows us to account for any POVMs that an agent may choose to measure the system at hand, even if the measurement outcome lies outside the set of $\mathcal{X}$. This is relevant to our case as we will soon demonstrate the equivalence between a general POVM and the work extraction operation in the context of output statistics. 
The probability $\Pr( O_t = o_t | X_t = x, K_{t-1} = \gbm\eta_{t-1} )$ appearing in 
Eq.~\eqref{eq:generalise_update} is a straightforward calculation of the probability that the observed POVM outcome $o_t$ should be obtained, given that the system was prepared as $\sigma^{(x)}$; one can simply apply Born rule. Conditioning on
the previous state of knowledge $K_{t-1}$ is important as the belief of the agent can influence the choice of POVM applied at time $t$.

\section{Operation of the agent}
\label{sec:op_agent}
To harvest the free energy locked up in correlations, the agent's internal memory should somehow become correlated with the latent state of the source during its work extraction operation. However, directly measuring each quantum system potentially costs energy and may disturb its state, thereby depleting free energy~\cite{araki1960measurement}.
Rather, our agent updates its memory conditioned on the extracted work value $W_t$ at each time.  The change in the battery's energy thus serves as the observable $O_t = W_t$ for updating the belief state mentioned in Section \ref{sec:synchronization}.

At each time step, conditioned on its memory state, the agent performs a \emph{work extraction protocol}---a unitary transformation of
the composite $Q_t$, $B$, and $R$ supersystem, designed to transfer free energy from $Q_t$ to $B$. We will use the $\rho^*$-ideal work extraction protocol discussed in Section~\ref{sec:rho_ideal} in order to extract all non-equilibrium free energy from specific states the agent tailors to.

In particular, Eq.~\eqref{eq:WorkExtractionProbs} shows the probability distribution for work extracted when the protocol optimized for $\rho_t^{*}$ actually operates on $\sigma^{(x)}$.
Regardless of how the belief state influences the choice of $\rho_t^*$, we can now leverage Theorems.~\ref{thm1} and \ref{thm2} to rewrite the belief update as
\begin{equation}
    \label{eq:update_pc2}
\gbm\eta_{t+1} =
\frac{\sum_{x \in \mathcal{X} }
	\sum_n \delta_{ w_{t+1}, w^{(n)} } \bra{\lambda_n}  \sigma^{(x)} \ket{ \lambda_n} \, 
	\gbm\eta_{t} T^{(x)} }{
	\sum_n \delta_{ w_{t+1},  w^{(n)}  } 
	 \bra{\lambda_n} \xi_{t} \ket{\lambda_n} }~.
\end{equation}
Here, $\xi_t = \sum_{x\in\mathcal{X}} \gbm\eta_t T^{(x)}\mathbf{1} \sigma^{(x)}$, is the quantum state that the agent \emph{expects} conditioned on his belief $\gbm\eta_t$, where $\gbm\eta_t T^{(x)}\mathbf{1}$ is the prior over the next emitted quantum state, $\sigma^{(x)}$, conditioned on the belief of the agent at this time step.

From Eq.~\eqref{eq:WorkExtractionProbs}, we can find that the work-induced transitions between belief states have probabilities
\begin{equation}
  \Pr(W_{t+1} = w | K_t = \gbm\eta_t )
= \sum_n \bra{\lambda_n}  \xi_t \ket{\lambda_n }
\delta_{w, w^{(n)}} ~.
\label{eq:BeliefStateTransitionProbs}  
\end{equation}
This is an important quantity to consider as it characterizes the meta-dynamic of the belief states transitions, and will be especially useful when we analyze the asymptotic work extraction rate later on.

Both the belief states and transitions between them
can be derived from the HMM of the known source. Belief states can thus be explicitly represented in the memory $M$ of an autonomous work-harvesting device.
The memory states $\{ (\gbm\eta , \varepsilon) \}_{\gbm\eta, \varepsilon}$
should also store the last measured energy state $\varepsilon$ of the battery.
A memory-controlled unitary can implement memory-assisted quantum work extraction, as depicted in the circuit diagram of Fig.~\ref{fig:schem_eng}. Subsequent measurement of the battery state then gives access to the extracted work and allows an autonomous update of the memory, according to the above-outlined transition rule of Bayesian prediction.
This prediction-extraction cycle continues repeatedly, as was shown in Fig.~\ref{fig:Bayes_Net}.
\begin{figure} 
\begin{center}
\includegraphics[width=\columnwidth]{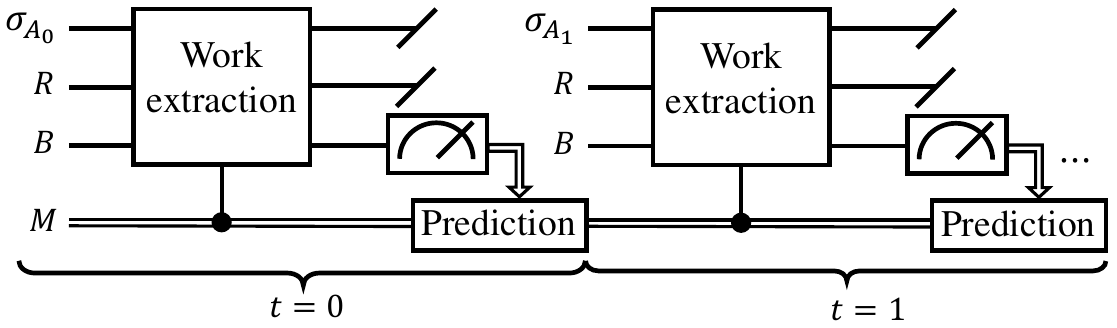}
\caption{Schematic diagram of the sequential work extraction.  At each time step, the process will take a quantum system, $\sigma_{Q_t}$, reservoir qubit, $R$, battery, $B$, and memory, $M$, as input. 
The `Work extraction' box
should be interpreted as 
a memory-dependent unitary. 
States of memory are recycled. The single wires represent quantum information being passed along, while the double wires represent classical information.}
\label{fig:schem_eng}
\end{center}
\end{figure}

\section{Choice of Policies}
\label{sec:local_action}
In the standard reinforcement learning literature, a general policy is defined as a sequence of conditional probability distributions $\Lambda=\{\Lambda_t\}_{t\in T}$ over the action set, given by:
\begin{equation}
\label{eq:general_policy}
    \Lambda_t(A_t|a_1,o_1,\ldots,a_{t-1},o_{t-1})~,
\end{equation}
where $A_{t}$ is the random variable representing the action taken at step $t$, with specific realizations $a_t$, while $o_t$ denotes the observed outcome at time $t$. As mentioned in Sec.~\ref{sec:belief}, in our setting, the agent retains memory only of its belief, which simplifies the policy definition to: 
\begin{equation}
\label{eq:belief_policy}
    \Lambda_t(A_t|\gbm\eta_{t-1})~.
\end{equation}
This is also known as a belief-based policy in the POMDP literature~\cite{arcieri2024pomdp}. It can be interpreted as choosing the appropriate state $\rho_t^*$ to tailor the protocol according to the agent’s belief. Accordingly, we define the set of all possible belief states as:
\begin{gather}
    \mathcal{K}\coloneqq \{\gbm \eta^{(i)}\}_i, \quad \gbm \eta^{(i)} \in \Delta^{n-1}\\
    \Delta^{n-1} = \{(p_1,p_2,\ldots,p_{n})\in \mathbb{R}^{n-1}\}
\end{gather}
where $n$ is the number of latent states of the HMM. The set $\Delta^{n-1}$ is a $n-1$-dimensional probability simplex, with $p_i\geq0 , \hquad\text{for}\hquad  i\in[n]$ and $\sum_{i=1}^np_i=1$. Each $p_i$ represents the probability of the system being in the $i$-th latent state of the HMM.

Constrained by sequential access and the absence of a quantum memory, the agent is restricted to a specific class of operations: a $\rho^*$-ideal protocol succeeded by a projective measurement on the battery to quantify the extracted work $W_t$ at each time step. Formally, each operation $\mathcal{W}$ is constructed from a global unitary $\mathcal{U}_{QRB}$—representing the $\rho^*$-ideal protocol—followed by a projective measurement $\mathcal{M}_B$ onto the energy eigenbasis of the battery:
\begin{equation}
    \mathcal{W}= (\mathcal{I}_Q\otimes\mathcal{I}_R\otimes\mathcal{M}_B)\circ\mathcal{U}_{QRB}~.
\end{equation}
Moving forward, we denote $\mathcal{W}_\rho$ as the state-specific operation wherein the underlying unitary $\mathcal{U}_{QRB}$ is uniquely optimized for a given state $\rho$.
We restrict to projective measurements in the energy eigenbasis rather than general POVMs, as any measurement misaligned with the eigenbasis may induce uncontrolled energy exchange between the system and the measurement device—an effect we seek to avoid for precise thermodynamic accounting~\cite{araki1960measurement}.

We therefore define the set of possible actions $\A$ to be within this class of operations, each tailored for a different state $\rho^{(j)}$ that has the same dimension as the potential outputs. More formally,
\begin{gather}
    \A \coloneqq \{\W_{\rho^{(j)}}\}_j\\
    \rho^{(j)} \in \mathcal{S}_d
\end{gather}
Note that both the inter-belief state transitions and the work extracted depend on the choice of policy. 

In this thesis, we primarily discuss two types of policies:
\begin{enumerate}
    \item \textbf{Local-optimizing, $\Lambda_{\text{LO}}$}: The agent selects actions based on its prediction of the expected emitted state $\xi_t$. 
    \item \textbf{Global-Optimum, $\Lambda^*$}: The agent selects actions based on an optimal policy obtained via dynamic programming.
\end{enumerate}
For the remainder of this chapter, we focus on the local-optimizing policy and reserve Chapter~\ref{ch:4} for a detailed discussion of the global-optimum policy. While the local-optimizing policy is not optimal, it carries physical interpretation: it describes an ``short-sighted" agent who only aims to extract the most work in the immediate time step with no concern for its future performance. On top of that, it yields interesting physical results and serves as a basis of comparison for the global-optimum policy.

An important caveat is that optimal policies are not necessarily time-independent. Even for a stationary process, one where the transition probabilities and rewards remain constant, the optimal policy can still be time-dependent in a finite-horizon setting (finite $T$), due to boundary effects. However, in the infinite-horizon limit ($T\to\infty$), the optimal policy becomes time-independent~\cite{sutton1998reinforcement,scherrer2012use}. In this chapter, we study a time-independent policy motivated by physical considerations. The distinction between finite and infinite horizons, along with the associated boundary effects, will be further discussed in Chapter~\ref{ch:4}.

We begin by discussing the local-optimizing policy. While this is \emph{not} the globally optimal policy, it has a clear physical intuition. At time $t$, the agent holds the belief $\gbm\eta_t$, from which it infers the expected emitted quantum state $\xi_t$. The agent can then choose to tailor the protocol specifically for this expected state. In principle, this enables the agent to extract more work than if it had tailored the protocol to $\xi_0$, the local reduced state of the multipartite quantum state in Eq.~\eqref{eq: fuel}—since $\xi_t$ typically contains more free energy than $\xi_0$. 

\begin{proposition}
    By self-consistency, the local reduced state of $\rho_{\overleftrightarrow A}$ must be a probabilistic mixture of $\sigma^{(x)}$, such that 
\begin{equation}
    \langle\xi_t\rangle_{K_{t-1}} = \sum_{x \in \mathcal{X}} \gbm\pi  T^{(x)}  \mathbf{1} \, \sigma^{(x)} = \xi_0~.
\end{equation} 
Since $\xi_0$ is the average of $\xi_t$ over belief states, the expected state $\xi_t$ has more free energy than $\xi_0$ on average, a consequence of the concavity of entropy and convexity of relative entropy. Specifically:
\begin{equation}
\label{eq:Nonneg}
\begin{split}
    \langle\text{D}(\xi_t \| \gamma)\rangle_{K_{t-1}} - \text{D}(\xi_0 \| \gamma)
&= 
\langle\text{D}(\xi_t \| \gamma)\rangle_{K_{t-1}} - \text{D}(\langle\xi_t\rangle_{K_{t-1}} \| \gamma ) \\
&= 
S(\langle\xi_t\rangle_{K_{t-1}}) - \langle S(\xi_t)\rangle_{K_{t-1}} \\
& \geq 0 
~,
\end{split}
\end{equation}
where, $\langle\text{D}(\xi_t \| \gamma)\rangle_{K_{t-1}}$ denotes the average relative entropy between $\xi_t$ and the Gibbs state $\gamma$,  taken over the belief state distribution at time $t$, as defined in Eq.~\eqref{eq:xi_t_expectation}.
\end{proposition}

By choosing the tailored state to be the expected state $\xi_t$, Eq.~\eqref{eq:update_pc2} simplifies to:
\begin{equation}
    \gbm\eta_{t+1} =
\frac{\sum_{x \in \mathcal{X} }
	\sum_n \delta_{ w_{t+1}, w^{(n)} } \bra{\lambda_n}  \sigma^{(x)} \ket{ \lambda_n} \, 
	\gbm\eta_{t} T^{(x)} }{
	\sum_n \delta_{ w_{t+1},  w^{(n)} } \lambda_n} ~,
\end{equation}
and Eq.\eqref{eq:BeliefStateTransitionProbs} becomes:
\begin{equation}
    \Pr(W_{t+1} = w | K_t = \gbm\eta_t )=\sum_n \lambda_n \, \delta_{ w,  w^{(n)}  }~.
\end{equation}
Combining Eqs.\eqref{eq:WorkExtractionValues} and \eqref{eq:BeliefStateTransitionProbs}, the expected work extracted at time step $t$ is given by:
\begin{equation}
\langle W_t \rangle =\sum_{\gbm\eta, w} w \Pr(K_{t-1} \! = \gbm\eta) \Pr(W_t = w | K_{t-1} \! = \gbm\eta)    ~,
\end{equation}
which simplifies to: 
\begin{equation}
\begin{split}
    \label{eq:xi_t_expectation}
    \langle W_t \rangle &=\sum_{\gbm\eta,i}\Pr(K_{t-1} \! = \gbm\eta)\lambda_i\left(\bra{\lambda_i}H\ket{\lambda_i}+\beta^{-1}\lambda_i\ln\lambda_i -F\right)\\
    &=\sum_{\gbm\eta,i}\Pr(K_{t-1} \! = \gbm\eta)\lambda_i\bra{\lambda_i}H\ket{\lambda_i}-\beta^{-1}S(\xi_t) -F \\
    & = \sum_{\gbm\eta}\Pr(K_{t-1} \! = \gbm\eta)\Tr(\xi_t(H-\mathcal{F}_{\text{eq}}\id))-\beta^{-1}S(\xi_t)\\
    &= \beta^{-1}\sum_{\gbm\eta}\Pr(K_{t-1} \! = \gbm\eta)[-\Tr(\xi_t\ln\gamma)+S(\xi_t)]\\
    &= \beta^{-1} \langle \text{D}(\xi_t \| \gamma)  \rangle_{\Pr(K_{t-1})} ~,
\end{split}
\end{equation}
The implications of this expression are significant: it shows that the average work extracted at any given time step depends on the instantaneous distribution over the agent’s belief states. Consequently, determining the agent’s work extraction rate requires characterizing the full meta-dynamics governing the evolution of these belief states. Drawing on prior studies of information ratchets~\cite{boyd2017transient}, we identify two distinct operational regimes for the agent: the transient regime and the recurrent regime.

The transient regime corresponds to the initial phase in which the agent is still synchronizing its belief with the latent states of the hidden Markov model (HMM). This phase is typically sensitive to the initial belief and is characterized by complex, history-dependent meta-dynamics. In contrast, the recurrent regime describes the long-term behavior after synchronization has occurred. In this regime, the belief dynamics become simpler and largely independent of initial conditions.
As the influence of the transient regime diminishes over time, our analysis focuses on the recurrent regime to determine the asymptotic work extraction rate.

\section{Meta-dynamics of belief}
\label{sec:metadynamic}
As established in Sec.~\ref{sec:local_action}, analyzing the asymptotic work-extraction rate requires a deeper understanding of the recurrent meta-dynamics of belief. We first observe that the Bayesian update in Eq.~\eqref{eq:update_pc2} can be interpreted as a nonlinear return map on the belief state—i.e., a nonlinear function that maps one belief state to another.
For the specific example of the Perturbed Coin illustrated in Fig.~\ref{fig:PC_diagram}(a), the belief states can be parametrized by a single variable as $\gbm\eta_t = (1/2+\epsilon_t,1/2-\epsilon_t)$, where $\epsilon_t\in [-0.5,0.5]$. Let $p$ denote the “flipping” probability of the coin, and $r$ denote the fidelity between the two possible quantum states $\sigma^{(0)}$ and $\sigma^{(1)}$. By varying these parameters, we can observe their impact on the update rule. Applying the update yields the return maps shown in Fig.~\ref{fig:update_map}.
\begin{figure}
    \centering
    \includegraphics[width=0.95\linewidth]{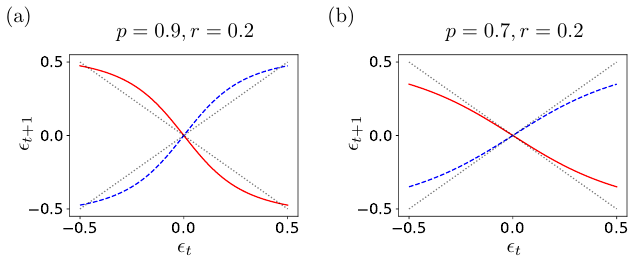}
    \caption{The update rule for different parameters. Panel (a) shows the update map for parameters $p=0.9,r=0.2$, while panel (b) shows the update for $p=0.7,r=0.2$. The red and blue solid lines represent the update functions corresponding to the two work values $\{w^{(i)}\}_{i=0,1}$, and the black dotted line represents the identity map.}
    \label{fig:update_map}
\end{figure}
The two colored lines (blue and red) correspond to the return maps conditioned on the outcome of the observation, while the black dotted line represents the identity map. Any point lying on the identity line maps back to itself under the update, meaning that such belief states remain unchanged—these are the recurrent belief states. 

Interestingly, varying the parameters $p$ and $r$ reveals two classes of return maps: one featuring two stable and one unstable equilibrium points (Fig.~\ref{fig:update_map}(a)), and another with only a single stable equilibrium point (Fig.~\ref{fig:update_map}(b)). Stability here refers to whether a small perturbation $\epsilon$ away from the equilibrium point results in convergence back to the point (stable), or divergence away (unstable), after repeated updates.

We begin by analyzing the case with only one stable equilibrium. The equilibrium point corresponds to the belief state $\gbm \eta_0 = (1/2,1/2)$, which is precisely the stationary distribution of the underlying HMM. This indicates that the agent gains no additional information over time—its belief does not evolve beyond the stationary prior. That is, the information extracted from past observations does not improve future predictions. The agent's predictive performance is thus equivalent to that of a memoryless agent. We refer to this regime as \emph{``memory-apathetic"}, since the agent’s performance is independent of whether or not it retains memory of past measurements.

On the other hand, in the regime with two stable and one unstable equilibrium points, the stationary point $\gbm\eta_0=\gbm\pi$ becomes unstable. Two new stable equilibrium points emerge as attractors. In this case, small perturbations away from $\gbm\pi$ cause the belief to evolve toward one of the two stable points, depending on the observed outcomes. The agent can therefore partially synchronize with the latent process, and its behavior differs from that of a memoryless agent. Based on our earlier analysis in Eq.~\eqref{eq:Nonneg}, this synchronization leads to an increase in the extracted work. We refer to this as the \emph{``memory-advantageous"} regime.

\section{Results}
\subsection{Benchmarking}
\label{sec:benchmark}
To demonstrate our \emph{memory-assisted quantum} approach, we employ a quantum work-extraction protocol tailored to the work-observation-induced expected state $\rho_t^* = \xi_t$. We apply this protocol to the quantum perturbed-coin process illustrated in Fig.~\ref{fig:PC_diagram}(a), and compute the agent’s asymptotic work extraction rate as was defined in Eq.~\eqref{eq:work_rate} under this policy (Approach~\ref{approach:1}). We then compare its performance against three alternative strategies (Approaches~\ref{approach:2}–\ref{approach:4}):
\begin{enumerate}
    \item 
    \label{approach:1}
    \emph{memory-assisted quantum}---the quantum work-extraction protocol is optimized for the work-observation-induced expected state $\rho_t^* = \xi_t$;
    \item 
    \label{approach:2}
    \emph{memory-assisted classical}--—the protocol cannot extract work from quantum coherence and is instead optimized for the energy-dephased version of the expected state, \(\rho_t^* = 
\xi_t^{\text{dec}} \coloneqq \sum_i \ket{E_i}\bra{E_i} \bra{E_i} \xi_t \ket{E_i}\), in some predetermined energy eigenbasis $\{\ket{E_i}\}_i$;

\item 
\label{approach:3}
\emph{memoryless quantum processing}—--the memory is never updated by observations, and the protocol is optimized for the time-averaged quantum state, \(\rho^* = \xi_0 = \sum_{x} \gbm\pi  T^{(x)}  \mathbf{1} \, \sigma^{(x)}\); and
\item 
\label{approach:4}
\emph{overcommitment to most probable quantum state}—--the protocol is optimized for \(\rho_t^* = \sigma^{(\text{argmax}_{x} \gbm\eta_t  T^{(x)}  \mathbf{1})}\).
\end{enumerate}

Note that Approach~\ref{approach:2} also corresponds to the amount of work extractable by an agent restricted to using thermal operations (with a fixed energy eigenbasis), as described in the resource-theoretic framework of Sec.~\ref{sec:thermal_ops}, where free energy stored in coherence remains inaccessible. In contrast, Approach~\ref{approach:3} represents the performance of an optimal non-adaptive agent—one that applies the same protocol at every time step, without updating its strategy based on past observations.

\begin{figure*}[htbp]
\begin{center}
\includegraphics[width=0.9\textwidth]{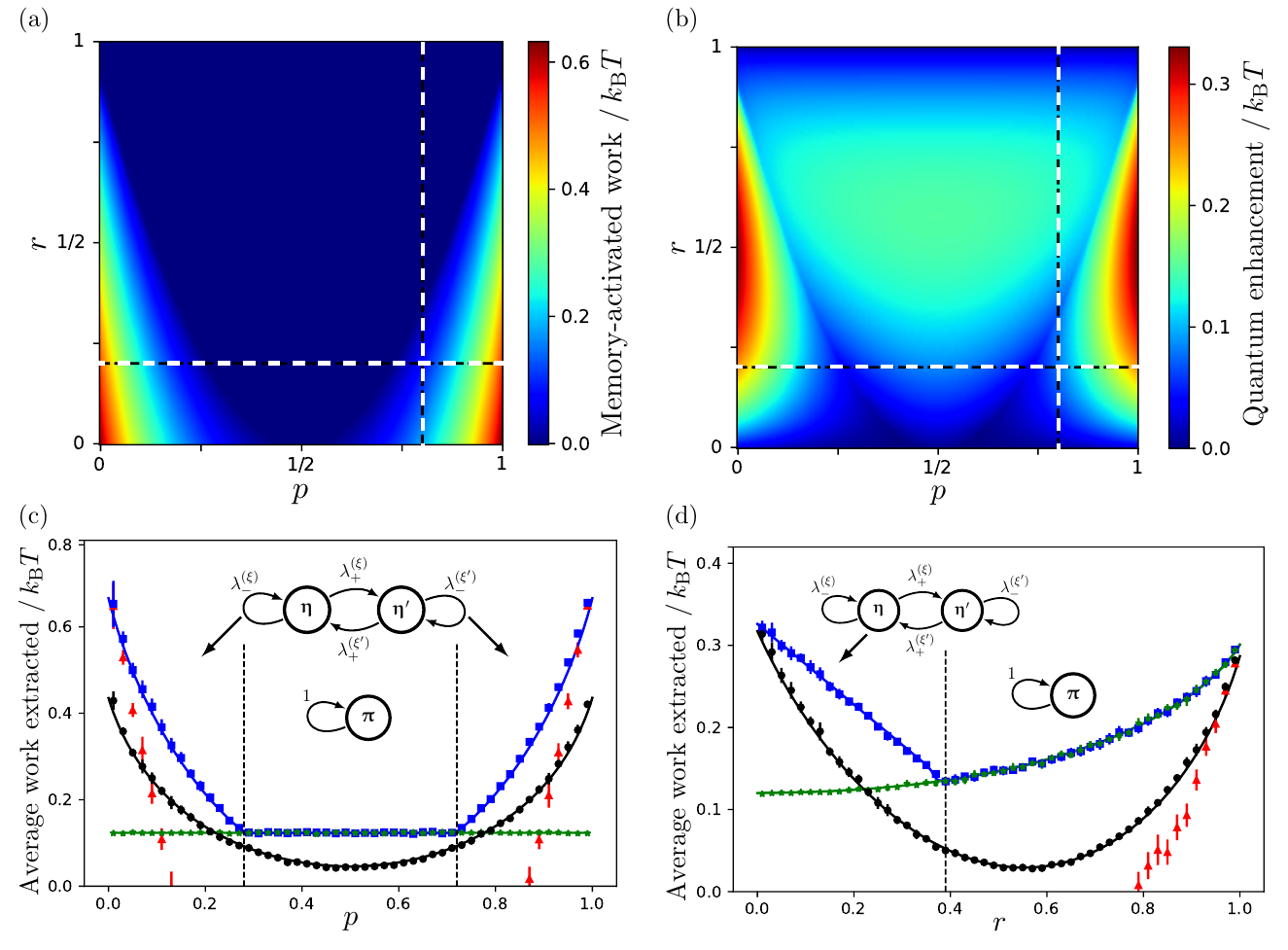}
\caption{Comparison of average work-extraction rates across different approaches.
The parameter $p$ characterizes the transition probability between the two latent states in the perturbed-coin process, while $r$ quantifies the overlap between the corresponding quantum outputs.
(a) illustrates the enhancement in work extraction due to memory, and (b) shows the quantum advantage in work extraction. Panels (c) and (d) present phase transitions in memory enhancement, shown via cross-sections of the parameter space, taken along the dashed black-white line. Solid lines indicate analytical results, while markers represent numerical simulations using the protocol shown in Sec.~\ref{sec:actual_protocol}.
Blue squares represent Approach~\ref{approach:1}; black circles, Approach~\ref{approach:2}; green stars, Approach~\ref{approach:3}; and red triangles, Approach~\ref{approach:4}.}
\label{fig:Deg}     
\end{center}
\end{figure*}
The asymptotic work-extraction rates from the different approaches are compared—both analytically and via numerical simulations—in Fig.~\ref{fig:Deg}, 
for the perturbed-coin process.
In Fig.~\ref{fig:Deg}(a), memory-activated work is defined as the difference between the memory-assisted quantum approach (Approach~\ref{approach:1}) and the memoryless quantum approach (Approach~\ref{approach:3}).
In Fig.~\ref{fig:Deg}(b), quantum enhancement is defined as the difference between the memory-assisted quantum approach (Approach\ref{approach:1}) and the memory-assisted classical approach (Approach~\ref{approach:2}).

As shown in Fig.~\ref{fig:Deg}, our memory-assisted quantum approach (Approach\ref{approach:1}) consistently performs at least as well as all other approaches, and strictly outperforms them across large regions of parameter space.
In these examples, the ideal input state $\rho_t^*$
is a qubit density matrix with eigenvalues $\lambda_{\pm}^{(\rho_t^*)}$, where $\lambda_+^{(\rho_t^*)} \geq\lambda_-^{(\rho_t^*)} \geq 0$.
According to Eqs.~\eqref{eq:WorkExtractionValues} and \eqref{eq:WorkExtractionProbs}, each distinct belief state therefore leads to one of two possible work values, $w^{(\pm)}$.

Approaches~\ref{approach:1}–\ref{approach:3} share several useful properties. From Eq.~\eqref{eq:WorkExtractionProbs}, assuming distinct work values $w^{(+)} \neq w^{(-)}$, the probability of observing each possible work value is directly given by the corresponding eigenvalue of the optimal input state:
\begin{align}
\Pr(W_{t+1} = w^{(\pm)} | K_t = \gbm\eta_{t}) = \lambda_{\pm}^{(\rho_t^*)} ~.
\end{align}
Combining this with Eq.~\eqref{eq:WorkExtractionValues}, the expected work extracted at each time step can be written as
\begin{align}
\label{eq:rhostar_expectation}
    \langle W_t \rangle =
    \beta^{-1} \langle \text{D}[\rho_t^* \| \gamma]  \rangle_{\Pr(K_t)}~,
\end{align}
which holds for Approaches~\ref{approach:1}–\ref{approach:3}. Note, however, that both $\rho_t^*$  and the distribution over belief states differ between these approaches: specifically, $\rho_t^* = \xi_t$ for Approach~\ref{approach:1}, $\rho_t^*=\xi_t^\text{dec}$ for Approach~\ref{approach:2}, and $\rho_t^*=\xi_0$ for Approach~\ref{approach:3}.
The non-negativity of the relative entropy ensures that the expected work extraction remains non-negative in all three cases. Moreover, Eq.~\eqref{eq:rhostar_expectation} generalizes Eq.~\eqref{eq:xi_t_expectation}, making it possible to compare the work-extraction performance across the three approaches. Additional details on the analytic derivation of expected work can be found in Appendix~\ref{app:analytic_work}.

Without memory or adaptivity, the extractable structure is limited to the time-averaged statistical bias of the output~\cite{Boyd17PRE}, which explains why the memoryless work extraction depends on $r$ but not $p$. Although it might seem intuitive to commit to the most likely outcome, the overcommitment approach (Approach~\ref{approach:4}) performs the worst. This is because any reset operation to $\gamma$ with minimal entropy production for a pure-state input causes infinite heat dissipation when applied to any other input~\cite{riechers2021initial, riechers2021impossibility}.
Consequently, this leads to infinite negative work extraction, $\langle W_t\rangle = -\infty$ in case where 
$\rho_t^* \in \{ \ket{0} \! \bra{0} , \, \ket{\psi} \! \bra{\psi} \}$.
This divergence can also be understood via Eq.~\eqref{eq:WorkExtractionValues}, since $w^{(-)} \sim \ln \lambda_- \to \ln 0 = -\infty$.
In our numerical simulations, following Ref.~\cite{SSP14}, the minimal eigenvalue $\lambda_-$
is inversely proportional to the number $M$ of bath interactions, causing the overcommitment work penalty to diverge as $-\beta^{-1} (1-r) \min( p, 1-p ) \ln M$.

\subsection{Phase transition}
As seen in Fig.~\ref{fig:Deg}, there exists a blue inner region of panel \ref{fig:Deg}(a) where the memoryless quantum approach \ref{approach:3} achieves the same performance
as our memory-assisted quantum approach \ref{approach:1}. 
\emph{There exists a sharp phase boundary within which the use of memory does not boost performance.} 
As seen clearly in panels \ref{fig:Deg}(c) and \ref{fig:Deg}(d), 
the phase boundary exhibits a discontinuity in the first derivative of work extraction with respect to process parametrization.
This phase boundary is not unique to the perturbed-coin process and, indeed, also occurs
in the 2-1 golden-mean process in Fig.~\ref{fig:PC_diagram}(b). This is consistent with the prediction we saw in Section~\ref{sec:metadynamic}. We summarized the meta-dynamic of the agents using different approaches (\ref{approach:1},\ref{approach:2},\ref{approach:3}) in Table.~\ref{tab:meta}. 
We focus our discussion on the first column of the table, as we previously predicted, the performance of the agent can be classified into 2 different regimes, the ``memory-apathetic" and the ``memory-advantageous" regime.

\begin{sidewaystable}[htbp]
    \begin{center}
    \caption{Summary of meta-dynamics in different regimes. The update function shows the nonlinear relationship between $\epsilon_t$ and $\epsilon_{t+1}$. The belief evolution shows the evolution of $\epsilon_t$ over iterations, which gives rise to the corresponding work series, with two possible work values per belief state. The recurrent belief states show the recurrent meta-dynamic of the different regimes.}
    \label{tab:meta}
\includegraphics[width=1\textwidth]{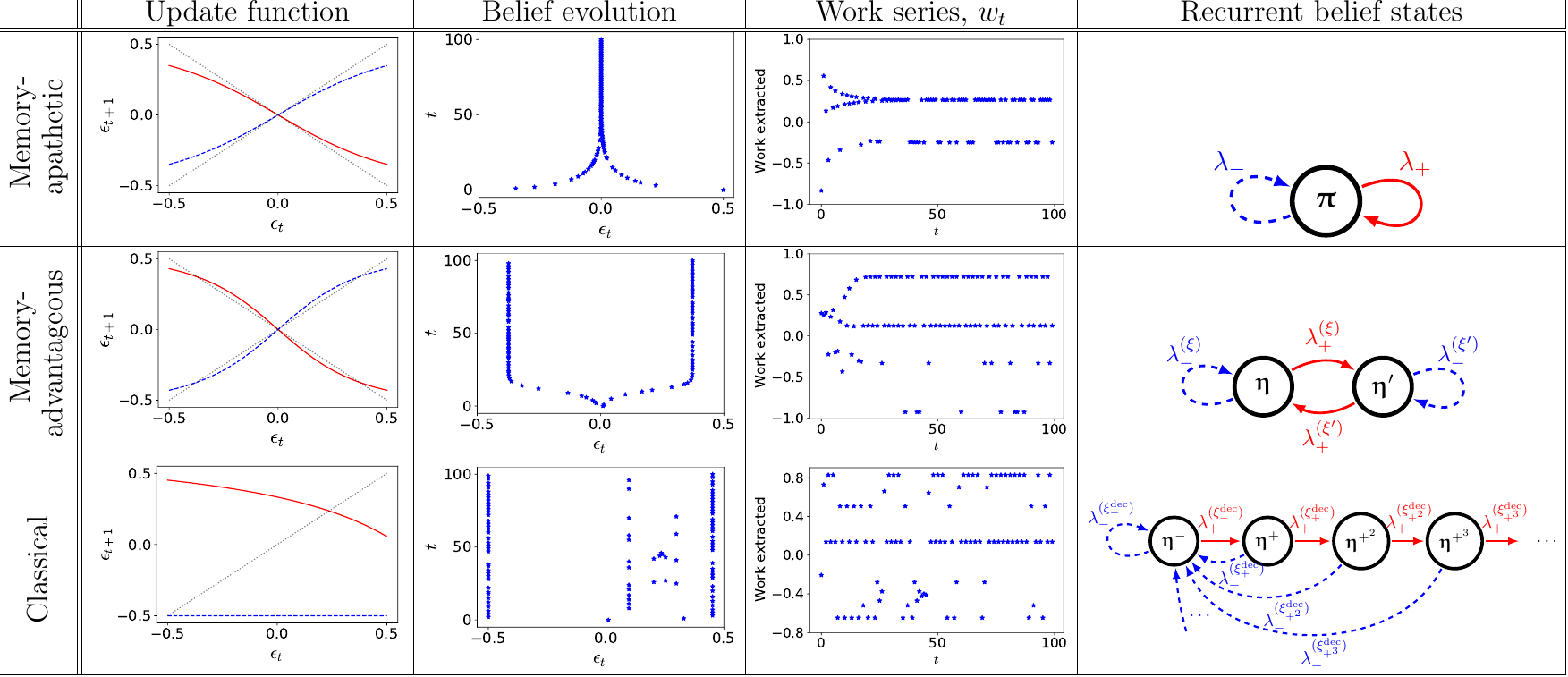}
\end{center}
\end{sidewaystable}
At the phase boundary
in the process' parameter space,
the fixed point at $\gbm\pi$
becomes unstable,
and new attractors emerge for each map.
However, the coexistence of the maps introduces competition between the attractors, as the maps are selected stochastically with probabilities $\lambda_{\pm}$.
The two maps (red solid and blue dashed) interact to induce a steady-state meta-dynamic over recurrent belief states, shown as a Markov process in each row of Table \ref{tab:meta}'s last column.
In this memory-advantageous region,
work extraction supplies sufficient evidence to 
inform an agent about the hidden state of the process, which in turn allows for more work to be extracted. Note that such a phase transition is not unique to the Perturbed Coin process; we have found a similar phenomenon for more complicated processes, such as the ``2-1 Golden Mean" process shown in Fig.~\ref{fig:PC_diagram}(b). Whether this exists for every process remains an open question.

The elegance of memory-assisted quantum work extraction 
is reflected in the simple one and two-state recurrent memory structures.
In comparison, the classical extractor (Approach \ref{approach:2}) not only harvests less work, but requires more memory to achieve its relatively meager returns.
Note the infinite number of recurrent memory states in the  classical-processing case,
in the last row of Table.~\ref{tab:meta}.
The return maps in the first column Fig.~\ref{tab:meta}, as well as the belief evolution in the second column, also demonstrate that the asymptotic behaviour of the agents is independent of the initial conditions; regardless of where the initial belief state is, the return map eventually moves them to a stable equilibrium. The work series in the third column also validates our claim that for each distinct belief state, there exist 2 possible work values, $w^{(\pm)}$.
\section{Limited quantum memory}
We now examine the implications of equipping the agent with a finite quantum memory, enabling the storage and coherent processing of $L$-qubit blocks. Operating on an $L$-qubit block strictly increases the work yield, as it allows the agent to extract coherence that is otherwise inaccessible. This enhancement is evident in Fig.~\ref{fig:finite_memory}, which demonstrates that introducing even a single qubit of quantum capacity universally increases the work extraction, irrespective of whether the agent utilizes a classical memory policy. Notably, the phase boundary separating the memory-advantage and memory-apathetic regimes remains strictly invariant under this extension. This invariance is a direct consequence of the Bayesian update rule (Eq.~\eqref{eq:update_pc2}), which depends exclusively on the eigenbasis of the tailored state and is independent of its eigenvalues. Because the measurement basis for an $L$-qubit block factorizes into the local bases of the individual subsystems, the classical information acquired from a joint $L$-qubit measurement is entirely equivalent to that obtained from $L$ sequential local measurements. Consequently, while a finite quantum memory provides a distinct thermodynamic advantage by unlocking temporal coherence, it offers no supplemental informational advantage for state prediction. The addition of quantum memory thereby improves the physical efficiency of the extraction process without altering the agent's fundamental predictive power.

\begin{figure}
    \centering
    \includegraphics[width=0.8\linewidth]{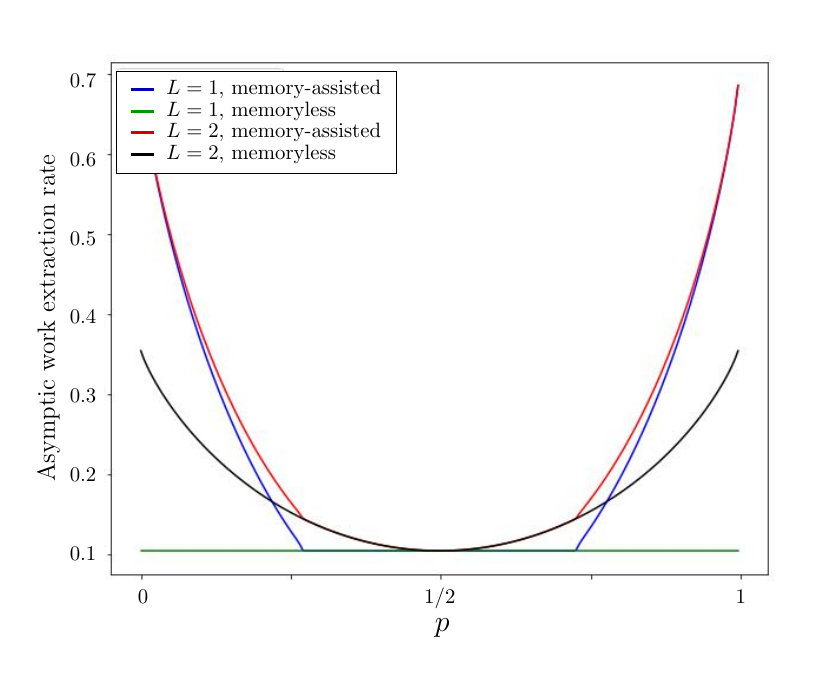}
    \caption{Comparison of the asymptotic work extraction rate of agents with varying block-length $L$, both memory-assisted and memoryless. The work extraction rate no doubt increased when we consider higher $L$, but notice that the phase boundary memory-advantageous and memory-apathetic region remains invariant with $L$. }
    \label{fig:finite_memory}
\end{figure}
\section{Summary}
In summary, now we have established a physically motivated method of extracting work from temporally correlated quantum states. Namely, tailoring the protocol to the state that the agent expects based on its belief. Clearly, this is not the only way to carry out the extraction, nor did we claim that it is the optimum. In fact, we hypothesize that there could exist other protocols that achieve higher overall work extraction over $T$ time steps but at each step incur greater local dissipation. How then can one reach the optimal method to extract work from such quantum states? We will discuss this in the next chapter and also coin a new term ``Time-ordered free energy" as an achievable upper bound on such work.

%% file: Chapters/Chapter5.tex

\chapter{Time-ordered free energy} 
\label{ch:4}
\noindent\rule{\textwidth}{0.4pt}
\vspace{1.5em} 
\setlength{\parindent}{4ex}

\noindent \textit{In this chapter, we show that the previous approach to extract work from temporally correlated quantum states can be further optimized. We make use of dynamic programming to reach the optimal work extraction rate and provide proof for the optimality claim. This can be interpreted as the maximal free energy an agent can extract constraint to obeying causality. This new free energy quantity is dubbed as \emph{``Time-ordered Free Energy"}.}
\newpage

\section{Choice of policies, revisited}
\label{sec:opt_action}
As discussed in Sec.~\ref{sec:local_action}, the agent has the freedom to choose an action based on its belief state; this choice was described as a policy. In \cref{ch:3}, we analyzed the agent’s performance under a local-optimizing (LO) policy. In this chapter, we address the question: What is the optimal policy that enables the agent to extract the maximal amount of work from these temporally correlated quantum states?

To address this, we utilize a technique from reinforcement learning—dynamic programming (DP), specifically backward induction. Before explaining how DP applies to our work-extraction problem, we first provide a brief review of dynamic programming in Sec.~\ref{sec:DP}.
In its original formulation for Markov decision processes (MDPs), actions are functions of the system’s state at time $t$. In our setting, however, the agent must condition its actions on its belief state $\gbm \eta_t$, since the true underlying quantum state $\sigma_{Q_t}$ is not directly accessible. Our goal is to determine the optimal policy $\Lambda^*$ that maximizes the work extracted. that maximizes work extraction. Formally, the maximum extractable cumulative work over $T$ time steps is given by
\begin{equation}
\label{eq:max_work}
\max_{\Lambda} \mathbf{W}(T,\Lambda)=\max_{\Lambda}\mathbb{E}_{\Lambda}\sum_{t=1}^T \mathbb{E}(W_t|A_t=\Lambda(K_{t-1}))
\end{equation}
and its asymptotic rate 
\begin{equation}
\label{eq:max_rate}
\max_{\Lambda} \lim_{T\to\infty}\frac{1}{T}\mathbf{W}(T,\Lambda)=\max_{\Lambda}\lim_{T\to\infty}\frac{1}{T}\mathbb{E}_{\Lambda}\sum_{t=1}^T \mathbb{E}(W_t|A_t=\Lambda(K_{t-1}))
\end{equation}

\section{Dynamic programming - DP}
\label{sec:DP}
Before we apply dynamic programming, we will provide a short overview of what dynamic programming is and how optimization can be carried out. Dynamic programming is a mathematical optimization technique that recursively decomposes a complex problem into a sequence of simpler, interdependent subproblems. This recursive structure enables efficient computation by avoiding redundant calculations, often at the cost of increased memory usage to store intermediate results. In the context of general stochastic optimal control problems, dynamic programming methods have been shown to yield optimal solutions~\cite{bellman1965dynamic}.

A stochastic optimal control problem involves determining a control strategy—interpreted as an agent's sequence of actions—that optimizes an objective function, typically expressed as the expected cumulative reward or cost, over a stochastic process evolving through time. This process extends a standard stochastic trajectory, $\cdots, X_{t-1}, X_{t}, X_{t+1},\cdots$ by incorporating actions $A_{t}$ at each time step, where each action influences the evolution of future states. The resulting controlled process can thus be represented as
\begin{equation}
    \cdots, X_{t-1}, A_{t}, X_{t}, A_{t+1}, X_{t+1},\cdots
\end{equation}
with specific realizations denoted by $\cdots, x_{t-1}, a_{t}, x_{t}, a_{t+1}, x_{t+2}, \cdots$.

We consider a Markovian process with a discrete, finite time horizon $t=0,1,\cdots, T$, where $T \in \mathbb{Z}^{+}$. At each time step $t$, the environment occupies a state $X_{t}\in\mathcal{X}=\{x^{(i)}\}_i$ and an actions $A_{t}\in\A=\{a^{(j)}\}_j$ is selected. Here, subscripts such as $t$ denote temporal indices, while superscripts in parentheses, e.g., $x^{(i)}$, indicate specific realizations within the corresponding set. For example, we denote the specific state realization time as $x_t^{(i)}$. The evolution of the system is governed by the Markov property: the random variable $X_{t}$ depends solely on the previous state $X_{t-1}$ and the previous action $A_{t}$, i.e.,
\begin{equation}
    \Pr(X_t|X_{t-1},A_{t}\cdots,A_1,X_0) = \Pr(X_t|X_{t-1},A_{t})~.
\end{equation}
The action $a_{t}$ at any time $t$ is determined by the current state $x_{t-1}$, and we may write this explicitly as $a_{t}(x_{t-1})$, though we often use the shorthand $a_t$ for brevity. as it is only dependent on the state $x_{t-1}$. Similarly, the reward function at each time step is denoted by $r( x_{t-1},a_{t}, X_{t})$, reflecting its dependence on the current state, $x_{t-1}$, the action taken $a_t$, and the resultant next state $X_{t}$.

Given this framework, one can define the objective function to be optimized as
\begin{equation} \label{eq:obj-func}
        J(x_{0},A_{1:T}) = \mathbb{E} \left[ r(x_0,A_1,X_1)+\sum^{T-1}_{t=1} g^{t} \cdot r( X_{t}, A_{t+1}, X_{t+1} ) \right]
\end{equation}
where $x_{0}$ is the initial state for $t=0$, $X_{t}$ and $A_{t}$ are random variables for all $t \geq 1$. The expectation is taken over the stochastic trajectories of the process. The parameter $g\in (0,1]$ is a discount factor that scales the contribution of future rewards. This factor is typically introduced in infinite-horizon problems ($T\to\infty$) to ensure convergence of the objective function and to model time preference, where future rewards are valued less than immediate ones.
The choice of $g$ reflects the trade-off between short-term and long-term gains: setting $g\to0$ prioritizes immediate rewards exclusively, while $g=1$ assigns equal weight to rewards across all time steps. 

The value function is defined as the maximized objective function over all admissible control sequences,
\begin{equation} \label{eq:mark-val-func}
    V(x_{0}) := \sup_{\{A_{t}\}^{T}_{t=1} \in {\mathcal{A}}} J(x_{0},A_{1:T})
\end{equation}
This formulation seeks the optimal control sequence that maximizes the expected cumulative reward, given an initial state $x_0$.
However, directly solving this optimization problem is generally computationally intractable due to the high dimensionality: it involves optimizing over $T$ decision variables and accounting for the joint distribution of $2T$ random variables (states and actions). To address this complexity, dynamic programming can be employed, in particular, the backward induction approach.

\subsection{Backward induction}
In the backward induction approach, the objective is to determine the optimal policy defined in Eq.~\eqref{eq:general_policy}, such that the objective function defined in Eq.~\eqref{eq:obj-func} is maximized, thereby attaining the optimal value function in Eq.~\eqref{eq:mark-val-func}. However, the \emph{environment} is now a fully-observable agent---the stochastic trajectory, $\{X_t\}_t$, is a classical variable; this, along with the Markovian assumption, allows us to rewrite the policy as
\begin{equation}
    \pi_t(A_t|X_{t-1})~.
\end{equation}
This optimization is achieved by recursively solving subproblems starting from the final time step and proceeding backward to the initial state. 
At each time step, the policy prescribes the action $A_t$ to be taken given the current state $X_{t-1}$. 

We now define a time-dependent version of the objective function, $J_{t} (x_{t-1}, a_{t})$, and the corresponding value function $V_{t} (x_{t-1})$ for each time step $t\in\{1,\cdots, T\}$. This function $J_{t} (x_{t-1}, a_{t})$ represents the expected rewards starting from a given state $x_{t-1}$ and action $a_{t}$ at time $t$, incorporating both the immediate and future rewards. It is defined as
\begin{equation} \label{eq:mark-obj-func-t}
    J_{t} (x_{t-1}^{(i)}, a_{t}^{(j)}) := \mathbb{E} \left[ r\left(x_{t-1}^{(i)},a_{t}^{(j)},X_{t}\right) + V_{t+1}(X_{t}) \big|X_{t}\right]~.
\end{equation} 
where the expectation is taken over the distribution of the next state $X_{t}$ given the current state-action pair $(x^{(i)}_{t-1},a_t^{(j)})$.
The corresponding value function at time $t$ is obtained by maximizing the objective over all admissible actions:
\begin{equation}
    V_{t} (x_{t-1}) = \max_{a_{t}^{(j)} \in \mathcal{A}} J_{t}(x_{t-1}^{(i)}, a_{t}^{(j)})~,
\end{equation} 
and the optimal action at time $t$ for state $x_{t}^{(i)}$ is given by 
\begin{equation}
    a^*_{t}(x^{(i)}_{t-1}) := \argmax_{a_{t}^{(j)} \in \mathcal{A}} J_{t}(x_{t-1}^{(i)}, a_{t}^{(j)})~.
\end{equation}
The backward induction algorithm proceeds as follows: we initialize the terminal value function by setting $V_{T}(x^{(i)}_{T-1})=0$ for all $x^{(i)}_{T-1}$, reflecting the assumption that no further rewards are obtained beyond the final time step $T$.
Starting from $t=T-1$,  we compute $J_{t} (x_{t-1}^{(i)}, a_{t}^{(j)})$ for all state-action pairs and derive the value function $V_t(x_{t-1}^{(i)})$ by maximizing over actions. Simultaneously, we record the optimal action $a^*_{t}(x_{t-1}^{(i)})$ corresponding to each state. This procedure is repeated recursively backward in time until $t=0$.
The output of the algorithm is the optimal policy, $\pi^*$, which achieves the maximum expected cumulative reward in  Eq.~\ref{eq:mark-val-func}. 

\section{DP for work-extraction}
\label{sec:DP_work}
In order to utilize DP, we have to first define a reward function. In terms of work-extraction, the reward will naturally be the amount of work extracted, $W_t$, at each time step $t$. The actual work extracted at each time $W_t$ is inherently probabilistic as both the identity of the next emitted state and the measured work are stochastic quantities. Consequently, the average reward defined in terms of its expectation over these underlying random work values is a better quantity to maximize. When a work extraction protocol is applied to an arbitrary state, the resulting extracted work can be modeled as a random variable, $W_t( K_{t-1}, A_t)$, governed by a probability distribution, $\mathcal{D}$. We assume that $\mathcal{D}$ depends solely on $K_{t-1}$ and $A_t$. This allows us to define the average reward
\begin{equation}
 \label{eq:DP_reward}
      r(K_{t-1},A_t) \coloneqq \Ex_{\mathcal{D}}\left(W_t(K_{t-1},A_t)\right)~.
\end{equation}
Note that during the actual extraction, the extracted work $W_t$ also depends on the quantum state $\sigma_{Q_t}^{(x_t)}$. However, during optimization, we can only model this dependency using the belief state $K_{t-1}$ since it most accurately reflects the agent's inference of the underlying distribution of latent state given the past outcomes. A schematic diagram of the evolutions and relations between the process (environment) and the agent can be found in Fig.~\ref{fig:Bayes_Net}

Recall that for any work extraction protocol aimed at extracting all non-equilibrium free energy from a quantum state $\rho$. The protocol has to be tailored for $\rho$, and if it is applied to a state $\sigma \neq \rho^*$, we will see a dissipation of \emph{at least} $\D(\sigma\|\rho^*)$~\cite{riechers2021initial}. 
Since the agent does not know the quantum states, the best it can do is rely on its belief to predict the distribution $P_x$, over $\{\sigma^{(x)}\}_{x\in\X}$.
Eq.~\eqref{eq:DP_reward} can then be written as 
\begin{equation}
\label{eq:final_rewards}
\begin{split}
      r(\gbm \eta^{(i)},a^{(j)}) &= \beta^{-1}\left[\sum_{x\in\mathcal{X}} P_x^{(i)}   \D(\sigma^{(x)}\|\gamma) - \D(\sigma^{(x)}\|\rho^{(j)})\right] \\
     &=  \beta^{-1}\left[\D(\xi_{\gbm\eta^{(i)}}\|\gamma) - \D(\xi_{\gbm\eta^{(i)}}\|\rho^{(j)})\right]~,
\end{split}
\end{equation}
 where we defined the expected state, $\xi_{\gbm\eta^{(i)}}\coloneqq \sum_{x\in\X} P_x^{(i)}\sigma^{(x)}$ , where $P_x^{(i)}=\gbm\eta^{(i)}T^{(x)}\mathbf{1}$ is the prior distribution over $\{\sigma^{(x)}\}_i$ induced by the belief state $\gbm\eta^{(i)}\in\K$. 
The form of this particular reward matches precisely with the expected work extracted from a $\rho^*$-ideal protocol that is tailored for $\rho^{(j)}$ applied onto the expected state $\xi_{\gbm\eta^{(i)}}$, and thus it makes sense for us to consider this class of operation. 
Conversely, one can always try to use different formalisms of thermodynamic such as thermal operations, but as we have already discussed in Sec.~\ref{sec:benchmark}, thermal operations can be re-cast into this formalism by tailoring the protocol to a state that is diagonalized in the energy eigenbasis of the Hamiltonian. Here, we wish to investigate the ultimate upper bound, which thermal operation will not be able to achieve due to its inability to extract work from coherence with respect to the energy eigenbasis.

To better understand the expression in Eq.~\eqref{eq:final_rewards}, observe that the second term captures the expected dissipation incurred by the agent when performing an operation tailored to the state $\rho^{(j)}$. A natural but premature conclusion might be that maximizing the reward function simply entails minimizing this dissipation term—achieved when $\rho^{(j)}=\xi_{\gbm\eta^{(i)}}$. Indeed, this strategy was precisely adopted in Chapter~\ref{ch:3}. However, this approach is not globally optimal. 
The reason lies in the fact that the agent’s action at time step $t$ affects not only the immediate reward but also the potential rewards in subsequent steps. This temporal interdependence arises through the Bayesian update rule governing the evolution of belief states, as given in Eq.~\eqref{eq:update_pc2}. In particular, the update mechanism depends on the eigenbasis of the tailored state, which is itself determined by the chosen action. Consequently, choosing a tailored state $\rho^{(j)} \neq \xi_{\gbm\eta^{(i)}}$ might reduce the immediate reward but can strategically shape future belief states in ways that enhance long-term rewards.

From this perspective, the divergence term $\D(\xi_{\gbm\eta^{(i)}}\|\rho^{(j)})$ quantifies the trade-off between maximizing short-term gains and fostering favorable conditions for future reward accumulation. 

Following the reward, we can define a global value function 
\begin{equation}
\label{eq:overall}
    \tilde{V_1}\left(\gbm\eta^{(i)},A_{1;T}\right) \coloneqq \mathbb{E}\left[\sum_{t=1}^{T} r(K_{t-1},A_t) \Bigg|K_0=\gbm\eta^{(i)}\right]~,
\end{equation}
where $A_{1;T}\coloneqq A_1,A_2\cdots A_{T}$. Physically, this expression represents the average work extracted by an agent over $T$ time steps starting with a belief state $\gbm\eta^{(i)}$. The expression in Eq.~\eqref{eq:max_work} can then be rewritten a maximization, 
\begin{equation}
\label{eq:overall_value}
    \tilde V^*_1 \left(\gbm\eta^{(i)}\right)\coloneqq \max_{A_{1;T}\in\A} \tilde{V_1}\left(\gbm\eta^{(i)},A_{1;T}\right)~,
\end{equation}
which reflects the maximization of the average work extracted over all possible sequences of action $A_1\cdots A_T$.

We use backward induction to perform the maximization, which guarantees the optimal policy, $\pi^*$. The pseudocode for the algorithm can be found in Alg.~\ref{alg:DP_for_work}  
\begin{algorithm}[htbp!]
	\caption{\textsf{Pseudocode for work extraction}} 
    \label{alg:DP_for_work}
    \textbf{input:} discretized set of belief states $\mathcal{K}\coloneqq \{\gbm \eta^{(i)}\}_{i=1}^{N}$, set of possible operations $\A \coloneqq \{a^{(j)}\}_{j=1}^{M}$, reward function $r(\gbm\eta_{t-1},a_t)$, the underlying HMM, and the total time $T$ of the process.
    \\

    \emph{Set final value to 0 as there are no more future rewards} \\
    Set $V^{*}_{T+1}(\gbm \eta^{(i)}_{T}) \coloneqq 0$ for all $\gbm \eta^{(i)}$ \\

    \For {$t = T, T-1, \cdots, 1$}{
    \For {belief state $\gbm \eta^{(i)}_{t-1} \in \mathcal{K}$ from $i=1$ to $N$}{
    \For {operation $a_{t}^{(j)} \in \A$ from $j=1$ to $M$}{
    
    \emph{Find next state following Theorem~\ref{thm1}} \\
    
    Find $\gbm \eta^{(k)}_{t}$ from $\gbm \eta^{(i)}_{t-1}$ and $a_{t}^{(j)}$ \\

    \emph{Calculate expected value for each action} \\
    $V_{t}(\gbm \eta^{(i)}_{t-1}, a^{(j)}_{t}) \coloneqq r(\gbm \eta^{(i)}_{t-1},a^{(j)}_t) + V^{*}_{t+1}(\gbm \eta^{(k)}_{t})$
    }   

    \emph{Find best value and action from all actions} \\
    Set $V^{*}_{t}(\gbm \eta^{(i)}_{t-1}) \coloneqq \max_{a^{(j)}_{t}}{V_{t}(\gbm \eta^{(i)}_{t-1}, a^{(j)}_{t})}$ \\
    Set $a^{*}_{t}(\gbm \eta^{(i)}_{t-1}) \coloneqq \argmax_{a^{(j)}_{t}}{V_{t}(\gbm \eta^{(i)}_{t-1}, a^{(j)}_{t})}$
    }
    }

    \textbf{output:} policy $\pi = \{ \mathcal{K}_{t-1}\to \A_{t} \}_{t=1}^{T}$
\end{algorithm}

\section{Search Space}
Unlike the previous approach in \cref{ch:3}, where the actions chosen always fall within the convex sum of the emitted quantum state $\{\sigma^{(x)}\}_{x\in\mathcal{X}}$. There is no reason why the optimal state for tailoring of the protocol should also fall within the convex sum. Instead, the global maximization is obtained over all possible states within the Bloch sphere; it is easy to see that such optimization is not feasible or at least extremely time-consuming. Instead, we will utilize the following theorem to narrow down the search space. 
\begin{theorem}
    Given any set of eigenbasis $\{|\psi_i^{(\gbm\theta)}\rangle\}_i$, parametrized $\gbm\theta$,
    The optimal state to tailor the protocol for is then given by:
    \begin{equation}
    \label{eq:opt_state}
        \rho^* = \sum_i\lambda_i\ket{\psi_i}\!\bra{\psi_i}, \hquad \lambda_i=\bra{\psi_i}\xi_{\gbm\eta^{(k)}}\ket{\psi_i}~,
    \end{equation}
    where $\xi_{\gbm\eta^{(k)}}$ is the expected state given a belief state $\gbm\eta^{(k)}$
\end{theorem}

\begin{proof}
Notice that the Bayesian update in Eq.~\eqref{eq:update_pc2} only depends on the eigenstates of the tailored state and the actual quantum state, i.e., $\bra{\psi_i}\sigma^{(x)}\ket{\psi_i}$, it has no dependence on the eigenvalues $\lambda_i$ of the tailored state. This means that the future inference and updating of belief states are independent of the eigenvalues of the local tailored state; hence, we have the freedom to choose the eigenvalues to minimize any local dissipation. Accordingly, based on the reward function that the agent aims to maximize in Eq.~\eqref{eq:final_rewards}, in other words, the agent aims to minimize expected local dissipation, $\D(\xi_{\gbm\eta^{(i)}}\|\rho^{(j)})$. We can rewrite this expression as
\begin{equation}
    \begin{split}
        \D(\xi_{\gbm\eta^{(i)}}\|\rho^{(j)}) & = -\tr(\xi_{\gbm\eta^{(i)}}\ln\rho^{(j)}) - S(\xi_{\gbm\eta^{(i)}})\\
        & = -\sum_i \ln \lambda_i \bra{\psi_i}\xi_{\gbm\eta^{(i)}}\ket{\psi_i}-S(\xi_{\gbm\eta_t})~,
    \end{split}
\end{equation}
notice that the entropic term is fixed for any given belief state, so we only have the freedom to minimize the first term over all possible $\{\lambda_i\}_i$, this can be done via Lagrange multipliers by imposing constraints on the eigenvalues or recognizing that this term resembles the definition of Shannon entropy and hence will be maximized when $\lambda_i = \bra{\psi_i}\xi_{\gbm\eta_i}\ket{\psi_i}$. Hence given any set of eigenbasis $\{\ket{\psi_i}\}_i$, the optimal state to tailor the extraction protocol to is 
\begin{equation}
        \rho^* = \sum_i\lambda_i\ket{\psi_i}\!\bra{\psi_i}, \hquad \lambda_i=\bra{\psi_i}\xi_{\gbm\eta^{(k)}}\ket{\psi_i}
    \end{equation}
\end{proof}
On top of this, if the HMM emits only two distinct quantum states,  the search for an optimal measurement basis can be significantly simplified. Specifically, we need only consider bases parameterized by
\begin{equation}
     \ket{\psi_\theta}=\cos\theta/2 \ket{0_{\phi_0}}+\sin\theta/2\ket{1_{\phi_0}}~,
\end{equation}
where $\phi_0$ denotes the angle between the plane $\mathbf{P}$—spanned by the Bloch vectors of $\sigma^{(0)}$ and $\sigma^{(1)}$—and the plane spanned by the computational basis. The rotated basis states $\ket{i_{\phi_0}}$ are defined via the unitary transformation $U_{\phi_0}$ acting on the computational basis: $\ket{i_{\phi_0}} = U_{\phi_0} \ket{i}$.

This constraint reduces the search space from the full Bloch sphere (a three-dimensional space) to a one-dimensional subspace lying within the plane $\mathbf{P}$. The optimal states defined in Eq.~\eqref{eq:opt_state} can be interpreted as projections of $\xi_{\gbm\eta^{(k)}}$ onto this chosen eigenbasis. This methodology can also be extended to higher-dimensional systems (qudits). For a more detailed derivation, see Appendix~\ref{app:narrow_search}.

\section{Proof of Optimality of DP}
Now we prove that the policy obtained from DP is indeed optimal in maximizing Eq.~\eqref{eq:max_work}.
\begin{theorem}
The optimal policy obtained from backward value iteration in Alg.~\ref{alg:DP_for_work}, $\Lambda^*$, is the optimal policy in achieving the highest global value function, $\tilde V^*_1 (\gbm\eta^{(i)})= \max_{a_{1:T}}\tilde{V_1}(\gbm\eta^{(i)},a_{1;T})$ stated in Eq.~\eqref{eq:overall_value}.
\end{theorem}
\begin{proof}
    To prove that the policy obtained via DP is indeed optimal. We first note that $\tilde V_1^{(\pi^*)} (\gbm\eta^{(i)})\leq \tilde V_1^*(\gbm\eta^{(i)})$ trivially as $\tilde V_1^*(\gbm\eta^{(i)})$ is the maxima by definition. We then have to prove that $\tilde V_1^{(\pi^*)}(\gbm\eta^{(i)}) \geq \tilde V_1^*(\gbm\eta^{(i)})$. If this is satisfied, then $\tilde V_1^{(\pi^*)}(\gbm\eta^{(i)}) = \tilde V_1^*(\gbm\eta^{(i)})$ and $\pi^*$ is the optimal policy that yields the highest global value function. 

For any $a_{1:T}$ in any policy $\Lambda$ and all $\gbm \eta^{(i)}$, the global value function is formally written as
\begin{equation}
    \tilde V_1^{\Lambda}(\gbm\eta^{(i)}) = \Ex\left(\sum_{t=1}^{
    T}r( K_{t-1},a_t)+V_{T+1}\Bigg|K_0=\gbm\eta^{(i)}\right)
\end{equation}\\
 We can rewrite this as a conditional expectation on the second last belief state, $K_{T-1}$, which can only be done since the metadynamics of belief states are Markovian. This can then be done iteratively backwards in time, as demonstrated in Eq.~\eqref{eq:proof}.
\begin{equation}
    \begin{split}
        \label{eq:proof}
        V_1^{(\Lambda)}(\gbm\eta^{(i)}) &= \Ex\left(\sum_{t=1}^{T}r( K_{t-1},a_t)+V_{T+1}\Bigg|K_0=\gbm\eta^{(i)}\right)\\
        &=\Ex\left[\Ex\left(\sum_{t=1}^{T-1}r(K_{t-1},a_t)+r(K_{T-1},a_{T})+V_{T+1}\Bigg|K_0=\gbm\eta^{(i)}\right) \Bigg |K_{T-1}\right]\\
        &= \underbrace{\Ex\left(\sum_{t=1}^{T-1}r(K_{t-1},a_t)\Bigg |K_0 = \gbm\eta^{(i)}\right)}_{I_{T-1}} \\
        &\quad\quad\quad\quad\quad\quad\quad\quad+\Ex\left[\underbrace{\Ex\left(r(K_{T-1},a_{T})+V_{T+1}\Bigg| K_{T-1}\right)}_{J_{T}(K_{T-1},a_{T})}\Bigg | K_0=\gbm\eta^{(i)}\right]\\
        &\leq I_{T-1} + \Ex\left[\max_{a\in\A} J_{T}(K_{T-1},a_{T})\Bigg | K_0=\gbm\eta^{(i)}\right]\\
        &= I_{T-1} + \Ex\left[J_{T}(K_{T-1},a^*_{T})\Bigg | K_0=\gbm\eta^{(i)}\right]\\
        &= I_{T-1}+\Ex\left(V_{T}(K_{T-1})\Big|K_0=\gbm\eta^{(i)}\right)\\
        &=\Ex\left(\sum_{t=1}^{T-2}r(K_{t-1},a_t)\Bigg |K_0 = \gbm\eta^{(i)}\right)+ \\
        &\quad\quad\quad\quad\qquad\Ex\Big[\Ex\left(r(K_{T-2},a_{T-1})+V_{T}(K_{T-1})\Big|K_{T-2}\right)\Big|K_0=\gbm\eta^{(i)}\Big]\\
        &\leq I_{T-2}+\Ex\Big[\max_{a\in\A}J_{T-1}(K_{T-2},a_{T-1})\Big|K_0=\gbm\eta^{(i)}\Big]\\
        &=I_{T-2}+\Ex\left(V_{T-1}(K_{T-2})\Big|K_0=\gbm\eta^{(i)}\right)\\
        &\vdots\\
        &=V_1^{(\Lambda^*)}(\gbm\eta^{(i)})~,
    \end{split}
\end{equation}
where $J_{t} (\gbm\eta^{(i)}, a_{t}) := \mathbb{E} \left[ r\left(\gbm\eta^{(i)},a_{t}\right) + V_{t+1}(K_{t}) |K_{t-1}=\eta^{(i)}\right]$ is the objective function and $V_t(\gbm\eta^{(i)}) = \max_{a_{t} \in \mathcal{A}} J_{t} (\gbm \eta^{(i)}, a_{t})$ is the value function of the belief state $\gbm\eta^{(i)}$ at time $t$.
The first inequality comes from the maximization of $J_{T}$ over all actions, and the 6th line follows from the definition of $V_t$. Notice that after each maximization, the optimal action can be generated, which will form the optimal policy $\Lambda^*$. This proof can extend till $t=1$, which then proves that $\tilde V_1^{(\Lambda)}(\gbm\eta^{(i)})\leq \tilde V^{(\Lambda^*)}_1(\gbm\eta^{(i)})$ for any possible policy. Since $\tilde V_1^*$ is the defined as $\max_{a\in\A} \tilde V^{(\Lambda)}_1(\gbm\eta^{(i)})$, this then imply that $\tilde V_1^*(\gbm\eta^{(i)})\leq \tilde V_1^{(\Lambda^*)}(\gbm\eta^{(i)})$. This together with the previous statement that $\tilde V_1^{(\Lambda^*)} (\gbm\eta^{(i)})\leq \tilde V_1^*(\gbm\eta^{(i)})$, proves that $\tilde V_1^{(\Lambda^*)} (\gbm\eta^{(i)})= \tilde V_1^*(\gbm\eta^{(i)})$. Hence $\Lambda^*$ is indeed the optimal policy.
\end{proof}

\section{Results}
\subsection{Asymptotic work-extraction rate}
As mentioned in Sec.~\ref{sec:local_action}, the optimal policy for a finite-horizon MDP, with stationary distribution, is usually time-dependent. Specifically, it can be divided into two different phases: the stationary and non-stationary phases. The non-stationary phase occurs when the future rewards diminish towards the end of the process, i.e., $t$ close to $T$. This can be illustrated when we consider a process that only lasts for 2 time steps.

At time $t=2$, which is the last step of the process. The value function $V_{3}$ is 0 regardless of belief state. Hence, the action taken at $t=2$ for a DP-agent will be one that minimizes the expected dissipation,  

\begin{equation}
\label{eq:last_step}
 a_{\text{2}}^*(K_1=\gbm\eta^{(i)})=\argmax_{a^{(j)}\in\A}\left[-k_BT\D(\xi_{\gbm\eta^{(i)}}\|\rho^{(j)})\right]~.
\end{equation}
which achieves maximum value of 0 when $\rho^{(j)}=\xi_{\gbm\eta^{(i)}}$.
Now we move on to the first time step, $t=1$. The value function that the agent optimizes becomes 
\begin{equation}
\begin{split}
    V_1(K_0=\gbm\eta^{(i)}) &= \max_{a^{(j)}\in\A}\bigg\{k_BT\left[\D(\xi_{\gbm\eta^{(i)}}\|\gamma)-\D(\xi_{\gbm\eta^{(i)}}\|\rho^{(j)})\right]\\
    &+\Ex\left[k_BT\D(\xi_{K_1}\|\gamma)\bigg|K_0 = \gbm\eta^{(i)},a^{(j)}\right]\bigg\}
\end{split}
\end{equation}
where the expectation is taken over the conditional distribution of $\Pr(K_1|K_0=\gbm\eta^{(i)},a^{(j)})$. Therefore, the action taken by the DP-agent at the first time step will be given by 
\begin{equation}
\begin{split}
\label{eq:action_of_DP}
    a_{1}^*(K_0=\gbm\eta^{(i)})&= \argmax_{a^{(j)}\in\A}\bigg\{-k_BT\D(\xi_{\gbm\eta^{(i)}}\|\rho^{(j)})\\
    &+\Ex\left[k_BT\D(\xi_{K_1}\|\gamma)\bigg|K_0 = \gbm\eta^{(i)},a^{(j)}\right]\bigg\}~.
\end{split}
\end{equation}
The first term represents the expected dissipation of the current action $a_1$, while the second term represents the expectation of future reward. Since the relative entropy is concave, that is, $p\D(\rho_1\|\gamma)+(1-p)\D(\rho_2\|\gamma)\geq \D(p\rho_1+(1-p)\rho_2\|\gamma)$. The first term is convex, and the second term will be concave; this results in a trade-off between the current reward and the future rewards that the DP-agent has to optimize. In general, an action at the end of the process differs from those at the start, since the change in expected future reward becomes significant towards the end. This also depends on the expectation of future rewards; the closer the future reward is to a constant function, the smaller the boundary effect. 

For the case of the perturbed coin, this can be illustrated in Fig.~\ref{fig:analytical_action}.
\begin{figure}
    \centering
    \includegraphics[width = 0.7 \linewidth]{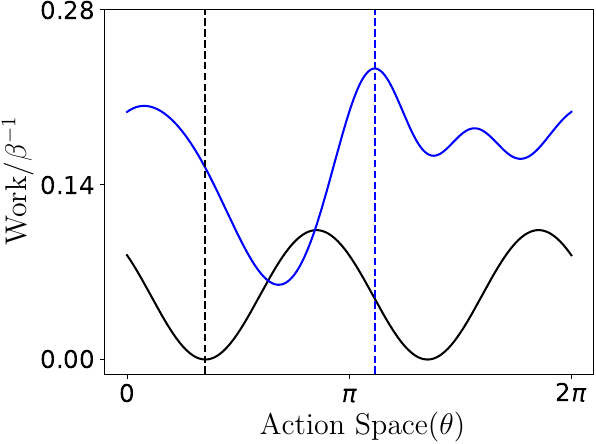}
    \caption{Graph of reward and dissipation, conditioned on the belief state $K = \gbm\pi$. The action space is parametrized by $\theta\in[0,2\pi]$. The blue line represents $V_1(K_0=\gbm\pi)$ represented in Eq.~\eqref{eq:action_of_DP}, the black line represents the dissipation incurred at the second time step, in Eq.~\eqref{eq:last_step}. The blue and black dotted lines correspond to the optimal action taken at $t=1$ and $t=2$, respectively.}
    \label{fig:analytical_action}
\end{figure}
We can see that the action taken in the second time step $a_2$ is chosen such that the dissipation is minimized, whereas $a_1$ is chosen at the point where the value function $V_1$ is maximized. It is evident that for a finite-horizon optimization, the optimal policy will be time-dependent since $a_2^*(\gbm\eta^{(i)})\neq a_1^*(\gbm\eta^{(i)})$.
As previously mentioned, the boundary effect diminishes when the expectation of future rewards becomes constant with respect to actions. In the case of the perturbed coin, the expectation reaches a constant in 3 cases. 
\begin{enumerate}
    \item When $p=0.5$, the process becomes a purely random process. The expected distribution of the future is uniform regardless of current belief. 
    \item When $r=0$, the classical limit where all states are orthogonal, allowing an agent to obtain perfect knowledge of the future regardless of actions and observations.
    \item When $r=1$, trivial limit where process emits identical states, hence regardless of action, the expectation of the future is the same.
\end{enumerate}
An interesting thing to note is that when the boundary effect completely disappears, even for the finite-horizon case, the optimal policy for all time steps becomes the same as the local-optimizing policy, since there is no longer a need to be concerned about the future reward.

On the other hand, far from the end of the process, the optimal policy tends to be time-independent, hence stationary. In this work, we focus on the stationary phase to avoid complications introduced by the boundary effect, as shown above. Furthermore, as we increase the time frame of the process, the boundary effect of the non-stationary phase diminishes. In the stationary phase, the agent's optimal actions are time-independent. As such, the transitions between belief states, conditioned on the optimal actions, essentially form a finite-state Markov chain with a fixed transition matrix. This permits at least 1 class of recurrent belief states. In the event that the transitions are aperiodic and irreducible, then there exists a unique limiting distribution, $\gbm\pi_{\lim}$, over the recurrent belief states. At each time step $t\in[T]$, the work extracted is $W_t$, then the asymptotic work extraction rate can be defined as
\begin{equation}
    w^* =\beta^{-1}\Ex_{\gbm\eta^{(k)}\in\gbm \pi_{\lim}} \left[\D(\xi_{\gbm\eta^{(k)}}\|\gamma_\beta)-\D(\xi_{\gbm\eta^{(k)}}\|\rho_{\Lambda^*(k)})\right]~,
\end{equation}
where the expectation is taken over $\gbm\pi_{\text{lim}}$, and $\rho_{\pi^*(k)}$ is the state that the DP-agent chooses to tailor the protocol to when its belief state is $\gbm\eta^{(k)}$.

\section{Time-ordered Free energy}
We define a new concept, \emph{time-ordered free energy} (TOFE). As discussed in Sec.~\ref{sec:collective}, when the agent is constrained to sequential access and possesses only classical memory, it is generally unable to extract the full non-equilibrium free energy of a multi-time quantum state. This motivates the definition of TOFE: the maximum free energy that can be extracted under a strict time-ordering constraint, where the agent's action at time $t$ can only depend on classical outcomes from time $t-1$. In other words, TOFE characterizes the extractable free energy available to an agent restricted to time-ordered operations.
This setting is precisely that imposed on the dynamic programming (DP). We can therefore define TOFE using the work expression in Eq.~\eqref{eq:max_work},
\begin{equation}
\label{eq:TOFE}
    \F_{\text{TO}}^{(1:T)} \coloneq \max_{\Lambda} \mathbf{W}(T,\Lambda)=\max_{\Lambda}\mathbb{E}_{\Lambda}\sum_{t=1}^T \mathbb{E}(W_t|A_t=\Lambda(K_{t-1}))
\end{equation}
Similarly, we can define the asymptotic rate of TOFE as that achieved by the agent under optimal policy $\Lambda^*$
\begin{equation}
\label{eq:TOFE_rate}
    f_{\text{TO}} \coloneqq \lim_{T\to\infty}\frac{1}{T} \F_{\text{TO}}^{(1:T)} = w^*~.
\end{equation}
For any multi-time state that can be generated by such a classical HMM, investigated in the form of 
\begin{equation}
\label{eq:multi_time}
    \rho^{(1:T)}\coloneqq \sum_{x_{1:T}} \Pr(x_{1:T})\bigotimes_{t=1}^T\sigma^{(x_t)}_{Q_t}~,
\end{equation} 
its non-equilibrium free energy is given by the relative entropy, 
\begin{equation}
   \F_{\text{noneq}}^{(1:T)}\coloneqq \F_\text{eq} + \beta^{-1} \D(\rho^{(1:T)}\|\gamma^{\otimes T})~,
\end{equation}
which sets the ultimate upper bound on any form of work extraction. 
Accordingly, for any bounded, finite-dimensional Hamiltonian, the non-equilibrium free energy rate of a quantum state generated by a stationary quantum process can be defined as 
\begin{equation}
\label{eq:FE_rate}
    f_{\text{noneq}} \coloneqq \lim_{t\to\infty}\frac{1}{t}\mathcal{F}^{(1:t)}_{\text{noneq}}~.
\end{equation}
We would also like to compare the rate of work extraction of a LO agent, one which minimizes the expected immediate dissipation in Eq.~\ref{eq:final_rewards}, $w(\Lambda_\text{LO})$ in Eq.~\ref{eq:rhostar_expectation} as well as its finite-time realization $\mathbf{W}(T,\Lambda_{\text{TO}})$.
Non-equilibrium free energy will always serve as an upper bound, whereas TOFE will always upper bound the performance of an LO agent by definition, since it has been optimized via DP.

\begin{proposition}[Thermodynamic Hierarchy]
\label{prop:hierachy}
For a sequence of temporally correlated quantum states shown in Eq.~\eqref{eq:multi_time}, both the cumulative work extracted as well as the asymptotic rate of work extraction of an adaptive agent with classical memory following different policies achieve the following relationship.
\begin{equation}
\label{eq:hierachy}
\F_{\text{noneq}}^{(1:T)} \geq \F_{\text{TO}}^{(1:T)}=\mathbf{W}(T,\Lambda^*)\geq\mathbf{W}(T,\Lambda_{\text{TO}}), \quad f_{\text{noneq}} \geq f_{\text{TO}}=w(\Lambda^*)\geq w(\Lambda_\text{LO})~,
\end{equation}
where $\Lambda^*$ is the optimal policy and $\Lambda_{\text{LO}}$ is the local-optimizing policy.
\end{proposition}

There is generally a separation between the Time-Ordered Free Energy (TOFE) and the true non-equilibrium free energy. A quantity known as thermal discord, introduced in Sec.~\ref{sec:multi-partite}, has long been used to characterize this gap in bipartite systems. Notably, even separable states—those without entanglement—can exhibit non-zero quantum discord. This discord reflects quantum correlations that cannot be accessed through local operations once any form of measurement is performed on earlier subsystems~\cite{ollivier2001quantum,henderson2001classical,zurek2003quantum,brodutch2010quantum}. Prior work has shown that thermal discord can quantify the discrepancy in work extraction between a fully quantum agent and a classical agent restricted to one-way classical communication in bipartite scenarios~\cite{zurek2003quantum}. This notion of thermal discord has been extended to multipartite systems to quantify the work deficit resulting from sequential access~\cite{braga2014maxwell}. However, these formulations do not account for adaptivity—that is, the agent's ability to condition future measurements on past outcomes. In this work, we propose a novel quantity, which we call causal dissipation, to properly quantify this gap.

\subsection{Causal dissipation}
\label{sec:adaptive_discord}
Causal dissipation quantifies the work potential of quantum correlations that are inevitably destroyed by the operations inherent to any time-ordered, sequential protocol~\cite{ollivier2001quantum,henderson2001classical,zurek2003quantum,brodutch2010quantum}. To obtain predictive information, an agent must use a series of adaptive measurements $\vec{\Pi}:=\{\Pi_{Q_{t}|O_{1:t-1}}\}_{t=1}^{N-1}$.
For each step, the entropy of the agent's memory increases by the Shannon entropy of the measurement outcomes:
\begin{equation}
    H\left(p_{\Pi_{Q_t}}|\Pi^{(O_{1:t-1})}_{Q_{1:t-1}}\right)\!\coloneqq\!-\sum_{o_{1:t}}\! p^{(o_{1:t-1})} p^{(o_t)}\!\ln p^{(o_t)},
\end{equation}
with $p^{(o_{1:t-1})}\coloneqq \tr(\Pi^{(o_1)}_{Q_1}\otimes\ldots\Pi^{(o_{t-1})}_{Q_{t-1}|O_{1:t-2}}\rho^{(Q_{1:T})})$ is the probability of obtaining outcome $o_{1:t-1}$, and $H(p_{\Pi_{Q_1}}) \coloneqq -\sum_{o_1}p^{(o_1)}\ln p^{(o_1)}$.

The entropy of the final system, conditioned on previous outcomes, is given by:
\begin{equation}
    S\left(\rho_{Q _T}\Big|\Pi^{(O_{1:T-1})}_{Q_{1:T-1}}\right)\!\coloneqq\!\sum_{o_{1:T-1}}\!p^{(o_{1:T-1})}S(\tilde{\rho}_{Q_T}(o_{1:T-1}))
\end{equation}
where $\tilde{\rho}_{Q_T}(o_{1:T-1})$ is the post-measured state at time $T$. We define causal dissipation as the minimal difference between the total entropy introduced by these measurements and the state's original entropy:
\begin{equation}
\label{eq:adaptive_discord}
    \delta(\overrightarrow{Q_1:Q_T})\coloneqq \min_{\gbm{\overrightarrow{\Pi}}}\sum_{t=1}^{T-1}H\left(p_{\Pi_{Q_t}}|\Pi^{(O_{1:t-1})}_{Q_{1:t-1}}\right)+S\left(\rho_{Q _T}\Big|\Pi^{(O_{1:T-1})}_{Q_{1:T-1}}\right)-S(\rho^{(Q_{1:T})})~.
\end{equation}

\subsection{Properties of causal dissipation}
Here, we try to establish some mathematical properties of causal dissipation.
\begin{enumerate}
    \item $\delta(\overrightarrow{Q_1:Q_T})=0$ when the subsystems are not correlated, i.e., 
        \begin{equation}
            \rho^{(1:T)} = \rho^{\otimes T}~,
        \end{equation} 
        then, by definition of discord, the measurement on the preceding systems will not affect the state of the systems after the measurement, hence the multipartite discord will be $0$. This can be seen when $p=1/2$ for the perturbed coin process, the state takes the form of $\rho^{(1:T)} = \left[\frac{1}{2}(\sigma^{(0)}+\sigma^{(1)})\right]^{\otimes T}$.

    \item $\delta(\overrightarrow{Q_1:Q_T})=0$ if the multipartite state is of the form 
        \begin{equation}
            \rho^{(1:T)} = \sum_{i_1\cdots i_{T}} \Pr(i_1\cdots i_{T})\bigotimes_{t=1}^{T-1}\ket{i_t}\!\bra{i_t}\otimes \rho^{(i_T)}_{Q_T}~,
        \end{equation}
        where $\ket{i_t}$ are orthonormal basis. The agent can always measure in the orthonormal eigenbasis of the first $T-1$ systems without disturbing the succeeding system. Note that this can be viewed as a stricter constraint compared to the bipartite analogue of a classical-quantum state. This property also demonstrates the \emph{asymmetry} of causal dissipation, i.e. $\delta(\overrightarrow{Q_1:Q_T})\neq\delta(\overleftarrow{Q_1:Q_T})$.
    \item $\delta(\overrightarrow{Q_1:Q_T})\neq0 \nRightarrow \lim_{T\to\infty}\frac{1}{T}\delta(\overrightarrow{Q_1:Q_T})\neq0$, i.e.,  the \emph{rate} of causal dissipation can asymptotically approach 0, even though the cumulative dissipation for a finite time step is non-zero. To illustrate this, we consider the state produced by the perturbed coin in the case where $p=0$, the quantum will take the following form.
    \begin{equation}
        \rho^{(1:T)} = \frac{1}{2}\sigma^{(0)\otimes T} + \frac{1}{2}\sigma^{(1)\otimes T},
    \end{equation}
    where $\sigma^{(0)}$ and $\sigma^{(1)}$ are not orthogonal. Clearly, if we just consider the case where $T$ is not too large, the multipartite discord would not be 0, since $\sigma^{(0)}$ and $\sigma^{(1)}$ do not commute, and any measurement will disturb the systems after the measurements. Conceptually, discord arises precisely because measurement on one subsystem produces indistinguishable conditional states on the other subsystems. Unlike the finite $T$, when we consider the asymptotic limit, it is always possible for the agent to distinguish the 2 states, $\sigma^{(0)\otimes T}$ or $\sigma^{(1)\otimes T}$. Further measurements no longer contain any useful information, resulting in rate of discord approaching 0 in the limit of $T\to\infty$, the rate of which the rate of discord approaching 0 \emph{could} be dependent on the error probability of the agent misidentifying the quantum state $\sigma^{(0)\otimes T}$ from $\sigma^{(1)\otimes T}$, which can be upper bounded by the quantum Chernoff bound:
    \begin{equation}
        P_{e,T} = \exp(-T\xi_{QCB}), \quad \xi_{QCB}=-\log(\min_{0\leq s\leq1}\tr(\rho^s\sigma^{1-s}))~.
    \end{equation}
    where $P_{e,T}$ is the error probability of misidentification, here we take $\rho,\sigma=\sigma^{(0)\otimes T}, \sigma^{(1)\otimes T}$. Note that this a loose bound as it assumes collective measurement is allowed, if $\rho$ and $\sigma$ are very close to each other, the Chernoff exponent $\xi_{QCB}$ will be very small, indicating a slow decay rate for the rate of discord, whereas for if the 2 states are far apart, the exponent will be large, resulting in a quick convergence to 0.
\end{enumerate}
\begin{proof}
We provide a detailed proof for Property 3 of causal discord. The proofs for Properties 1 and 2 are omitted here, as they follow directly from the definition of causal dissipation, utilizing reasoning analogous to that of bipartite quantum discord.

Formally, we wish to demonstrate that
\begin{equation}
    \lim_{T\to\infty} \delta(\overrightarrow{Q_1:Q_T}) = 0~,
\end{equation}
even though multipartite discord for any finite time steps is clearly non-zero. We can move around the terms to obtain that
\begin{equation}
    \lim_{T\to\infty} \min_{\overrightarrow{\gbm\Pi}}\sum_{i=1}^{T-1}H(p_{\Pi_{Q_i}}|\Pi_{Q_{1:i-1}},O_{1:i-1}) + S(\rho_{Q_T}|\Pi_{Q_{1:i-1}},O_{1:i-1}) = \lim_{T\to\infty} S(\rho^{(1:T)})~,
\end{equation}
where the LHS is the entropy rate of the quantum state.
Here we consider the case where the HMM has parameters of $p=0,1$ to demonstrate. 
First, we realize that in both cases where $p=0$ or $p=1$, the total quantum state can be written as 
\begin{equation}
    \rho^{(1:T)} = \frac{1}{2}\rho_0 + \frac{1}{2}\rho_1~,
\end{equation}
where $\rho_0,\rho_1 = \sigma^{(0)\otimes T},\sigma^{(1)\otimes T}$ when $p=0$  and $\rho_0,\rho_1 = (\sigma^{(0)}\otimes\sigma^{(1)})^{\otimes T/2}, (\sigma^{(1)}\otimes\sigma^{(0)})^{\otimes T/2}$ when $p=1$.

We first prove that the actual entropy rate of the total quantum state tends to 0.
Here, we consider a quantity 
    \begin{equation}
        \chi_T = S(\rho^{(1:T)})-\frac{1}{2}S(\rho_0)-\frac{1}{2}S(\rho_1)~,
    \end{equation}
this quantity is known as the Holevo information or Holevo quantity,  which is upper-bounded by the classical description of the ensemble, i.e., $0\leq \chi_T\leq H_2(\frac{1}{2})$ since both states appear with $1/2$ probability, $H_2$ here is the binary entropy. 
We then rewrite the entropy rate with this quantity.
\begin{equation}
    h \coloneqq \lim_{T\to\infty}\frac{1}{T}S(\rho^{(1:T)}) = \lim_{T\to\infty}\frac{\chi_T}{T} + \frac{1}{2T}S(\rho_0)+\frac{1}{2T}S(\rho_1)~.
\end{equation}
Since the Holevo information is bounded, this results in
\begin{equation}
    h = \lim_{T\to\infty}\frac{1}{2T}S(\rho_0)+\frac{1}{2T}S(\rho_1)~.
\end{equation}
For our case, this admits a closed form solution, and equals 0 in the event both $\sigma^{(0)}$ and $\sigma^{(1)}$ are pure. 

Now we have to prove that 
\begin{equation}
    \lim_{T\to\infty} \frac{1}{T}\min_{\overrightarrow{\gbm\Pi}}\sum_{i=1}^{T-1}H(p_{\Pi_{Q_i}}|\Pi_{Q_{1:i-1}},O_{1:i-1}) + S(\rho_{Q_T}|\Pi_{Q_{1:i-1}},O_{1:i-1}) =0
\end{equation}
We explicitly write out the definition of the terms
\begin{equation}
\label{eq:branch_entropy}
    H(p_{Q_i}|\Pi_{Q_{1:i-1}|O_{1:i-1}}) = -\sum_{o_{1:i-1}} \Pr(O_{1:i-1}=o_{1:i-1}|\Pi_{Q_{1:i-1}}) \sum_{o_i}p^{(o_i)}\ln p^{(o_i)}~,
\end{equation}
and 
\begin{equation}
    S(\rho_{Q_T}|\Pi_{Q_{1:T-1}},O_{1:T-1}) = \sum_{o_{1:T-1}}\Pr(O_{1:T-1}=o_{1:T-1}|\Pi_{Q_{1:T-1}})S(\rho_{Q_T|\Pi_{Q_1:Q_{T-1},O_{1:T-1}}})~.
\end{equation}
Both of these expressions are positive, which implies they must all converge to 0 independently. 

From results in symmetric state discrimination, given 2 pure quantum states $\rho$ and $\sigma$ with equal prior, one may distinguish them asymptotically. I.e., there exist a sequence of measurements $\{\Pi_1^{(i_1)},\Pi_2^{(i_2)},\ldots,\Pi_T^{(i_T)}\}$ such that
\begin{equation}
\lim_{T\to\infty}\tr(\Pi_1^{(i_1)}\otimes\Pi_2^{(i_2)}\ldots\otimes\Pi_T^{(i_T)} \rho^{\otimes T}) = 1 
\end{equation}
and 
\begin{equation}
\lim_{T\to\infty}\tr(\Pi_1^{(i_1)}\otimes\Pi_2^{(i_2)}\ldots\otimes\Pi_T^{(i_T)} \sigma^{\otimes T}) = 0 
\end{equation}
In fact, we can measure along the eigenbasis of one of the states, such that $\tr(\Pi\rho) = 1$ and $\tr(\Pi\sigma)=r$. As long as the measurement result correspond to $\Pi$, we keep measuring in that basis, on the other hand, if we observe an outcome corresponding to $\id-\Pi$, we choose to measure along the eigenbasis of the other quantum state, i.e. $\Pi'$ such that $\tr (\Pi'\sigma)=1,\tr(\Pi'\rho)=r$. We now focus on the case where we observe the outcome corresponding to $\Pi$ for the first $k$ times, and the post-measured state will be
\begin{equation}
\label{eq:post_state}
\begin{split}
    \tilde\rho^{(1:T)} = \frac{(\Pi^{\otimes k}\otimes \id^{\otimes(T-k)})\rho^{(1:T)})}{\tr((\Pi^{\otimes k}\otimes \id^{\otimes(T-k)})\rho^{(1:T)})} &= \frac{\frac{1}{2}\rho^{\otimes T}+\frac{1}{2}r^k\rho^{\otimes T}\otimes \sigma^{\otimes (T-k)}}{\frac{1}{2}+\frac{1}{2}r^k}\\
    &= \frac{1}{1+r^k} \rho^{\otimes T} + \frac{r^k}{1+r^k}\rho^{\otimes k}\otimes \sigma^{\otimes (T-k)}
\end{split}
\end{equation}
In the limits $T\to\infty$ and $k\to\infty$, $\tilde\rho^{(1:T)}$ becomes increasingly purified to become $\rho^{\otimes T}$, the probability $p^{(o_i)}$ appearing in Eq.~\eqref{eq:branch_entropy} converges to 1 in this regime. In fact, the moment an observation corresponding to $\id-\Pi$ appears, the agent already knows that the state \emph{cannot} be $\rho$, subsequent measurements in the eigenbasis of $\sigma$ would result in $p^{(o_i)}$ equal to 1. Formally, one can see that
for Eq.~\eqref{eq:branch_entropy}, after measuring $k$ outcomes the $k+1$-th measurement will yield an entropy of
\begin{equation}
    H(p_{Q_{k+1}}|\Pi_{Q_{1:k}|O_{1:k}}) = \frac{1+r^k}{2}H_2\left(\frac{1+r^{k+1}}{1+r^k}\right) + \frac{1-r^k}{2}H_2(1)
\end{equation}
The first term represents the sequence of outcomes corresponding to only $\Pi$, while the second term represents all other sequences that involve at least 1 measurement $\id-\Pi$. As $k\to\infty$, the term $\frac{1+r^{k+1}}{1+r^k}\to 1$, resulting in the whole expression converging to 0.
This then implies that 
\begin{equation}
    \lim_{T\to\infty}\frac{1}{T}\sum_{i=1}^{T-1}H(p_{\Pi_{Q_i}}|\Pi_{Q_{1:i-1}},O_{1:i-1}) = 0
\end{equation}
for some measurement, which is sufficient since this achieves the theoretical minimum.

Using similar measurement strategy, the final conditional von Neumann entropy $S(\rho_{Q_T|\Pi_{Q_1:Q_{T-1},O_{1:T-1}}})$ will just be given by $S(\rho)$ or $S(\sigma)$ since the identity of the final subsystem can be obtain unambiguously in the asymptotic limit. Formally, the final conditional state will be given by 
\begin{equation}
\begin{split}
    \lim_{T\to\infty} \frac{1}{T}\min_{\overrightarrow{\gbm\Pi}} S(\rho_{Q_T}|\Pi_{Q_{1:T-1}},O_{1:T-1}) &= \lim_{T\to\infty} \frac{1+r^{T-1}}{2}S(\tilde\rho^{(T)})+\frac{1-r^{T-1}}{2}S(\sigma)\\
    &= \frac{1}{2}S(\rho)+\frac{1}{2}S(\sigma)~,
\end{split}
\end{equation}
where $\tilde\rho^{(T)} = \tr_{/T}\tilde\rho^{(1:T)}$ is the reduced state of the final system post-measurement as in Eq.~\eqref{eq:post_state} as $T,k\to\infty$.

For our scenario, since both $\rho$ and $\sigma$ are pure,
\begin{equation}
  \lim_{T\to\infty} \frac{1}{T}\min_{\overrightarrow{\gbm\Pi}} S(\rho_{Q_T}|\Pi_{Q_{1:T-1}},O_{1:T-1}) = 0  
\end{equation}
This completes the proof.
\end{proof}

\section{Case study:  Perturb-coin}
For illustration purposes, we implemented DP on a simple process of ``Perturbed Coin" as illustrated in Fig.~\ref{fig:PC_diagram}(a). Without loss of generality, the two emitted states $\sigma^{(0)}$ and $\sigma^{(1)}$ are assumed to be pure and have fidelity, $F(\rho,\sigma)=\left(\tr\sqrt{\sqrt{\rho}\sigma\sqrt{\rho}}\right)^2$ of $r$. We once again adopt the work extraction protocol proposed in~\cite{skrzypczyk2014work}, an example of the \emph{$\rho^*$-ideal extraction protocol}. Under the assumption of a degenerate Hamiltonian, this protocol also belongs to the family of thermal operations. When it is tailored to a state $\rho = \sum_i \lambda_i\ket{\lambda_i}\bra{\lambda_i}$ and applied onto a state $\sigma$, the resultant probability distribution takes the form.
\begin{gather}
    \label{eq:work_val}w^{(i)} \coloneqq  \beta^{-1}(\ln{\lambda_i}+\ln2)\\
    \label{eq:work_prob}\Pr(W=w^{(i)}) = \bra{\lambda_i}\sigma\ket{\lambda_i}~,
\end{gather}
where $W$ is the random variable of the work extracted and $\{w^{(i)}\}_i$ are the realizations. The DP-agent will extract the free energy based on the optimal policy, $\pi^*$, whereas the LO-agent will execute the action by minimizing the expected dissipation term.

In order for us to compare the free energy and the TOFE, we have to be able to compute the free energy, which may appear as a trivial task, but it involves finding eigenvalues of a Hermitian matrix of $2^T\times2^T$ dimension, a task that scales exponentially with $T$. Instead, we focus on finding the rate of non-equilibrium free energy per time step in Eq.~\eqref{eq:FE_rate}. Once again, calculating the von Neumann entropy rate, $s_{vN}\coloneqq\lim_{T\to\infty}\frac{1}{T}S(\rho^{1:T})$ is difficult due to eigen-decomposition, but we can use the data-processing inequality to derive a lower bound on it. Here we demonstrate how the upper bound on the entropy rate can be achieved; readers interested in the lower bound may refer to Appendix~\ref{app:free_energy_rate_bound} for a detailed exposition.

\subsection{Lower bound on free energy rate}
\label{sec:FE_lowerbound}
First, notice that for any multi-time state $\rho^{(1:T)}$, we can rewrite it into a bi-partite state $\rho_{AB}$ where $B$ represents all the subsystems in the past time $Q_1,Q_2,\ldots,Q_{T-1}$ and $A$ represents the subsystem at $Q_T$. The joint entropy can then be written as
\begin{equation}
    S(\rho_{AB}) = S(\rho_B) + S(A|B)_{\rho}~,
\end{equation}
if $s_{vN}$ does indeed exist, then as $T\to\infty$, $S(\rho_{AB})=Ts_{vN}$ and $S(\rho_{B})=(T-1)s_{vN}$ which leaves us with 
\begin{equation}
    s_{vN} = S(A|B)_{\rho}~.
\end{equation}
Alternatively, we can express it as follows: the purpose will become clear. 
\begin{equation}
    s_{vN} = S(\rho_A)-I(A;B)_\rho=S(\rho_A)-\D(\rho_{AB}\|\rho_A\otimes\rho_B)~.
\end{equation}
With this, we can apply the data-processing inequality, which states that given any 2 quantum states $\rho$ and $\sigma$ and any quantum channel $\mathcal{E}$,
\begin{equation}
\label{eq:data_processing}
    \D(\rho\|\sigma)\geq \D(\mathcal{E}(\rho)\|\mathcal{E}(\sigma))~.
\end{equation}
We let be $\mathcal{E}:B\mapsto D$ be some channel that acts only on $B$ leaving $A$ unchanged to obtain
\begin{equation}
    \sigma_{AD} = (\id_A\otimes\mathcal{E}_B)(\rho_{AB})~.
\end{equation}
Then by Eq.~\eqref{eq:data_processing}, we have that
\begin{equation}
    \begin{split}
        S(A|D)_{\sigma} &= S(\sigma_A) - \D(\sigma_{AD}\|\sigma_A\otimes\sigma_D)\\
        &=S(\rho_A)-\D((\id_A\otimes\mathcal{E}_B)(\rho_{AB})\|(\id_A\otimes\mathcal{E}_B)(\rho_A\otimes\rho_B))\\
        & \geq S(\rho_A) - S(\rho_{AB}\|\rho_A\otimes\rho_B)\\
        &=S(A|B)_{\rho}~,
    \end{split}
\end{equation}
this provides an upper bound on the entropy rate. Now, we have to find a meaningful channel which reduces $S(A|D)_{\sigma}$ so the bound is not too loose. One way to do this is via a POVM,$\{\mathcal{M}_B^{(i)}\}_i$, on $B$, also a valid quantum channel. This can be interpreted as finding a measurement such that the uncertainty of the quantum state in the next time step is as low as possible. Since the state is generated by an HMM, the best way to predict the next emitted state would be to consider the states in the past, i.e., the subsystems in $B$ as much as possible. An obvious choice would be to use Helstrom measurements; this may or may not be optimal, but it is good enough for us. 

Based on the structure of the HMM, this boils down to distinguishing all multi-time states ending from $\sigma^{(0)}$ from those of $\sigma^{(1)}$. For the case of the perturbed coin, it boils down to distinguishing the following states,
\begin{equation}
    \tau_{t}^{(s)} = (1-p)\tau_{t-1}^{(s)}+p\tau_{t-1}^{(s)}\quad s\in \{0,1\},\quad \tau_1^{(s)}=\sigma^{(s)}~.
\end{equation}
This implies that the original state $\rho_{AB}$ can be expressed as
\begin{equation}
    \rho^{(t+1)}_{AB}=\sum_{s}\gbm\pi(s)\tau_{t,B}^{(s)}\otimes\xi_A^{(s)}
\end{equation}
where $\gbm \pi$ is the stationary distribution of the HMM, $\xi^{(s)}$ is the expected state conditioned on the latent state of the HMM, specifically,
\begin{equation}
    \xi^{(s)} = (1-p)\sigma^{(s)}+p\sigma^{(1-s)}~.
\end{equation}
The application of POVM $\{\mathcal{M}^{(i)}\}_i$ on to $\rho_{AB}$ results in a post-measurement state
\begin{equation}
\begin{split}
\label{eq:classic_quantum}
    \eta^{(t+1)}_{AD} &= \sum_{s,y}\gbm\pi(s)\tr(M^{(y)}\tau_t^{(s)})\ket{y}\!\bra{y}_D\otimes \xi^{(s)}_A\\
    &= \sum_yp^{(y)}_t\ket{y}\!\bra{y}_D\otimes\chi_A^{(y)}
\end{split}
\end{equation}
where $p^{(y)}_t $ is the probability distribution of the measurement and index $s$ now run through the latent states of the HMM as well as the last emitted state. Note that this can only be done due to the unique dynamic of the Perturbed Coin. $\chi_t^{(y)}$ is the post-measured state in system $A$ associated with outcome $y$. 
The classical-quantum form of Eq.~\eqref{eq:classic_quantum} allows us to express the resultant upper bound as 
\begin{equation}
    S(A|B)_{\rho} \leq S(A|D)_{\eta} \leq \sum_yp^{(y)}_tS(\chi_t^{(y)})~,
\end{equation}
The equality of the second inequality can be satisfied with $t\to\infty$, but this provides a tight bound even for small $t$s.
This then allows us to calculate the lower bound on $f_{\text{noneq}}$. The asymptotic rate of non-equilibrium free energy and that of TOFE can then be compared. 

\subsection{Comparison}
We begin by analyzing the gap between the true non-equilibrium free energy rate and the time-ordered free energy (TOFE) rate. We will make use of the lower bound obtained in Sec.~\ref{sec:FE_lowerbound}. As illustrated in the right panel of Fig.~\ref{fig:LOvsDP},  $f_{\text{TO}}$ coincides with $f_{\text{noneq}}$ in regions where temporal correlations are trivial—namely, when $p\in\{0,1,1/2\}$, or $r\in\{0,1\}$. The $r=0$ case also recovers the results in the classical limit in~\cite{garner2017thermodynamics}. Outside these regions, a nonzero gap emerges. 
\begin{figure*}
    \centering
    \includegraphics[width=\linewidth]{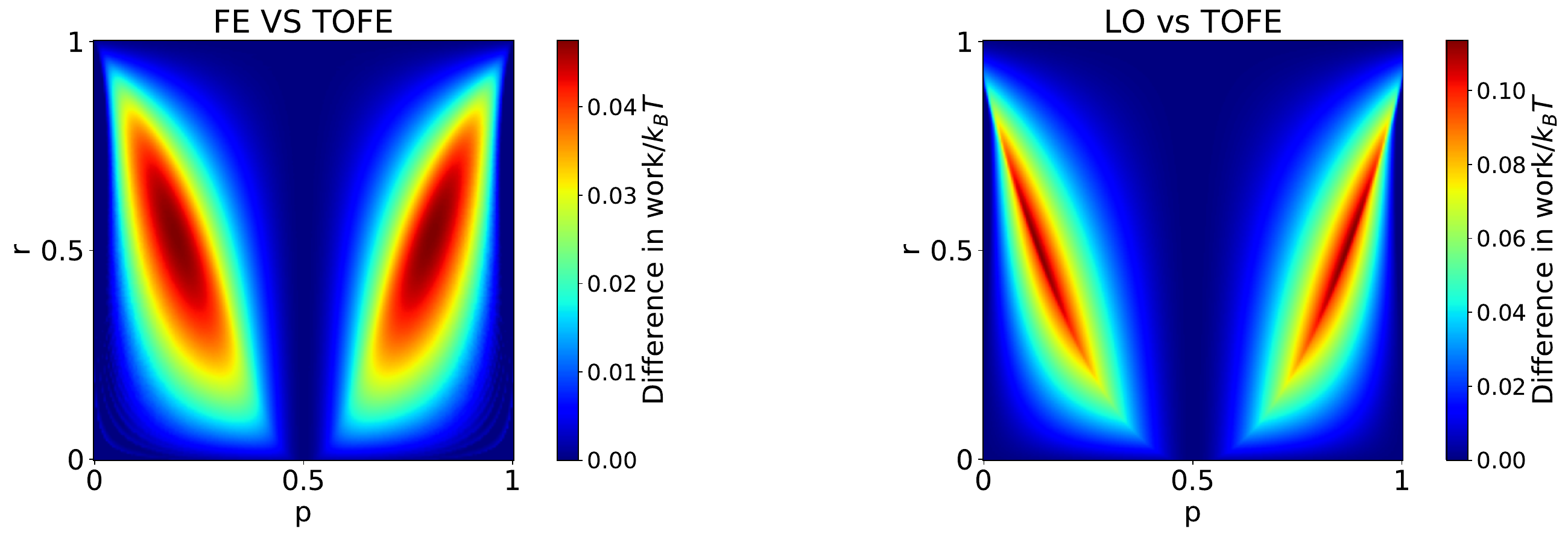}
    \caption{Comparison of asymptotic work extraction rate. The left panel shows the difference between the non-equilibrium free energy rate in Eq.~\eqref{eq:FE_rate} and the asymptotic TOFE rate in Eq.~\eqref{eq:TOFE_rate}. The right panel shows the difference between the asymptotic TOFE rate and the asymptotic work extraction rate of an LO-agent that just aims to minimize immediate expected dissipation.}
    \label{fig:LOvsDP}
\end{figure*}

We next compare the performance of the LO-agent and the DP-agent. As seen in the right panel of Fig.~\ref{fig:LOvsDP}, the most substantial improvement by the DP-agent occurs in the previously identified region where the LO-agent performance was not enhanced by memory---the \emph{memory-apathetic} region, as we discussed in \cref{ch:3}. This demonstrates that the DP-agent exhibits greater predictive power than the LO-agent; this is achievable only by permitting non-zero expected local dissipation. In contrast, within the navy blue region, even the DP-agent loses its advantage and performs equivalently to the LO-agent. This region corresponds to processes that are intrinsically hard to predict based on past observations, either due to high fidelity between the two quantum states, high stochasticity in the process, or both.

\begin{figure}
    \centering
    \includegraphics[width=1\linewidth]{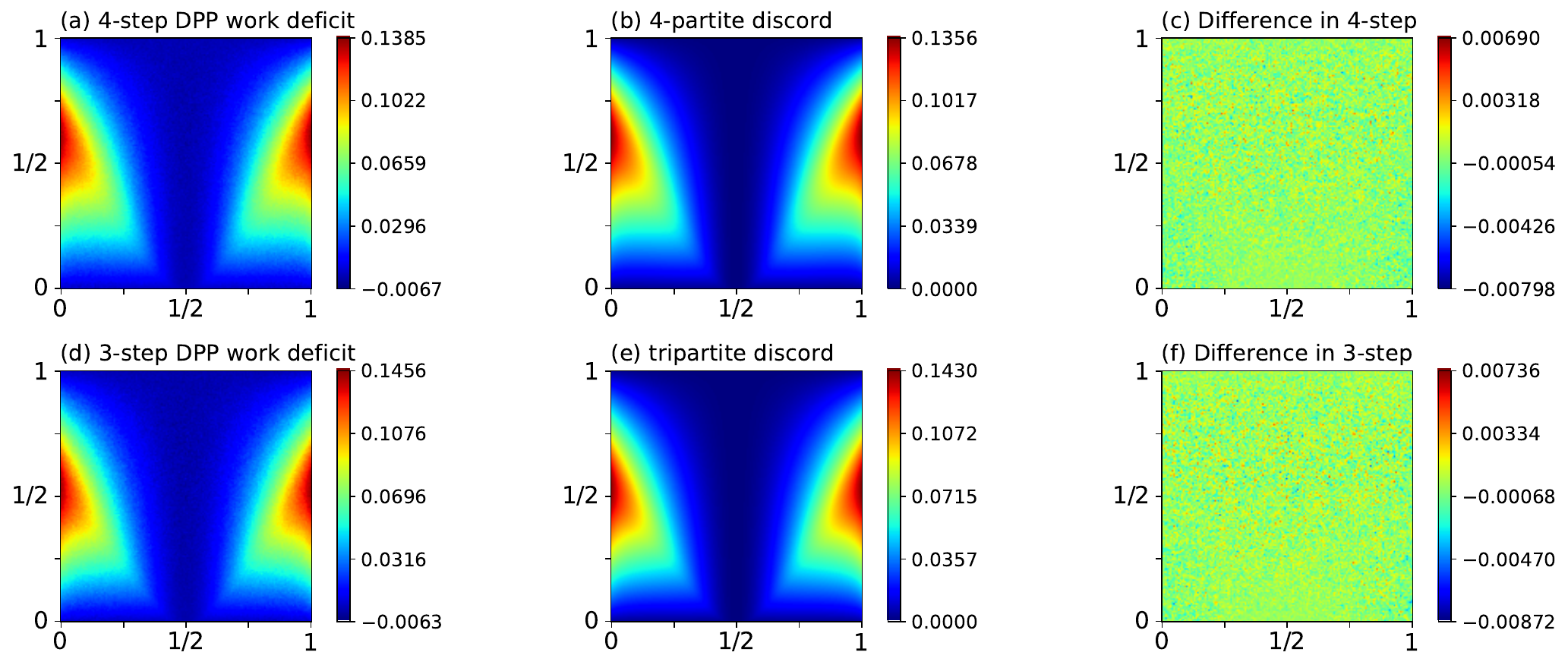}
    \caption{Comparison between simulated work extraction using dynamic programming (DP) and the analytical adaptive multipartite discord defined in Eq.~\eqref{eq:adaptive_discord}. The top row corresponds to a four-subsystem state, while the bottom row corresponds to a tripartite system. Panels (a) and (d) show the simulated work deficit under the optimal adaptive policy. Panels (b) and (e) show the corresponding analytical values of the adaptive multipartite discord. Panels (c) and (f) display the pointwise difference between the work deficit and the discord.}
    \label{fig:enter-label}
\end{figure}
As discussed in Sec.~\ref{sec:adaptive_discord}, we compare the work deficit incurred when an agent extracts work using the optimal adaptive policy to the analytical value of the adaptive multipartite discord defined in Eq.~\eqref{eq:adaptive_discord}. As shown in panels (c) and (f), the difference between the two quantities appears to be dominated by stochastic noise without any clear structure. This observation supports our conjecture that the adaptive multipartite discord provides a good quantitative description of the gap between the time-ordered free energy (TOFE) and the true non-equilibrium free energy, at least for finite systems.

\subsection{Understanding the policy} 
Above all, we wish to understand the physical meaning of the optimal policy, $\gbm\pi^*$---specifically, how the agent’s belief state relates to its optimal action. To this end, we examine the difference between the expected state $\xi_{\gbm\eta^{(k)}}$ and the tailored state $\rho_{\gbm\pi^*(k)}$ associated with a given belief state $\gbm\eta^{(k)}$. This comparison is made via their Bloch vectors, defined as:
\begin{equation}
    \label{eq:bloch_vect}
    \vec{b} = \begin{pmatrix}
        x\\y\\z
    \end{pmatrix},\hquad x= \tr(X\rho),\hquad y= \tr(Y\rho), \hquad z= \tr(Z\rho)~,
\end{equation}
where $X,Y,Z$ are the 3 Pauli matrices. The Bloch vectors with varying belief states can be found in Fig.~\ref{fig:compare_action}. 
\begin{figure}
    \centering
    \includegraphics[width=\linewidth]{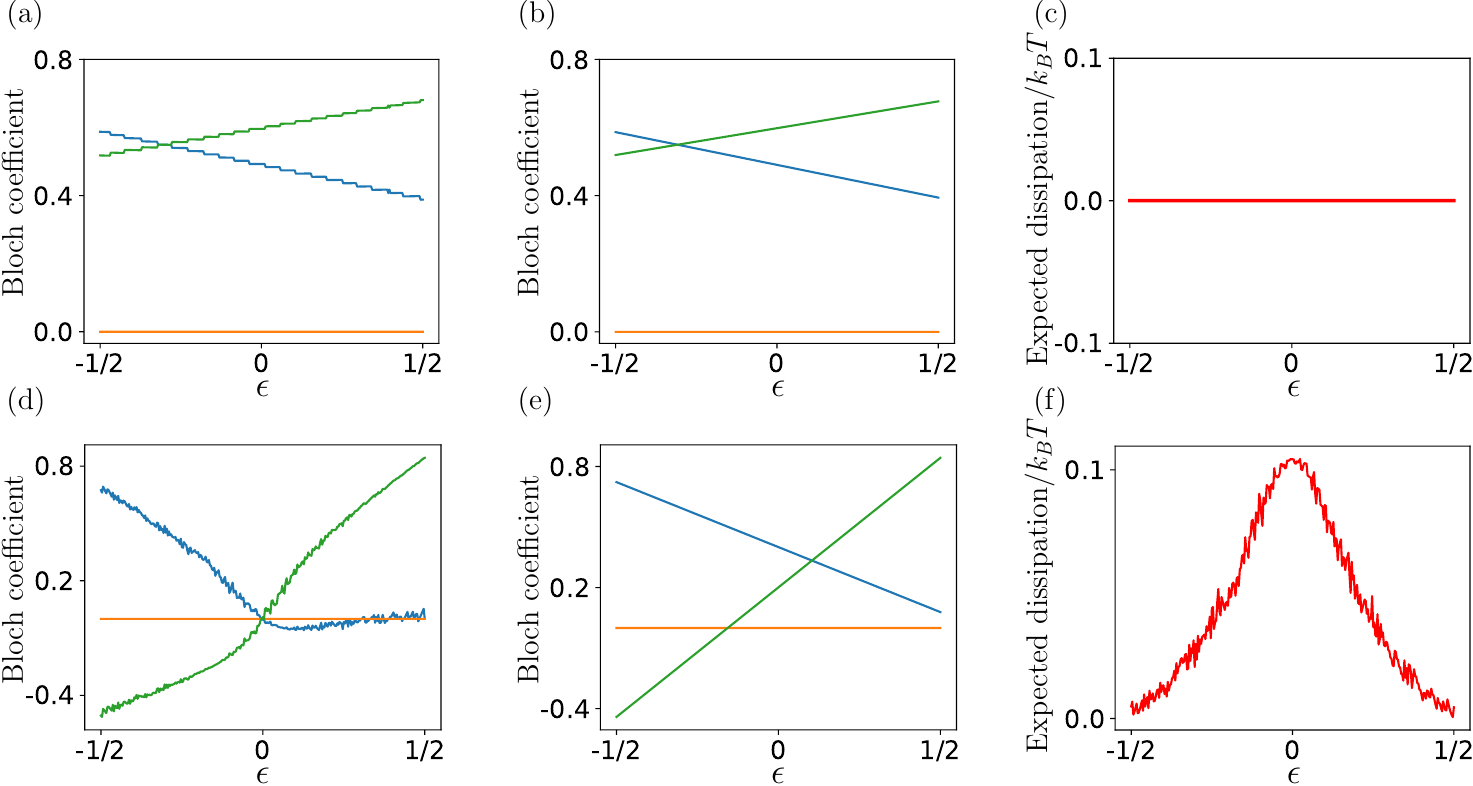}
    \caption{Comparison of Bloch vectors of expected states in panel (b) and (c) with tailored states in panel (a) and (d) under different parameters against varying belief states parametrized by $\gbm\eta=(1/2+\epsilon,1/2-\epsilon)$. Panel (a)(b)(c) has parameter $p=0.6,r=0.6$, panel (e)(d)(f) has parameter $p=0.9,r=0.2$. Panels (c) and (f) show the expected dissipation. The orange line represents the $Y$-component, blue the $X$-component, and red the $Z$-component.}
    \label{fig:compare_action}
\end{figure}
We begin by analyzing the central navy blue region in Fig.~\ref{fig:LOvsDP}(c), where the DP-agent shows no advantage over the LO-agent (Fig.~\ref{fig:compare_action}(a), (b)).  The Bloch vectors of the expected states closely match those of the tailored states. Plotting the local expected dissipation as a function of the belief state (Fig.~\ref{fig:compare_action}(c)) confirms this: the dissipation is nearly zero across all belief values. This implies that the DP, LO, and memoryless agents all behave similarly here, as the underlying process is too stochastic for any agent to gain a predictive advantage. In such cases, minimizing local dissipation is the best possible strategy. In contrast, in the regime where the DP-agent outperforms the LO-agent (Fig.~\ref{fig:compare_action}(d),(e)), the Bloch vectors of the tailored and expected states no longer align clearly. Yet, the plot of local expected dissipation reveals a distinctive trend (Fig.~\ref{fig:compare_action}(f)).  At $\gbm\eta=(1/2,1/2)$, representing maximal uncertainty, the local dissipation peaks, indicating that the agent strategically sacrifices immediate work extraction to gain more in future steps. As $\gbm\eta$ approaches $(1, 0)$ or $(0, 1)$, the local dissipation drops to zero, showing that the agent becomes increasingly confident about the next state and reverts to minimizing immediate dissipation. This behavior highlights the temporal foresight encoded in the optimal policy; it dynamically balances present and future rewards depending on belief certainty.

\section{Comment on robustness of result}
As established in this chapter and Chapter~\ref{ch:3}, the average work extracted by the agent is primarily governed by its adopted policy. Yet, the overall efficacy of the extraction process relies equally on the broader classical-causal strategy detailed in Sec.~\ref{sec:classical_causal}. Modifying the set of allowed operations—for example, generalizing beyond thermal operations—or restricting the agent's memory capacity to limit its predictive capability will correspondingly shift the thermodynamic yield. Despite these potential variations, provided that the four defining components of the strategy are rigorously specified, dynamic programming remains a robust mathematical tool for evaluating the average extractable work.

It is important to note, however, that our analysis currently assumes the underlying state evolution proceeds independently of the agent's interventions. Generalizing this framework to scenarios where actions actively disturb future dynamics necessitates the formalism of true quantum stochastic processes, such as quantum process tensors or combs. We explore this open trajectory further in \cref{ch_conclusions}

\section{Summary}
By integrating reinforcement learning techniques with the agent-based belief framework from computational mechanics, we establish a fundamental upper bound on the rate of work extraction from temporally correlated quantum states generated by classical hidden Markov models (HMMs). This approach leads us to define the concept of \emph{Time-Ordered Free Energy} (TOFE), which captures causal constraints in a manner not accounted for by conventional free energy. We demonstrate a distinct separation between TOFE and the actual free energy of the quantum state—--one that we conjecture to be quantified by a new measure of discord in the multipartite regime. Furthermore, we elucidate the physical intuition underlying the optimal control policy that governs the agent’s actions. Our results deepen the understanding of sequential processing in non-i.i.d. temporal quantum systems and provide a classical benchmark against which quantum agents can be evaluated in comparable tasks. Although the current optimization algorithm used in dynamic programming scales exponentially with the discretization of the belief and action spaces, future improvements may be achievable by leveraging quantum dynamic programming to mitigate this computational overhead.

%% file: Chapters/Chapter6.tex

\chapter{Learning of unknown state} 
\label{ch:learning}
\noindent\rule{\textwidth}{0.4pt}
\vspace{1.5em} 
\setlength{\parindent}{4ex}

\noindent \textit{In this chapter, we explore a novel setting for work extraction. Instead of considering a sequence of temporally correlated quantum systems, we focus on an i.i.d. sequence of unknown pure quantum states $\psi$. and ask if an agent without any prior knowledge is able to extract work from such a sequence. We investigate whether an agent, with no prior knowledge of the state, can simultaneously extract work and learn the underlying quantum state. Remarkably, we demonstrate that not only is this dual task feasible, but the adaptive strategy employed by the agent also leads to an exponential reduction in thermodynamic dissipation compared to tomographical approaches.}

\newpage

So far, we have been focusing on extracting free energy from temporally correlated quantum systems with the assumption that the agent knows the full description of the generator of these quantum systems. This assumption is not necessarily realistic. Rather than enforcing on the agent knowing the full description of the underlying stochastic generator, is it possible for the agent, without the knowledge of the process, to extract work and simultaneously learn the dynamics? This question may be a lot to tackle, and hence we start with an easier question. Suppose we are given an i.i.d sequence of unknown pure quantum states of length $N$, i.e.,
\begin{equation}
\rho^{(1:N)} = \rho^{\otimes N}~.  
\end{equation}
\emph{Is it possible to extract work while simultaneously learn the identity of the unknown quantum state, $\rho$}? 
While numerous protocols have been proposed for work extraction in quantum settings \cite{allahverdyan2004maximal,aaberg2013truly,brandao2013resource,skrzypczyk2014work,elouard2018efficient}, the majority assume that the agent has full knowledge of the quantum state---these are so-called \emph{state-aware} protocols. In practice, we may not always know how the state is prepared, leading us to the \emph{state-agnostic} scenario at hand. In this case, how can we design a protocol for work extraction?

A natural strategy might involve first performing quantum state tomography to estimate the unknown state, followed by work extraction based on this estimate \cite{vsafranek2023work,watanabe2024black,watanabe2025universal}. Indeed, knowledge of the quantum state is essential for efficient work extraction; consider, for example, Szilard's engine, where the information about the system's configuration---such as the position of a particle---enables work to be extracted \cite{szilard1929entropieverminderung}. However, with only finitely many samples available, any estimate of the true state has statistical uncertainty~\cite{o2016efficient,haah2016sample}, resulting in unavoidable heat dissipation during work extraction \cite{riechers2021initial}. Furthermore, quantum systems that are measured during state tomography are no longer available for work extraction and therefore contribute to irreversible heat dissipation. Combining these observations, we will show that any two-step procedure that first uses tomography to inform an efficient work extraction procedure will lead to a cumulative dissipation that scales at least as $\Omega(\sqrt{N})$ in the number $N$ of copies of the unknown state.
Such strategies are only minimally adaptive, though, and this raises a natural question: \emph{Can more adaptive strategies, which simultaneously balance learning and work extraction, offer better performance in terms of cumulative dissipation?} We answer this in the affirmative, showing an exponential improvement for pure qubit states. This result builds on a novel connection between adaptive quantum control and classical reinforcement learning~\cite{lattimore2020bandit}, which we believe to be of independent interest. More precisely, we show that the dissipation in our framework can be interpreted as the \emph{regret} incurred during the learning of pure quantum states, formulated within a multi-armed bandit setting~\cite{lumbreras22bandit,lumbreras24pure}.
 
\section{Multi-armed bandit}
\label{sec:MAB}
A classical framework for adaptive decision-making in environments with unknown stochastic dynamics is the multi-armed bandit (MAB) problem. In its simplest form, the discrete classical MAB consists of a finite set of actions (or arms) and an unknown environment.
At each time step, $t$, the agent selects an action, $a_t\in\A$, and a corresponding reward drawn from a probability distribution conditioned on both the action and the environment.

We consider a quantum generalization of this framework, in which both the action space and the environment possess quantum structure. Specifically, actions correspond to a finite set of quantum observables, while the environment is characterized by an unknown quantum state. This leads us to the following formal definition:
\begin{definition}[Quantum multi-armed bandit]
A $d$-dimensional \emph{discrete multi-armed quantum bandit} is defined over a finite set $\A\in\mathcal{O}_d$ of observables that we denote as the \emph{action set}. the bandit is in an \emph{environment}, an unknown quantum state $\rho$ that is taken from a set of potential environments $\Gamma\in \mathcal{S}_d$. The bandit problem is fully characterized by the tuple $(\A,\Gamma)$.
\end{definition}
A policy in this framework is a sequence of conditional probability distributions $\pi=\{\pi_t\}_{t\in\mathbb{N}}$ over the action set, defined by Eq.~\ref{eq:general_policy}.
This quantum bandit model is particularly relevant in quantum state tomography, where the actions typically correspond to choices of measurement bases, and the outcomes $o_t\in\{0,1\}$ represent the binary measurement results. We are particularly interested in the case where the unknown environment is characterized by a single unknown pure state, $\psi=\ket{\psi}\!\bra{\psi}$, and all the measurements are projective. In that case, the measurement outcomes at time $t$, using measurement $\Pi_{a_t}$ is distributed according to:
\begin{equation}
    \Pr(o_t|\Pi_{a_t})=\begin{cases}
        \bra{\psi}\Pi_{a_t}\ket{\psi}, \hquad\hquad\hquad \text{if}\hquad x=1\\
        1-\bra{\psi}\Pi_{a_t}\ket{\psi} \hquad \text{if}\hquad x=0\\
        0\quad\quad \quad\quad\quad\quad\hquad\text{else}
    \end{cases}
\end{equation}
where $\Pi_{a_t}$ correspond to the projective measurement associated with the action $a_t$.
Other than the definitions and interactions between the agent and the multi-armed bandit, we are required to define the objective of the agent. In this thesis, we focus on the reward maximization over a finite time horizon $T\in\mathbb{N}$. In particular, we aim to minimize a metric in reinforcement learning called \emph{cumulative regret} over the course of $T$ steps. Formally, the cumulative regret is defined as
\begin{equation} 
    R_T(\A,\psi,\pi) \coloneqq \sum_{t=1}^T\max_{a_t\in \A}\bra{\psi}\Pi_{a_t}\ket{\psi}-\bra{\psi}\Pi_{a_t}\ket{\psi}
\end{equation}
The expected cumulative regret $\Ex_{\pi}(R_T)$ can be obtained by taking the expectation value with respect to the probability distribution of 
\begin{equation}
    P_{\pi}(a_1,o_1,\ldots,a_T,o_{T}) \coloneqq \prod_{t=1}^T \pi_t(a_t|a_1o_1\ldots,a_{t-1},o_{t-1})P(o_t|a_t)~,
\end{equation}
governed by the policy. In the pure state scenario, the regret can be restated using fidelity and trace distance as:
\begin{equation}
\label{eq:cumu_regret}
\begin{split}
    R_{T}(\A,\psi,\pi) &= \sum_{t=1}^T 1-\bra{\psi}\Pi_{A_t}\ket{\psi}\\
    &=\sum_{t=1}^T1-F(\psi,\Pi_{A_t})\\
    &=\frac{1}{4}\sum_{t=1}^T \|\ket{\psi}\!\bra{\psi}-\Pi_{A_t}\|_1^2~,
\end{split}
\end{equation}
where the $\|A\|_1= \sqrt{A^\dagger A}$ is the trace norm (also called operator 1-norm or Schatten 1-norm), and $F(\psi,{\Pi_{A_t}})=\bra{\psi}\Pi_{A_t}\ket{\psi}$ can be interpreted as the quantum fidelity since $\psi$ is a pure state and $\Pi_{A_t}$ is a projective measurement.  The overarching goal is to identify a policy, $\pi$, that minimizes this cumulative regret, thereby ensuring that the sequence of actions converges to optimal measurement directions with respect to the underlying unknown quantum state. Recent work in QMAB has provided bounds and shown that for pure quantum states~\cite {lumbreras22bandit,lumbreras24pure}, there exists a learning algorithm that can upper bound the expected cumulative regret 
\begin{equation}
    \Ex\left[R_{T}(\A,\psi,\pi)\right] = O(\log^2(T))~.
\end{equation}

\section{Set-up}
Consider a model where the goal is to extract energy from a qubit system in contact with a thermal bath~\cite{skrzypczyk2014work,huang2023engines} and store it in a battery. The battery is modeled as the potential energy of a weight that can be raised and lowered. We will now formally introduce our model, starting with the three physical subsystems involved in the process.

\begin{enumerate}
    \item Unknown pure state source (System $Q$): A qubit in a pure state $\psi_Q = \ketbra{\psi}{\psi}$ with Hamiltonian $\mathcal{H}_Q=\omega\id$. This is the system from which free energy is to be extracted.
    
    \item Battery (System $B$): A semi-classical weight described by a continuous-variable state $\varphi(x) \in L^2(\mathbb{R})$. The battery Hamiltonian is defined as $\mathcal{H}_B \varphi(x) = x \varphi(x)$, where $x$ represents the height of the weight. The energy of the battery can be changed by translating the weight up by a certain height $\mathrm{d} x$, described by the translation operator $\Gamma^B_{\mathrm{d} x} \varphi(x) = \varphi(x-\mathrm{d} x)$.
    
    \item Thermal reservoir (System $R$): A heat bath at a fixed inverse temperature $\beta$, modeled as a supply of qubit states $\gamma_\beta(\nu) = Z(\nu)^{-1} e^{-\beta \mathcal{H}_R(\nu)}$, where the Hamiltonian is $\mathcal{H}_R(\nu) = \nu \ketbra{1}{1}$ and $Z(\nu)$ 
    is the partition function. Here, $\ket{0}$ and $\ket{1}$ are energy eigenstates and $\nu$ is an energy gap that can be tuned.
\end{enumerate}
We consider the system Hamiltonian $H_A$ to be fully degenerate. Discussion on the consequences of a non-degenerate Hamiltonian can be found back in~\cref{ch:work_extraction}.

Now, assume that the agent has oracle access to the unknown system $Q$ over a finite number of rounds $N\in\mathbb{N}$. The goal is to extract the maximum amount of free energy from $Q$ and store it in the battery $B$ through interactions with a thermal reservoir $R$. However, since the state $\psi$ of the system is unknown, the agent cannot extract work optimally from the outset. Instead, it must gradually improve its strategy by learning from each round, where the learning is done by measuring the battery. The work extraction protocol is defined by three key components: a policy that updates the agent’s guess $\psi_k$ of the unknown state $\psi$ at each round $k \in [N]$, a sequence $\{ \epsilon_k \}_{k=1}^N$ that predicts the accuracy of these guesses and the number of iteration $M\in\mathbb{N}$ that need to be performed to sufficiently approximate a quasi-static interaction between system, battery and thermal reservoir.
These parameters together determine the dissipation at each round, and we will discuss their choice subsequently. A diagrammatic representation of the operation of the agent can be found in Fig.~\ref{fig:protocol}.
\begin{figure}
    \centering
    \includegraphics[width=0.6\linewidth]{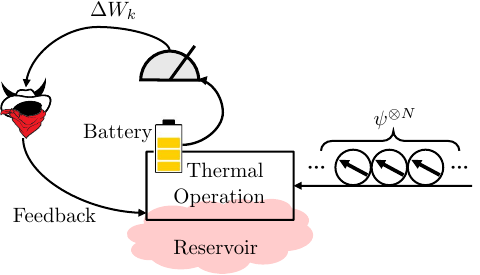}
    \caption{ Sketch of the sequential work extraction protocol with a thermal reservoir. At each time step $k \in [N]$, the agent receives a copy of an unknown qubit state $\psi$ and performs a thermal operation involving the reservoir and a battery. A measurement of the battery system is carried out to determine the extracted work $\Delta W_k$, which is then used as feedback to improve the extraction strategy in subsequent rounds. }
    \label{fig:protocol}
\end{figure}
\section{Operation of agent}
The details on how the sequence of $\{\epsilon_k\}_{k=1}^N$ and $M$ are determined will be discussed later on, along with justifications for such a choice. 
For now, suppose the agent is already given a sequence $\{\epsilon_k\}_{k=1}^N$ and fixed $M$, the general operation of the agent can be described as follows:

\begin{enumerate}
    \item The agent receives a sample with unknown state $\psi$.
    
    \item The agent selects a direction $\psi_k$ on the Bloch sphere that defines a basis $\lbrace \psi_k , \psi_k^\perp \rbrace$ for system $Q$. The direction $\psi_k$ is a function of past measurement choices and measured battery energies (see~Eq.~\eqref{eq:action_general_update}). One may think of $\psi_k$ as an estimate of $\psi$, but as we will see, it actually slightly deviates from the agent's best guess of $\psi$ so that the subsequent battery measurement reveals more information about $\psi$.
    
    \item The agent performs a unitary on the system qubit in the form of 
        \begin{equation}
            U_k = \ketbra{0}{\psi_k} + \ketbra{1}{\psi_k^\perp},
        \end{equation}
        satisfying $[\mathcal{H}_Q,U_k]=0$, note that if $\psi_k = \psi$, this unitary simply diagonalizes the system qubit in the computational basis.
    \item 
    The agent then performs a thermal operation by repeatedly applying an energy-conserving unitary on a combined system $Q B R$ to transfer energy from a reservoir qubit $R$ to the battery $B$ mediated by $Q$. 
    Specifically, for each iteration $\ell \in [M]$, the agent:
    \begin{enumerate}
        \item sets the energy gap to $\nu (\ell,\epsilon_k )$ (see Eq.~\eqref{eq:intro_gap_parametrization_2}) and selects a reservoir qubit in the thermal state $\gamma_\beta(\nu(\ell,\epsilon_k))$;
        \item applies the energy-conserving unitary
        \begin{align}\label{eq:unitary_intro_ch6}
            V_{\psi_k , \ell} = \sum_{i,j} \ket{i}\!\bra{j}_Q \otimes \ket{j}\!\bra{i}_R \otimes \Gamma^B_{(i-j)\nu(\ell,\epsilon_k)} \,;
        \end{align}
       \item discards the reservoir qubit.
    \end{enumerate}
    Note that these $M$ iterations of the swap operation are designed to mimic a quasi-static process, thereby minimizing entropy production.
    
    \item Finally, the agent measures the energy of the battery in its eigenbasis and records the outcome $\mu_k$.
\end{enumerate}
Note that steps 2-5 of this protocol are exactly what was presented in~\cref{ch:work_extraction} as a realization of the $\rho^*$-ideal work extraction protocol. Aside from the choice of the energy gap mentioned in step 4, where the accuracy parameter $\epsilon_k$ comes into play. Here, the energy gap is defined as
\begin{align}\label{eq:intro_gap_parametrization_2}
    \nu (\ell,\epsilon_k ) = \beta^{-1} 
    \ln \left( \frac{1 - \frac{\ell}{2M}- \left(1 - \frac{\ell}{M} \right) \epsilon_k }{\frac{\ell}{2M}+\left(1 - \frac{\ell}{M} \right) \epsilon_k} \right),
\end{align}
where $\epsilon_k \in [0,1]$ controls the mismatch. The full interaction of systems $Q$, $B$, and $R$ through the unitary $V_{\psi_k,\ell}$ is illustrated in Figure~\ref{fig:protocol_pic}.
\begin{figure}
    \centering
    \includegraphics[width=0.9\linewidth]{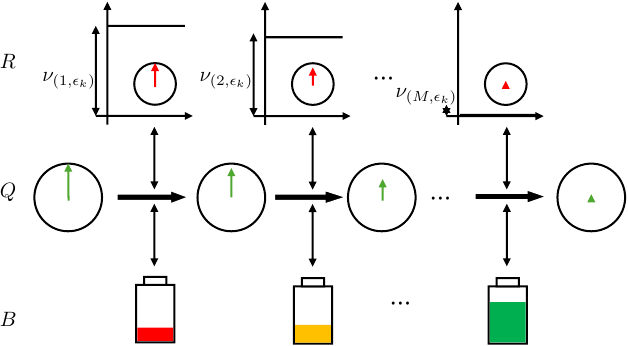}
    \caption{Illustration of the iterations of the thermal operation in the full system $QBR$, where arrows represent Bloch vectors of states, showing that the system qubit becomes more and more mixed as the process goes. The energy gaps $\{\nu_{k,i}\}_i$ of successive reservoir Hamiltonian forms a strictly decreasing sequence, making the successive thermal states more mixed. At each step, we take a new qubit from the reservoir and swap the system qubit with the reservoir qubit; the energy from the reservoir will flow into the battery.
    At the end of the process, the qubit in system $Q$ is the thermal state.}
    \label{fig:protocol_pic}
\end{figure}

Since this protocol is a member of the $\rho^*$-ideal work extraction protocol, the work distribution as well as the respective probability will be as stated in Theorem~\ref{thm2}. In other words, in the quasi-static limit $M\to\infty$, the work values are: 
\begin{align}
    \label{eq:work_values_intro}
    \begin{split}
        w_{k,0} &:=  \beta^{-1}(\D(\psi_{k}\|\id/2) +  \ln (1-\epsilon_k))~, \\
        w_{k,1} &:= \beta^{-1}(\D(\psi_{k}^\perp\|\id/2) +  \ln \epsilon_k)~.
    \end{split}
\end{align}
and the probability distribution follows
\begin{equation}\label{eq:work_distribution_intro}
  \mathrm{Pr}(  \Delta W_k = w_{k,i} ) = \begin{cases}
      |\braket{\psi}{\psi_k}|^2       &i=0\\
      1-|\braket{\psi}{\psi_k}|^2&i=1~,
  \end{cases}
\end{equation}
where $\Delta W_k = \mu_{k+1} - \mu_k$ is the energy injected into the battery at round $k\in[N]$. For the detailed analysis, including derivation of the work distribution, rate of convergence with respect to $M$ as well as the error probability, please refer to Appendix~\ref{app:convergence_rate}.

Since this extraction protocol is indeed a member of the $\rho^*$-ideal work extraction protocol, the expected work extracted at round $k\in[N]$ will be given by 
\begin{equation}
   \label{eq:avgwork}
 \Ex[\Delta W_k] = \beta^{-1} \left[ \D (\psi \| \id/2) - \D (\psi \| \Delta_{2\epsilon_k}(\psi_k)) \right], 
\end{equation}
where $\Delta_{\epsilon}(\rho) = (1 - \epsilon)\rho + \epsilon \id/2$ denotes the depolarizing channel. The maximal expected work per round, $\beta^{-1} \D(\psi \| \id/2)$, is achieved when the agent perfectly estimates the state, i.e., when $\psi_k = \psi$ and $\epsilon_k = 0$. Accordingly, we define the dissipation at round $k$ as:
\begin{align}
\dissipation^{(k)} :=&\  \beta^{-1} \D(\psi \| \id/2) - \mathbb{E}[\Delta W_k]\\
=&\ \beta^{-1} \D(\psi \| \Delta_{2\epsilon_k}(\psi_k)).
\end{align}
This expression highlights a key trade-off: to reduce dissipation, the agent must align $\psi_k$ with the true state $\psi$, but it cannot set $\epsilon_k$ too small unless the estimate is sufficiently accurate, as otherwise the divergence becomes unboundedly large. The parameter $\epsilon_k$ thus plays a dual role, quantifying both the uncertainty in the estimate and its thermodynamic penalty. The agent's objective is to extract the maximum amount of work into the battery using the $N$ copies of the unknown system. Equivalently, the agent aims to minimize the cumulative dissipation over $N$ rounds, which is given by
\begin{equation}
    \label{eq:dissipation_sc}
    \dissipation(N) = \beta^{-1} \sum_{k=1}^N \dissipation^{(k)}~.
\end{equation}
The question still remains, how should the agent decide on what the estimate $\psi_k$ should be, given all the past information, and what value should the parameter $\epsilon_k$ take to balance between minimization of dissipation and preventing the dissipation from diverging? 

\section{Extracting while learning}
To determine the choices of the estimate $\psi_k$, the accuracy $\epsilon_k$ and number of iterations within the extraction protocol, $M$ we start from the observation that the sequence of battery measurement outcomes follows the same probability distribution as the outcome of a projective measurement $\ket{\psi_k}\!\bra{\psi_k}$ performed on an unknown quantum state $\psi$. This equivalence means the extracted work statistics can be used analogously to measurement outcomes for performing state tomography on $\psi$. 
Thus, the agent aims to select $\psi_k$ and $\epsilon_k$ to minimize $\D(\psi \| \Delta_{2\epsilon_k} (\psi_k))$. For updating $\psi_k$, we adopt the fully adaptive state tomography algorithm introduced in~\cite{lumbreras24pure} which builds on the multi-armed bandit method from~\cite{pmlr-v247-lumbreras24a}. 
The full details of this algorithm are provided in Appendix~\ref{sec:learning_algo}, but here we give an overview and highlight the main results.

The sequence $\lbrace \epsilon_k \rbrace_{k=1}^N$ quantifies the accuracy of the state estimates $\psi_k$, while the parameter $M$ controls the success probability of the work extraction at each step. Their values are chosen to ensure a provable bound on the total dissipation, as detailed in the following subsections.

The learning algorithm assumes sequential access to $N$ identical copies of an unknown pure qubit state $|\psi\rangle$. At each round $k\in[N]$, it performs an adaptive single-copy measurement using a rank-1, two-outcome POVM. Specifically, the algorithm chooses a measurement direction $\psi_k$ and performs a projective measurement in the basis $\lbrace \psi_k, \psi_k^\perp \rbrace$.

The binary measurement outcome (or ``reward'') $r_k\in\lbrace 0,1\rbrace$ is distributed according to Born's rule
\begin{align}
    \label{eq:reward_intro_bandit}
    \mathrm{Pr} (R_k = r_k) = 
    \begin{cases}
    |\langle\psi_k|\psi\rangle|^2, \quad  & r_k=1~,\\
    1-|\langle\psi_k|\psi\rangle|^2, \quad & r_k=0~.
    \end{cases}
\end{align}
These outcomes exactly match the energy distribution from the work extraction protocol in Eq.~\eqref{eq:work_distribution_intro}, with $r_k = i$ if $\Delta W_k = w_{k,1-i}$ for $i\in\lbrace 0,1 \rbrace$.

To ensure the measurement directions $\psi_k$ remain close to the unknown state $\psi$, thereby bounding the relative entropy, we apply a bound adapted from~\cite[Proposition 2.34]{Flammia2024quantumchisquared} which yields:
\begin{equation}
    \D(\psi\|\Delta_{2\epsilon_k} (\psi_k))\leq16\epsilon_k (2-\ln \epsilon_k)~,
\end{equation} 
where the infidelity satisfies $1 - |\langle \psi | \psi_k \rangle|^2 \leq \epsilon_k \leq \frac{1}{2}$. This leads to an upper bound on the cumulative dissipation over $N$ rounds:
\begin{align}\label{eq:dissipaiton_infidelity}
       \dissipation(N) \leq \beta^{-1}\sum_{k=1}^N 16\epsilon_k (2-\ln \epsilon_k)~.
\end{align}
\subsection{Determination of estimator and parameters}
The update of $\psi_k$ reflects the classic exploration-exploitation trade-off described in Section \ref{sec:MAB}. Each $\psi_k$ depends adaptively on past measurement outcomes $\lbrace r_l \rbrace_{l=1}^{k-1}$ and measurement directions $\lbrace \psi_l \rbrace_{l=1}^{k-1}$. Instead of a detailed mathematical derivation, we will present the physical intuition to aid in understanding. The full algorithm, along with the mathematical details including the form of the estimator, can be found in Appendix~\ref{sec:learning_algo}; readers may also refer to \cite{lumbreras24pure}. 

A series of weighted least-squares estimators can be constructed from the history of past rewards and outcomes. As shown in recent work~\cite{pmlr-v247-lumbreras24a,lumbreras24pure}, these estimators effectively boost learning along directions of low reward variance—those most closely aligned with the unknown state $\psi$.  This allows the algorithm to ``exploit" its current knowledge in pursuit of higher rewards. The evolving confidence in these estimates can be represented by a shrinking confidence region $\mathcal{C}_t$, where $t$ is the number of estimators constructed. Simultaneously, the algorithm balances this exploitation with "exploration": it actively selects two Bloch vectors (equivalently, two measurement directions $\Pi^{(+)}_{k},\Pi^{(-)}_{k}$) corresponding to the directions of highest uncertainty in the reward variance. These measurements target regions where information about $\psi$ is most lacking, thereby improving subsequent estimates. This trade-off between exploration and exploitation is illustrated in Fig.~\ref{fig:bandit_ball}. It is important to note that while $\Pi^{(+)}_{k},\Pi^{(-)}_{k}$ are not themselves the updated state $\psi_k$, they influence how $\psi_k$ is constructed.

\begin{figure}
    \centering
    \includegraphics[width=0.4\linewidth]{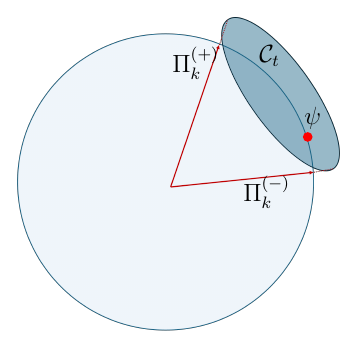}
    \caption{Diagrammatic representation of the learning process. Based on past observations, the algorithm constructs a confidence region $\mathcal{C}_t$ around the unknown state $\psi$. It then selects two directions $\{\Pi^{(+)}_{k},\Pi^{(+)}_{k}\}$to probe the space of maximum reward uncertainty, influencing the next state estimate $\psi_{k}$.}
    \label{fig:bandit_ball}
\end{figure}

Within the thermal work extraction protocol discussed in \cref{ch:work_extraction}, when $M$ is finite, i.e., the non-quasi-static regime, the reward distribution may deviate from the idealized form, potentially affecting the algorithm’s success probability. Nevertheless, the learning protocol guarantees success with probability at least $1-\delta$. By choosing $\delta =\frac{1}{N}$, the tolerance on the error probability for the reward will be upper bounded by $O(\frac{1}{N^2})$. Returning to the error condition established in Eq.~\eqref{eq:error_bound}, we can ensure this bound is satisfied by selecting $M=D N^2/\ln(N)^3$ for some constant $D>0$. This guarantees that the cumulative reward error remains below the error threshold of the tomography-based estimation.

For setting $\epsilon_k$, the algorithm guarantees that the infidelity between the estimated and true state is bounded with high probability. Specifically, with probability at least $(1-\delta)$, the following bound holds:
\begin{equation}
\label{eq:infidelity_bound}
    1-|\braket{\psi_k}{\psi}|^2 \leq C\frac{\ln (N/\delta)}{k}~,
\end{equation}
where $C\geq0$ is a large enough constant. This inequality also holds in expectation. Recalling Eq.~\eqref{eq:dissipaiton_infidelity}, where $\epsilon_k$ has to be bounded such that $1-|\langle\psi|\psi\rangle|^2\leq \epsilon_k\leq 1/2$. Using the bound we obtained from the algorithm, we can now set $\epsilon_k = \min (C\ln(N/\delta)/k,1/2)$. Substituting this into Eq.~\eqref{eq:dissipaiton_infidelity}, the total dissipation over $N$ steps is bounded by:
\begin{equation}
    \begin{split}
        \dissipation(N) &\leq \beta^{-1}\sum_{k=1}^N 16\epsilon_k (2-\ln \epsilon_k)\\
        &= \beta^{-1}\sum_{k=1}^N 16C\frac{\ln(N/\delta)}{k}\left(2-\ln \left(C\frac{\ln(N/\delta)}{k}\right)\right)\\
        &=\beta^{-1}\sum_{k=1}^N 16C\frac{\ln(N/\delta)}{k}\left(\underbrace{2-\ln C-\ln^2\left(\frac{N}{\delta}\right)}_{A}+\ln k\right)\\
        &=16C\beta^{-1}\ln\left(\frac{N}{\delta}\right)\left[A\sum_{k=1}^N\frac{1}{k}+\sum_{k=1}^N\frac{\ln k}{k}\right]
    \end{split}
\end{equation}
The first summation is a finite Harmonic sum, which can be approximated using $\sum_{k=1}^N \frac{1}{k}\approx \ln N+\gamma_E$, where $\gamma_E\approx0.577$ is the Euler–Mascheroni constant. The second summation can be approximated using integral $\sum_{k=1}^N\frac{\ln k}{k}\approx \int_{1}^N\frac{\ln x}{x} \mathrm{d} x=(\ln N)^2/2$. Combining these, we obtain that
\begin{equation}
    \begin{split}
        \dissipation(N) &\leq16C\beta^{-1}\ln\left(\frac{N}{\delta}\right)[A(\ln N+\gamma_E) + (\ln N)^2/2]~.
    \end{split}
\end{equation}
Notice that the expression in square brackets is dominated by $\frac{(\ln N)^2}{2}$, which then yields our main result.
\begin{theorem} [Cumulative Dissipation Scaling]
There exists an explicit protocol for the semi-classical battery model that adaptively updates the estimate $\hat{\psi}_k$ and the probe state $\psi_k$ based on the rewards $\lbrace r_s \rbrace_{s=1}^{k-1}$, achieving, with probability at least $1-\delta$ 
\begin{align} 
\dissipation (N) = O\left(\beta^{-1}\ln^2(N) \ln \left( \frac{N}{\delta} \right)  \right)~,
\end{align} 
where $N$ is the number of copies of the unknown pure state,$\psi$, available that provide the non-equilibrium free energy to be extracted.
\end{theorem}
One could go a step further to demand that $\delta=1/N$, which means as the number of sample $N$ increases, the acceptable error probability decreases accordingly. This would then yield a cumulative dissipation of
\begin{equation}
    \dissipation (N) = O\left(\beta^{-1}\ln^2(N) \ln(N)\right)~,
\end{equation}
with a probability of at least $1-\frac{1}{N}$.

\section{Extracting after learning}
One may naturally ask: can we achieve the same dissipation by using quantum state tomography to estimate $\psi$, and then extract work based on that estimate? The answer is no. Any tomography algorithm that achieves a relative entropy error $\zeta$ requires at least $N = \Omega(1/\zeta)$ copies of the unknown state~\cite{haah2016sample,Flammia2024quantumchisquared,chen_adaptivity}. Suppose we use a fraction $\alpha N$, with $0 \leq \alpha \leq 1$, to obtain an estimate with relative entropy error of at least $\zeta = \Omega(\frac{1}{\alpha N})$. The agent can then extract work based on that estimate for the remaining $(1-\alpha)N$ time steps, resulting in at least $\frac{1-\alpha}{\alpha }$ dissipation. 
The total dissipation will be lower bounded as
\begin{align}
\dissipation(N) = \Omega\left( \alpha N + \frac{(1 - \alpha)}{\alpha} \right) ,
\end{align}
and minimizing over $\alpha$ yields the optimal $\alpha=\frac{1}{\sqrt{N}}$, the dissipation then scales with $\dissipation(N) = \Omega(\sqrt{N})$. This bound is tight as it can be matched by an upper bound $\dissipation(N) = O(\sqrt{N})$ using an optimal learning algorithm for relative entropy~\cite{Flammia2024quantumchisquared}. This scaling reflects a fundamental limit imposed by sample complexity of state tomography. Achieving lower dissipation requires a more adaptive strategy that goes beyond standard tomography, like the one we exhibit here. 

These conclusions are inherently tied to the agent's reliance on classical memory. However, the landscape changes significantly when the agents are allowed to have quantum memory; as shown in~\cite{huang2022quantum}, a single qubit of memory can trigger an exponential leap in learning capability. Consequently, if the agent were permitted a finite quantum memory, we would expect a decrease in cumulative dissipation.

\section{Numerical simulation}
To empirically validate our findings, we conduct numerical simulations confirming that the cumulative dissipation of our adaptive algorithm is strictly bounded by $O(\polylog(N))$. The simulation comprises 50 independent trials. In each trial, the target pure qubit is drawn from a Haar-random distribution on the Bloch sphere. We randomly select four initial measurement bases and set the algorithmic parameter to $t \approx \ln(30000) = 10$. The combined learning and work extraction protocol is then executed as outlined in Appendix~\ref{sec:learning_algo}. We quantify the cumulative dissipation by summing the deficit between the theoretically maximal extracted work ($k_B T \ln 2$) and the actual work extracted at each step. The results are visualized in Fig.~\ref{fig:dissipation_scaling}, with mean values plotted as blue dots enclosed by translucent $95\%$ confidence intervals. As a benchmark, we plot the cumulative dissipation of a standard tomography-first approach (red markers and shading), which allocates a fixed fraction of the ensemble exclusively for state estimation. A least-squares regression on the adaptive protocol data rigorously supports the $O(\polylog(N))$ scaling, yielding a coefficient of determination of $R^2 = 0.9958$. 
\begin{figure*}[t]
    \centering
\includegraphics[width=\textwidth]{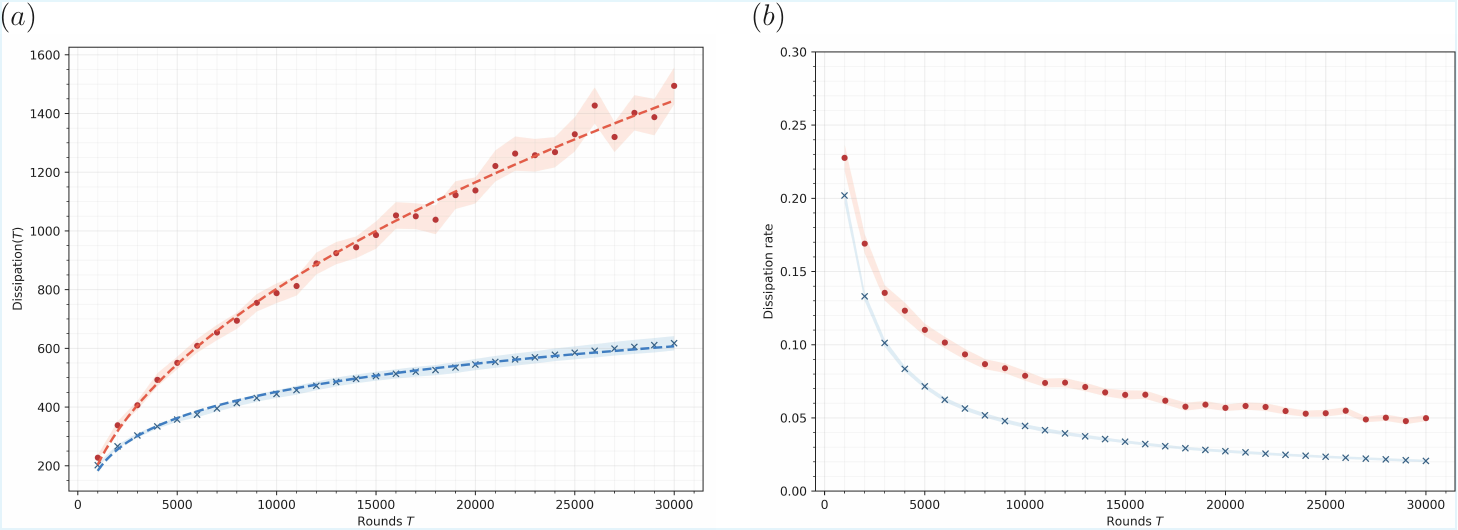}
 \caption{\textbf{Performance scaling of the adaptive work extraction protocol.}
Cumulative dissipation (a) and dissipation rate (b) versus the number of rounds $T$ (rate = average dissipation per copy). Blue: our adaptive protocol. Red: a tomography-first baseline that uses $O(1/\sqrt{T})$ of the available copies for learning and then applies the state-aware extraction protocol on the remainder.
For our protocol, we probe four directions per update and repeat each direction for $t=5$ trials, hence the estimator is updated every $4t=20$ rounds. While $t$ is mildly horizon-dependent in the theory (to guarantee high-probability confidence bounds), empirically $t=5$ already yields the same scaling, making the implementation effectively horizon-independent.
Points are averages over 50 independent trials with a random pure state per trial; for the tomography-first baseline, each $T$ is simulated as an independent experiment (since the strategy depends explicitly on the horizon).
Dashed lines show least-squares fits: polylogarithmic $a(\log_{10} T)^2+b$ for our protocol ($R^2=0.9958$, $a=38.4\pm0.47$, $b=-163\pm8$) and square-root $a\sqrt{T}+b$ for tomography-first ($R^2=0.9930$, $a=8.76\pm0.14$, $b=-74.0\pm17$). Shaded regions indicate 95\% confidence intervals; regression statistics were computed with \texttt{scipy.stats.linregress}.}
    \label{fig:dissipation_scaling}
\end{figure*}

\section{Robustness of result}
Notice that in this model, measurements in the energy eigenbasis of the battery were performed. According to Landauer's principle, such measurements reduce entropy and thus incur an energetic cost associated with information erasure. This form of dissipation --- arising from the cost of resetting memory---should also be taken into account. In Appendix~\ref{sec:landauer} we show that the dissipation required by Landauer's principle also scales as $O(\polylog(N))$ when our adaptive algorithm is employed. Furthermore, this algorithm is not specific to the work extraction protocol discussed above. We have also applied it to an alternative protocol modeled using the Jaynes-Cummings Hamiltonians, where the battery is represented as a uniform energy ladder, and the work dissipation similarly exhibits $O(\polylog(N))$ scaling. Details can be found in Appendix~\ref{apd:jc_protocol}.

\section{Summary}
In this chapter, we introduced a novel perspective by framing sequential work extraction from an unknown quantum state as an instance of the exploration–exploitation trade-off. We quantified the cumulative dissipation in the finite-copy regime and established a rigorous upper bound on the total dissipated work. Crucially, we constructed an explicit, adaptive algorithm that saturates this bound and achieves exponential improvement over conventional tomography-based approaches.

%% file: Chapters/Chapter7.tex

\chapter{Conclusions and Outlook} 
\label{ch_conclusions}
\noindent\rule{\textwidth}{0.4pt}
\vspace{1.5em} 
\setlength{\parindent}{4ex}

\noindent \textit{This concluding chapter summarizes the work presented throughout the thesis, highlighting the substantial advancements made in the field of not only quantum thermodynamics but also that of computational mechanics. It provides a comprehensive overview of the key findings and also paves the way for future explorations. We will not only focus on some of the opening questions that arise during the research process but also discuss how the results and techniques developed in the paper can be used for future works.}
\newpage

\section{Conclusions}
This thesis presents a comprehensive exploration of agential approaches to work extraction in the quantum regime. The results are not only relevant to foundational questions in quantum thermodynamics, but also establish promising connections with techniques from computational mechanics and reinforcement learning, suggesting new directions for interdisciplinary applications.

In \cref{ch:work_extraction}, we introduced a new class of work extraction protocols, termed the $\rho^*$-ideal work extraction protocol, motivated by the perspective of an agent lacking complete knowledge of the underlying quantum state.  We demonstrated both the simplicity of the resulting work distributions and the physical interpretability of the final extracted work. The limitations of this protocol were also discussed, along with its role in the broader framework developed throughout the thesis.

\cref{ch:3} formally introduced the agential framework for work extraction, emphasizing causal constraints and the absence of quantum memory. This framework was applied to the problem of extracting work from temporally correlated states. We introduced the concept of an agent’s belief state and its critical role in determining the work extraction protocol. Notably, we highlighted the interplay between the meta-dynamics of the agent’s belief and the hidden-state dynamics of an underlying hidden Markov model (HMM), and how this interplay governs the agent's performance. We showed that such an adaptive agent outperforms the best non-adaptive strategy, and uncovered an apparent phase transition that demarcates regimes where adaptivity significantly enhances performance.

In \cref{ch:4}, we extended our analysis by deriving a fundamental upper bound on work extraction from temporally correlated quantum states under causal constraints. We introduced the concept of Time-Ordered Free Energy (TOFE), a non-trivial bound that is strictly lower than the conventional non-equilibrium free energy in general settings. This gap in performance could potentially be quantified by a new measure of discord. To achieve this bound, we proposed a constructive algorithm based on dynamic programming (DP), enabling an agent to optimally exploit temporal correlations. Furthermore, we identified the agent introduced in \cref{ch:3} as a local-optimizing (LO) agent, and demonstrated that the DP-based agent outperforms the LO agent in both extractable work and predictive power. Finally, we offered physical insights into how the structure of the underlying quantum process guides the emergence of optimal policies.

\cref{ch:learning} addressed the setting where an agent interacts with a sequence of identically prepared but unknown pure quantum states.  We framed this scenario as an instance of the exploration–exploitation trade-off and established a connection to multi-armed bandit problems. Leveraging this connection, we developed an adaptive algorithm for optimizing work extraction. Remarkably, the resulting cumulative dissipation over $N$ copies exhibits exponential improvement compared to conventional tomography-based approaches.

\section{Future works}
\subsection*{Questions in computational mechanics}
This thesis investigated the optimal strategy for an agent to extract work from states generated by arbitrary classical hidden Markov models (HMMs). In the domain of predictive modeling, it is well-established that a fundamental asymmetry often exists between modeling a process forward versus backward in time, even when both directions yield sequences with the same entropy rate~\cite{marzen2015informational,thompson2018causal,kechrimparis2023causal}. This phenomenon, known as causal asymmetry, reflects a difference in statistical complexity between prediction and retrodiction. Beyond its information-theoretic implications, causal asymmetry can also be viewed through a thermodynamic lens. For instance, in bipartite quantum systems, the quantum discord obtained by measuring subsystem $A$ differs from that obtained by measuring subsystem $B$, i.e.,
\begin{equation}
    \delta(A;B) \neq \delta(B;A)~.
\end{equation}
Such asymmetries suggest an inherent directional bias in the underlying correlations. By extension, it would seem natural that the adaptive ordered discord would also exhibit such asymmetry. This idea can be further explored by computing the Time-Ordered Free Energy (TOFE) for both the forward and backward processes. Unlike the entropy rate, which remains invariant under time reversal, TOFE captures directional differences in extractable work, providing a refined thermodynamic signature of causal asymmetry.

A compelling direction for future work is the analytical characterization of the phase boundaries observed in the performance of both the local-optimizing and globally optimal agents. The presence of such macroscopic transitions hints at fundamental information-theoretic limits inherent to quantum learning tasks. To rigorously derive these boundaries, one could leverage the framework of Ref.~\cite{liu2022partially}, which establishes that a quantum process is efficiently learnable only if it is sufficiently $\alpha$-revealing. In this context, the parameter $\alpha$ is governed by the minimum eigenvalue of the associated observation matrix—a mathematical condition that directly quantifies the physical distinguishability of the constituent quantum states within the process.

\subsection*{Relation to fluctuation theorem or thermodynamic uncertainty relations}
The standard fluctuation theorem encodes that the probability of entropy production $\Sigma$ in a forward process and a backward process can be expressed as
\begin{equation}
    \frac{\Pr(\Sigma)}{\Pr(-\Sigma)}=e^\Sigma~,
\end{equation}
which physically implies that a forward trajectory that produces entropy is exponentially more likely than a backward trajectory that reduces said entropy, a mathematically rigorous way one may take to look at the arrow of time~\cite{jarzynski1997nonequilibrium,crooks1999entropy,campisi2011colloquium}. In the standard Fluctuation theorem, the different trajectories (say a state evolving from $\rho_0$ to $\rho_\tau$) taken during the evolution are assumed to be i.i.d in the sense that the other trajectory sampled has no correlation to this trajectory. It would be very interesting to see if the sequence of initial state ($\{\rho^{(i)}_0\}_{i=1}^N$) exhibits temporal correlation might break this particular bound, since now the agent's memory is somewhat synchronized to the process, this correlation might affect and break the standard fluctuation theorem and require some additional factor such as information flow within the exponent.

Thermodynamic uncertainty relations, on the other hand, aim to characterize the precision of integrated thermodynamic current $J$ within an evolution that is upper bounded by the entropy production, $\Sigma$, incurred in the transformation~\cite{barato2015thermodynamic,gingrich2016dissipation}. I.e.
\begin{equation}
    \frac{\text{Var}(J)}{\langle J\rangle ^2}
    \geq \frac{2k_B}{\Sigma}
\end{equation} 
In our work, the cumulative work extracted can be considered one of these currents, and its fluctuation or uncertainty would be bounded by entropy production. However, the standard TUR is assumed to be Markovian, whereas in our scenario, the agent actually carried memory from the last extraction to the next time step. This would induce non-Markovianity and could cause the TUR to be violated. 
\subsection*{Continuously monitored system}

In the current framework, the HMM is assumed to be discrete in time and hence emits a quantum state at fixed time-boxes $t$. It would be interesting to see the generalization to a continuous-time stochastic process, where the frequency of the agent's interaction may affect the final work output. There are a few ways to do this; we will briefly mention a few.

The most intuitive way is to consider a truly quantum stochastic process where the temporal correlation is generated by some unknown environment~\cite{strasberg2019operational}. At any time $\tau$, the agent can choose to interact with the environment and carry out some joint operation on the environment and perhaps a battery. In this scenario, the more frequently the agent interacts with the environment, the less evolution is going to occur, and hence the temporal correlation may no longer exist. This freezing of dynamics is known as the Quantum Zeno effect~\cite{misra1977zeno,itano1990quantum}. This can also be avoided by considering weaker coupling between the environment and the battery~\cite{facchi2008quantum}. 

On the other hand, a discrete sequence of quantum states can be mapped onto the continuum limit of a quantum collision model~\cite{strasberg2017quantum}. In this framework, the non-i.i.d. source is modeled as a continuous field, where infinitesimal 'time-bins' of duration $\text{d}t$ act as the individual quantum subsystems~\cite{gardiner1985input}. A classical continuous-time Hidden Markov Model modulates the preparation of these time bins, embedding temporal correlations into the continuous field. The agent then continuously couples to the passing time bins via some local interactions. This formulation avoids the Zeno-freezing of the source, as the agent continuously interacts with fresh, independent temporal modes of the field. However, optimizing work extraction in this regime requires moving beyond discrete Bayesian updates to the formalism of continuous quantum filtering, where the agent must continuously infer the hidden state of the HMM from weak measurements on the post-interaction field~\cite{wiseman2009quantum}.
\subsection*{Extension to a true quantum stochastic process}
So far, we have focused on quantum states generated by classical stochastic processes. A natural extension of this analysis involves considering genuinely quantum stochastic processes, formalized through the framework of quantum combs or process tensors~\cite{chiribella2008quantum,pollock2018non,milz2021quantum}. 

In this setting, the underlying dynamics are governed by a sequence of unitary interactions with an unknown environment, which acts as a memory system that mediates temporal correlations. Crucially, in such quantum processes, the agent's actions can influence the future evolution of the system—a feature absent in the classical HMM-based models considered thus far. This opens up new avenues to study causal influence, feedback control, and adaptive strategies in the quantum regime. A particularly promising direction is to adopt the operational framework proposed in~\cite{strasberg2019operational}, which reformulates the first and second laws of thermodynamics within the process tensor formalism, allowing for a consistent thermodynamic interpretation of memory effects and agent-environment interactions in non-Markovian quantum processes.
Furthermore, this extended framework provides a natural platform to analyze trade-offs between predictive synchronization of the agent’s belief and the amount of extractable work.

\subsection*{Application of quantum algorithms}
The current implementation of the dynamic programming (DP) algorithm, while systematic and effective, suffers from an exponential scaling with the number of time steps $T$, leading to significant computational overhead. This computational bottleneck poses a challenge for practical applications involving long time horizons. To address this, one promising direction is to explore quantum algorithms that may offer improved efficiency. In particular, it may be possible to construct a quantum-enhanced dynamic programming framework that leverages quantum parallelism to reduce the computational complexity of optimization. 

\subsection*{Generalized resources}
The learning algorithm presented in \cref{ch:learning} does more than provide an efficient way to learn quantum states—it establishes a pathway for work extraction from unknown and dynamic quantum sources, potentially extending well beyond the domain of free energy. This opens the door to extracting other forms of quantum resources such as:
\begin{itemize}
    \item Quantum coherence, which captures superposition relative to a reference basis,
    \item Entanglement, enabling non-classical correlations across subsystems, and
    \item Quantum magic, relevant to fault-tolerant quantum computation and resource theories of non-stabilizerness.
\end{itemize}
These extensions would involve integrating appropriate monotones and operational criteria from their respective resource theories, thus embedding work extraction in a broader theoretical landscape.
\subsection*{Learning of non-i.i.d. sequences}
While our current formulation assumes access to identically prepared but otherwise unknown quantum states, a significant and promising future direction lies in relaxing this assumption. By extending the learning strategy to temporally correlated quantum states—particularly those exhibiting non-Markovian or memoryful dynamics—we can begin to quantify and exploit the thermodynamic value of temporal correlations themselves.
Learning algorithm which accomplishes these already exist, but has yet to be applied to the scenario of thermodynamics~\cite{fanizza2023learning}. Such an extension would allow us to move beyond static state learning toward dynamic adaptation, where the history of interaction plays a crucial role in optimizing performance. This line of inquiry could ultimately lead to a unified framework for understanding how memory, predictability, and quantum coherence jointly determine the extractable thermodynamic resources in sequential quantum processes.

%% file: Appendices/Appendix3.tex

\chapter{Proofs for Chapter 4} 

\label{AppendixA} \section{Synchronizing to a memoryful quantum source}

Inferring the latent state of a known memoryful quantum source
allows maximal work extraction
when operating serially on the quantum states of the process.
The optimal state of knowledge, given a sequence of observations $o_1 o_2 \dots o_t$ obtained via interventions on the sequence of quantum systems $\sigma^{(x_1)}, \sigma^{(x_2)}, \dots \sigma^{(x_t)}$ is the conditional probability distribution induced by these interventions,
\begin{align}
\gbm\eta_t := \Pr(S_t | O_1 \dots O_t = o_1 \dots o_t, S_0 \sim \gbm\pi ) ~.
\label{eq:MxStDef}
\end{align}
The last condition $S_0 \sim \gbm\pi$ means that the initial latent state of the generator is distributed as $\gbm\pi$.
This can be rewritten as 
$\gbm\eta_t = \sum_s \gbm\pi(s) \Pr(S_t | O_1 \dots O_t = o_1 \dots o_t, S_0 =s)$ for $t>0$.
Recall that $\gbm\pi = \gbm\pi \sum_{x \in \mathcal{X}} T^{(x)}$ is the stationary distribution over the states of the generator.
Thus, $\gbm\eta_0 = \gbm\pi$.


If we 
introduce a new random variable $K_t$ to denote the optimally updated state of knowledge about the latent state of the
pattern generator, 
then we can replace the condition $S_{t-1} \sim \gbm\eta_{t-1}$
with
$K_{t-1} = \gbm\eta_{t-1}$.
The condition on the state of knowledge is relevant to the extent that the choice of POVM is influenced by the state of knowledge.
We 
remind the reader that in our framework
the POVM on the current
quantum output is chosen as a function of the state of knowledge $K_t$.

Note that the current
quantum output only depends on the current latent state of the process.  Accordingly,
the next observation---which is the outcome of the 
POVM on the current
quantum output---is conditionally independent of all previous outputs, given the current latent state and given the state of knowledge induced by all previous outputs.

We will now show that the optimal state of knowledge 
is recursive.
I.e., we will show that: 
\begin{align}
\gbm\eta_t 
& = 
\Pr(S_t | O_t = o_t, S_{t-1} \sim \gbm\eta_{t-1} ) ~.
\label{eq:RecursiveMxSt}
\end{align}
This follows from marginalizing over intervening latent states, employing Bayes' rule, and recognizing that the belief state $\gbm\eta_t$ is a function of the observations up to that time $o_{1} \dots o_t$. Starting from Eq.~\eqref{eq:MxStDef}, we find:
\begin{align}
\gbm\eta_t 
&:= 
\Pr(S_t | O_{1:t}= o_{1:t}, S_0 \sim \gbm\pi ) 
\nonumber \\
&= 
\sum_s 
\Pr(S_t, S_{t-1} = s | O_{1:t} = o_{1:t}, S_0 \sim \gbm\pi ) \\
&= 
\frac{
\sum_s 
\Pr(S_t, O_t = o_t, S_{t-1} = s | O_{1:t-1} = o_{1:t-1}, S_0 \sim \gbm\pi ) }{\Pr(O_t = o_t | O_{1:t-1} = o_{1:t-1}, S_0 \sim \gbm\pi )}
\\
&= 
\frac{
\sum_s 
\Pr(S_{t-1} = s | O_{1:t-1} = o_{1:t-1}, S_0 \sim \gbm\pi )
\Pr(S_t, O_t = o_t | O_{1:t-1} = o_{1:t-1}, S_0 \sim \gbm\pi , S_{t-1} = s) }{\sum_{s'} \Pr(O_t = o_t, S_{t-1} = s' | O_{1:t-1} = o_{1:t-1}, S_0 \sim \gbm\pi )}\\
&= 
\frac{
\sum_s 
\gbm\eta_{t-1}(s)
\Pr(S_t, O_t = o_t | K_{t-1} = \gbm\eta_{t-1}, S_{t-1} = s) }{
\sum_{s'} 
\gbm\eta_{t-1}(s')
\Pr(O_t = o_t | K_{t-1} = \gbm\eta_{t-1}, S_{t-1} = s') }
\\
&=
\Pr(S_t | O_t = o_t, K_{t-1} = \gbm\eta_{t-1} )
~.    
\end{align}
Hence, we have obtained Eq.~\eqref{eq:RecursiveMxSt}
from Eq.~\eqref{eq:MxStDef} as promised.

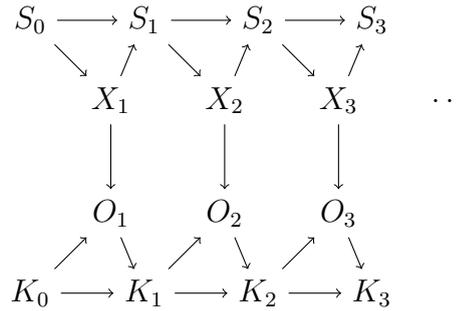
\begin{figure}[h]
\centering
\begin{tikzpicture}[node distance={15mm},main/.style = {draw, circle}] 
\node (1) {$S_0$}; 
\node (2) [below right of =1]{$X_1$};
\node (3) [below of =2]{$O_1$};
\node (4) [right of =1]{$S_1$};
\node (5) [below right of =4]{$X_2$};
\node (6) [below of =5]{$O_2$};
\node (7) [right of =4]{$S_2$};
\node (8) [below right of =7]{$X_3$};
\node (9) [below of =8]{$O_3$};
\node (10) [right of =7]{$S_3$};
\node (11) [right of =8]{$\dots$};
\node (12) [below left of =3]{$K_0$};
\node (13) [right of =12]{$K_1$};
\node (14) [right of =13]{$K_2$};
\node (15) [right of =14]{$K_3$};

\draw[->](1) to (2);
\draw[->](1) to (4);
\draw[->](2) to (3);
\draw[->](2) to (4); 
\draw[->](4) to (5);
\draw[->](4) to (7);
\draw[->](5) to (6);
\draw[->](5) to (7);
\draw[->](7) to (8);
\draw[->](7) to (10);
\draw[->](8) to (9);
\draw[->](8) to (10); 
\draw[->](12) to (13);
\draw[->](13) to (14);
\draw[->](14) to (15);
\draw[->](12) to (3);
\draw[->](3) to (13);
\draw[->](13) to (6);
\draw[->](6) to (14);
\draw[->](14) to (9);
\draw[->](9) to (15);
\end{tikzpicture} 
 
\caption{Bayesian network showing the structure of conditional independencies among latent states $S_t$ of the quantum source, the type $X_t$ of quantum state produced, the observable $O_t$ attained from interaction, and the state of knowledge $K_t$ that influences the work extraction protocol.}
\label{fig:BayesNet}
\end{figure}

Further manipulations, using the rules of probability and the conditional independencies indicated in the Bayesian network depicted in Fig.~\ref{fig:BayesNet}, allow us to express the optimal state of knowledge in terms of both conditional work distributions and simple linear algebraic manipulations of the generative HMM representing the memoryful source. We find
\begin{align}
\gbm\eta_t 
& = 
\Pr(S_t | O_t = o_t, S_{t-1} \sim \gbm\eta_{t-1} ) \\
& = 
\sum_{x \in \mathcal{X} }
\Pr(S_t , X_t = x | O_t = o_t, S_{t-1} \sim \gbm\eta_{t-1} ) \\
& = 
\sum_{x \in \mathcal{X} }
\Pr( X_t = x | O_t = o_t, S_{t-1} \sim \gbm\eta_{t-1} )
\Pr(S_t | X_t = x ,  S_{t-1} \sim \gbm\eta_{t-1} ) \\
& = 
\sum_{x \in \mathcal{X} }
\Pr( X_t = x | O_t = o_t, S_{t-1} \sim \gbm\eta_{t-1} )
\frac{\gbm\eta_{t-1} T^{(x)}}{\gbm\eta_{t-1} T^{(x)}\mathbf{1}} \\
& = 
\frac{\sum_{x \in \mathcal{X} }
\Pr( X_t = x , O_t = o_t | S_{t-1} \sim \gbm\eta_{t-1} )
\gbm\eta_{t-1} T^{(x)} / \gbm\eta_{t-1} T^{(x)}\mathbf{1} }{
\sum_{x' \in \mathcal{X} }
\Pr( X_t = x' , O_t = o_t | S_{t-1} \sim \gbm\eta_{t-1} )
} \\
& = 
\frac{\sum_{x \in \mathcal{X} }
\Pr( O_t = o_t | X_t = x, S_{t-1} \sim \gbm\eta_{t-1} ) \, \gbm\eta_{t-1} T^{(x)} }{
\sum_{x' \in \mathcal{X} }
\Pr( O_t = o_t | X_t = x', S_{t-1} \sim \gbm\eta_{t-1} ) \, \gbm\eta_{t-1} T^{(x')} \mathbf{1}
}
~.
\end{align}

\section{Using this to build a predictive work-extraction engine}

Rather than repeatedly calculating these ideal belief states on the fly for a specific realization of the process,
we can alternatively
systematically build up the set of all such belief states,
together with the
observation-induced transitions among them, to inform the design of an autonomous engine. 
There will be both a set of transient belief states and a set of recurrent belief states.
Both of these sets may be either finite or infinite.
In the case that only finitely many belief states are induced by observations,
we can explicitly build out the transition structure among them.  If there are infinitely many such states, then we would need to truncate unlikely states in the design of our finite physical engine~\cite{marzen17nearly}.

The physical memory system of our proposed engine should have at least one distinguishable state corresponding to every observation-induced belief state.
In fact, the memory must encode both the belief state and the most recent energy of the battery, so that conditioning on the new state of the battery is sufficient to supply the change in battery energy. 
These will likely be encoded with some finite precision, to avoid storing real numbers.
Conditioned on the state of the memory encoding $\gbm\eta$, the work extraction protocol will operate jointly on the quantum system, thermal reservoirs, and battery, to optimally extract work from the expected state $\xi = \sum_{x \in \mathcal{X}} \gbm\eta T^{(x)} \mathbf{1} \sigma^{(x)}$.

The subsequently observed work value $w$ 
uniquely updates
the memory from 
the state encoding $\gbm\eta$ to
the state encoding $\gbm\eta' = \frac{\sum_{x \in \mathcal{X} }
\Pr( W_t = w | X_t = x, S_{t-1} \sim \gbm\eta ) \, \gbm\eta T^{(x)} }{
\sum_{x' \in \mathcal{X} }
\Pr( W_t = w | X_t = x', S_{t-1} \sim \gbm\eta ) \, \gbm\eta T^{(x')} \mathbf{1}
}$.
Once the next quantum system arrives, the predictive quantum work extraction cycle begins again.

\section{Expected work extraction from the four approaches}
\label{app:analytic_work}

Here, we derive analytical expressions for the expectation value of work extraction from the various approaches compared in the main text.
We derive these expressions for the ($p$, $r$)-parametrized family of perturbed-coin processes 
of classically correlated quantum states discussed in the main text.

Recall that 
the quantum states $(\sigma^{(x)} )_{x\in \mathcal{X} }$ are the outputs of a Mealy HMM with labeled transition matrices $( T^{(x)} )_{x \in \mathcal{X}}$.  
An element of the labeled transition matrix $T_{s \to s'}^{(x)} = \Pr(X_t = x, S_{t} = s' | S_{t-1} = s)$ gives the joint probability of producing quantum state $\sigma^{(x)}$ and arriving at latent state $s'$, given that the HMM begins in state $s$.

For the perturbed-coin example, the HMM's labeled transition matrices are
\begin{equation}
    T^{(0)} = 
	\begin{bmatrix}
	1-p & 0 \\
	p & 0
	\end{bmatrix}
\qquad
\text{and }
\quad
	T^{(1)} = 
	\begin{bmatrix}
		0 & p \\
		0 & 1-p
	\end{bmatrix} ~.
\end{equation}
The stationary distribution over the latent states is $\gbm\pi=\left[\frac{1}{2},\frac{1}{2}\right]$ and the two different quantum states created are 
\begin{align}
    \sigma^{(0)} = \ket{0} \bra{0} =
 	\begin{bmatrix}
 		1 & 0 \\
 		0 & 0
 	\end{bmatrix}
 	\quad
 	\text{and }
 	\quad
 	\sigma^{(1)} = \ket{\psi} \bra{\psi} =
 	\begin{bmatrix}
 		r & \sqrt{r (1-r) }  \\
 		\sqrt{r (1-r)}  & 1 - r
 	\end{bmatrix} ~.
\end{align}
For any state of knowledge, $\gbm\eta_t=\left[ \frac{1}{2}+\epsilon_t, \frac{1}{2}-\epsilon_t \right]$ parameterized by $\epsilon_t \in [-\frac{1}{2},\frac{1}{2}]$, the induced expected state is 
\begin{align}
    \xi_t
	&
	= \rho^{(\epsilon_t)} :=
\sum_x  
	\begin{bmatrix}
		\tfrac{1}{2} + \epsilon_t &  \tfrac{1}{2} - \epsilon_t 
	\end{bmatrix}	
	T^{(x)} \mathbf{1} \sigma^{(x)} 
		= 
	\tfrac{1}{2}
	\begin{bmatrix}
		1 + r + \epsilon_t' \sqrt{1-r}  \,\, & \sqrt{r (1-r) }  - \epsilon_t' \sqrt{r}   \\
		\sqrt{r (1-r) }  - \epsilon_t' \sqrt{r}   \,\, & 1 - r - \epsilon_t' \sqrt{1-r}  
	\end{bmatrix} ~,
\end{align}
where $\epsilon' :=2 \epsilon(1-2p) \sqrt{1-r}\; \propto \epsilon$.

\subsection{Memory-assisted quantum processing}\label{app:memfull}

In the memory-assisted quantum approach,
we utilize work-extraction protocols 
that are thermodynamically optimized for the expected quantum state
\begin{align}
    \rho_t^* = \xi_t = \rho^{(\epsilon_t)} ~.
\end{align}

Recall that the eigenvalues and eigenstates of $\rho_t^*$
play a prominent role in the work-extraction statistics.
We find that the eigenvalues of $\rho^{(\epsilon)}$ are
\begin{align}
    \lambda_{\pm}^{(\rho^{(\epsilon)})} 
	= \tfrac{1}{2} \pm \tfrac{1}{2}  \sqrt{r +  \epsilon'^2}~,
\end{align}
with corresponding eigenstates
\begin{align}
	\ket{\lambda_{\pm}^{(\rho^{(\epsilon)})}} = 
	2^{-1/2}
	\Bigl[ r + \epsilon'^2 \mp \bigl(r + \epsilon' \sqrt{1-r} \, \bigr) \sqrt{r + \epsilon'^2 }  \, \Bigr]^{-1/2}
	\begin{bmatrix}
		\sqrt{r(1-r)} - \epsilon' \sqrt{r}   \\
		- r -\epsilon' \sqrt{1-r}  \pm  \sqrt{r+\epsilon'^2}
	\end{bmatrix}  
	~.
\end{align}

For all times after $t=0$,
the update rule for belief states simplifies to the following
\begin{align}
    \gbm\eta_{t+1} \Bigr|_{ W_{t+1} = w^{(\pm)} }
	& = 
	\frac{
		\sum_{x \in \mathcal{X} }
		 \bra{ \lambda_{\pm}^{(\xi_t)}  } \sigma^{(x)} \ket{ \lambda_{\pm}^{(\xi_t)}} \, 
		\gbm\eta_{t} T^{(x)} 
	}{
		\lambda_{\pm}^{(\xi_t)}
	} ~,
\end{align}
which can be expressed explicitly in terms of $p$, $r$, and $\epsilon_t$.

When $\epsilon_t=0$, we find that $\gbm\eta_{t+1}=\Bigl[ \frac{1}{2},\frac{1}{2} \Bigr] = \gbm\pi$. I.e., the stationary distribution is a fixed point for this dynamic over belief states.
Because of this, we break the initial symmetry by setting $\epsilon$ to a small non-zero value to obtain useful knowledge. 
In other words,
for the very first work-extraction protocol, we choose some $\rho_0^* \neq \xi_0$
to avoid an unstable fixed point of the update rule.
However, for all subsequent time steps, we choose $\rho_t^* = \xi_t$.

For the perturbed coin, the metadynamic of the belief state in the long run will yield two different results, depending on which regime the system is in, ``memory-apathetic regime" or ``memory-advantageous regime". 
The reason for this separation comes from the shape of their update function. For the memory-apathetic region, the update function has gradient less than unity, making $\epsilon=0$ an attractor. For the memory-advantageous region, the gradient of the update function exceeds unity, therefore making $\epsilon=0$ a repellor, at the same time two other points become part of a new attractor. 


In the long run,  transient belief states die out, leaving only the steady-state dynamics among the recurrent states of knowledge; any initial distribution over belief states generically converges to the stationary measure $\boldsymbol{\pi}_{\mathcal{K}}$. 
Hence the steady-state rate of work extraction is given by 
\begin{equation}
\lim_{t \to \infty}
    \langle{W_t}\rangle = \beta^{-1} \langle \text{D} [ \xi_t \| \gamma ]\rangle_{\Pr(K_t) = \boldsymbol{\pi}_{\mathcal{K}}} ~,
\end{equation}

The expected extracted work for the memory-apathetic region coincide with that of memoryless extraction and is given by
\begin{equation}
    \left\langle W^{\text{apathetic}}\right\rangle= \beta^{-1}\text{D}[\xi_0 \| \gamma]= \left\langle W^{\text{memoryless}} \right\rangle ~.
\end{equation}
On the other hand, in the regime where memory enhances the performance of the protocol, the stationary distribution over the two recurrent belief states $\gbm\eta$ and $\gbm\eta'$, with corresponding expected quantum states $\xi$ and $\xi'$, is
\begin{equation}
    \boldsymbol{\pi}_{\mathcal{K}} =\frac{1}{\lambda^{(\xi)}_++\lambda^{(\xi')}_+} [\lambda^{(\xi')}_+,\lambda^{(\xi)}_+] ~.
\end{equation}
Hence, the work extraction rate is given by 
\begin{equation}
    \left\langle W^{\text{advantage}} \right\rangle =\frac{\beta^{-1} }{\lambda^{(\xi)}_+ + \lambda^{(\xi')}_+} \Bigl( \lambda^{(\xi')}_+ \text{D}[\xi \| \gamma]+\lambda^{(\xi)}_+ \text{D}[\xi' \| \gamma] \Bigr) ~.
\end{equation}

\subsection{Classical approach}\label{app:class}
The derivation for the memory-assisted classical approach is similar to that of the memory-assisted quantum approach illustrated above. However rather than operating on the induced expected state $\xi_t$, the classical approach uses work-extraction protocols that are thermodynamically optimized for the decohered state 
\begin{equation}
    \rho_t^* = \xi_t^{\text{dec}}= 
	\frac{1}{2}\begin{bmatrix}
	1 + r + \epsilon_t' \sqrt{1-r}  \,\, & 0  \\
	0  \,\, & 1 - r - \epsilon_t' \sqrt{1-r}  
\end{bmatrix} ~.
\end{equation}

The eigenstates of $\rho_t^*$ are thus $\ket{0}$ and $\ket{1}$, independent of time in this case.
In the classical approach,
$\gbm\pi$ is no longer a fixed point of the belief-state update maps.
The transition probabilities between belief states are now given by
\begin{align}
\lambda_{\pm}^{(\xi_t^\text{dec})} = \tfrac{1}{2} \Bigl[ 1 \pm \bigl( r + \epsilon_t' \sqrt{1-r} \, \bigr) \Bigr] ~.
\end{align}

The metadynamic of belief in the classical case behaves as a reset processes. Unlike the quantum case with only two recurrent belief states, the 
classical protocol induces an infinite set of 
recurrent belief states.
To construct a finite-state autonomous engine, we could choose to truncate those states within some small $\delta$ distance from another recurrent state, or truncate belief states with negligible probability, with vanishing work-extraction penalty. 

We find that the work-extraction rate can again be computed by averaging the relative entropy-now between the decohered expected state and thermal state-over all recurrent states of knowledge:
\begin{align}
    \langle W_t^\text{classical} \rangle =
    \beta^{-1} \langle \text{D}[\xi_t^\text{dec} \| \gamma]  \rangle_{\Pr(K_t^\text{classical})} ~.
\end{align}


     

\subsection{Overcommitment to the most likely outcome}\label{app:naive}

The ``overcommitment'' approach used for comparison in the main text bets exclusively on the most likely outcome in $\{ \sigma^{(x)} \}_x$.

The expected thermodynamic cost of misaligned expectations during work extraction can be quantified exactly via the relative entropy $\text{D} [ \rho_0 \| \alpha_0 ]$ between the actual input 
$\rho_0$ and the anticipated input $\alpha_0$ that the protocol is optimal for, if we assume that the final state is independent of the initial state~\cite{riechers2021initial, riechers2021impossibility}.  Hence, if we design the protocol for a pure state, but operate on a mixed state, we will encounter divergent thermodynamic penalties.

Accordingly, we can observe divergent thermodynamic costs when we design the Skrzypczyk work extraction protocol to be optimal for operation on a pure state.  

Using the Skrzypczyk protocol (with $N$ relaxation steps)
to extract work from the pure state bet upon,
we see that the first bath state swapped with the system for energy extraction is not exactly pure, but rather satisfies $\gamma_\text{B} = \Bigl( 1 - \frac{e^{-\beta E_0}}{N(e^{-\beta E_0} + e^{-\beta E_1})} \Bigr) \ket{0} \bra{0}
+ \frac{e^{-\beta E_1}}{N(e^{-\beta E_0} + e^{-\beta E_1})} \ket{1} \bra{1}$.
(Recall that $H$ is the Hamiltonian for the system, not of the bath.) Any purity of the actual input beyond this initial bath purity is wasted.
The input state leading to minimal entropy production under this protocol is thus 
a unitary rotation of $\gamma_\text{B}$.

Thus, for this use case of the Skrzypczyk protocol, the minimally dissipative state $\alpha_0$ becomes pure as $N \to \infty$.  As $N \to \infty$, we observe the battery's final expected energy diverging (but only logarithmically in $N$) to negative infinity, when this protocol acts on any other state.  I.e., 
$\langle W \rangle  \sim - \beta^{-1} \ln N$.

More specifically, we can leverage 
Eqs.~\eqref{eq:WorkExtractionValues}
and \eqref{eq:WorkExtractionProbs}
to calculate the expected value of work
for the overcommitment approach.
We find that 
\begin{align}
    \Pr(W=w^{(-)} | \sigma^{(\text{argmin}_{x} \gbm\eta_t  T^{(x)}  \mathbf{1} )}) = |\braket{1 | \psi}|^2 = 1-r ~.
\end{align}
With
$\lambda_- = \frac{e^{-\beta E_1}}{N(e^{-\beta E_0} + e^{-\beta E_1})}$,
$w^{(\lambda_-)} \sim - \beta^{-1} \ln N$, 
and 
$\text{min}_x  \gbm\eta_t  T^{(x)}  \mathbf{1}  \sim
\text{min}(p,1-p)$ when $\gbm\eta_t$ is close to either latent state,
we anticipate that the overcommited work penalty diverges as $- \beta^{-1} (1-r) \min(p, 1-p) \ln N$,
as observed.

Interestingly, for a finite number of bath interactions, some work can be extracted on average within certain regimes.  But other regions of parameter space would yield very negative work-extraction averages.

Unlike the other approaches,
the expectation value of work in the 
overcommitment approach cannot be written as a relative entropy.
Hence, whereas the other approaches were guaranteed to have non-negative work extraction on average,
the overcommitment approach enjoys no such guarantee of non-negativity.
Indeed in the limit of many bath interactions,
the overcommitment approach leads to infinitely negative work extraction.

%% file: Appendices/Appendix4.tex

\chapter{Proofs for Chapter 5} 

\label{AppendixB} 
\section{Narrowing down search space}
\label{app:narrow_search}
In this section we derive the optimal states that each protocol should be tailoring to. Since the update of belief state is entirely dependent on the set of probability $\{\bra{\lambda_i}\sigma^{(x)}\ket{\lambda_i}\}_i$. We will first show that for any set of eigenbasis $\{\ket{\lambda_i}\}_i$ the protocol tailors to, the optimal state that results in least dissipation is given by 
\begin{equation}
\label{eq:best_state}
    \rho_k^*=\sum_ip^*_i\ket{\lambda_i}\bra{\lambda_i},\quad p_i^*=\bra{\lambda_i}\xi_k\ket{\lambda_i} \quad\text{for}\quad i\in {0,1}
\end{equation}
where $\xi_k$ is the expected state formed from $\gbm\eta^{(k)}$. Then we show that the complex degree of freedom can be ignored when choosing the eigenbases.

First, notice that the Bayesian update is characterized solely by the work distribution which are given by Eq.~\eqref{eq:work_prob}. Which also means that if the work distribution for 2 different protocols, $\mathcal{W}_{\rho}$ and $\mathcal{W}_{\rho'}$ are the same, then all the future statistics should remain the same. The only difference will then only be in the dissipation incurred. 
For any $\rho$ that the protocol is tailored to, we can write the expected dissipation term as
\begin{equation}
    \D(\xi_k\|\rho) = -\tr(\xi\ln\rho) - S(\xi_k)~.
\end{equation}
Since the second term is independent of $\rho$, we simply wish to minimise the first term, i.e., we wish to find $\rho^*$ that maximize $\tr(\xi_k\ln\rho)$. We expand the expression using the fact that $\rho=\sum_i\lambda_i\ket{\lambda_i}\bra{\lambda_i}$
\begin{equation}
\label{eq:optimality_proof}
    \tr(\xi\ln\rho)  = \sum_i\ln \lambda_i\bra{\lambda_i}\xi_k\ket{\lambda_i}
\end{equation}
Now in order to ensure 2 protocols, $\mathcal{W}_{\rho}$ and $\mathcal{W}_{\rho'}$ induce the same future statistics in the Baysian sense, we require that $\bra{\lambda_i}\xi_k\ket{\lambda_i} = \bra{\lambda'_i}\xi_k\ket{\lambda'_i}$ for all $i\in\{0,1\}$ and $\{\lambda_i\}_i,\{\lambda'_i\}_i$ are the eigvectors of $\rho$ and $\rho'$ respectively. The only way to maximize is then via changing the eigenvalues. A simple analysis of Eq.~\eqref{eq:optimality_proof} shows that it reaches maximum value when $\lambda_i = \bra{\lambda_i}\xi_k\ket{\lambda_i}$, which then shows the validity of our claim in Eq.~\eqref{eq:best_state}.

In order to search through all possible eigenbasis, the search has to go through all $\theta$ value as well as $\phi$, this is more tractable compare to the searching through the whole Bloch sphere but can still be computationally difficult. Instead, we argue that the $\phi$ is redundant for this optimization. Given any 2 quantum states in the Bloch sphere, it is always possible for us to define a plane such that both of the quantum states as well as the maximally mixed state are co-planar, let this plane be defined by some $\phi_0$ away from the computational basis. Without loss of generality, we can always set $\phi_0=0$. Which means both of the quantum states has no imaginary components. Notice again that the Bayesian update is essentially looking at the expression of
\begin{equation}
\begin{split}
\label{eq:bayes_ratio}
     \frac{\Pr(O=o_t | \sigma_{A_t}=\sigma^{(x)})}{\sum_{x\in\mathcal{X}}\Pr(O=o_t | \sigma_{A_t}=\sigma^{(x)})} = \frac{\bra{\lambda_{\theta,\phi}}\sigma^{(x)}\ket{\lambda_{\theta,\phi}}}{\sum_{x\in\mathcal{X}}\bra{\lambda_{\theta,\phi}}\sigma^{(x)}\ket{\lambda_{\theta,\phi}}}
\end{split}
\end{equation}
Again, without the loss of generality, we can assume one of the quantum state , $\sigma^{(0)}$ is diagonalized in the computational basis such that 
\begin{equation}
    \sigma^{(0)} = \begin{pmatrix} a&0\\ 0&1-a
    \end{pmatrix} ,\quad \sigma^{(1)} = \begin{pmatrix} b&c\\ c&1-b
    \end{pmatrix}~.
\end{equation}
Doing so helps us to analyze in expression in Eq.~\eqref{eq:bayes_ratio}, it can be rewritten as
\begin{equation}
    \frac{\Pr(O=o_t | \sigma_{A_t}=\sigma^{(x)})}{\sum_{x\in\mathcal{X}}\Pr(O=o_t | \sigma_{A_t}=\sigma^{(x)})} = \frac{a\cos\theta+\sin^2\frac{\theta}{2}}{(a+b)\cos\theta+\sin^2\frac{\theta}{2}+c\sin\theta\cos\phi}
\end{equation}
for $\theta\in[0,2\pi],\phi\in[0,\pi]$~. Notice that this function is smooth and continuous over the range of $\theta$ and can only reach its maximum range when $\cos\phi=1$. This indicates that $\phi = n\pi$, i.e., the eigenbasis searched through should always lie on the plane itself, hence searching through different $\phi$ is not required if we are already searching through all $\theta$ values.

For any $N$ dimensional eigenstate, $\ket{\varphi}$ it can be parametrized by 
$N-1$ parameters, i.e.
\begin{equation}
    \ket{\varphi}= \sum_{i=0}^{N-1}\alpha_i\ket{i}
\end{equation}
where 
\begin{equation}
    \begin{split}
    \alpha_0 &= \cos\frac{\theta_1}{2}\\
    \alpha_1 &=  \sin\frac{\theta_1}{2}\cos\frac{\theta_2}{2}\\
    &\vdots\\
    \alpha_{N-1}&=\sin\frac{\theta_1}{2}\sin\frac{\theta_2}{2}\cdots\sin\frac{\theta_{N-1}}{2}
    \end{split}
\end{equation}
where each $\theta_i \in [0,2\pi]$. For each of the eigenstate, we will choose to tailor the protocol according to the state in Eq.~\eqref{eq:optimality_proof}. This can then guarantee the optimal work extracted.

\section{Bounds on rate of free enegry}

\label{app:free_energy_rate_bound}
It is necessary for us to show that there exist a separation between the best possible adaptive local strategy and the global collective strategy. To compare these, we can compare the asymptotic work extraction rate based on policy obtained via DPP and the free energy rate of quantum state. In the previous section we have already discussed on the asymptotic rate of DP-agent, here we show that even the lower bound of the free energy rate can be shown be to higher, hence a separation. 

To start, we first define the so-called free energy rate. Consider a quantum state in the form of 
\begin{equation}
    \rho_{Q_1\cdots Q_N}=\sum_{x_{1:N}}\Pr(x_{1:N})\bigotimes_{t=1}^N\sigma^{(x_t)}_{Q_t}.
\end{equation}
The free energy rate can then be defined as
\begin{equation}
    r = \lim_{N\to\infty}\frac{\beta^{-1}S(\rho_{Q_1\cdots Q_N}\|\gamma^{\otimes N})}{N}
\end{equation}
if the limit on the right hand side exists and is well behaved. Now suppose it is then we can expand the expressions to find that 
\begin{equation}
\begin{split}
    r=&\lim_{N\to\infty}\frac{\beta^{-1}S(\rho_{Q_1\cdots Q_N}\|\gamma^{\otimes N})}{N}\\
    =& \lim_{N\to\infty}\frac{\beta^{-1}}{N}\left[-\Tr(\rho_{Q_1\cdots Q_N}\ln\gamma^{\otimes N})-S(\rho_{Q_1\cdots Q_N})\right]\\
    =& \lim_{N\to\infty}\frac{\beta^{-1}}{N}\left[-\Tr(\rho_{Q_1\cdots Q_N}\beta(F\mathbb{I}-H)^{\otimes N})-S(\rho_{Q_1\cdots Q_N})\right]\\
    =&\lim_{N\to\infty}\frac{1}{N} \left[\Tr(\rho_{Q_1\cdots Q_N}H^{\otimes N})-S(\rho_{Q_1\cdots Q_N})\right]-\mathcal{F}_\text{eq}
\end{split}
\end{equation}
Suppose we are dealing with a degenerate Hamiltonian in the form of $H=\omega\id$ then the expression can be further simplified to
\begin{equation}
\label{eq:free_energy_rate}
    r = \lim_{N\to\infty} \frac{-S(\rho_{Q_1\cdots Q_N})}{N}+\ln2
\end{equation}
So then what we really need to find is just the expression $\lim_{N\to\infty} \frac{-S(\rho_{Q_1\cdots Q_N})}{N}$, which is essentially the entropy rate of the quantum system. Now to we can likewise denote the combine systems of $Q_1\cdots Q_{N-1}$ as $B$ and $Q_N$ as $A$. Then by the definition of quantum conditional entropy, we can obtain that
\begin{equation}
    S(\rho_{BA})=S(\rho_B)+S(A|B)_\rho
\end{equation}
Now if the limit on the von Neumann entropy rate, $S_{vN}$, does exist, we can rewrite them to be $S(\rho_B)=s_{vN}(N-1)$ and $S(\rho_{BA})=Ns_{vN}$. Then we conclude that
\begin{equation}
    s_{vN} = S(A|B)_\rho~.
\end{equation}
The von Neumann entropy rate is not an easy quantity to calculate, numerically it is also difficult to compute since it involves eigen-decomposition which scales exponentially with $N$. We therefore turn to finding the upper bound of the entropy rate, which will then give us the lower bound on the free energy rate in Eq.~\ref{eq:free_energy_rate}.

To find a meaning upper bound on the entropy rate, we turn to data processing inequality, which essentially states that any data processing cannot make 2 probability distribution further apart. Here we can define a quantum channel that acts on $B$ while leaving $A$ untouched i.e.,
\begin{equation}
    \eta_{AD} = (\text{id}_A\otimes \mathcal{E}_B)\rho_{AB} , \quad \eta_A\otimes\eta_D = (\text{id}_A\otimes \mathcal{E}_B)(\rho_A\otimes\rho_B)~.
\end{equation}
We can then show the quantum data processing inequality
\begin{equation}
    \begin{split}
        S(A|D)_{\eta} &= S(\eta_A)-S(\eta_{AD}\|\eta_A\otimes\eta_D)\\
        &=S(\rho_A) - S\left[ (\text{id}_A\otimes \mathcal{E}_B)(\rho_{AB})\| (\text{id}_A\otimes \mathcal{E}_B)(\rho_A\otimes\rho_B)\right]\\
        &\geq S(\rho_A)-S (\rho_{AB}\| \rho_A\otimes\rho_B)\\
        &= S(A|B)_\rho
    \end{split}
\end{equation}
Conveniently, any quantum measurement does constitute a channel as well, however we need to ensure that the measurement is minimally disturbing to obtain a tight bound. For example, a trivial measurement, which always return an outcome $0$, will render the bound loose since such measurement gives us no information about the state in $B$. A possible candidate for such a measurement could be the Helstrom measurement. This may or may not be the optimal measurement but it is sufficient for the lower bound in our case.

The Helstrom measurement is the optimal measurement to distinguish 2 quantum states that results in the lowest error probability. Based on the structure of the quantum state, the best an agent can know about the current latent state of the system is if he is able to distinguish the very last quantum state, to be either $\sigma^{(0)}$ which implies the current latent state to be in $S$ or $S'$ if $\sigma^{(1)}$ is the last state. Here we denote $\tau_n^{(s)},s={0,1}$ to be the quantum states which are $n$-partite quantum states ending with $\sigma^{(s)}$. 
In this case, the Helstrom measurement can be expressed as
\begin{equation}
    M_0 = \sum_{\lambda_i>0} \ket{\lambda_i}\bra{\lambda_i},\quad M_1 = \sum_{\lambda_i<0} \ket{\lambda_i}\bra{\lambda_i}~,
\end{equation}
where $p_0\tau_n^{(0)} -p_1\tau_n^{(1)}=\sum_i\lambda_i\ket{\lambda_i}\bra{\lambda_i}$ is the eigen-decomposition. 
The post-measurement state can then be written as
\begin{equation}
\begin{split}
    \eta_{AD}&=\sum_{s,y} \gbm\pi(s)\Tr\left[M_y\tau_n^{(s)}\right] \ket{y}\bra{y}_D\otimes \xi_A^{(s)}\\
    &=\sum_yp_y\ket{y}\bra{y}_D\otimes \chi_A^{(y)}
\end{split}
\end{equation}
where $p_y=\sum_s\gbm\pi(s)\Tr\left[M_y\tau_n^{(s)}\right]$ and $\chi^{(y)} = \frac{1}{p_y}\sum_s\gbm\pi(s)\Tr(M_y\tau_n^{(s)})\xi^{(s)}$. The upper bound on the entropy rate can then be written as
\begin{equation}
    s_{vN} = S(A|B)_\rho \leq S(A|D)_\eta = \sum_y p_yS\left[\chi^{(y)}\right]
\end{equation}
which in turn provide us with the lower bound on the free energy rate,
\begin{equation}
    r \geq \lim_{n\to\infty} \ln 2 - S(A|D)_\eta~,
\end{equation}
where the dependence on $n$ comes from $\tau_n^{(s)}$. For the numerical simulation, taking $n=12$ is enough to create a gap in performance.

%% file: Appendices/Appendix5.tex
\chapter{Proofs for Chapter 6} 

\label{AppendixC} 
\section{Learning Algorithm}
\label{sec:learning_algo}
and at each time step $k\in\{1,\ldots,N\}$, a reward measurement is performed on the direction $\ket{\psi_k}$. Formally the reward measurement is described by a rank-1 two-outcome POVM $\lbrace \psi_k, \psi_k^\perp \rbrace$ where $\psi_k=\ketbra{\psi_k}{\psi_k}$ corresponds to $R_k=1$ and $\psi_k^\perp=\mathbb{I}-\psi_k$ corresponds to $R_k=0$. The observed reward $R_k$ is distributed accordingly to Born's rule, i.e 
\begin{align}\label{eq:prob_quantum_reward}
    \mathrm{Pr} \left(R_k = r_k \right) = 
    \begin{cases}
    |\langle \psi_k | \psi \rangle |^2, \quad & r_k=1,\\
    1-|\langle \psi_k | \psi \rangle |^2, \quad & r_k=0. 
    \end{cases}
\end{align}
The goal of this work was to design an algorithm that uses measurements that is minimally disturbing to the unknown state $\ket{\psi}$. The figure of merit used in the paper was a quantity termed \emph{regret}, which is defined as
\begin{align}\label{eq:infidelity_regret}
    \regret (N) = \sum_{k=1}^N (1 - |\langle \psi_k | \psi \rangle|^2).
\end{align}
This also corresponds to the cumulative sum of infidelities between the unknown state $|\psi\rangle$ and the selected direction $\psi_k$. The algorithm is formulated in terms of the Bloch vector which means that at each time step $k$, the algorithm uses the previous and current information of the reward $\lbrace r_1,\ldots , r_{k} \rbrace$ as well as the corresponding measurement directions $\lbrace \ket{\psi_1},\ldots , \ket{\psi_{k}} \rbrace $ to output a Bloch vector $a_{k+1}\in\mathbb{R}^3$ that is normalized, i.e., $\| a_{k+1} \| =1$. This will then be used for the next reward measurement described by a two-outcome POVM  $\lbrace \psi_{k+1} , \psi_{k+1}^\perp \rbrace$ (or equivalently the direction $\ket{\psi_{k+1}}$), given by:
\begin{align}
    \psi_k = \frac{1}{2} (\mathbb{I} + a_k \cdot \sigma ),
\end{align}
where $\sigma = (\sigma_x , \sigma_y , \sigma_z )$ are the Pauli matrices. The pseudo-code for how the algorithm updates the reward measurements can be found in Algorithm~\ref{alg:linucb_vn_var} (\textsf{LinUCB Vanishing Variance Noise}). We now describe all the involved quantities. Note that the original algorithm presented in~\cite{lumbreras24pure} was formulated for a more general $d$-dimensional qudits, but here we focus on the specific qubit case. 

\begin{algorithm}
	\caption{\textsf{LinUCB-VVN}} 
	\label{alg:linucb_vn_var}
 
        Require: $\lambda_0\in\mathbb{R}_{>0}$, $t\in\mathbb{N}$
        
        Set initial design matrix $V_0 \gets \lambda_0\mathbb{I}$. 
        
        Set initial estimate of variance $\hat{\sigma}^2_{1} \gets 1$. 
        
        Set initial Bloch vectors $a_{1,1}=(\frac{1}{\sqrt{2}},0,\frac{1}{\sqrt{2}})^\intercal$, $a_{1,2}=(-\frac{1}{\sqrt{2}},0,\frac{1}{\sqrt{2}})^\intercal$, $a_{1,3}=(0,\frac{1}{\sqrt{2}},\frac{1}{\sqrt{2}})^\intercal$, $a_{1,4}=(0,-\frac{1}{\sqrt{2}},\frac{1}{\sqrt{2}})^\intercal$.

        \For{$l=1,2,\ldots$}{
            \vspace{1mm}
            \textit{Optimistic action selection}
            \vspace{1mm}
            
            \For{$i = 1,2,3,4$}{               
                \vspace{1mm}
                \textit{Perform $t$ independent measurements for each $a_{l,i}$} 
                \vspace{1mm}
                
                \For{$j=1,...,t$}{
                    Measure the unknown $|\psi\rangle$ in the Bloch vector directions given by $a_{l,i}$ and receive outcomes $r_{l,i,j}$ 
                }
            }

            \vspace{1mm}
            \textit{Update design matrix}
            \vspace{1mm}
            
            $V_{l} \gets V_{l-1} + \frac{1}{\hat{\sigma}_{l}^2} \sum_{i=1}^{4} a_{l,i}  a_{l,i}^\intercal $        
            
            \vspace{1mm}
            \textit{Update Least Square Estimator for each subsample}
            \vspace{1mm}

            \For{$j=1,2,...,t$:}{
                $\tilde{\theta}_{l,j}^\text{w} \gets V_{l}^{-1}  \sum_{s = 1}^{l} \frac{1}{\hat{\sigma}_{l}^2} \sum_{i=1}^{4} a_{s,i} (2r_{s,i,j}-1)$
            } 
            
            \vspace{1mm}
            \textit{Update Bloch vectors for estimates}
            \vspace{1mm}
            
            Compute ${\theta}_{l}^{\text{\tiny wMOM}}$ using $\lbrace \tilde{\theta}_{l,j}^\text{w} \rbrace_{j=1,\ldots t}$ according to Eq.~\eqref{eq:median_of_means}

            \vspace{1mm}
            \textit{Update Bloch vectors for measurements}
            \vspace{1mm}
            
            Select Bloch vectors $a_{l+1,i}$ using ${\theta}_{l}^{\text{\tiny wMOM}}$ according to Eq.~\eqref{eq:action_general_update}

            \vspace{1mm}
            \textit{Update estimator of variance for $a_{l+1,i}$}
            \vspace{1mm}
            
            $\hat{\sigma}^2_{l+1} \gets \frac{2\zeta}{\sqrt{\lambda_{\max}(V_{l})}}$
        }
\end{algorithm}
The algorithm presented in Algorithm~\ref{alg:linucb_vn_var} accepts a parameter $t$, which controls the success probabilities. The protocol operates on $N$ identically prepared copies of an unknown pure quantum state, with each copy corresponding to a measurement round indexed globally by $k\in\{1,\ldots,N\}$. The measurements are organized into blocks of $4t$ rounds, with each block indexed by $l\in 1,\ldots,L$, where $L=N/(4t)$ denotes the total number of blocks. Each block $l$ consists of four sub-blocks, indexed by $i\in\{1,2,3,4\}$, and each of the sub-block contains $t$ measurements, indexed by $j\in\{1,\ldots, t\}$.
To precisely identify each individual measurement round, we define a composite index $(l,i,j)$, with components:
\begin{equation}
    (l,i,j) \quad\text{where}\quad\begin{cases}
        l\in\{1,\ldots,L=\frac{N}{4t}\}  &\quad \text{(block index)}\\
        i\in\{1,2,3,4\}&\quad \text{(sub-block index)}\\
        j\in\{1,\ldots,t\}&\quad \text{ (iteration within sub-block)}~.
    \end{cases}
\end{equation}
Each composite index $(l,i,j)$, maps to a unique global round index $k\in\mathbb{N}$ via the function:
\begin{equation}
    k = \underbrace{4t(l-1)}_{\text{offset for previous blocks}}+\underbrace{4(j-1)}_{\text{offset within current block}}+\underbrace{i}_{\text{position with sub-block}}
\end{equation}
Within each group $l$, the algorithm selects a set of four Bloch vectors $\lbrace a_{l,i} \rbrace_{i=1}^4$. For each sub-block, the algorithm performs $t$ independent measurements on $t$ distinct copies of the unknown quantum state. All measurements within sub-block $i$ are defined by the Bloch vector $a_{l,i}$.

The vectors $a_{l,i}$ are updated recursively, incorporating information from previous measurement outcomes along with the measurements themselves. This update is performed after completing each block of $4t$ measurements, allowing the algorithm to adaptively refine its estimates of the underlying quantum state.

At each block $l\in\{1,\ldots,L\}$, the algorithm computes an estimate of the Bloch vector associated with the unknown quantum state using weighted least-squares estimators.  Specifically, for each of the $t$ measurement rounds within the block, indexed by $j\in\{1,\ldots,t\}$, the algorithm constructs
\begin{align}
\tilde{\theta}_{l,j}^\text{w} = V_{l}^{-1}  \sum_{m = 1}^{l} \frac{1}{\hat{\sigma}_m^2} \sum_{i=1}^4a_{m,i}( 2r_{m,i,j}-1),
\end{align}
where:
\begin{itemize}
    \item $r_{m,i,j}\in\lbrace 0 ,1 \rbrace$ denotes the binary outcome of the quantum measurement along direction $a_{m,i}$ at round $(m,i,j)$,
    \item $\hat{\sigma}_{m}^2$ is an upper bound on the variance of  $r_{m,i,j}$ in block $m$, defined in Eq.~\eqref{eq:weight_choice},
    \item $V_{l}$ is the design matrix capturing the accumulated measurement directions, recursively defined by: 
\end{itemize}
\begin{align}
  V_{l} = V_{l-1} + \frac{1}{\hat{\sigma}_{l}^2} \sum_{i=1}^{4} a_{l,i}  a_{l,i}^\intercal, \quad \text{with}\hquad V_0 = \lambda_0 \mathbb{I},\hquad \lambda_0 > 0~.
\end{align}
To obtain enhance the robustness of the estimate of the Bloch vector, the algorithm applies the Median-of-Means (MoM) technique. It selects a particular weighted least-squares estimator among the $t$ candidates using:
\begin{equation}
    j^* = \argmin_{j\in[t]} \text{median}\lbrace \|\tilde{\theta}^\text{w}_{l,j} - \tilde{\theta}^\text{w}_{l,j'} \|_{V_{l}}: j'\in [k]/j\rbrace,
\end{equation}
where $\| x \|^2_{V_{l}} = \langle x, V_{l} x\rangle $ is the weighted norm and $\langle \cdot , \cdot \rangle$ is the standard inner product between real vectors. The weighted median-of-means (wMoM) estimator is then:
\begin{equation}
\label{eq:median_of_means}
        \tilde{\theta}_{l}^{\text {wMoM}} := \tilde{\theta}_{l,j^*}, \hquad \text{and}\hquad {\theta}^{\text{wMoM}}_{l} = \frac{\tilde{\theta}^{\text{wMoM}}_{l}}{\| \tilde{\theta}^{\text{wMoM}}_{l} \|_2} ~,
\end{equation}
where $\theta^{\text{wMoM}}_{l}$ is the normalized wMoM estimator used for subsequent measurement updates.

Using the estimate $\theta_{l}^{\text {wMoM}}$, the algorithm determines the next set of measurement directions
$\{a_{l+1,i}\}_{i=1}^4$ for block $l+1$ via:
\begin{equation}
\label{eq:action_general_update}
    a_{l+1,i} = \frac{\tilde{a}_{l+1,i}}{\| \tilde{a}_{l+1,i}\|_2}, \hquad \text{where}\hquad  \tilde{a}_{l+1,i} = {\theta}^{\text{wMoM}}_{l} - \frac{(-1)^i}{\sqrt{\lambda_{\min}(V_{l})}}v_{l,\lceil\frac{i}{2}\rceil}~.
\end{equation}
Here $\lbrace v_{l,0},v_{l,1}\rbrace$ are eigenvectors of $V_{l}$ associated with its two smallest eigenvalues, and $\lambda_{\min}(V_{l})$ denotes the smallest eigenvalue. This construction perturbs the estimated Bloch vector along the least certain directions, while ensure that the infelidelity between the estimate and the unknown state are still small, balancing exploration and exploitation.

The variance estimate $\hat{\sigma}^2_{l}$ is chosen to upper-bound the variability of measurement outcomes and is defined as:
\begin{align}\label{eq:weight_choice}
    \hat{\sigma}^2_{l} = \frac{2\zeta}{\sqrt{\lambda_{\max}(V_{l})}}~,
\end{align}
where $\lambda_{\max}(V_{l})$  is the largest eigenvalue of the design matrix and $\zeta$ is any constant satisfying $\zeta \geq 334812\sqrt{2}+1296\sqrt{6}$. The specific choice of variance estimator $\hat{\sigma}^2_{l}$ in Eq.~\eqref{eq:weight_choice} ensures that it serves as a upper bound on the true variance of the measurement outcomes $r_{l,i,j}$ with high probability of $1-\delta$, which result from projecting the unknown quantum state onto the directions $a_{l,i}$. Importantly, as demonstrated in~\cite{lumbreras24pure}, this formulation permits the derivation of rigorous high-probability concentration bounds. In particular, it guarantees that the weighted Median-of-Means estimator $\tilde{\theta}^{\text{wMoM}}_{l} $ remains a statistically reliable and robust approximation of the true Bloch vector of the unknown quantum state.

We are now prepared to present the main theoretical result from~\cite{lumbreras24pure},  which characterizes the performance of the algorithm in terms of two central metrics:
\begin{enumerate}
    \item Regret scaling as a function of the total number of measurement rounds $N$, and
    \item Infidelity behavior between the true quantum state $\ket{\psi}$ and the sequence of adaptive measurement directions $\ket{\psi_k}$
\end{enumerate}
The theorem formally establishes high-probability bounds that quantify how efficiently the algorithm learns the underlying quantum state, both in terms of cumulative performance (regret) and per-round estimation accuracy (infidelity).
\begin{theorem}[{\cite[Theorem 9 and 11]{lumbreras24pure}}]\label{th:fidelity_gurantee_bandit}
 Fix $L\in\mathbb{N}$,  $t= \lceil 24\ln \left( L / \delta \right) \rceil $ for some $\delta > 0$ and time horizon $N = 4t L$. Then we have that the quantum state tomography Algorithm~\ref{alg:linucb_vn_var} over an unknown state $\ket{\psi}$ achieves with probability at least $1-\delta$ the regret Eq.~\eqref{eq:infidelity_regret} scaling
\begin{align}
    \regret (N) \leq C_1 \log \left( \frac{N}{\delta} \right)\log ( N ),
\end{align}
 for some universal constant $C_1 > 0$. Also for all $k\in\{1,\ldots,N\}$ the selected 2-outcome POVM's given by the rank-1 projector $\psi_k = \ketbra{\psi_k}{\psi_k}$ achieve infidelity
 \begin{align}
   1- |\langle \psi_k | \psi \rangle|^2 \leq C_2 \frac{\log\left(\frac{N}{\delta} \right)}{k},
 \end{align}
 for some universal constant $C_2 > 0$. Moreover setting $\delta = \frac{1}{N}$  it holds
 \begin{align}
     \Ex [\regret (N)] = C_3\log^2 (N), \quad \Ex [1- |\langle \psi_k | \psi \rangle|^2] \leq C_4 \frac{\log (N) }{k},
 \end{align}
 for some universal constants $C_3 , C_4 > 0$ and the expectation is taken over the probability distribution of outcomes and measurements induced by the policy. 
\end{theorem}

\section{\texorpdfstring{$\rho^*$}{}-work extraction protocol}
\label{app:convergence_rate}
Here we discuss more details about the work-extraction protocol that is used, we adapted the protocol formalized in Skrzypczyk's paper~\cite{skrzypczyk2014work}. Just as in their formulation, the expected work extracted from a known state $\rho$ will precisely be given by the state's non-equilibrium free energy, which equals the relative entropy between the state and Gibbs' state, $\gamma_\beta$ with inverse temperature $\beta$, i.e.
\begin{equation}
    \label{eq:relative_entropy}
    \beta\Ex(W) = \D(\rho\|\gamma_\beta). 
\end{equation}
We will be applying the protocol to a degenerate Hamiltonian.

We focus on a specific time step within the $N$ rounds of extraction, doing so simplifies notation by removing the $k$ index. In a specific round, the agent is given a partially unknown qubit system state $\psi$, and a classical description of the direction $\hat{\psi}$ and an accuracy $\epsilon$. The agent will then choose to optimize the protocol for a state $\rho^* = (1-\epsilon)\hat{\psi}+ \epsilon\hat{\psi}^\perp$.
Along the unknown state $\psi$, he also has access to a heat bath at inverse temperature $\beta$ and a battery state $\varphi(x)$. The Hamiltonian of the system is $H_A = \omega \id$, setting $\hbar=1$. The heat bath can provide any amount of thermal state with any Hamiltonian at inverse temperature $\beta$. We will mainly consider qubit thermal states $\gamma_\beta(\nu) = \frac{1}{Z_R(\nu)}e^{-\beta H_R(\nu)}$ where $Z_R(\nu) = \tr(e^{-\beta H_R(\nu)})$ with Hamiltonian $H_R(\nu)=\nu\ketbra{1}{1}$ where $\{\ket{i}\}_{i=0,1}$ are the energy eigenstates. The battery is modeled as a weight at a certain height whose state is described by $\varphi(x)\in L^2(\mathbb{R})$ and Hamiltonian $H_B$ such that $H_B \varphi(x-\mathrm{d} x) = x\varphi(x-\mathrm{d} x)$. We will assume $\varphi(x-\mathrm{d} x)$ to be a battery state whose energy is sharply centered at $\mathrm{d} x$. The energy of the battery can be changed by translating the weight up by a certain height $\mathrm{d} x$, described by the translation operator $\Gamma_{\mathrm{d} x}^B \varphi(x) = \varphi(x-\mathrm{d} x)$. We aim to extract work from the partially unknown system system state into the battery. We will design the work extraction protocol for the state $\rho^*$, i.e. the protocol optimally extracts work from $\rho^*$, see Algorithm~\ref{alg:sc_work_extraction}. 

 To simplify calculation as well as maintain generality later, we will denote $\psi$ as $\rho$, $\hat{\psi}$ as $\ket{\phi_0}$ and $\hat{\psi}^\perp$ as $\ket{\phi_1}$, likewise we denote $p_0 =1-\epsilon$ and $p_1=\epsilon$, so $\rho^* = \sum_i p_i\ketbra{\phi_i}{\phi_i}$. 

\begin{algorithm}[H]
	\caption{\textsf{$\rho^*$-ideal work extraction}} 
	\label{alg:sc_work_extraction}
 
        Require: An unknown system state $\psi$, a classical description of the state $\rho^*$, a battery state $\varphi(x)$ with the battery energy $\mu=0$, a reservoir at inverse temperature $\beta$

        Set $\{\phi_{i}\}_i$ and $\{p_{i}\}_i$ to be the eigenvectors and eigenvalues of $\rho^*$

        \vspace{1mm}
        \textit{Unitary Rotation}
        \vspace{1mm}
        
        Apply unitary $U=\sum_i \ketbra{i}{\phi_i}$ to try to diagonalize the system qubit in computational basis.
        \vspace{1mm}
        
        \For{$\ell=1,2,\ldots,M$}{
            \vspace{1mm}
            \textit{Prepare a fresh a reservoir qubit and exchange it with the system}
            \vspace{1mm}
            
            Take a fresh thermal qubit $\gamma_\beta(\nu(\ell,\epsilon))=\frac{1}{Z_R(\nu(\ell,\epsilon))}e^{-\beta H_R(\nu(\ell,\epsilon))}$ where $\nu(\ell,\epsilon)=\beta^{-1}\ln\frac{p_{0,\ell}}{p_{1,\ell}}$, $p_{i,\ell} = p_i-(-1)^i \ell \delta p$ and $\delta p=\frac{1}{M}(p_0-\frac{1}{2})$ from the reservoir.  
            
            Apply the swap unitary $V_{\rho^*,\ell}= \sum_{ij}\ketbra{i}{j}_A\otimes\ketbra{j}{i}_R \otimes \Gamma_{(i-j)\nu(\ell,\epsilon)}$ on the system, the battery and the reservoir qubit. 
            
            Discard the reservoir qubit. 
        }
        
        \vspace{1mm}
        \textit{Measure the extracted work}
        \vspace{1mm}
            
        Measure the battery energy, obtain the battery energy $\mu'$ and compute the extracted work $\Delta W=\mu'-\mu$. 
\end{algorithm}

In the first stage of the protocol, we rotate the unknown qubit via unitary 
\begin{equation}
    U = \sum_i\ketbra{i}{\phi_i}~.
\end{equation}
This operation attempts to diagonalize the system qubit in the computation basis.
We then interact the system with the battery. The state of the system together with the battery is 
\begin{equation}
\label{eq:intial_joint_state}
    \rho_{AB} = \sum_{ij}\bra{\phi_i}\rho\ket{\phi_j}\ketbra{i}{j}_A\otimes \varphi(x)_B.
\end{equation}

In the second stage of the protocol, we perform $M$ repetitions of the following process. In repetition $\ell$, we take a fresh thermal qubit $\gamma_\beta(\nu(\ell,\epsilon))$ with Hamiltonian $H_{R}(\nu(\ell,\epsilon))$ from the reservoir where $\nu(\ell,\epsilon)= \beta^{-1} \ln \frac{p_{0,\ell}}{p_{1,\ell}}$, $p_{i,\ell}= p_{i} - (-1)^i \ell \delta p$ and $\delta p = \frac{1}{M}(p_{0} - \frac{1}{2})$. Note that the reservoir qubit we take depends on which repetition we are in. Then we perform the swap unitary 
\begin{align}
    V_{\rho^*,\ell}  = \sum_{ij}\ketbra{i}{j}_A\otimes\ketbra{j}{i}_R \otimes \Gamma_{(i-j)\nu(\ell,\epsilon)}~.
\end{align}
This unitary swaps the system and the fresh qubit from the reservoir, extracts work into the battery due to the different energy gap between $\{\ket{i}_A\}_{i=0,1}$ and $\{\ket{i}_R\}_{i=0,1}$ and conserves energy of the system, the qubit from the reservoir and the battery. Finally, the qubit from the reservoir is discarded. At the end of each repetition $\ell$, the reduced state is 
\begin{align}
    \rho_{AB,\ell}  = \tr_R\left(V_{\rho^*,\ell}\left(\rho_{AB,\ell-1}\otimes \gamma_\beta(\nu(\ell,\epsilon))\right) V_{\rho^*,\ell}^\dagger\right),
\end{align}
After the first repetition, we obtain 
\begin{align}
    \label{eq:joint_state_1}
    \rho_{AB,1} = \sum_{i} p_{i,1} \ketbra{i}{i}_A \otimes \rho_{B,i,1},
\end{align}
where 
\begin{align}
    \label{eq:battery_state_1}
    \rho_{B,i,1} = \sum_j \bra{\phi_{j}}\rho\ket{\phi_{j}} \varphi(x-(i-j)\nu(\ell,\epsilon)), 
\end{align}
and after repetition $\ell$ where $\ell\geq 2$, we obtain
\begin{align}
    \label{eq:joint_state_tau}
    \rho_{AB,\ell} = \sum_{i} p_{i,\ell} \ketbra{i}{i}_A \otimes \rho_{B,i,\ell}, 
\end{align}
where
\begin{align}
    \label{eq:battery_state_tau}
    \rho_{B,i,\ell} = \sum_j  p_{j,\ell-1} \Gamma_{(i-j)\nu(\ell,\epsilon)} \rho_{B,j,\ell-1} \Gamma_{(i-j)\nu(\ell,\epsilon)}^\dagger. 
\end{align}
From Eq.~\eqref{eq:joint_state_tau}, we observe that the reduced state of the system changes gradually, which resembles a quasi-static process in thermodynamics. This is the reason why we take the swap unitary in repetition.

\subsection{Work distribution}
\label{apd:work_distribution}
In this section, we will use the Lagrange mean value theorem and the first mean value theorem for definite integrals~\cite{strang2019calculus} as follows:
\begin{theorem}
    \label{thm:lag_mean_value_theorem}
    Let $f:[a, b] \to \mathbb{R}$ be a continuous on the closed interval $[a,b]$ and differentiable on the open interval $(a,b)$. Then there exists $c\in (a, b)$ such that
    \begin{align}
        f(b)-f(a) = f'(c) (b-a). 
    \end{align}
\end{theorem}
\begin{theorem}
    \label{thm:int_mean_value_theorem}
    Let $f:[a, b] \to \mathbb{R}$ be a continuous function on the closed interval $[a,b]$. Then there exists $c\in (a, b)$ such that
    \begin{align}
        \int_a^b f(x) \textnormal{d} x = f(c) (b-a). 
    \end{align}
\end{theorem}

We will show the following theorem in the following subsection.
\begin{theorem}
    \label{thm:work_distribution}
    Let $\lbrace\phi_i\rbrace_{i=0,1}$ and $\lbrace p_i\rbrace_{i=0,1}$ be the eigenvectors and eigenvalues of $\rho^*$, $\Delta W$ be the extracted work (which is a continuous random variable) and $M$ be the number of repetitions as in Algorithm~\ref{alg:sc_work_extraction}. It holds that: if the extraction protocol is operated on a state $\rho$ that is the eigenstate of $\rho^*$, i.e., $\rho=\phi_{i}$, then then the expected extracted work $\Ex[\Delta W]$ converges to a fixed value $w_i$ and the extracted work $\Delta W$ converges in probability to its expectation $\Ex[\Delta W]$. To be precise, it means
        \begin{align}
            \lim_{M\to \infty} \Ex[\Delta W] = w_i,
        \end{align}
        where
        \begin{equation}
        \label{eq:thm3work_values}
            w_{i} \coloneqq \beta^{-1} (\D(\phi_{i}\|\id/2) + \ln p_{i}),
        \end{equation} 
        and for any $\epsilon>0$
        \begin{align}
            \lim_{M\to\infty} \Pr[|\Delta W - \Ex[\Delta W]|\geq \epsilon ] = 0.
        \end{align}

\end{theorem}
\begin{proof}
We consider the case where $\rho=\phi_{i}$. We assume that $p_{0} > \frac{1}{2}$ as we deal with an estimate for a pure state in Algorithm~\ref{alg:sc_work_extraction}, although similar proof holds for other cases. According to Eq.~\eqref{eq:joint_state_1}, the state after the first repetition can be viewed as a classical state described as follows: the state after repetition $1$ is $\phi_{x_1}$ where $x_1$ is a random bit sampled from $\{0,1\}$ according to the probability distribution $(p_{0,1},p_{1,1})$; the extracted work after the first repetition conditioned on $x_1$ is $(x_1-i)\nu(1,\epsilon)$. According to Eq.~\eqref{eq:joint_state_tau}, the evolution in repetition $\ell$ where $\ell\geq 2$ can be viewed as a classical process described as follows: the state after repetition $\ell$ is $\phi_{x_\ell}$ where $x_{\ell}$ is a random bit sampled from $\{0,1\}$ according to the probability distribution $(p_{0,\ell}, p_{1,\ell})$; the extracted work in repetition $\ell$ conditioned on $x_{\ell-1}x_\ell$ and is $(x_\ell - x_{\ell-1}) \nu(\ell,\epsilon)$. Suppose that the random bits sampled during the above process is $x_1\ldots x_M$ after $M$ repetitions. The extracted work after repetition $M$ conditioned on $x_1\ldots x_M$ is
\begin{align}
\label{eq:B9}
    \Delta W&= (x_1-i)\nu(1,\epsilon) + \sum_{\ell=2}^M (x_\ell - x_{\ell-1})\nu(\ell,\epsilon) \\
    &= - i  \nu(1,\epsilon) + \sum_{\ell=1}^{M-1} x_\ell (\nu(\ell,\epsilon) - \nu(\ell+1,\epsilon)) + x_M \nu(M,\epsilon), 
\end{align}
recall that 
\begin{align}
\label{eq:nu_tau}
    \nu(\ell,\epsilon) = \beta^{-1} \ln \frac{p_{0} - \ell \delta p }{p_{1} + \ell \delta p} . 
\end{align}
where $\delta p = \frac{1}{M}(p_{0}-\frac{1}{2})$.
The expected extracted work is 
\begin{align}
    \Ex[\Delta W] =   - i  \nu(1,\epsilon) + \sum_{\ell=1}^{M-1} \Ex[x_\ell] (\nu(\ell,\epsilon) - \nu(\ell+1,\epsilon)) + \Ex[x_M] \nu(M,\epsilon).
\end{align}
Notice that, from the definition of $x_\ell$, $\Ex[x_\ell] = p_{1,\ell}=p_{1} + \ell \delta p$, 
we obtain
\begin{align}
\label{eq:B12}
    \Ex[\Delta W] & =   - i \beta^{-1}\ln \frac{p_{0} - \delta p }{p_{1} + \delta p} + \beta^{-1} \sum_{\ell=1}^{M-1}  \left(  \ln \frac{p_{0} - \ell \delta p }{p_{1} + \ell \delta p} - \ln \frac{p_{0} - (\ell+1) \delta p }{p_{1} + (\ell+1) \delta p}\right) (p_{1}+\ell\delta p ). 
\end{align}
We now use the definition of $\nu(\ell,\epsilon)$ in Eq.~\eqref{eq:nu_tau} as well as the Lagrange mean value theorem in Theorem~\ref{thm:lag_mean_value_theorem} to obtain 
\begin{align}
    \beta (\nu(\ell,\epsilon) - \nu(\ell+1,\epsilon)) = \ln \frac{p_{0} - \ell \delta p }{p_{1} + \ell \delta p} -\ln \frac{p_{0} - (\ell+1) \delta p }{p_{1} + (\ell+1) \delta p} = \frac{1}{\xi_\ell(1-\xi_\ell)}  \delta p, 
\end{align}
for some $\xi_\ell\in [p_{0}-(\ell+1) \delta p, p_{0}-\ell\delta p ]$. 

Therefore, Eq.~\eqref{eq:B12} can be simplified to 
\begin{equation}
\begin{split}
    \label{eq:expected_extracted_work}
    \Ex[\Delta W]  & =  - i  \beta^{-1} \left(\ln \frac{p_{0} }{p_{1}} -\frac{\delta p}{\xi_0 (1-\xi_0)}  \right)+ \beta^{-1} \sum_{\ell=1}^{M-1}   \frac{p_{1}+\ell\delta p }{\xi_\ell(1-\xi_\ell)}  \delta p   \\
    & =  - i  \beta^{-1} \ln \frac{p_{0} }{p_{1}} + \beta^{-1} \sum_{\ell=1}^{M} \frac{p_{1}+\ell\delta p }{\xi_\ell(1-\xi_\ell)}  \delta p \\
    & \quad + i  \beta^{-1} \frac{\delta p}{\xi_0 (1-\xi_0)}  - \beta^{-1} \frac{\delta p}{2\xi_M (1-\xi_M)}~.
\end{split}
\end{equation}

We will approximate the sum in the second line of Eq.~\eqref{eq:expected_extracted_work} with an integration, where the remainder is bounded due to first mean value theorem for definite integrals as in Theorem~\ref{thm:int_mean_value_theorem}. Namely, 
\begin{align}
    \beta^{-1} \sum_{\ell=1}^{M}  \frac{p_{1}+\ell\delta p }{\xi_\ell(1-\xi_\ell)}  \delta p &=\beta^{-1} \int_{\frac{1}{2}}^{p_{0} } \frac{\text{d} p}{p} + \beta^{-1}R_1(\delta p) \\
    &= \beta^{-1}\ln p_{0}+ \beta^{-1}\ln(2) + \beta^{-1}R_1(\delta p),
\end{align}
where $R_1(\delta p)$ is the remainder given by 
\begin{align}
    R_1(\delta p) & = \sum_{\ell=1}^{M}  \frac{p_{1}+\ell\delta p }{\xi_\ell(1-\xi_\ell)}  \delta p -  \int_{\frac{1}{2}}^{p_{0}} \frac{\text{d}p}{p} \\
    &= \sum_{\ell = 1}^{M} \left( \frac{p_{1}+\ell\delta p }{\xi_\ell(1-\xi_\ell)}  \delta p - \int_{p_{0}-\ell \delta p}^{p_{0}-(\ell-1)\delta p}\frac{\text{d} p}{p} \right) \\ 
    & = \sum_{\ell=1}^{M} \left(\frac{p_{1}+\ell\delta p }{\xi_\ell(1-\xi_\ell)} -  \frac{1}{\xi_\ell'} \right)\delta p =  \sum_{\ell=1}^{M} \frac{\xi_\ell(1-\xi_\ell) - \xi_\ell'(p_{1}+\ell\delta p)}{\xi_\ell'\xi_\ell(1-\xi_\ell)} \delta p, 
\end{align}
where from the first line to the second line, we have used the first mean value theorem for definite integrals~Theorem~\ref{thm:int_mean_value_theorem} that
\begin{align}
    \int_{p_{0}-\ell \delta p}^{p_{0}-(\ell-1)\delta p}\frac{\text{d} p}{p} = \frac{1}{\xi_\ell'} \delta p,
\end{align}
for some $\xi_\ell'\in [p_{0}-\ell \delta p, p_{0}-(\ell-1)\delta]$. Therefore, the remainder satisfies
\begin{align}
    |R_1(\delta p )| & \leq \sum_{\ell=1}^{M} \left|\frac{\xi_\ell(1-\xi_\ell) - \xi_\ell'(p_{1}+\ell\delta p)}{\xi_\ell'\xi_\ell(1-\xi_\ell)} \right|\delta p \leq \sum_{\ell=1}^{M} \left|\frac{p_{1}+\ell\delta p + \xi_\ell}{\xi_\ell'\xi_\ell(1-\xi_\ell)} \right| (\delta p)^2 \\
    &\leq \sum_{\ell=1}^{M}\frac{4}{p_{1}} (\delta p)^2 \leq  (2p_{0}-1)\frac{2}{p_{1}}\delta p. 
\end{align}

Therefore, the second line in Eq.~\eqref{eq:expected_extracted_work} is finite while the third line is infinitesimal, and we obtain
\begin{align}
    \Ex[\Delta W] & =  - i  \beta^{-1} \ln \frac{p_{0} }{p_{1}} + \beta^{-1}\ln p_{0}+ \beta^{-1}\ln(2) + O(\delta p) \\
    & = -\beta^{-1} \tr(\phi_{i}\ln\id/2)  - i  \beta^{-1} \ln \frac{p_{0} }{p_{1}} + \beta^{-1}\ln p_{0} + O(\delta p) \\
    & = -\beta^{-1} \tr(\phi_{i}\ln\id/2) + \beta^{-1}\ln p_{i} + O(\delta p) \\
    & = \beta^{-1}\left[\D\left(\phi_{i}\middle\|\id/2\right) + \ln p_{i} \right] + O(\delta p). 
\end{align}
Since $\delta p \propto\frac{1}{M}$, we then obtain that
\begin{align}
    \Ex[\Delta W] =  \beta^{-1}\left[\D\left(\phi_{i}\middle\|\id/2\right) + \ln p_{i} \right] + O\left(\frac{1}{M}\right).
\end{align}
Taking $M\to \infty$, we obtain
\begin{align}
    \lim_{M\to \infty}\Ex[\Delta W] =  \beta^{-1}\left[\D\left(\phi_{i}\middle\|\id/2\right) + \ln p_{i} \right].
\end{align}
Now we demonstrate the convergence of $\Delta W$ towards its expectation value, recall from Eq.~\eqref{eq:B9}, we have that 
\begin{align}
    \Delta W =  - i  \nu(1,\epsilon) + \sum_{\ell=1}^{M-1} x_\ell (\nu(\ell,\epsilon) - \nu(\ell+1,\epsilon)) + x_M \nu(M,\epsilon).  
\end{align}
By the Lagrange mean value theorem,  
\begin{align}
     \nu(\ell,\epsilon) - \nu(\ell+1,\epsilon) = \ln \frac{p_{0} - \ell \delta p }{p_{1} + \ell \delta p} -\ln \frac{p_{0} - (\ell+1) \delta p }{p_{1} + (\ell+1) \delta p} = \frac{\delta p}{\xi_\ell(1-\xi_\ell)} , 
\end{align}
for some $\xi_\ell \in [p_{0}-(\ell+1)\delta p ,p_{0}-\ell\delta p ] $ and $\ell=1,\ldots, (M-1)$ satisfying 
\begin{align}
    |\nu(\ell,\epsilon) - \nu(\ell+1,\epsilon)| \leq \frac{2}{p_{1}} \delta p.
\end{align}
We thus obtain that $x_\ell(\nu(\ell,\epsilon) - \nu(\ell+1,\epsilon))\in [0,\frac{2}{p_{1}} \delta p]$ for $\ell=1,\ldots, (M-1)$. Besides, $\nu(M,\epsilon)=0$. The convergence rate to the expectation, by the Hoeffding inequality, is given by 
\begin{align}
    \label{eq:work_concentration}
    \Pr[|\Delta W - \Ex[\Delta W]|\geq \zeta ] \leq 2 e^{-\frac{\zeta^2}{\sum_{\ell=1}^{M-1} \left(\frac{2}{p_{1}} \delta p \right)^2}} \leq 2 e^{-\frac{p_{1}^2 \zeta^2 M}{(2p_{0}-1)^2}}. 
\end{align}
Taking $M\to \infty$, we obtain
\begin{align}
    \lim_{M\to\infty}\Pr[|\Delta W - \Ex[\Delta W]|\geq \zeta ] =0. 
\end{align}
\end{proof}

Theorem~\ref{thm:work_distribution} demonstrates that the extracted work is close to either $w_{0}$ or $w_{1}$,  
with probability close to $\bra{\phi_{0}}\rho\ket{\phi_{0}}$ and $\bra{\phi_{1}}\rho\ket{\phi_{1}}$ respectively. Therefore, measuring the extracted work $\Delta W$ from the state $\rho$ in Algorithm~\ref{alg:sc_work_extraction} is effectively measuring the state $\rho$ in the basis $\{\phi_{i}\}_{i=0,1}$ up to an error probability exponentially vanishing with respect to the number of repetitions $M$ in Algorithm~\ref{alg:sc_work_extraction}. When $\rho$ is a pure state i.e., $\rho=\ketbra{\psi}{\psi}$, the energy measurement of the battery is equivalent to the reward measurement in the quasi-static limit of Algorithm~\ref{alg:linucb_vn_var}, and their correspondence is (without loss of generality assuming $w_{0}\geq w_{1}$)
\begin{align}
    \label{eq:reward_finite}
    r = \begin{cases}
        1,\quad & \Delta W \geq \frac{w_{0}+w_{1}}{2}, \\
        0,\quad & \Delta W \leq \frac{w_{0}+w_{1}}{2}. 
    \end{cases}
\end{align}
The distribution of the reward is 
\begin{align}
    \Pr[R=r] = \begin{cases}
        |\langle \phi_0|\psi\rangle|^2+\epsilon_{\text{error}},\quad & r=1, \\
        1- |\langle\phi_{0}|\psi\rangle|^2 - \epsilon_{\text{error}},\quad & r=0, 
    \end{cases}
\end{align}
where 
\begin{align}
    |\epsilon_{\text{error}}| \leq 2e^{-\frac{p_1^2\zeta^2 M}{(2p_0-1)^2}}
\end{align}
In the limit of $M\to\infty$, the correspondence reduces to 
\begin{align}\label{eq:reward_thermal}
    r = \begin{cases}
        1,\quad & \Delta W = w_{0}, \\
        0,\quad & \Delta W = w_{1}, 
    \end{cases}
\end{align}
and the distribution of the reward reduces to 
\begin{align}
   \Pr[R=r] = \begin{cases}
        |\langle \phi_0|\psi\rangle|^2,\quad & r=1, \\
        1- |\langle\phi_{0}|\psi\rangle|^2 ,\quad & r=0, 
    \end{cases}
\end{align}

The above claims takes $M\to\infty$. In reality, this means $M$ to be sufficiently large. Now we explain how large $M$ should be using Algorithm~\ref{alg:linucb_vn_var}. Without loss of generality, we assume the input state is $\rho = \phi_0$. We are supposed to obtain the correct reward $r=1$. According to Eq.~\eqref{eq:reward_finite}, we obtain a wrong reward $r=0$ only if 
\begin{align}
    \label{eq:error_condition}
    |\Delta W - \Ex[\Delta W]| \geq \frac{1}{2}(w_0-w_1) - |\Ex[\Delta W] - w_0|. 
\end{align}
Note that as Algorithm~\ref{alg:sc_work_extraction} proceeds with increasing $k$, $\epsilon_k= \Theta(\ln(N)/k)$ and, by definition, $p_1 =\Theta(\ln(N)/k)$ as well. Substituting $\epsilon_k$ into Eq.~\eqref{eq:thm3work_values}, we obtain 
\begin{align}
    \Ex[\Delta W] & = w_0 + O\left(\frac{1}{M}\right),\\
    w_0 & = \beta^{-1} \left(\frac{1}{2}+ \ln (1-\Theta\left(\frac{\ln(N)}{k}\right))\right),\\
    w_1 & = \beta^{-1} \left(\frac{1}{2}+ \ln (\Theta\left(\frac{C\ln(N)}{k}\right))\right). 
\end{align}
Substituting above values into Eq.~\eqref{eq:error_condition} and noting $\Theta(\ln(N)/k)$ is small for large $k$, we obtain a wrong reward $r=0$ only if
\begin{align}
    |\Delta W - \Ex[\Delta W]| \geq \frac{\beta^{-1}}{2} \ln\Theta\left(\frac{k}{\ln(N)}\right) - O\left(\frac{1}{M}\right) = \Theta(\ln(k)).
\end{align}
In the concentration bound in Eq.~\ref{eq:work_concentration}, setting $\zeta= \Theta(\ln (k))$, the error probability is upper bounded by 
\begin{align}
\label{eq:wrong_prob}
    \Pr[\textnormal{Wrong Reward}] \leq 2 e^{-\frac{p_1^2\zeta^2 M}{(2p_0-1)^2}}\leq 2 e^{-\Theta\left(\frac{\ln(N)^2\ln(k)^2M}{k^2} \right)},
\end{align}
where we have substituted $p_1 = \Theta(\ln(N)/k)$, $\zeta = \Theta(\ln(k))$ and $2p_0-1=\Theta(1)$. At the same time, Algorithm~\ref{alg:linucb_vn_var} succeeds in all rounds with probability $1-\frac{1}{N}$ if the error probability scales as $\Pr[\textnormal{Wrong Reward}] \leq O(\frac{1}{N^2})$. Recalling that $k\in [N]$, this indicates that we have to choose
\begin{align}
    M = \Theta\left(\frac{N^2}{\ln(N)^3}\right). 
\end{align}
The number of iterations $M=\Theta(N^2/\ln(N)^3)$ in each round leads to the time $T = \Theta(N^3/\ln(N)^3)$ cost by the $N$-round work extraction algorithm.

\subsection{Extracted work for different inputs}
\begin{theorem}\label{th:work_differnt_input}
\label{thm:6} 
Let $\lbrace\phi_i\rbrace_{i=0,1}$ and $\lbrace p_i\rbrace_{i=0,1}$ be the eigenvectors and eigenvalues of $\rho^*$ and $w_i$ be the value of work extracted defined in Theorem.~\ref{thm:work_distribution}. It holds that,
\begin{enumerate}
    \item When applying the protocol to any state $\rho$, the probability of measuring $\Delta W = w_i$ 
        is given by 
        \begin{equation}
        \label{eq:thm3work_probs}
            \Pr(\Delta W=w_i)=\bra{\phi_{i}}\rho\ket{\phi_i}~.
        \end{equation}

    \item When the extraction protocol is operated on a state $\rho$ where $\rho\neq \rho^*$, the expected work extracted is given by
        \begin{equation}
            \Ex[\Delta W]=\beta^{-1}\left[\D(\rho\middle\|\id/2)-\D(\rho\|\rho^*)\right]~,
        \end{equation}
        where the second term can be defined as the the dissipation due to the agent's imperfect knowledge of $\rho$.
\end{enumerate}
\end{theorem}
\begin{proof}
We first observe that the off-diagonal term $\bra{\phi_i}\rho\ket{\phi_j} \ketbra{i}{j}_A\otimes \varphi(x)_B$ of the join state in Eq.~\eqref{eq:intial_joint_state} does not affect Eq.~\eqref{eq:joint_state_1}. Therefore, it is identical for the case where the input is $\rho$ and the case where the input is $\sum_i \bra{\phi_i}\rho\ket{\phi_i} \phi_i$. The latter case can be viewed as a probabilistic mixture of cases where the input is $\phi_{i}$ with probability $\bra{\phi_{i}}\rho\ket{\phi_{i}}$. Therefore, Statement 1 in Theorem~\ref{thm:6} holds, i.e.,
\begin{equation}
    \Pr(\Delta W=w_i)=\bra{\phi_{i}}\rho\ket{\phi_i}~.
\end{equation}

Next, to prove Statement 2, we consider the case where a protocol optimized for $\rho^*=\sum_ip_i\ketbra{\phi_i}{\phi_i}$ is applied onto an arbitrary state $\rho$ with possibly $\rho\neq\rho^*$. Using Eq.~\eqref{eq:thm3work_values} and~\eqref{eq:thm3work_probs}, the extracted work is given by 
\begin{align}
    \Ex[\Delta W] & = \bra{\phi_{0}} \rho \ket{\phi_{0}} \beta^{-1} \left[\D\left(\phi_{0}\|\id/2\right) + \ln p_{0}\right] + \bra{\phi_{1}} \rho \ket{\phi_{1}} \beta^{-1} \left[\D\left(\phi_{1}\|\id/2\right) + \ln p_{1}\right] \\
    & = \beta^{-1}\left[\bra{\phi_{0}} \rho \ket{\phi_{0}}  (-  \tr(\phi_{0}\ln\id/2)+  \ln p_{0})- \bra{\phi_{1}} \rho \ket{\phi_{1}}   (-\tr(\phi_{1}\ln \id/2)+  \ln p_{1} )\right] \\
    & = \beta^{-1}\left[\tr(\mathcal{P}(\rho)\ln\rho^*) -\tr(\mathcal{P}(\rho)\ln \id/2) \right] = \beta^{-1} \left[ \tr(\rho\ln\rho^*) - \tr(\rho\ln\id/2) \right],
\end{align}
where $\mathcal{P}(\rho) = \sum_i \phi_{i} \rho \phi_{i}$ the pinching map.
We can also express the expected work in term of relative entropy:
\begin{align}
\label{eq:b42}
    \Ex[\Delta W]  =  \beta^{-1}\left[ \tr(\rho\ln\rho^*) - \tr(\rho \ln\id/2)\right] = \beta^{-1}\left[\D\left(\rho\|\id/2\right)-\D(\rho\|\rho^*)\right]~.
\end{align}
\end{proof}

Note that in the event $\rho^*=\rho$, i.e., the agent is fully aware of the identity of the quantum state and is able to ensure the work extraction protocol to be entirely quasi-static. Then the extract work is given by  
\begin{align}
    \Ex[\Delta W] & = p_{0} \beta^{-1} \left[\D\left(\phi_{0}\middle\|\id/2\right) + \ln p_{0}\right] + p_{1} \beta^{-1} \left[\D\left(\phi_{1}\|\id/2\right) + \ln p_{1}\right] \\
    & = \beta^{-1}(- p_{0} \tr(\phi_{0}\ln\id/2)+ p_{0}\ln p_{0} -   p_{1}  \tr(\phi_{1}\ln\id/2)+ p_{1}\ln p_{1} ) \\
    & = \beta^{-1}\left[ \tr(\rho \ln \rho) - \tr(\rho\ln\id/2)\right] = \beta^{-1} \D\left(\rho\|\id/2\right). 
\end{align}
We retrieve the full non-equilibrium free energy.
Therefore, the dissipation due to agent's imperfect knowledge of the true state $\rho$ can be quantified as 
\begin{align}
    \dissipation = \max_{\rho^*}\Ex[ \Delta W] - \Ex[ \Delta W] = \beta^{-1} \D(\rho\|\rho^*). 
\end{align}

\subsection{Cumulative dissipation}

In this section, we consider the setting where we have oracle sequential access to an unknown pure qubit state $\psi$, and our goal is to extract the maximal amount of work into a battery system. To achieve this, we can use Algorithm~\ref{alg:linucb_vn_var} with the rewards~\eqref{eq:reward_thermal} to learn an approximate direction of the state, and then run Algorithm~\ref{alg:sc_work_extraction} to extract work based on the approximate input. In general, we can consider mixed-state estimator $\hat{\rho}_k$. Assuming sequential access to the unknown state over $N$ rounds, and using the expected extracted work from Theorem~\ref{th:work_differnt_input}, we define the dissipation at round $k \in [N]$ with respect to the optimal protocol as
\begin{align}
\dissipation^{k} := \beta^{-1} \D(\psi \| \hat{\rho}_k ),
\end{align}
and the cumulative dissipation over all $N$ rounds as
\begin{align}\label{eq:apendix_dissipation_sc}
\dissipation(N) := \beta^{-1} \sum_{k=1}^N \dissipation^{k} = \beta^{-1} \sum_{k=1}^N \D(\psi \| \hat{\rho}_k).
\end{align}

\begin{algorithm}[H]
	\caption{\textsf{Thermal work extraction}}
	\label{alg:cum_dissipation}
        Require: sequence $\lbrace \epsilon_k \rbrace_{k=1}^\infty$
          
	\For {$k=1,2,\ldots$}{
            Receive unknown $ | \psi \rangle$ and couple to battery state $\varphi(x-\mu_{k-1})$
           
            Compute direction $|\psi_k \rangle$ with Algorithm~\ref{alg:linucb_vn_var} using $\lbrace \psi_s , r_s \rbrace_{s=1}^{k-1}$

            Set $\hat{\rho}_k = \Delta_{2\epsilon_k} ( \psi_k )$
           
            Extract work using Algorithm~\ref{alg:sc_work_extraction} with input $\hat{\rho}_k$ and get extracted work $\Delta W_k$ and energy $\mu_k$
           
            Set reward $r_k = \lbrace 0 , 1 \rbrace$ according to~\eqref{eq:reward_thermal}
        }
\end{algorithm}

To minimize the cumulative dissipation, we use Algorithm~\ref{alg:cum_dissipation}, which takes as input a sequence of accuracies $\lbrace \epsilon_k \rbrace_{k=1}^\infty$. At each round, the estimator uses $\psi_k$, the direction output by Algorithm~\ref{alg:linucb_vn_var}, and sets $\hat{\rho}_k = \Delta_{2\epsilon_k}(\psi_k)$, where $\Delta_{2\epsilon_k}$ is the completely depolarizing channel. If $\epsilon_k$ is a good approximation of the infidelity between the true state $\psi$ and the estimate $\psi_k$, then the dissipation $\dissipation^k$ is controlled by $\epsilon_k$. This is formalized in the following theorem.
\begin{theorem}[Theorem 2.34 in~\cite{Flammia2024quantumchisquared}]\label{th:relative_entropy_fideity} 
Let $\rho$ and $\hat{\rho}$ be $d$-dimensional quantum states achieving infidelity $1 - F(\rho,\hat{\rho}) \leq \epsilon \leq \frac{1}{2}$. Then we have
\begin{align}
    \D ( \rho \| \Delta_{2\epsilon} ( \hat{\rho}) ) \leq 16\epsilon \left( 2 + \ln \left( \frac{d}{2\epsilon} \right) \right) . 
\end{align}
\end{theorem}

Given the above bound we can use the fidelity guarantee of Algorithm~\ref{alg:linucb_vn_var} in Theorem~\ref{th:fidelity_gurantee_bandit} to prove a bound on the cumulative dissipation.

\begin{theorem} Given a finite time horizon $N\in\mathbb{N}$ and $\delta\in (0,1 )$ there exists an explicit sequence of accuracies $\lbrace \epsilon_k \rbrace_{k=1}^\infty $ such that Algorithm~\ref{alg:cum_dissipation} achieves
\begin{align} 
\dissipation (N) = O\left(\beta^{-1}\ln^2(N) \ln \left( \frac{N}{\delta} \right)  \right) .
\end{align} 
\end{theorem}

\begin{proof}
    We can choose
    \begin{align}
        \epsilon_k = \min \left\lbrace C \frac{\ln\left(\frac{N}{\delta} \right)}{k} ,\frac{1}{2} \right\rbrace ,
    \end{align}
    where $C$ is the constant in Theorem~\ref{th:fidelity_gurantee_bandit} for the fidelity bound of the direction $\psi_k$. Then we can use Theorem~\ref{th:relative_entropy_fideity} combined with the fidelity guarantee of Theorem~\ref{th:fidelity_gurantee_bandit} to get that with  probability at least $1-\delta$ we have
    \begin{align}
        \dissipation (N) \leq \beta^{-1} \sum_{k=1}^{N} 16 \epsilon_k (2-\ln \epsilon_k ).
    \end{align}
The result follows by noting that, for a sufficiently large constant $k^*$, we have $\epsilon_k = C \ln\left(\frac{N}{\delta} \right)/k$ for all $k \geq k^*$. The dissipation incurred during the first $k^*$ rounds contributes a constant term. For the remaining rounds $k \geq k^*$, using $-\ln\epsilon_k\leq \ln N$ for $k\leq N$ and summing the corresponding dissipation terms yields the claimed polylogarithmic scaling.
\end{proof}

\section{Jaynes-Cummings work extraction protocol}
\label{apd:jc_protocol}

\begin{algorithm}[H]
        \caption{Jaynes-Cummings work extraction} 
	\label{alg:jc_work_extraction}
        Require: A sequence of unknown states $\ket{\psi}$
       
        \For {$k=1,2,\ldots$}{
            Receive the unknown state $\ket{\psi}$
            
            Compute direction $\ket{\psi_k}$ using Algorithm~\ref{alg:linucb_vn_var}
            
            Expose it to a field that induces Hamiltonian $H_A=\omega \ketbra{\psi_k}{\psi_k}$. 
            
            Turn on the interaction between the system and the battery whose interaction Hamiltonian is $H_I = \frac{\Omega}{2}(a \otimes \ketbra{\psi_k}{\psi_k^\perp} + a^\dagger \otimes \ketbra{\psi_k^\perp}{\psi_k})$ for a time $t_k=\pi\Omega^{-1}(n_k+1)^{-\frac{1}{2}}$. 
            
            Measure the battery energy to obtain $n_{k+1}$

            Set reward $r_k=\lbrace0,1\rbrace$ according to Eq.~\eqref{eq:reward_jc}
	}
\end{algorithm}

We consider the Jaynes-Cummings work extraction protocol in Algorithm~\ref{alg:jc_work_extraction} which extracts the energy from a field into a battery.

The systems involved in this protocol includes a system in an unknown state $\ket{\psi}$ with a tunable Hamiltonian in the form of $H_A = \nu \ketbra{\phi}{\phi}$ where $\nu$ and $\phi$ can be controlled by the strength and the direction of a field, respectively, as well as a battery system with the Hamiltonian of the battery given by $H_B = \omega a^\dagger a$ where $a=\sum_{n=1}^{\infty} \sqrt{n} \ketbra{n-1}{n}$ and $a^\dagger =\sum_{n=0}^{\infty} \sqrt{n+1} \ketbra{n+1}{n}$. We are also allowed to turn on an interaction $H_I=\frac{\Omega}{2}(a \otimes \ketbra{\psi_k}{\psi_k^\perp} + a^\dagger \otimes \ketbra{\psi_k^\perp}{\psi_k})$ between the system and the battery. 

The protocol works in round $k=1,2,\ldots,N$. 

In each round $k$, we possess a battery state $\ket{n_k}$, receive an unknown state $\ket{\psi}$ and compute a direction $\ket{\psi_k}$ from previous records of $\lbrace n_s\rbrace_{s=1}^{k}$. 

We then expose the system to a field that induces Hamiltonian $H_A=\omega \ketbra{\psi_k}{\psi_k}$, which transfer energy from the field to the system. 

We turn on the interaction $H_I$ between the system and the battery for time $t_k = \pi\Omega^{-1}(n_k+1)^{-\frac{1}{2}}$. The time evolution of the system and the battery is under the total Hamiltonian
\begin{align}
    H= \omega (\ketbra{\psi_k}{\psi_k} + a^\dagger a) + \frac{\Omega}{2}(a \otimes \ketbra{\psi_k}{\psi_k^\perp} + a^\dagger \otimes \ketbra{\psi_k^\perp}{\psi_k}). 
\end{align}
This is exactly given by the famous Jaynes-Cummings model~\cite{Jaynes_1963}, which is nowadays a textbook model~\cite{Scully_1997}. One may verify that the eigenstates of the total Hamiltonian are $\ket{0}\ket{\psi_k^\perp}$ as well as 
\begin{align}
    \ket{n,+}  = \frac{1}{\sqrt{2}}(\ket{n-1}\ket{\psi_k} + \ket{n}\ket{\psi_k^\perp}), \quad 
    \ket{n,-}  = \frac{1}{\sqrt{2}}( \ket{n-1}\ket{\psi_k} - \ket{n}\ket{\psi_k^\perp}),
\end{align}
for $n=1,2,\ldots$ and the eigenvalues are respectively $E_0=0$ and 
\begin{align}
    E_{n+}  =  n\omega + \frac{\Omega}{2} \sqrt{n}, \quad 
    E_{n-}  =  n\omega - \frac{\Omega}{2} \sqrt{n}.
\end{align}
for $n=1,2,\ldots$. The state $\ket{n_k}\ket{\psi}$ is decomposed into $4$ eigenstates $\ket{n_k,\pm}$ and $\ket{n_k+1,\pm}$, i.e.
\begin{align}
    \ket{n_k}\ket{\psi} = \frac{1}{\sqrt{2}}\langle \psi_k^\perp |\psi \rangle(\ket{n_k,+}-\ket{n_k,-}) + \frac{1}{\sqrt{2}} \langle \psi_k |\psi \rangle (\ket{n_k+1,+}+\ket{n_k+1,-}). 
\end{align}
These eigenstates gains phase factors during the time evolution. After time $t_k=\pi\Omega^{-1}(n_k+1)^{-\frac{1}{2}}$, the state evolves to, up to an irrelevant global phase, 
\begin{align}
    e^{-i H t_k}\ket{n_k}\ket{\psi} & = \frac{1}{\sqrt{2}}\langle \psi_k^\perp |\psi \rangle(e^{i\theta_k}\ket{n_k,+}-e^{-i\theta_k}\ket{n_k,-}) \\
    &\quad+ \frac{1}{\sqrt{2}} e^{i\frac{\omega\pi}{\Omega\sqrt{n_k+1}} } \langle \psi_k |\psi \rangle(i\ket{n_k+1,+}-i\ket{n_k+1,-}) \\
    & = \langle \psi_k^\perp |\psi \rangle( i \sin\theta_k \ket{n_k-1}\ket{\psi_k}+ \cos\theta_k \ket{n_k}\ket{\psi_k^\perp}) \\
    &\quad+  \frac{1}{\sqrt{2}} i e^{i\frac{\omega\pi}{\Omega\sqrt{n_k+1}} } \langle \psi_k |\psi \rangle \ket{n_k+1}\ket{\psi_k^\perp}, 
\end{align}
where $\theta_k = \frac{\pi}{2}\sqrt{\frac{n_k}{n_k+1}}$. 

We finally measure the battery energy. The measurement outcome is $n_{k+1}$. The probability distribution of $n_{k+1}$ is given by 
\begin{align}
    \Pr[n_{k+1}] = \begin{cases}
        |\langle \psi_k |\psi \rangle |^2, \quad & n_{k+1} = n_k+1,\\
        |\langle \psi_k^\perp |\psi \rangle |^2 \cos^2\theta_k, \quad & n_{k+1} = n_k,\\
        |\langle \psi_k^\perp |\psi \rangle |^2 \sin^2\theta_k, \quad& n_{k+1} = n_k-1. 
    \end{cases}
\end{align}
The extracted work is defined as $\Delta W_k =\omega( n_{k+1}-n_k) $. The expected extracted work is given by 
\begin{align}
    \Ex[\Delta W_k] = \omega( |\langle \psi_k |\psi \rangle |^2 - |\langle \psi_k^\perp |\psi \rangle |^2 \sin^2\theta_k )= \omega( |\langle \psi_k |\psi \rangle |^2 (1+\sin^2\theta_k) - \sin^2\theta_k), 
\end{align}
where we have used $|\langle \psi_k^\perp |\psi \rangle |^2 = 1 - |\langle \psi_k |\psi \rangle |^2 $. It is obvious that $\Ex[\Delta W_k]$ increases as $|\langle \psi_k |\psi \rangle |^2$ increases, therefore, $\Ex[\Delta W_k]$ is maximized at $|\langle \psi_k |\psi\rangle|=1$, 
\begin{align}
    \max_{\ket{\psi_k}} \Ex[\Delta W_k] = \omega. 
\end{align}
The dissipation in this round is thus given by the difference between the maximal and the actual expected extracted work
\begin{align}
    \dissipation^{\text{jc},k} = \max_{\ket{\psi_k}} \Ex[\Delta W_k] -  \Ex[\Delta W_k] = \omega (1+\sin^2\theta_k^2)(1- |\langle \psi_k|\psi\rangle |^2 ) \leq 2 \omega (1-|\langle \psi_k|\psi\rangle |^2). 
\end{align}
The correspondence between the reward measurement and the battery energy measurement is given by
\begin{align}
    \label{eq:reward_jc}
    r_k = \begin{cases}
        1, \quad & n_{k+1} = n_k+1, \\
        0, \quad & \text{otherwise. }
    \end{cases}
\end{align}
The probability distribution of the reward is 
\begin{align}
    \Pr[R_k=r_k] = \begin{cases}
        |\langle \psi_k |\psi \rangle |^2, \quad & r_k=1,\\
        |\langle \psi_k^\perp |\psi \rangle |^2 , \quad & r_k=0.
    \end{cases}
\end{align}
The regret in this round is thus 
\begin{align}
    \regret^k = 1-|\langle \psi_k |\psi \rangle |^2. 
\end{align}
The dissipation and the regret in this round is thus related by 
\begin{align}
    \dissipation^{\text{jc},k} \leq 2\omega \regret^k .
\end{align}
The cumulative dissipation over $N$ rounds is thus 
\begin{align}
    \dissipation^{\text{jc}}(N)  = \sum_{k=1}^{N}  \dissipation^{\text{jc},k} \leq 2 \omega\sum_{k=1}^{N} (1-|\langle\psi_k|\psi\rangle|^2). 
\end{align}
Which is related to the cumulative dissipation in extracting work from knowledge via 
\begin{align}
    \dissipation^{\text{jc}}(N)\leq 2\omega \regret (N), 
\end{align}
\begin{theorem} 
There exists an explicit protocol for extracting work from knowledge that adaptively updates the estimate $\ket{\psi_k}$ based on the rewards $\lbrace r_s \rbrace_{s=1}^{k-1}$, achieving, with probability at least $1-\delta$, 
\begin{align} 
\dissipation^{\text{jc}} (N) = O\left( \omega \ln (N) \ln \left( \frac{N}{\delta} \right)  \right). 
\end{align} 
\end{theorem}

\section{Cost of measurement and erasure}
\label{sec:landauer}
The measurement as well as erasure of memory of the agent does not come for free. The cost of measurement along with the cost of memory erasure can be lower bounded by a quantity known as \emph{QC-mutual
information}, $I_{QC}$, which is a measure of how much information a measurement carries about the quantum system \cite{Sagawa_2009}. $I_{QC}$ is then upper bounded by the entropy of the classical memory register itself. Hence the total cost of measurement and erasure can be loosely lower bounded by just the entropy of the memory register itself.
By Landauer's principle~\cite{Sagawa_2009,Goold_2015,van2022finite}, the heat dissipation required to erase the register is lower bounded by the entropy change $\Delta S$ of the memory register via
\begin{align}
    \beta Q \geq \Delta S.
\end{align}
Furthermore, it is widely accepted that the lower bound can be achieved in the quasi-static limit when the state to be erased is known~\cite{riechers2021initial,Jun_2014,Miller_2020}. In our work extraction model, we repeatedly perform measurements in the energy eigenbasis, which requires memory erasure in order to store new measurement outcomes. With the assumption on the probability distribution of work values scaling linearly with fidelity, when $\ket{\psi}$ is known, the measurement outcome is deterministic and there is no energy dissipation. However, when $\ket{\psi}$ is unknown, the measurement outcome is stochastic and there is energy dissipation. This may also be taken into account in the dissipation, that is, 
\begin{align}
    \label{eq:regret_modified}
    \dissipation'(N) = \dissipation(N) + \beta^{-1} \sum_{k=1}^N  \Delta S_k, 
\end{align}
where $\Delta S_k$ is the entropy change of the memory register in round $k$. Let $\epsilon_k = 1 - |\langle \psi_k | \psi\rangle |^2$ be the infidelity in round $k$. $\Delta S_k$ is closely related to $\epsilon_k$ in both models we consider. In the semi-classical battery model, the entropy change is upper bounded by 
\begin{align}
    \Delta S_k = -(1-\epsilon_k) \ln (1-\epsilon_k) - \epsilon_k \ln \epsilon_k \leq  \epsilon_k - \epsilon_k \ln \epsilon_k, 
\end{align}
and the dissipation is upper bounded by 
\begin{align}
    \dissipation^{\text{sc},*}(N) = \dissipation^{\text{sc}}(N) + \beta^{-1} \sum_{k=1}^N (\epsilon_k- \epsilon_k\ln\epsilon_k). 
\end{align}
In the Jaynes-Cummings battery model, the entropy change in round $k$ is upper bounded by 
\begin{align}
    \Delta S_k &= -(1-\epsilon_k) \ln (1-\epsilon_k) - \epsilon_k \ln \epsilon_k - \epsilon_k (\cos^2\theta_k \ln \cos^2\theta_k + \sin^2\theta_k \ln \sin^2\theta_k )\\
    &\quad\quad\quad\quad\quad\quad\quad\quad\quad\quad\quad\quad\quad\quad\quad\quad\quad\quad\quad\quad\quad\quad\quad\leq 2\epsilon_k -\epsilon_k\ln \epsilon_k, 
\end{align}
and the dissipation is upper bounded by 
\begin{align}
    \dissipation^{\text{jc},*}(N) = \dissipation^{\text{jc}}(N) + \beta^{-1} \sum_{k=1}^N (2\epsilon_k - \epsilon_k\ln\epsilon_k). 
\end{align}

The quantum state tomography algorithm in~\cite{lumbreras24pure} ensures a polylogarithmic cumulative infidelity with a high probability. Now we focus on the case where the cumulative infidelity is upper bounded by $\ln\frac{N}{\delta}\ln N$ with high probability $1-\delta$, that is, for some constant $C$, 
\begin{align}
    \sum_{k=1}^N \epsilon_k \leq C \ln \frac{N}{\delta}\ln N.
\end{align}
By taking the derivative $(-\epsilon_k\ln\epsilon_k)' = -  \ln\epsilon_k - 1$, it is obvious that when $0<\epsilon_k \leq e^{-1}$, $-\epsilon_k\ln  \epsilon_k$ increases with $\epsilon_k$ while $e^{-1}\leq \epsilon_k\leq 1$, $-\epsilon_k\ln\epsilon_k$ decreases with $\epsilon_k$. Now we consider two cases: 

\textbf{Case 1: $C (\ln (N/\delta) \ln N)/N \leq e^{-1}$.} This happens when $N$ is large enough. As a result, $-\epsilon_k\ln\epsilon_k$ monotonically increase with respect to $\epsilon_k$ when $\epsilon_k\leq C  (\ln (N/\delta) \ln N)/N$. On the one hand, for $\epsilon_k\leq  C(\ln (N/\delta) \ln N)/N $, 
\begin{align}
    -\epsilon_k\ln\epsilon_k\leq \frac{C\ln\frac{N}{\delta} \ln N}{N} \ln \frac{N}{C\ln\frac{N}{\delta} \ln N},
\end{align}
On the other hand, for $\epsilon_k > C (\ln (N/\delta) \ln N)/N$, 
\begin{align}
    -\epsilon_k\ln\epsilon_k \leq  \epsilon_k \ln\frac{N}{C\ln \frac{N}{\delta} \ln N}, 
\end{align}
Then, 
\begin{align}
    -\sum_{k=1}^{N} \epsilon_k \ln\epsilon_k & = - \sum_{\epsilon_k:\epsilon_k\leq \frac{C \ln \frac{N}{\delta} \ln N}{N}} \epsilon_k \ln\epsilon_k - \sum_{\epsilon_k:\epsilon_k >\frac{C \ln \frac{N}{\delta} \ln N}{N}}\epsilon_k \ln\epsilon_k \\
    & \leq \sum_{\epsilon_k:\epsilon_k\leq \frac{C \ln \frac{N}{\delta} \ln N}{N}} \frac{C \ln \frac{N}{\delta} \ln N}{N} \ln \frac{N}{C \ln \frac{N}{\delta} \ln N}\\
    &\qquad\qquad\qquad+ \sum_{\epsilon_k:\epsilon_k>\frac{C \ln \frac{N}{\delta} \ln N}{N}} \epsilon_k  \ln\frac{N}{C \ln \frac{N}{\delta} \ln N}  \\
    & \leq  C \ln \frac{N}{\delta} \ln N \ln \frac{N}{C \ln \frac{N}{\delta} \ln N} + C \ln \frac{N}{\delta} \ln N \ln \frac{N}{C \ln \frac{N}{\delta} \ln N} \\
    &\quad \quad\quad\quad\quad\quad\quad\quad\quad\leq 2 C \ln \frac{N}{\delta} \ln N \ln \frac{N}{C \ln \frac{N}{\delta} \ln N}. 
\end{align}

\textbf{Case 2: $C (\ln (N/\delta) \ln N)/N > e^{-1}$.} This happens when $N$ is not very large. Therefore, $N/e\leq  C \ln (N/\delta) \ln N$ and certainly $N \geq 2$. Using the fact that $-\epsilon_k\ln\epsilon_k\leq e^{-1}$, we obtain
\begin{align}
    \sum_{t=1}^{N} \epsilon_k \ln \frac{1}{\epsilon_k} \leq \frac{N}{e} \leq C \ln \frac{N}{\delta} \ln N . 
\end{align}
Combining both cases, we conclude that 
\begin{align}
    - \sum_{t=1}^{N} \epsilon_k \ln\epsilon_k \leq O((\ln N)^3). 
\end{align}
Substituting into Eq.~\eqref{eq:regret_modified}, we obtain the dissipation for for the semi-classical model,  
\begin{align}
    \dissipation^{\text{sc},*}(N)  \leq  \dissipation^{\text{sc}}(N) + \beta^{-1} O((\ln N)^3) = O((\ln N)^3). 
\end{align}
and the Jaynes-Cummings battery model, 
\begin{align}
    \dissipation^{\text{jc},*}(N)  \leq  \dissipation^{\text{jc}}(N) + 2\beta^{-1} O((\ln N)^3) = O((\ln N)^3), 
\end{align}
Therefore, the regret still scales as $O(\polylog(N))$ even if we take the energy dissipation due to Landauer's principle into account. 

As mentioned above, in order to achieve Landauer's bound, the erasure process must necessarily be done quasi-statically. This is not something that a sequential adaptive agent can accomplish since it needs to measure and decide on its action in real time. To circumvent this, an array of empty memory registers can be first prepared. Suppose the agent will be extracting work for $N$ steps, we first prepare a memory register $M_0$, consisting of $N$ empty registers.
\begin{equation}
    M_0 = \{\underbrace{0,0,\ldots,0}_N\}
\end{equation}
At every time step, rather than erasing the old memory that has the previous outcome remembered, the agent simply records the outcome of the measurement on the battery into a new empty register. In order words, after time step $t$, the memory register takes the form
\begin{equation}
    M_t = \{r_1,r_2,\ldots,r_t,0,\ldots,0\} \quad \text{for}\quad t\in \lbrace{1,2,\ldots N \rbrace}~.
\end{equation}
At the end of all the extractions, the memory $M_N$ can then be quasi-statically reset. As previously calculated, as $t\to N$, the distribution of $r_t$ becomes more peaked and the total entropy of the memory registers scale with $O(\ln N)^3$, likewise for the cost of measurement. The resultant dissipation therefore still scales with $O(\polylog(N))$.